\pdfoutput=1
\documentclass[acmsmall,screen,nonacm,svgnames,table,appendix]{acmart}
\startPage{1}

\setcopyright{none}

\usepackage[capitalize]{cleveref}
\crefname{line}{line}{lines}
\Crefname{line}{Line}{Lines}
\crefrangelabelformat{line}{#3#1#4--#5#2#6}

\usepackage[inline]{enumitem}
\usepackage{floatpag}
\usepackage{wrapfig}

\usepackage{style/mathpartir}
\usepackage{style/notation}
\usepackage{style/TaDA}
\usepackage{style/tikzstyles}
\let\ifdraft\iffalse
\usepackage{style/appendix}

\theoremstyle{acmdefinition}
\newtheorem{remark}{Remark}
\newtheorem{convention}{Convention}

\floatpagestyle{standardpagestyle}
\rotfloatpagestyle{empty}

\newcommand{\adjustfigure}[1][\small]{\centering#1\belowdisplayskip=0pt\belowdisplayshortskip=0pt\abovedisplayskip=0pt\abovedisplayshortskip=0pt}
\let\mathfigsize\footnotesize
\newenvironment{mathfig}[1][tb]{\begin{figure}[#1]\adjustfigure[\mathfigsize]}{\end{figure}}

\newcommand{\paperstats}{\small \textcolor{red}{TOTAL~PAGES:~\pageref*{paper-last-page}}}

\ifappendix
\usepackage[atend]{bookmark}
\BookmarkAtEnd{\bookmarksetup{startatroot}\bookmark[level=0,dest=None]{RULES INDEX}}
\fi

\ifdraft
\fi

\begin{document}

\title[TaDA Live]{TaDA Live: Compositional Reasoning for Termination of Fine-grained Concurrent Programs}

\author{Emanuele D'Osualdo}
\orcid{0000-0002-9179-5827}
\affiliation{
  \institution{Imperial College London}
}
\affiliation{
  \institution{MPI-SWS Saarbrücken}
}
\email{dosualdo@mpi-sws.org}

\author{Julian Sutherland}
\affiliation{
\institution{Imperial College London}
}
\email{julian.sutherland10@ic.ac.uk}

\author{Azadeh Farzan}
\affiliation{
\institution{University of Toronto}
}
\email{azadeh@cs.toronto.edu}

\author{Philippa Gardner}
\affiliation{
\institution{Imperial College London}
}
\email{pg@doc.ic.ac.uk}

\begin{abstract}
  We present TaDA Live, a concurrent separation logic for reasoning compositionally about the termination of blocking fine-grained concurrent programs.
  The crucial challenge
is
  how to deal with \emph{abstract atomic blocking}:
  that is, abstract atomic operations that have blocking behaviour
  arising from busy-waiting patterns as found in, for
  example, fine-grained spin locks.
  Our fundamental innovation is
  with the design of abstract specifications that capture this  blocking
  behaviour 
  as liveness assumptions on the environment.
  We design a logic that can reason about
  the termination of clients which use such 
operations without breaking their abstraction boundaries,
  and
  the correctness of the  implementations of the operations with
  respect to their  abstract specifications.
  We introduce a novel semantic model using 
  layered subjective obligations  to express liveness invariants,
and a proof system that is sound with respect
  to the model. 
  The subtlety of our specifications and reasoning is illustrated using
  several case studies.
\end{abstract}




\maketitle
\renewcommand{\shortauthors}{E.\,D'Osualdo,
  J.\,Sutherland,
  A.\,Farzan,
  P.\,Gardner}

\expandpqsets

\section{Introduction}
\label{sec:intro}

Compositional reasoning for fine-grained concurrent programs interacting with
shared memory is a fundamental, open research problem.  
We are beginning to obtain a good understanding of how to reason about
\emph{safety properties} of concurrent programs:
i.e.~if the program terminates and the input satisfies the precondition,
then the program does not fault and the result satisfies the postcondition.
O'Hearn and Brookes~\cite{OHearn04,Brookes04} introduced
concurrent separation logic for reasoning compositionally
about course-grained concurrent programs.
Since then, there has been a flowering of work
on modern concurrent separation logics for reasoning compositionally
about safety properties of fine-grained concurrent programs:
e.g.~CAP~\cite{cap},
  TaDA~\cite{tada},
  Iris~\cite{iris} and
  FCSL~\cite{NanevskiLSD14}.
With these modern logics,
it is possible to provide abstract specifications
that match the intuitive software interface understood by the developer,
and to verify both implementations and client programs.

We have comparatively little understanding of  how to reason
compositionally 
about \emph{progress (liveness) properties} for fine-grained
concurrent algorithms: i.e. something good eventually happens.
Examples of progress properties include termination, livelock-freedom,
or that every user request is eventually served.
The intricacies of the design of concurrent programs
often arise precisely from the need to make the program correct
with respect to progress properties.
The goal of this paper is to design a program logic
to reason compositionally about the safety and termination 
of fine-grained concurrent programs:
i.e.~to be able to prove that
if the input satisfies the precondition,
then the program terminates without faulting 
and the result satisfies the postcondition.
As with safety,
the aim is to provide abstract specifications, and
to verify implementations and clients.

A truly compositional approach would achieve  \emph{proof scalability}
through the reduction of large complex proofs
into a composition of smaller, more tractable proofs, 
and \emph{proof reuse}  through the ability to define 
abstract interfaces between independent sub-proofs.
Proof scalability for concurrent systems is achieved through 
\emph{thread-local} reasoning:
i.e.~the proof of the parallel composition of threads
should be the composition of smaller, separate proofs of each thread.
Proof reuse is achieved when the right \emph{abstract interface} for a
module is identified, so that the proof of correctness of the
implementation of the module and the proof of its clients is
decoupled: a proof of a client can be reused when swapping the
implementation of the module for one satisfying the same
specification; a proof of an implementation can be reused when the
specification is general enough to support arbitrary correct clients.

For safety, thread-local reasoning can be obtained through 
rely/guarantee proofs: a protocol on shared state is specified in
terms of the set of \emph{allowed} updates, and each thread is
verified to respect the protocol under the assumption that the
environment respects the protocol.
There have been successful attempts at using
rely/guarantee reasoning to prove progress properties,
such as termination,
of \emph{non-blocking} concurrent programs~\cite{HoffmannMS13,GotsmanCPV09,GotsmanY11,CookPR07,LiangFS14,total-tada},
which are the programs where
the progress of a thread does not depend on the progress of other threads.
For example, the Total~\tada\ concurrent separation logic~\cite{total-tada}
was introduced to provide compositional reasoning
about the safety and termination of non-blocking programs.
It provided thread-local reasoning and
abstract specification of module interfaces,
without the need to extend the rely/guarantee reasoning.

Standard rely/guarantee reasoning is not enough
to prove progress properties for \emph{blocking} programs.
In a blocking program, termination of a thread may depend on
other threads performing some updates to the shared state.
For example, if a thread~$t$ is requesting a lock that has been
acquired by another thread, then 
the lack of progress of the thread currently owning the lock
would hinder the progress of~$t$.
Thread~$t$ is blocked, waiting for the lock owner to release the lock.
In such situation, a safety abstraction of the environment is insufficient
to support a termination argument for~$t$:
knowing that the release of the lock is \emph{always allowed to happen}
does not imply that it is  \emph{eventually happening}.

There has been some work~\cite{Kobayashi06,BostromM15,Jacobs18toplas}
on proving progress properties for programs
where blocking is caused solely by \emph{blocking primitives}
such as built-in locks or channels.
However, it is very common, especially for fine-grained programs,
to use ad hoc busy-waiting patterns.
For example, consider a thread running
\code{while(v!=1)\{v:=[x]\}}. 
The termination of this thread is entirely dependent on the environment
eventually storing~1 in~\p{x}.
This form of blocking is completely different from a call to a
blocking primitive that cannot take a step in the current state.
It instead corresponds to code executing steps without making real progress.
We call this pattern of behaviour \emph{abstract blocking}.

We have identified two ways to reason about progress in the presence of
abstract blocking in the literature:
the history-based approach  and the refinement-based approach. 
The history-based approach~\cite{OwickiL82,Shao17,Shao18}
is very general but results in complex and indirect specifications
with complicated reasoning involving  explicit trace manipulations.
We discuss this approach further in~\cref{sec:relwork}.
In the refinement-based approach,
the LiLi logic~\cite{LiangF16,LiangF18}
is the work most closely related to our goals.
LiLi extends rely/guarantee with liveness information,
to prove a \emph{progress-preserving} contextual refinement
between the implementation of a module's operations and simpler code representing their specifications.
LiLi's extension of rely/guarantee requires, however, heavy use of
global auxiliary shared state manipulated through ghost code,
which makes the proofs less local.
Moreover,
the specification code associated with abstractly atomic operations
that are blocking 
is not atomic and exposes implementation details,
which hinders scalability and reuse.
We give a detailed comparison with our work and LiLi in
\cref{sec:overview,sec:compare-lili-proof,sec:relwork}.

The refinement approach does not prove termination directly, 
but instead relates termination of implementation code
with termination of specification code.
By contrast, our goal is to develop a program logic
with which we are able to verify
specifications that describe termination directly,
without the manipulation of histories,
with proofs that keep auxiliary state as local as possible
without requiring the addition of ghost code,
and with specifications that allow the abstraction of implementation details
while representing precisely the abstract termination guarantees.

\paragraph{Contributions.}
Our starting observation is that
just as safety rely/guarantee arguments are centred around \emph{invariants},
i.e.~facts of the form \emph{always~P}, so liveness rely/guarantee
arguments  for proving progress in the presence of blocking should be  centered around 
\emph{liveness invariants},
i.e.~facts of the form \emph{always eventually~P}.
\tadalive's design is based on the idea that this is not a fluke:
the dependence on liveness invariants  might be considered 
a \emph{definition} of abstract blocking.
To capture  this observation within a program logic,
we introduce a number of key innovations:
\begin{itemize}
  \item \textbf{subjective obligations},
    a new form of logical ghost state to  express liveness invariants
    in a thread-local way without the need for ghost code;
\item \textbf{obligation layers},
    to express  dependencies between liveness invariants
    and avoid unsound circular reasoning;
  \item \textbf{abstract specifications for atomic blocking operations},
    to  express termination guarantees conditionally
    on an \emph{environment liveness assumption} of the form
    ``always eventually~$P$''.
\end{itemize}

We obtain \tadalive, a concurrent separation logic
which  uses liveness invariants  to provide compositional reasoning 
for establishing safety and termination 
for blocking programs.
The logic makes extensive use of 
 abstract specifications for atomic blocking operations  to achieve
proof scalability and reuse.
This paper presents the following contributions:
\begin{itemize}
  \item the \tadalive\ logic and its specification format;
  \item a novel semantic model and soundness proof for the logic:
    the new model is a substantial re-definition of the \tada\ model
to  allow
    for the non-trivial extensions needed
    to incorporate the liveness content of the \tadalive\ specifications;
  \item \tadalive\ proofs for several paradigmatic case studies:
    two fine-grained implementations of locks showcase abstraction in
    the specifications and the obligation mechanism;
    a program mixing locks and busy-waiting
    illustrates common proof patterns for clients;
    two counter modules illustrate \tadalive's
    ability to hide internal blocking and proof reuse; and 
    a set module using a lock-coupling pattern illustrates the generality
    of the layer system.
\end{itemize}

\paragraph{Outline.}
\cref{sec:overview} provides  an example-driven overview of the main innovations
of \tadalive.
\Cref{sec:model} introduces 
the assertion language and the semantics of the \tadalive\
specifications. 
\Cref{sec:rules} presents the crucial proof rules of \tadalive,
 with 
a running example to illustrate their use. \cref{sec:evaluation}
presents \tadalive\ proofs of several key case studies  and a
discussion on the 
limitations of the \tadalive\ reasoning. \cref{sec:relwork} contains 
related work and \cref{sec:concl}
ends with conclusions and future work. 
 \section{An Overview of \tadalive}
\label{sec:overview}

We introduce the main ideas of \tadalive\  in this section, leaving 
the complex technical details  for the following sections. 
Consider a simple example program with non-primitive blocking behaviour:

\[
\color{gray}\cmd_1\left\{\color{black}
\begin{parall}\begin{threadcode}[gobble=2]
  var v = 0 in
  while(v!=1){
    v:=[x]
  }
  \end{threadcode}
  \PARALLEL
  \begin{threadcode}[gobble=2]
  [x]:=1
  \end{threadcode}
\end{parall}
\color{gray}
\raisebox{-.9pt}{$\big\}\;\cmd_2$}\right.
\]

We use a first-order, fine-grained concurrent while language
for manipulating shared state.
The shared state comprises heap cells
which have addresses and store values (addresses, integers, booleans).
The \code{[x]}~notation denotes the value
stored at the heap cell with address \code{x}.
The thread on the left ($\cmd_1$) is busy-waiting
on the value stored at the shared heap cell at \p{x}.
Under fair scheduling, the program is guaranteed to terminate:
eventually, the right-hand thread ($\cmd_2$) will be scheduled, 
and will set the heap cell to 1;
after that, eventually the left-hand thread
will read the value 1 into the local variable \pvar{v}
and the while loop will terminate.
Since we are aiming at a thread-local proof method,
we should be able to break the proof of termination of the program
into two separate proofs for the two threads.

We first explore how to provide a thread-local proof of \emph{safety}
for this example program 
using the \tada\ logic~\cite{tada}. We then  extend the reasoning with the
ingredients needed 
to prove termination.
\tada\ is a concurrent separation logic so it uses the standard
separation logic assertions.
Let us assume the precondition
$ P = \exists v\st\p{x} \mapsto v $
and, for simplicity,  aim at the postcondition $\True$.
\tada\ uses the  standard parallel rule for concurrent separation logics, where the precondition is separated
into two preconditions $P = P_1 * P_2$, one for each thread.
Since both threads dereference \p{x},
we need a means to share the heap cell in the assertions,
turning $\p{x} \mapsto v$ into a duplicable assertion,
called a \emph{shared region} in \tada.
For our example,  we define a shared region $\region{ex}{\rid}(\p{x},v)$
with an associated \emph{interpretation}
$\rInt(\region{ex}{\rid}(x,v)) \is x \mapsto v$,
which specifies which resource is being shared.
The \emph{region type} $\rt[ex]$ (for ``$\rt[ex]$ample'')
is the  name associated with this interpretation,
and the \emph{region identifier} $\rid$ is an abstract identifier associated
with this specific instance of the region type $\rt[ex]$.
The arguments $ (\p{x},v) $ of the region
are called the \emph{abstract state} of the region.
The definition of a region is completed by an
\emph{interference protocol}~$\regLTS[ex]$
which  restricts, in rely/guarantee style,
the allowed updates to the abstract state.
Here,  we  encode the facts that
\begin{enumerate*}[label=(\alph*)]
  \item only $\cmd_2$ can update $\p{x}$ and 
    \label{fact:cmd2-v}
  \item $v$ can only be updated to 1.
    \label{fact:v-to1}
\end{enumerate*}
Although such strong invariants are not required to just prove safety,
they will be useful for the termination proof later.
To encode fact~\ref{fact:cmd2-v},
we introduce a form of ghost state called a \emph{guard},~$\gEx$,
which gives $\guard{e}$xclusive permission to update $\p{x}$.
Formally, guards
(probably first introduced in deny-guarantee reasoning~\cite{deny-guarantee})
form a partial commutative monoid~(PCM), where in this case
$\gEx \guardOp \gEx$ is undefined to capture exclusive permission:
if a thread owns $\gEx$ then no other thread can own it at the same time.
To link $\gEx$ with the ability to change \p{x},
the protocol $\regLTS[ex]$ allows the guarded update
\begin{equation}
  \gEx : (\p{x}, v) \interfTo (\p{x}, 1)
  \label{prot:ex-safety}
\end{equation}
Fact~\ref{fact:v-to1} is encoded by  this being  the only allowed update.

In \tada\ and other modern separation logics such as Iris,
implication is generalised to the \emph{viewshift} construct~($\vshift$)
from~\cite{views},
which can be used to consistently update ghost information,
purely within the logic (as opposed to through ghost {code}).
Here,  it can be used to turn the owned resource $x \mapsto v$
into a shared resource
$
  P = \exists v\st x \mapsto v
    \vshift
      \exists \rid\st(
        \exists v\st \region{ex}{\rid}(\p{x},v) * \guardA{\gEx}{\rid}
      )
    \equiv
      \exists \rid\st (P_1 * P_2)
$
where
$
  P_1 = \exists v\st \region{ex}{\rid}(\p{x},v)
$
and
$
  P_2 = \exists v\st \region{ex}{\rid}(\p{x},v) * \guardA{\gEx}{\rid}
$.
The \emph{guard assertion}~$\guardA{\gEx}{\rid}$ indicates ownership of the guard $\gEx$ for the region with identity~$\rid$.
Using standard reasoning, one can then prove
$ \TRIPLE |- {P_1} {\cmd_1} {\True} $
and
$ \TRIPLE |- {P_2} {\cmd_2} {\True} $,
which entails, by the  parallel rule
$ \TRIPLE |- {P_1 * P_2} {\cmd_1 \parallel \cmd_2} {\True * \True} $.
By consequence and existential elimination on~$\rid$, we obtain our goal
$ \TRIPLE |- {P} {\cmd_1 \parallel \cmd_2} {\True} $.

Let us now turn  to termination.
A thread-local approach would  proceed
by first proving that~$\cmd_1$ and~$\cmd_2$ terminate separately,
and then concluding  that their parallel composition terminates.
In the case of \emph{non-blocking} code,
it is  possible to obtain a proof of this form:
by definition, a non-blocking thread does not need
the progress  of another thread 
in order to terminate.
For non-blocking code, a rely/guarantee protocol that only asserts
safety facts about the extent of the interference of the threads 
is all that is needed to prove termination.
This is exploited by virtually all the program logics which prove 
total specifications for non-blocking programs
\cite{HoffmannMS13,CookPR07,LiangFS14,total-tada}.
The non-blocking case allows the use of
a while rule which is essentially the one of total Hoare logics:
\begin{inlineproofrule}
  \infer*[right={WhileNB}]{
    \label{rule:while-non-blocking}
    \forall \beta \leq \beta_0\st
      \TRIPLE |-
        {P(\beta) \land \bexp}
          \cmd
        {\exists \gamma.P(\gamma) \land \gamma < \beta}
  }{
    \TRIPLE |-
      {P(\beta_0)}
        {\acode{while(BEXP)\{CMD\}}}
      {\exists \gamma.
          P(\gamma)
          \land \neg\bexp
          \land \gamma \leq \beta_0}
  }
\end{inlineproofrule}
Here $\beta$ is an \emph{ordinal-valued variant}
which is shown to strictly decrease after each iteration.
By well-foundedness of ordinals, there can only be finitely many iterations,
and  hence the loop terminates.
However, this rule is completely inadequate for blocking code: in our example,
the loop of $\cmd_1$ admits no variant,
since the iterations do not achieve any sort of progress.
Indeed, none of the cited works can handle this simple example.
Reasoning about progress for blocking programs
requires a whole set of new reasoning principles,
and a genuine extension of rely/guarantee with liveness information.

In \tadalive, the while rule has  a more general form:\footnote{We simplify the rule for this introductory section, informally using  the
  standard LTL notation
  $\Always\,P$
  for \emph{always P}
  (i.e.~$P$ holds at every point of a trace)
    and
  $\Eventually\,P$ 
  for \emph{eventually P}
  (i.e.~$P$ holds at some point of a trace). The full rule is
given 
 in \cref{sec:while-rule}.}
\begin{inlineproofrule}
  \infer*[right={WhileB}]{\label{rule:while-blocking}
    \Always\,L \implies \Finally\,T
    \\\\\forall \beta \le \beta_0 \st\;
      \TRIPLE |-
        {P(\beta)\phantom{{} * T} \land \bexp}
        \cmd
        {\exists \gamma \st P(\gamma) \land \gamma \leq \beta}
    \\\\\forall \beta \le \beta_0 \st\;
      \TRIPLE |-
        {P(\beta) * T \land \bexp}
        \cmd
        {\exists \gamma \st P(\gamma) \land \gamma < \beta}
  }{\TRIPLE |-
      {P(\beta_0) * L}
        {\acode{while(BEXP)\{CMD\}}}
      {\exists \gamma \st P(\gamma) * L
        \land \neg \bexp
        \land \gamma \leq \beta_0}
  }
\end{inlineproofrule}
The crucial difference is that the rule uses  a set of
\emph{target} states~$T$:
when an  iteration starts  in a target state,
the variant must be shown to strictly decrease, $\gamma<\beta$
(i.e.~the iteration needs to produce measurable progress);
when an iteration starts from a non-target state, the variant is only required not to increase, $\gamma\leq\beta$
(i.e.~no progress is undone).
These two conditions alone do not prove the termination of the loop:
the execution may be constantly in a non-target state.
In our example, the $T$ is $\region{ex}{\rid}(\p{x},1)$.
To conclude that the loop terminates,
the first premise requires $\Always\,L \implies \Finally\,T$:
that is, 
in traces where $L$ holds constantly,
with the help of the environment,
we will be \emph{eventually always} in a target state.
The assertion~$L$ captures facts that hold at any point of the iterations
of the loop, as it is in the triple of the conclusion but 
framed off the triples in the premises. 
When~$T$ finally happens, by fairness of the scheduler the loop will
execute, and will do so from a state  where, by the third premise,
the iterations will make progress towards termination.

To make this reasoning work, the first problem we encounter is that none of the information in a standard rely/guarantee specification
supports proving $\Finally\,T$.
Indeed, nothing in the protocol defined by $\rt[ex]$ expresses the idea
that at some point the environment will help $\cmd_1$ by setting~\p{x} to~1.
A safety rely merely expresses  that an update is \emph{allowed},
not that it will be eventually executed.
In other words, a safety rely alone is too imprecise an abstraction:
it cannot distinguish between environments that make the local thread terminate
from the ones that do not.
The first question we have to answer is:
\emph{how can ``help'' from the environment be represented
  in a rely/guarantee proof?}

\subsection*{Innovation~1: Subjective Obligations For
  Liveness Invariants}
Safety arguments are centred around \emph{invariants}:
that is, facts of the form \emph{always P},
encoded using  regions in \tada.
\tadalive's basic observation is that to represent help from the environment,
all that is needed is \emph{liveness invariants}:
that is, facts of the form \emph{always eventually P}.
By combining liveness invariants and safety invariants one can encode
more complex progress conditions such as $\Finally\,T$.
To represent liveness invariants in a thread-local way,
\tadalive\ introduces a new kind of ghost state called \emph{obligations}.
Similarly to guards, they form a PCM.
The interference protocol is augmented by a component that explains
how an update affects the obligations.
In our example, we want to represent the liveness invariant
always eventually $\region{ex}{\rid}(\p{x},1)$
which,  together with the invariant that \p{x} can only be set to 1, 
implies $\Finally(\region{ex}{\rid}(\p{x},1))$.
We therefore introduce an obligation $\obl{u}$
(for~$\obl{u}$pdate-to-1),
where again $\obl{u} \oblOp \obl{u}$ undefined captures
exclusivity, and 
extend the protocol to link $\obl{u}$ to the update:
\begin{equation}
  \gEx : ((\p{x}, v), \obl{u}) \interfTo ((\p{x}, 1), \oblZero)
  \label{prot:ex-liveness}
\end{equation}
This transition to update the region  can be executed by a thread  with both  the $\gEx$
guard and the $\obl{u}$ obligation;
the effect of the update is to ``consume'' the~$\obl{u}$ resource, as the obligation resulting from the update is the unit $\oblZero$.
We say the update \emph{fulfils} the obligation~$\obl{u}$.

A safety rely, as expressed by specification~\eqref{prot:ex-safety}, says:
verify a thread under the assumption that
the environment steps will obey the protocol.
As a first approximation,
our liveness rely, as expressed by~\eqref{prot:ex-liveness}, additionally says:
verify a thread under the assumption that the environment will always eventually
fulfil the obligations it owns.
(We will refine this idea in the next section to avoid unsound circular reasoning).
We say an obligation~$O$ is \emph{assumed live}
if the environment always eventually fulfils~$O$.
In other words: if, at any time, the environment owns~$O$,
it eventually fulfils~$O$.

This idea introduces a complication; 
we need to locally keep track of which (relevant) obligations are owned
by the environment, in order to make use of the liveness rely assumption.
We solve this problem by taking inspiration from the concept
of subjective separation of~\cite{scsl}. We introduce subjective 
obligation assertions: 
  local obligations,~$\locObl{\obl{u}}{\rid}$,
    asserting local ownership of the obligation~$\obl{u}$ associated
    with region~$\rid$, and
  environmental obligations,~$\envObl{\obl{u}}{\rid}$,
    asserting environment ownership of the obligation $\obl{u}$.
What makes these assertions interesting  is the way they compose:
that is, 
$
  \locObl{\obl{u}}{\rid}
    \iff
  \locObl{\obl{u}}{\rid} * \envObl{\obl{u}}{\rid}
$.
If  we start with local obligation $\obl{u}$
and we want to fork into two threads,
we  use~$*$ to give responsibility of $\obl{u}$ to one thread
and knowledge that the environment has this responsibility to the other.

To complete the proof sketch for our example, we first need to extend the region interpretation
by adding the obligation protocol:\footnote{
  The assertion $\bexp \dotimplies Q$ stands for
  $ (\bexp \land Q) \lor (\neg \bexp \land \emp) $.
}
\[
  \rInt(\region{ex}{\rid}(\p{x},v)) \is
    x \mapsto v * ( v = 1 \dotimplies \locObl{\obl{u}}{\rid} )
\]
When the value at~\p{x} is 1, the obligation $\obl{u}$ is owned by the interpretation,
and hence owned by no thread.
A thread owning~$\obl{u}$ and setting~\p{x} to~1  fulfils the obligation
precisely by leaving it inside the interpretation.
There is no other way of losing ownership of an obligation
because we adopt a classical interpretation of separation: that is, 
$ P*\locObl{\obl{u}}{r} \not\implies P $.
For soundness, the interpretation of a region with id~$\rid$ is only allowed
to own obligations of~$\rid$.

The \tadalive\ proof  starts by using viewshift to
transform the resource in the precondition
into this new region that is  shared between the two threads:
\[
  \exists v\st x \mapsto v
    \quad\vshift\quad
      \exists \rid\st\bigl(
        \exists v\st \region{ex}{\rid}(\p{x},v) *
        \guardA{\gEx}{\rid} *
        (v\neq 1 \dotimplies \locObl{\obl{u}}{\rid})
      \bigr)
    \quad\equiv\quad
      \exists \rid\st (P_1 * P_2)
\]
where
$
  P_1 = 
        \exists v \st
          \region{ex}{\rid}(\p{x}, v) *
          v\neq 1 \dotimplies \envObl{\obl{u}}{\rid}
$
and
$
  P_2 =
        \exists v \st
          \region{ex}{\rid}(\p{x}, v) *
          \guardA{\gEx}{\rid} *
          v\neq 1 \dotimplies \locObl{\obl{u}}{\rid}   
$
are the preconditions of the proofs of~$\cmd_1$ and~$\cmd_2$ respectively.
To discharge $\Finally(\region{ex}{\rid}(\p{x},1))$ in
the proof of the while loop of~$\cmd_1$,
we can use
$ L =
    \exists v\st
      \region{ex}{\rid}(\p{x},v) * v\neq 1 \dotimplies \envObl{\obl{u}}{\rid}
$, which holds throughout the loop:
if we are in a state $\region{ex}{\rid}(\p{x},v)$ either~$v=1$,
in which case we are in a target state and the value of~$v$ will
remain~1 forever; or~$v\neq 1$ in which case we know
$\envObl{\obl{u}}{\rid}$.
By the liveness rely, when the environment owns~$\obl{u}$, it will eventually fulfil it, which by~\eqref{prot:ex-liveness} can only be done by setting~$v=1$.
\Cref{sec:rules} explains in detail how this argument
is carried out formally in \tadalive.

At this point, we are able to prove the \emph{total} triples
$ \TRIPLE |- {P_1} {\cmd_1} {\True} $
and
$ \TRIPLE |- {P_2} {\cmd_2} {\True} $.
However, the standard parallel rule is unsound in the sense that 
the two triples can be proven even with $\cmd_2 = \code{skip}$
but, in this case, the parallel composition would not terminate!
\tadalive's parallel rule can recover soundness by checking
that the postconditions of the two threads do not own pending obligations,
which we can show by proving the stronger triples
$ \TRIPLE |- {P_1} {\cmd_1} {\region{ex}{\rid}(\p{x},1)} $
and
$ \TRIPLE |- {P_2} {\cmd_2} {\region{ex}{\rid}(\p{x},1) * \guardA{\gEx}{\rid}} $.
This condition is  too restrictive in general,
and we will relax it appropriately in the next section.

\subsection*{Innovation~2: Obligation Layers To Avoid Circular Arguments}
Structuring liveness invariants through obligations, as sketched, 
presents a significant  problem for soundness due to the possibility
of making  unsound circular liveness assumptions.
Consider the following variant  of our  busy-waiting example:
\[
\color{gray}\cmd'_1\left\{\color{black}
\begin{parall}
  \begin{threadcode}[gobble=2]
  var v_1 = 0 in
  while(v_1!=1){
    v_1:=[x_1]
  }
  [x_2] := 1
  \end{threadcode}
  \PARALLEL
  \begin{threadcode}[gobble=2]
  var v_2 = 0 in
  while(v_2!=1){
    v_2:=[x_2]
  }
  [x_1] := 1
  \end{threadcode}
\end{parall}
\color{gray}
\right\}\;\cmd'_2
\]
There are two shared heap cells at \code{x_1} and \code{x_2} respectively.
The thread on the left ($\cmd'_1$) is busy-waiting on \code{x_1} which
is supposed to be set by the thread on the right ($\cmd'_2$),
and vice versa, causing a classic high-level deadlock\footnote{This liveness form of deadlock is also known as ``livelock''
  since every thread is always taking steps, although no global progress
  is made by any of those steps.
  This is not to be confused with the safety property of ``global'' deadlock
  as found in languages with blocking primitives.
  }
situation: the program does not terminate.

Let us try to replicate the argument we  used for the busy-waiting example.
We require  a region sharing   both cells, $ \region{dex}{\rid}(\p{x}_1,\p{x}_2,v_1,v_2) $,
where $v_i$ is the value stored at $\p{x}_i$.
We use two guards $\gEx_1$ and $\gEx_2$,
and two obligations, $\obl{u}_1$ and $\obl{u}_2$
linked to the update of \code{x_1} and \code{x_2} respectively:
\begin{align}
  \gEx_1 &:
    ((\p{x}_1,\p{x}_2, v_1, v_2), \obl{u}_1)
      \interfTo
    ((\p{x}_1,\p{x}_2,  1 , v_2), \oblZero)
  \label{prot:dex1}
  \\
  \gEx_2 &:
    ((\p{x}_1,\p{x}_2, v_1, v_2), \obl{u}_2)
      \interfTo
    ((\p{x}_1,\p{x}_2, v_1,  1 ), \oblZero)
  \label{prot:dex2}
\end{align}
Without additional precautions,
we would be able to derive the triples (for $i=1,2$)
\begin{equation}
  \TRIPLE |-
    {P_i}
    {\cmd'_i}
    {\region{dex}{\rid}(\p{x}_1,\p{x}_2,1,1) * \guardA{\gEx_i}{\rid}}
  \label{spec:dex-wrong}
\end{equation}
where
$
  P_i = \exists v_1,v_2 \st
          \region{dex}{\rid}(\p{x}_1,\p{x}_2, v_1,v_2) *
          \guardA{\gEx_i}{\rid} *
          \bigl( v_i    \neq 1 \dotimplies \envObl{\obl{u}_i}{\rid}     \bigr) *
          \bigl( v_{3-i}\neq 1 \dotimplies \locObl{\obl{u}_{3-i}}{\rid} \bigr)
$.
Given the interpretation we sketched earlier, these triples mean:
thread~$i$ terminates provided its environment (i.e.~thread~$3-i$)
always eventually fulfils obligation~$\obl{u}_{3-i}$.
This leads, in the application of the parallel rule,
to an unsound circular argument:
to show thread~$i$  fulfils obligation~$\obl{u}_{i}$,
thread~$i$ is relying on  the assumption
about the eventual fulfilment of~$\obl{u}_{3-i}$ by the environment,
which in turn relies on
the eventual fulfilment of~$\obl{u}_{i}$ by thread~$i$ itself.
The question is then:
\emph{how can we rule out circular arguments,
  while keeping the proof thread-local?}
In particular, we want a solution that
allows us to keep the abstraction of the environment as
local and abstract as possible,
without revealing unnecessary structure of the other threads.

Our solution is to specify dependencies between liveness invariants.
We do this by imposing a partial order on obligations:
each obligation~$O$ is associated with a \emph{layer}, denoted~$\lay(O)$,
which is an element of a user-defined well-founded partial order,~$\Layer$.
Using layers,  we can refine our reasoning principle and gain soundness:
to be allowed to assume~$O$ is live, one has to show all the locally owned obligations have layers greater than~$\lay(O)$.
The intuition is that
local fulfilment of~$O_2$ can depend
on the environment's fulfilment of~$O_1$
\emph{only if} $\lay(O_1) \laylt \lay(O_2)$.

In our deadlocking example, layers expose the circularity issue
and prevent the triples~\eqref{spec:dex-wrong} from being derivable.
Specifically, the proof of the loop of~$\cmd'_1$
requires us to prove $\Finally(\region{dex}{\rid}(\p{x}_1,\p{x}_2, 1, \wtv))$.
At this point we are continuously holding the obligation~$\obl{u}_2$
so, to be able to assume~$\obl{u}_1$ live, we require
${\lay(\obl{u}_1) < \lay(\obl{u}_2)}$.
However, the proof of the loop of~$\cmd'_2$
would require the symmetric constraint,
${\lay(\obl{u}_2) < \lay(\obl{u}_1)}$, leading to a contradiction.

If we replace $\cmd'_2$ with
$\cmd''_2 \is \bigl(\code{
  [x_1] := 1;
  var v_2 = 0 in
  while(v_2!=1)\{v_2:=[x_2]\}
}\bigr)$,
the program $\cmd'_1 \parallel \cmd''_2$  terminates
and indeed the proof  goes through with
${\lay(\obl{u}_1) < \lay(\obl{u}_2)}$.
This is because the first instruction of $\cmd''_2$ fulfils $\obl{u}_1$
so the loop no longer constantly owns it whilst  assuming $\obl{u}_2$ live.
The structure of $\cmd''_2$ does not impose any dependency on the two
liveness invariants.

The generalisation of the liveness rely to use  obligations
with layers, enables us to give  a general parallel rule:
instead of just forbidding pending obligations in the postconditions,
we require that the postcondition of each thread  only
owns 
obligations with  layers
greater than the layers of obligations assumed live in the  other thread's proof.

Let us contrast our layered obligations
with other solutions found in the literature.
The LiLi logic cannot verify the above examples
as it lacks support for parallel composition.\footnote{Indeed, LiLi's goal is limited to proving that
  a module's implementation refines its specification.
  The code of the module cannot fork threads
  but any multi-threaded client of the module
  is guaranteed not to be able to distinguish the implementation
  from the specification.
}
LiLi's while rule does share  the same high-level structure as
\ref{rule:while-blocking}, a structure that can be traced as far back as~\cite{OwickiL82}.
The main crucial difference is in how $\Finally\,T$ is proven.
LiLi proposes the idea of \emph{definite actions}, a reincarnation of
``leads-to'' assertions of~\cite{OwickiL82}, 
to build  a liveness rely.
Definite actions require the identification of a logical global ``queue''
of threads where the thread at the front is always able to execute its action
and that action implies global progress.
In LiLi, the target states are the ones where the local thread is at the head of this queue, and the $\Finally\,T$ condition is proven by showing that when the head of the queue executes an action, there is some local well-founded progress measure that decreases.
Definite actions have a number of drawbacks:
\begin{itemize}
  \item they require heavy introduction of ghost code for manipulating
        globally shared ghost state
        in order to construct the queue of threads; and 
  \item the progress reasoning on the queue requires analysing all possible
        ways the other threads may finally produce the target states.
\end{itemize}
Layered obligations are key to  resolving these problems:
\begin{itemize}
  \item they remove the need for ghost code altogether,
        and make liveness invariants local
        using the local/environmental obligation assertions;  and
  \item by only relying on the eventual fulfilment of layered  obligations,
        the proof of $\Finally\,T$ can ignore \emph{how} the environment
        is going to implement such fulfilment;
the only important fact to retain about the \emph{how}
        is which liveness invariants are  assumed to guarantee the fulfilment.
\end{itemize}

There has been work on proving various
  safety (e.g.~global dead\-lock-freedom)~\cite{Kobayashi06,Jacobs18}
  and progress (e.g.~dead\-lock-freedom, termination)~\cite{Kobayashi00,Leino10,BostromM15,Jacobs18toplas}
properties of concurrent programs,
which assume the only source of blocking behaviour
comes from the use of blocking primitives (e.g.~built-in locks or channels).
Although none of them can handle busy-waiting patterns
like our previous examples,
they typically detect deadlocks using ``tokens''
(often also called \emph{obligations})
which represent the responsibility to call a blocking primitive.
These tokens are arranged in an acyclic graph of dependencies.
Superficially, these tokens are related to our layered obligations
in that they both are devices to rule out cyclical dependencies.
There are, however, deep differences between the two.
Tokens are linked (ad hoc in the operational semantics and through ghost code)
to blocking primitive operations calls,
and dependencies between the tokens represent
causal dependencies between these primitive \emph{events}.
By contrast, our layers represent dependencies between
\emph{liveness assumptions},
and reflect a purely logical structure.
This makes our layered obligations particularly general and flexible:
they are able to express arbitrary high-level blocking patterns and not just primitive blocking operations,
enabling  truly abstract specifications.

\subsection*{Innovation~3: Abstract Atomic Specifications for Blocking Operations}

Understanding blocking behaviour as the need for an abstraction of the environment that includes liveness invariants unlocks a novel approach
in giving abstract, precise and reusable total specifications for abstractly atomic operations.
Building on Total~\tada,
we propose a new specification format that expresses
the atomic effect of a linearizable operation,
and succinctly states the liveness invariant
required for ensuring termination,
at the right level of abstraction.
To see the problem and our solution, let us consider the paradigmatic example
of two fine-grained implementations of a lock module.

\paragraph{Two Lock Implementations.}
Consider the \emph{spin lock} and the \emph{CLH lock} given in \cref{fig:locks}.
The implementations enable threads to compete
for the acquisition of a lock at address \code{x}
by running concurrent invocations of the \code{lock(x)} operation.
Only one thread will succeed, leaving the others to wait until
the \code{unlock(x)} operation is called by the winning thread.

\begin{figure}[tb]
  \centering\footnotesize \begin{minipage}{140pt}
    \begin{sourcecode}[style=framed,title={Spin Lock},gobble=3,xleftmargin=1.5em]
    def lock(x){
      var d = 0 in
      while(d = 0){
        d := CAS(x,0,1) @\label[line]{line:spin-lock-linpt}@
      }
    }

    def unlock(x) { [x] := 0 }

    def makeLock() { var x in
      x := alloc(1);
      [x] := 0;
      ret := x
    }
    \end{sourcecode}
  \end{minipage}
  \hfil
  \begin{minipage}{185pt}
    \begin{sourcecode}[style=framed,title={CLH Lock},gobble=3,xleftmargin=1.5em]
    def lock(x) { var c,p,v in
      c := alloc(1); [c] := 1;
      p := FAS(x+1, c);@\label[line]{line:clh-enq}@
      v := [p];
      while(v != 0) { v := [p] }@\label[line]{line:clh-wait}@
      [x] := c; @\label[line]{line:clh-lock-linpt}@
      dealloc(p)
    }
    def unlock(x) { var h in h := [x]; [h] := 0 }
    def makeLock() { var x,h in
      h := alloc(1); [h] := 0;
      x := alloc(2); [x] := h; [x+1] := h;
      ret := x
    }
    \end{sourcecode}
  \end{minipage}
  \caption{Two fine-grained lock implementations.}
  \label{fig:locks}
\end{figure}

The primitive commands, such as assignment, lookup and mutation,
are primitive atomic and non-blocking:
every primitive command, if given a CPU cycle, will terminate in one step.
Since reads and writes may race, the language is equipped with a
\emph{compare-and-swap} primitive command, \code{CAS(x,$v_1$,$v_2$)},
which checks if the value stored at \code{x} is $v_1$:
if so, it atomically stores $v_2$ at \code{x} and returns~1;
otherwise it just returns~0.
Similarly, the \emph{fetch-and-set} primitive command,
\code{FAS(x,$v$)},
stores $v$ at \code{x} returning the value that was stored at \code{x} just before overwriting it.

The spin lock in~\cref{fig:locks}  is  standard.
Its state comprises  a heap cell at \code{x}
which stores either 0 (unlocked) or 1 (locked).
The Craig-Landin-Hagersten (CLH) lock~\cite{ArtBook} in~\cref{fig:locks}
serves threads competing for the lock in a FIFO order.
It queues requests,
keeping a head and a tail pointer (at \p{x} and \p{x+1} respectively).
The predecessor pointers are stored in each thread's local state (in \p{p}).
The lock can be acquired by a thread once its predecessor signals
release of the lock by setting its queue node to~0.
Unlocking the lock corresponds to setting the queue's head node value to~0.

Let us focus on the lock operation of the CLH lock.
The interesting aspect is that \code{lock} displays blocking behaviour
that is observable by the client of the module
(it is indeed the quintessence of blocking).
We cannot just provide a total triple for it:
the operation does not always terminate.
The challenge is to design a specification format that
accurately captures the abstract functionality of the operation
and its subtle termination properties.

First off, one would like a specification that hides the implementation details
and only exposes the abstract state of the lock to the client:
a lock instance is represented by an abstract resource
$ \ap L(\p{x}, l) $\footnote{We omit the  region identifier to
  simplify the discussion.} where
  $l=1$ indicates the lock is locked, and
  $l=0$ means it is unlocked. It is worth noting that traditional 
Hoare triples are not able to represent the useful behaviour of
\code{lock(x)}. The triple 
$ \TRIPLE |- {\ap L(\p{x},0)} {\p{lock(x)}} {\ap L(\p{x},1)}
$
requires the client to establish that  the lock is unlocked \emph{before} calling
the operation, defying the very purpose of the operation's
functionality. The triple $ \TRIPLE |- {\ap L(\p{x},0) \lor \ap L(\p{x},1)} {\p{lock(x)}} {\ap
  L(\p{x},1)} $ allows the operation to be called in the locked state,
but is not precise enough since 
the same triple holds for a simple assignment $\p{[x]\:=1}$. It does not express the property that,
upon termination of the operation,
we can claim that we have acquired the lock.
A partial specification  of a lock  is already a challenge; a total
specification more difficult still. 

Proposed solutions in the literature can be divided into 
history-based, refinement-based  and abstract atomicity-based
approaches.
The history-based approach
(e.g.~\cite{SergeyHistories} for safety,
  \cite{Shao17,Shao18} for progress)
is expressive but at the price of
complex and indirect specifications;
the verification requires explicit manipulation of the histories,
complicating client reasoning.
The only progress-aware refinement-based  approach  that can
modularly verify the 
CLH lock is the LiLi logic~\cite{LiangF18}.
LiLi's refinement~$\refines$ is progress-preserving and contextual,
allowing the result to be reused in arbitrary client contexts.
For example, the LiLi proof for CLH lock (under weak fair scheduling) shows that
\[
  \code{lock(x)} \refines \code{spec_lock(x)}
\]
where \code{spec_lock(x)} is defined (in pseudocode) as\footnote{In~\cite{LiangF18}, this is the result of applying the appropriate
  wrapper to the lock specification:
  $\mathsf{wr}_{\mathsf{PSF}}^{\mathsf{wfair}}(\textbf{await}(\p{l=0})\{\p{l\:=cid}\})$.
}

\begin{sourcecode*}[morekeywords={await,self},xleftmargin=5em]
spec_lock(x) {
    enqueue(x.queue, self);
    await (head(x.queue) = self $\land$ x.state=0) {
      << x.state:=self; x.queue := tail(x.queue) >>
    }
}
\end{sourcecode*}The abstract state of the lock is represented by \code{x.state},
but to represent the fact that threads will not be starved,
an abstract FIFO queue at \code{x.queue} keeps track of the threads to be served;
\code[morekeywords={self}]{self} is the thread id of the caller.
The command \code[morekeywords={await}]{await($\bexp$)\{$\cmd$\}}
is a blocking \emph{primitive} introduced to express
the non-primitive blocking of the implementation.
The potential absence of progress of the implementation's busy-waiting steps
is represented by potential absence of a step ahead in the specification.

LiLi's specification style has three major drawbacks:
\begin{enumerate}
  \item the specification code is not much simpler
        than the original implementation,
        and is not able to hide the implementation detail of the thread queue;
  \item the specification code is \emph{not atomic}:
        it produces one step for entering the queue,
        and one step for acquiring the lock;
  \item since the termination properties are
        represented through the behaviour of code,
        a client proof that wants to make use of these properties
        must reprove them on the specification code
        before being able to use them in the argument.
\end{enumerate}
These problems limit the
  abstraction capabilities,
  proof reuse and
  scalability
of the approach.

The abstract atomicity approach has been pioneered by the \tada\ logic.
It directly influenced logical atomicity in Iris~\cite{iris}, 
and was extended to provide total specification
for non-blocking programs in Total~\tada~\cite{total-tada}.
The aim of the \tada\ approach  is to keep the Hoare-triple style of specification,
whilst being able to give precise and abstract specifications
to fine-grained code like CLH lock. The \tada\ solution
is to provide a Hoare triple for \code{lock} which embraces the fact
that, between the invocation of the operation
and the execution of the atomic update of the lock,
there is a phase of \emph{interference}
where the environment can change the value of the lock.
It is important to be able to distinguish
the imprecise precondition that holds during the interference phase, 
$ \ap L(\p{x},0) \lor \ap L(\p{x},1) $,
and the precise precondition,  $ \ap L(\p{x},0) $,  that holds 
\emph{just before} the atomic update performed by the lock operation
at its \emph{linearization point}~\cite{HerlihyW90}.

The \tada\ safety specification for \code{lock}
is the partial \emph{atomic triple}:
\begin{equation}
  \ATRIPLE |-
    \A l \in\set{0,1}.
      <\ap{L}(\p{x},l)>
        \code{lock(x)}
      <\ap{L}(\p{x},1) \land l=0>
  \label{spec:lock-tada}
\end{equation}
The \emph{interference precondition}
$\PQ{l \in\set{0,1}.} \color{atomic}\atombra{\ap{L}(\p{x},l)}$
describes the interference phase.
It states that the environment must preserve the existence of the lock at
\code{x} but may change the value of $l$, and the implementation of
the lock must tolerate these environmental changes.
The \emph{pseudo-quantifier} $\PQ{l \in \set{0,1}}$ is
unusual, behaving like an evolving universal quantifier in that the
environment is able to keep changing $l$ over time and behaving like
an existential quantifier in that the implementation can assume that
the lock always exists with $l \in \set{0,1}$.
The triple (\ref{spec:lock-tada}) states that, if the environment
satisfies the interference precondition and the operation terminates,
then the implementation guarantees that, just before
the linearization point, the lock must have been available  for
locking (${\color{atomic} l = 0}$) and, just afterwards, the lock has
been locked by the operation  (${\color{atomic} \ap{L}(\p{x}, 1)}$).
Exclusive ownership of the lock after the operation terminates
can be derived from the $l=0$ assertion in the postcondition:
just before we locked it, nobody else could claim that they owned the lock.
The \tada\ safety specification for \code{unlock} is the partial
atomic triple 
\[
  \ATRIPLE |- \aall l \in \set{1}.
    <\ap{L}(\p{x}, l)> \code{unlock(x)} <\ap{L}(\p{x}, 0)>
\]
This triple\footnote{We typically omit the pseudo-quantifier in
  the case where the set has just one element,
  e.g.~$
    \ATRIPLE |- <\ap{L}(\p{x}, 1)> \code{unlock(x)} <\ap{L}(\p{x}, 0)>
  $.}
states that,
to be used correctly, the unlock  operation requires the 
lock to be locked  and  not changed by the environment
during the interference phase; in return, the operation promises to atomically
set the lock to be unlocked.

\tadalive{} builds on the \tada\ specification format.  To turn the
\tada\ triple for \p{lock} into a total specification, the termination
guarantee must depend on the environment: if the environment decides
to hold the lock indefinitely, no lock implementation should allow the
\p{lock} operation to terminate.
Hence, we express blocking as a \emph{liveness condition}
on the \emph{environment} during the interference phase
of an abstractly atomic operation.
The CLH \p{lock} operation will terminate under weak fairness,
provided that, if the lock is locked by the environment
during the interference phase,
the environment will \emph{eventually} unlock it.
In general, a blocking operation will require
an environment that is \emph{live}:
it will always eventually bring the abstract state
to a \emph{good} (e.g.~unlocked)~state.

The TaDA Live total specification of the CLH \p{lock} operation is:
\begin{equation}
  \ATRIPLE |-
    \A l\in\set{0,1} \eventually \set{0}.
      <\ap{L}(\p{x},l)>\code{lock(x)}<\ap{L}(\p{x},1) \land l=0>
 \label{spec:lock-fair}
\end{equation}
The interference precondition is
$  \PQ{l\in\set{0,1} \eventually \set{0}.}
      {\color{atomic}\atombra{\ap{L}(\p{x},l)}} $
with the pseudo-quantifier now incorporating the environment liveness condition.
As well as stating that the environment can keep changing the lock,
the interference precondition also states that
if the lock is in a bad state ($l \in \set{0,1}\setminus\set{0}$)
then the environment must always
eventually change it to a good state ($l \in \set{0}$).
The implementation needs to ensure
termination under the assumption  that the lock  always eventually
returns to the unlocked state.
Note that  the environment is allowed to change~$l$ back to~1
arbitrarily many times, provided it always eventually sets it back to 0.
To see why this is enough to ensure termination,
consider \cref{fig:waves}\ref{fig:wave-fair}
where we chart the evolution
of the abstract state induced by a live environment
in the interference phase of \p{lock}.
Progress towards termination of \p{lock} is guaranteed
by the progress measure charted in \cref{fig:waves}\ref{fig:progress-fair}:
every time the environment unlocks, the value of $l$ decreases from 1
to 0; 
when the environment locks, although $l$ increases to 1,
the number $q$ of threads in front of us in the queue decreases.
One crucial aspect of our specification design 
is that we do not want to expose the progress argument to the client
\emph{unless part of the argument needs to be made by the client}.
With CLH, the part of the argument appealing to the queue of threads
is completely internal to the implementation of the operation,
while the argument for the environment's liveness must be provided
by the client (the implementation has no power over this).
We prove this formally in \cref{sec:evaluation}.

\begin{figure*}[tb]
  \adjustfigure \begin{tikzpicture}[wave,baseline=\waveoffset]
\setlabel{(a)}

\draw[main line]
  (4pt,0)
    foreach[count=\step] \time
      in {2,3,1,1,4,1,2,2} {
      \ifodd\step
        -- ++(\time,0)
        -- ++(\atomicsteplength/2,1)
      \else
        --coordinate[midway](good-\step) ++(\time,0)
        -- ++(\atomicsteplength/2,-1)
      \fi
    }
    coordinate (last)
;
\draw[main line]
  (last) -- ++(1,0) -- ++(.2*\atomicsteplength,.2) coordinate (last);
\draw[main line,dotted]
  (last) -- ++(.8*\atomicsteplength,.8) --coordinate[midway](good-last) ++(1,0) -- ++(.8*\atomicsteplength,-.8)
  ++(1,.2) coordinate (last);

\draw[axis] (0,-\waveoffset)--(last|-30,-\waveoffset)
  node[below left]{time};
\draw[axis,use as bounding box] (0,-\waveoffset)--(0,1.2+\waveoffset)
  node[above left](vaxis){abs.~state};
\draw[tick] (2pt,0) -- (-2pt,0) node[abs state](st1){$ \ap{L}(\p{x},1) $};
\draw[tick] (2pt,1) -- (-2pt,1) node[abs state](st0){$ \ap{L}(\p{x},0) $};
\node[left=-1ex] at (last|-vaxis) {\normalsize(a)};

\path (vaxis) -- node[annot,yshift=7pt,good](inf) {always eventually good state} (last|-vaxis);
\draw[annot,good] (inf)
  edge[in=90,out=200] (good-2)
  edge[in=80,out=220] (good-4)
  edge[in=110,out=270] (good-6)
  edge[in=90,out=330] (good-8)
  edge[in=90,out=345,densely dashed] (good-last)
;
\end{tikzpicture} \label{fig:wave-fair}
\begin{tikzpicture}[wave,baseline=-\waveoffset]
\setlabel{(c)}

\draw[main line]
  (4pt,0)
    foreach[count=\step] \time
      in {2,3,1,1,4,2,3} {
      \ifodd\step
        -- ++(\time,0)
        -- ++(\atomicsteplength/2,1)
      \else
        --coordinate[midway](good-\step) ++(\time,0)
        --coordinate[midway](bad-\step) ++(\atomicsteplength/2,-1)
      \fi
    }
    -- ++(2,0) coordinate (last)
;
\draw[main line,dotted]  (last) --coordinate (good-last) ++(2,0)
  ++(1,-\waveoffset)coordinate (last);

\draw[axis] (0,-\waveoffset)--(last|-30,-\waveoffset) node[below left]{time};
\draw[axis] (0,-\waveoffset)--(0,1.2+\waveoffset) node[above left](vaxis){abs.~state};
\draw[tick] (2pt,0) -- (-2pt,0) node[abs state](st1){$ \ap{L}(\p{x},1) $};
\draw[tick] (2pt,1) -- (-2pt,1) node[abs state](st0){$ \ap{L}(\p{x},0) $};
\node[left=-1ex] at (last|-vaxis) {\normalsize(c)};

\path[use as bounding box]
  (bad-2|-0,-\waveoffset) node(i1)[impedence] {}
  (bad-4|-0,-\waveoffset) node(i2)[impedence] {}
  (bad-6|-0,-\waveoffset) node(i3)[impedence] {}
;
\path[annot] (-2,-1*\waveoffset) node[below,bad,align=left,font=\scriptsize] {bounded\\[-.8ex]impedance}
  edge[red,in=-120,out=4] (i1)
  edge[red,in=-90,out=2] (i2)
  edge[red,in=-90,out=0,looseness=.8] (i3)
;
\path (vaxis) -- node[annot,yshift=7pt,good](inf) {always eventually good state} (last|-vaxis);
\draw[annot,good] (inf)
  edge[in=90,out=200] (good-2)
  edge[in=90,out=220] (good-4)
  edge[in=100,out=270] (good-6)
edge[in=90,out=345,densely dashed] (good-last)
;
\end{tikzpicture}
 \label{fig:wave-impeded}\par\bigskip
  \begin{tikzpicture}[wave,baseline=\waveoffset]
\setlabel{(b)}

\draw[main line]
  (4pt,2.3)
    foreach[count=\step] \time
      in {2,3,1,1,4,1,2,2} {
      \ifodd\step
        -- ++(\time,0)
        -- ++(\atomicsteplength/2,-.25)
      \else
        --coordinate[midway](good-\step) ++(\time,0)
        -- ++(\atomicsteplength/2,-.25)
      \fi
    }
    coordinate (last)
;
\draw[main line]
  (last) -- ++(1,0) -- ++(.2*\atomicsteplength,-.2) coordinate (last);
\draw[main line,dotted]
  (last) -- ++(.8*\atomicsteplength,-.2)
    --coordinate[midway](good-last) ++(1,0) -- ++(.8*\atomicsteplength,-.2)
  ++(1,-.2) coordinate (last);

\draw[axis] (0,-\waveoffset)--(last|-30,-\waveoffset)
  node[below left]{time};
\draw[axis,use as bounding box] (0,-\waveoffset)--
  node[left,align=right](vaxis)
  {variant\\[-.5ex]\phantom{abs. state}\llap{$ 2 q + l $}}
(0,2+\waveoffset);
\node[left=-1ex,yshift=1.5ex] at (last|-vaxis) {\normalsize(b)};

\end{tikzpicture} \label{fig:progress-fair}
  \begin{tikzpicture}[wave,baseline=-\waveoffset]
\setlabel{(d)}

\draw[main line]
  (4pt,2)
    foreach[count=\step] \time
      in {2,3,1,1,4,2,3} {
      \ifodd\step
        -- ++(\time,0)
        -- ++(\atomicsteplength/2,-.25)
      \else
        --coordinate[midway](good-\step) ++(\time,0)
        --coordinate[midway](bad-\step) ++(\atomicsteplength/2,-.25)
      \fi
    }
    -- ++(2,0) coordinate (last)
;
\draw[main line,dotted]  (last) --coordinate (good-last) ++(2,0)
  ++(1,-\waveoffset)coordinate (last);

\draw[axis] (0,-\waveoffset)--(last|-30,-\waveoffset) node[below left]{time};
\draw[axis] (0,-\waveoffset)--
  node[left,align=right](vaxis)
  {variant\\[-.5ex]\phantom{abs.~state}\llap{$ 2\budget + l$}}
(0,2+\waveoffset) node[yshift=.5ex] {};
\node[left=-1ex,yshift=1.5ex] at (last|-vaxis) {\normalsize(d)};

\end{tikzpicture}
 \label{fig:progress-impeded}
  \par
  \caption{Live environment~\ref{fig:wave-fair};
    measure of progress for CLH lock where
    $q$ is the number of threads ahead in the queue~\ref{fig:progress-fair};
    live environment with bounded impedance~\ref{fig:wave-impeded};
    measure of progress for spin lock~\ref{fig:progress-impeded}.
  }
  \label{fig:waves}
  \label{fig:lockspecs-tadalive}
\end{figure*}

Now let us consider the spin lock implementation.
The spin \p{lock} operation cannot promise to terminate
just by relying on a live environment.
The problem is that when the environment locks the lock,
there is no measure of progress that decreases:
we are genuinely delayed by this action.
We call this effect \emph{impedance}.
We conceptualise impedance as
a greater \emph{leaking} of the progress argument to the client.
In the spin lock example, the whole of the progress argument
needs to be provided by the client:
the client needs to ensure that the environment will always eventually unlock the lock, and that it will only impede the operation a bounded number of times.
To represent this extra \emph{bounded impedance} requirement (depicted
in \cref{fig:waves}\ref{fig:wave-impeded}),
we extend the abstract state of the lock with an ordinal~$\budget$,
an \emph{impedance budget} that strictly decreases when the lock
state is set to~1.
We arrive at the following \tadalive\ specification for spin \p{lock}:
\begin{small}\begin{equation}\ATRIPLE \forall\budgetfun\st |-
       \A l\in\set{0,1} \eventually \set{0}, \budget.
        <\ap{L}(\p{x},l, \budget) \land \budgetfun(\budget)<\budget>
          \code{lock(x)}
        <\ap{L}(\p{x},1,\budgetfun(\budget)) \land l=0>
   \label{spec:lock-impeded}
  \end{equation}\end{small}The lock is now represented by the predicate assertion
$\ap{L}(\p{x},l,\budget)$ with ordinal~$\budget$,
which can also be changed by the environment during the interference phase.
As well as expressing the dependency on a live environment on $l$,
this triple states that every
\code{lock} operation consumes the budget~$\budget$ by a non-trivial
amount, thus providing  a logical measure of progress from good to bad
states.
The initial value of the budget and the function $\ordfun$
from ordinals to ordinals is determined by the client, which must
demonstrate that the budget is enough to make all its calls.

The \tadalive\ total specification of \p{unlock}   for the CLH \p{lock}
is the same as the \tada\ partial specification. By contrast, the \tadalive\ specification 
of \p{unlock}    for the  spin lock needs to incorporate the ordinals:
$
  \ATRIPLE |-
    <\ap{L}(\pvar{x}, 1, \budget)>
      \code{unlock}(\pvar{x})
    <\ap{L}(\pvar{x}, 0, \budget)>.
$
The impedance budget~$\budget$ is preserved by unlock.
This encodes the fact that \p{unlock} does not impede the other operations,
but also that by unlocking we cannot increase the budget.
By combining these assumptions about the budget
(it decreases when locking, stays constant when unlocking), it is
possible to conclude that the implementation of the 
spin lock terminates 
using the progress measure in \cref{fig:waves}\ref{fig:progress-impeded}.
Crucially, for spin lock,  the \emph{whole} of the progress argument
is provided by (and thus visible to) the client.

The impedance budget technique
was first introduced to concurrent separation logics
for non-blocking operations in Total~\tada~\cite{total-tada}.
Here, we smoothly integrate ordinals into TaDA Live
that fully supports blocking.

\subsection{Abstraction and proof reuse}
\label{sec:proof-reuse}

The 
\tadalive\ program logic works with \emph{hybrid} triples of the
form:
\begin{equation*}
  \ATRIPLE |-
    \A x \in X \eventually X'.
    <\na{P}|\at{P}(x)>
    \cmd
<\na{Q}(x)|\at{Q}(x)>
    \label{spec:hybrid}
\end{equation*}
which generalises both Hoare triples and abstract atomic triples. This
triple comprises: a pseudo-quantifier with its environment liveness
condition; atomic pre/post-conditions $\at{P}(x)$ and $\at{Q}(x)$; and
Hoare pre/post-conditions, $\na{P}$ and $\na{Q}(x)$.  The Hoare
pre/post-conditions describe stable resources that are owned locally by $\cmd$ and can
be updated non-atomically.  Hoare triples correspond to the 
case where $X=X'=\set{1}$ and $\at{P} = \at{Q} = \emp$.
Abstract atomic triples correspond to the case when 
$\na{P}$ and $ \na{Q} (x) $ are empty.
We have omitted some details from the hybrid
triples, such as layers
  and levels, since they are not important for the ideas of this
  section; the full details are given in \cref{sec:specs}. 

The integration of the liveness annotations in triples
achieves the goal of keeping the specification abstract and atomic.
To obtain the goal of reuse of proofs,
there are two missing ingredients:
a mechanism to make use of the $X\eventually X'$ assumption in a proof
of an implementation of the specification;
and a way to reuse the specification in an arbitrary client context.

Imagine proving the CLH \code{lock} implementation correct with respect to
specification~\eqref{spec:lock-fair}.
The while loop needs to discharge that
  ``finally, the current thread is at the head of the queue,
             and the lock is unlocked''.
This can only be proven with the help of the
$l\in\set{0,1} \eventually \set{0}$ liveness assumption
coming from the lock specification.
To this end, in addition to liveness assumptions given by obligation assertions, 
\tadalive\ extends judgements to allow contexts with
$X\eventually X'$ liveness assumptions,
used to discharge the $\Finally\,T$ condition in the while
rule. The full details are given in \cref{sec:specs}.

Now consider proving a client of a lock
using the specifications of the lock
operations for the calls to these operations. 
This requires the \textsc{Live}ness \textsc{C}heck rule:
\begin{inlineproofrule}
  \infer*[right={LiveC'}]{
    \Always{L} \implies \AlwaysEv{T}
    \\
    \forall x \in X \st \VALID |= \at{P}(x) * T \implies x \in X'
    \\\\ \ATRIPLE |-
      \A x \in X \eventually X'.
        <\na{P}| \at{P}(x)> \cmd <\na{Q}(x)  |\at{Q}(x)>
  }{\phantom{{}\eventually X'}
    \ATRIPLE |-
      \A x \in X.
        <\na{P} * L | \at{P}(x)> \cmd <\na{Q}(x) * L | \at{Q}(x)>
  }
  \label{rule:livec-simpl}
\end{inlineproofrule}
The rule's crucial effect is to remove the liveness annotation
$X \eventually X'$,
which can only be done in a situation where the corresponding liveness assumption
$\Always{(x\in X)} \implies \AlwaysEv{(x\in X')}$ is satisfied.
Just like the \ref{rule:while-blocking} rule,
we frame an assertion~$L$ which is information that holds
for the duration of the call.
Typically,~$L$ asserts the existence of some shared region and that
the environment holds some obligations depending on the state of the region.
We also need to provide a set of target states~$T$
capturing when~$x\in X'$ (second premise).
The crucial check of the rule is the first premise,
which examines the traces where~$L$ holds everywhere,
and asks us to prove that in those traces we see~$T$
satisfied infinitely often (and thus $x\in X'$ infinitely often).
If that is true, we can conclude that the command terminates
in the current context without the extra assumption
in the pseudo-quantification.
The resulting triple can be then manipulated using standard \tada\ reasoning.

Take the typical use of (CLH) locks
  $\cmd = \code{lock(x);...;unlock(x)}$
in a client $\cmd \parallel \dots \parallel \cmd$.
To share the lock resource $\ap L(\p{x},l)$,
the client proof would specify some region $\region{client}{\rid}(\p{x},l)$
where $l$ is the abstract state of the lock. A typical client would include the abstract state of other shared resources too,
but for simplicity we focus here on the lock.
The client needs to specify in its protocol
that the lock will be always eventually unlocked by the threads sharing it.
We therefore introduce an exclusive obligation $\obl{k}$
(the $\obl{k}$ey of the lock)
which is obtained when locking the lock and fulfilled when unlocking it:
\begin{align*}
  & ((\p{x},0), \oblZero) \interfTo ((\p{x}, 1), \obl{k}) &
  & ((\p{x},1), \obl{k}) \interfTo ((\p{x}, 0), \oblZero)
\end{align*}
The protocol is mirrored in the region's interpretation
$
  \rInt(\region{client}{\rid}(x,l)) \is
    \ap L(x,l) * 
    \bigl(l = 0 \dotimplies \locObl{\obl{k}}{\rid}\bigr)
$.

With the application of standard \tada\ reasoning,
it is possible to derive
\[
  \infer{
    \ATRIPLE |-
      \A l\in\set{0,1} \eventually \set{0}.
        <\ap{L}(\p{x},l)>\code{lock(x)}<\ap{L}(\p{x},1) \land l=0>
  }{
    \ATRIPLE |-
      \A l\in\set{0,1} \eventually \set{0}.
        <\emp|\region{client}{\rid}(\p{x},l)>
          \code{lock(x)}
        <\locObl{\obl{k}}{\rid} | \region{client}{\rid}(\p{x},1) \land l=0 >
  }
\]
which amounts to saying that
if \code{lock(x)} atomically locks the lock region, then 
it also atomically updates the client region containing the lock.
Notice that the $\set{0,1} \eventually \set{0}$ annotation is propagated as is.
In other words,
the update on the lock is put in the context of the current client.
In such context,
we can set the frame~$L$ to be
$ \exists l \in \set{0,1} \st
    \region{client}{\rid}(\p{x},l) * l=1 \dotimplies \envObl{\obl{k}}{\rid} $:
according to the protocol of the current client,
the environment holds an obligation $\obl{k}$ when $l=1$.
Because of the liveness invariant encoded by~$\obl{k}$,
it is true that the environment will always eventually unlock the lock,
allowing us to discharge the side condition of~\ref{rule:livec-simpl}:
\[
  \Always\bigl(
    \exists l \in \set{0,1} \st
      \region{client}{\rid}(\p{x},l) * l=1 \dotimplies \envObl{\obl{k}}{\rid}
  \bigr)
  \implies
  \AlwaysEv(\region{client}{\rid}(\p{x},0))
\]
Indeed, if~$l=1$ the precondition gives us~$\envObl{\obl{k}}{\rid}$ which means
that the environment owns~$\obl{k}$ and will therefore eventually fulfil it,
which can only be done by setting~$l=0$.
The environment is allowed to then lock \p{x} again, but that is fine:
as we discussed, a CLH lock can promise termination under this milder condition.

Thanks to the smooth integration of liveness annotations in the specifications
and liveness invariants expressed as obligations,
\tadalive\ proofs can properly abstract and encapsulate behaviour.
Consider a module implementing a counter which can be safely used
concurrently thanks to the internal use of locks to protect access
to the shared cell holding the value of the counter.
For example, the increment operation can be implemented as
\begin{center}
\code{def incr(x)\{var v in lock(x); v:=[x+1]; [x+1]:=v+1; unlock(x)\}}
\end{center}
While the use of locks involves blocking behaviour,
the blocking is handled completely internally and a client of the
counter cannot observe it.
The \tadalive\ specification of the increment operation thus  does not leak this
implementation detail:
\[
  \ATRIPLE |-
    \A n\in\Nat.
    <\ap{C}(\p{x},n)>{\code{incr(x)}}<\ap{C}(\p{x},n+1)>
\]
A client of the counter does  not need to worry about the internal
blocking since 
the specification does not entail any liveness proof obligation.
The proof of \p{incr}  discharges the liveness assumption
of the specification of \p{lock} by using obligations analogous to~$\obl{k}$
above, specified in an internal protocol that is not exposed to the client proof.
In \cref{sec:spec-encapsulation},  we discuss
the encapsulation properties of \tadalive's specifications in more
detail.

Our approach contrasts significantly with
previous work~\cite{BostromM15,Jacobs18toplas}
where blocking is represented in specifications by the acquisition of
tokens acting as obligations. In this work, the 
specification style fixes an expected protocol to be followed by the client.
For example, the axiom for a built-in lock acquisition operation
returns a built-in token representing the need for calling a
lock release primitive.

In contrast, our lock specification does not
impose on the client any particular way in which
its environment liveness assumption should be enforced.
It is the job of the client to
devise a protocol that ensures
the environment liveness assumptions of the lock specifications will be provable.
For locks, this is indeed often achieved by making sure every thread that
locks a lock eventually unlocks it.
Such a protocol is encoded by the liveness invariants of the
client's region (e.g.~$\rt[client]$ in the example above)
and the \pre\obl{k}-obligation pattern.
The specification of the lock, however, does not transfer obligations to the client, leaving open the possibility for clients to use completely different protocols.
The following example client illustrates the added flexibility
of our approach:
\[
\begin{parall}
  \begin{threadcode}[gobble=2]
  lock(x);
  [y] := 1
  \end{threadcode}
  \PARALLEL
  \begin{threadcode}[gobble=2]
  var b = 0 in
  while(b != 1){ b:=[y] };
  unlock(x)
  \end{threadcode}
\end{parall}
\]
The code assumes a lock has been allocated at~\p{x},
and~\p{y} initially stores~0.
In the specification style where the expected (liveness) protocol is built-in,
the \p{lock} call in the left-hand thread would return a built-in token
which can only be consumed by calling \p{unlock}.
This in turn requires an extension of the logic
---as done e.g.~in~\cite{Hamin019}---
providing some mechanism for the sound delegation of tokens from one thread to the other.
In \tadalive, there is no need for such an extension.
The protocol of this client does not need to associate obligations
with the lock;
one can simply define an obligation
(owned initially by the left-hand thread)
that is fulfilled when~\p{y} is set to~1,
and use it to prove the appropriate environment liveness conditions for
the proof.

\medskip

In this informal overview, we used temporal logic formulas to
represent the key liveness conditions in the
\ref{rule:while-blocking} and \ref{rule:livec-simpl} rules.
The formal versions of these \tadalive\ rules, however,
implement those checks with what we call the 
\emph{environment liveness condition},
which reduces these liveness properties to safety checks
via a dedicated set of rules
(explained in \cref{sec:rules}).
Remarkably, the liveness checks of both rules can
be phrased in terms of the environment liveness condition,
which therefore provides
a uniform proof principle for blocking termination.

\subsection{A Guide For the Reader}

The rest of the paper proceeds by introducing the assertion language and the semantic model of \tadalive\ in full detail (\cref{sec:model}),
then presenting the proof rules through the proof of an example (\cref{sec:rules}),
then walking through the proofs of our case studies and commenting on limitations (\cref{sec:evaluation}),
and finally discussing related work (\cref{sec:relwork}).
A  reader interested in the proof rules can skim through
\cref{sec:assertions,sec:protocols,sec:interpr-reific,sec:specs}
and the beginning of \cref{sec:specs-semantics}
to familiarise with the basic definitions, and then move to
\cref{sec:rules} to understand how the rules themselves work,
and the typical proof patterns. \section{The \tadalive\ Semantic Model}
\label{sec:model}

We introduce the semantic model which justifies \tadalive, defining:
\begin{itemize}
  \item the operational semantics of commands and their fair traces;
  \item the assertion language, regions, guards, obligations and protocols;
  \item the semantics of assertions and viewshifts;
  \item the specification format; and
  \item the trace semantics of specifications. 
\end{itemize}

In \cref{sec:overview},
we introduced hybrid triples which generalise 
Hoare and atomic triples. For our formal semantics, we separate
triples into two components:
  the command~$\cmd$, and
  the specification~$\spec$
  comprising the pseudo-quantifier, the precondition and the
  postcondition. We
introduce the
\emph{semantic} judgement,
$\JUDGE |= \cmd : \spec$ with
$
  \spec = \SPEC |=
    \A x\in X\eventually X'.< \na{P} | \at{P}(x) > < \na{Q} (x) | \at{Q}(x)>
    $,
which captures the semantic properties
of a command that satisfies a specification:
i.e. safety and termination of its fair traces.
This required a complete reformulation of the model of \tada.
First,
  we give a trace semantics to specifications independently of commands.
  This enables us to define the semantic judgement to hold when
  $\sem{\cmd} \subseteq \sem{\spec}$:
  that is, when the concrete traces of a command
  are allowed by the specification traces.
  This approach is unusual for separation logics based on Hoare-style triples,
  and brings the semantics nearer to approaches based on refinement.
Second,
  the trace model is an ``open-world'' semantics
  where traces include both individual local steps made by the command
  and individual arbitrary environment steps.
  Other models typically model the environment interference  indirectly,
 representing a  sequence of environment steps as a single big jump.
  Our ``open-world'' approach is crucial to capture  the assumptions on the liveness of the environment stipulated by the specifications.
Third,
  the trace semantics of the specification is given in a style that
  is closely related to alternating automata~\cite{Vardi95}.
  The specification is seen as an automaton which traverses a concrete trace
  and only accepts those traces that satisfy the specification.
  This enables us  to cleanly separate the (alternating) safety constraints
  from the (linear time) liveness constraints, imposed by a specification.

\subsection{Notation}
We write $X \pto Y$ for the set of partial functions from~$X$ to~$Y$,
and  $X \finpto Y$ for the set of finite partial functions.
Given $f \from X \pto Y$, we write $f(x) = \bot$  if~$f$ is undefined
on~$x$,  and
$\dom(f) \is \set{x | f(x) \neq \bot}$.
We write  $\map{x_1 -> y_1;;x_n -> y_n}$
for the finite function that maps each of the~$x_i$ to~$y_i$
and is undefined on any other input.
We write $f\map{x -> y}$ for the partial function which coincides with~$f$
except on~$x$ where it returns~$y$, and write $f\map{x->\bot}$ analogously.
The disjoint union between partial functions $ f\dunion g$ is defined
if their domains are disjoint.
In contexts where the expected type is a function,
we write~$\emptyset$ for the empty function.

\subsection{Fair Trace Semantics of Commands}
\label{sec:cmd-semantics}

We present a standard first-order imperative language, called \While,
with shared-memory concurrency and fine-grained non-blocking primitives,
and define the fair concrete trace semantics of its commands.
Our \While\ language is parametrised by the following sets:
the \emph{Booleans}, ${\Bool \is \set{\p{true},\p{false}} \ni b}$;
the \emph{values}, ${\Val \is \Int \cup \Bool \ni v}$;
the \emph{program variables}, $\PVar \ni \pvar{x}, \pvar{y}, \dots$;
and the \emph{function names}, $\FName \ni \pvar{f}$.
The set~$\PVar$ contains a special element, $\pvar{ret}$,
the name of a local variable that holds a function's return value.

\begin{mathfig}
  $
    \begin{array}[t]{r@{\;}c@{\;}l@{\hspace{1em}}l}
      \cmd & ::= &
            \acode{skip}                       & (\text{skip})\\
      & | & \acode{x:=EXP}                     & (\text{assignment})\\
      & | & \acode{x:=[EXP]}                   & (\text{read})\\
      & | & \acode{[EXP]:=EXP}                 & (\text{write})\\
      & | & \acode{x:=CAS(EXP,EXP,EXP)}        & (\text{compare-and-swap})\\
      & | & \acode{x:=FAS(EXP,EXP,EXP)}        & (\text{fetch-and-set})\\
      & | & \acode{x:=alloc(EXP)}              & (\text{allocate})\\
      & | & \acode{dealloc(EXP)}               & (\text{deallocate})\\
      & | & \acode{CMD;CMD}                    & (\text{sequential composition})\\
      & | & \cmd \parallel \cmd                & (\text{parallel composition}) \\
      & | & \acode{let f(VARS)=CMD in CMD}     & (\text{function definition})\\
      & | & \acode{var x = EXP in CMD}         & (\text{local variable binding}) \\
      & | & \acode{if(BEXP)\{CMD\}else\{CMD\}} & (\text{if}) \\
      & | & \acode{while(BEXP)\{CMD\}}         & (\text{while loop})\\
      & | & \acode{x:=f($\vec{\vexp}$)}        & (\text{function call})\\
      & | & \acode{<<CMD>>}                    & (\text{primitive atomic block})
      \\
      \vexp & ::= &
                 v
        \;\mid\; \pvar{x}
        \;\mid\; \vexp \mathrel{\p{+}} \vexp
\;\mid\; \vexp \mathrel{\p{*}} \vexp
        \;\mid\; \cdots
        & (\text{numeric expressions})
      \\
      \bexp & ::= &
        \ghost{v}{b}
        \;\mid\; \pvar{x}
        \;\mid\; \neg \bexp
\;\mid\; \vexp \leq \vexp
        \;\mid\; \cdots
        & (\text{Boolean expressions})
    \end{array}
  $
  \hfil
  $
    \begin{array}[t]{r@{\;}l}
      \llap{Domains:}&\\[1ex]
      v      \in \Val   & \is \Int \union \Bool
        \supseteq\Addr    \is \Nat \\
      \store \in \Store & \is \PVar \pto \Val \\
      h      \in \Heap  & \is \Addr \finpto \Val \\
      \functxt \in \FunCtxt &\is
            \FName \pto (\PVar^{*}, \Cmd) \\
      \PState\ni\pstate &
        \begin{array}[t]{@{}c@{\;}l}
             ::= &  \checkmark
            \;\mid\;  \cmd
            \;\mid\;  \pstate \parallel \pstate
            \;\mid\;  \pstate; \cmd \\
            \;\mid\; & (\store, \pstate)
            \;\mid\;  \acode{let f(VARS) = CMD in\~$\pstate$}
        \end{array}
      \\
      \pconf\in \Type{PConf} & \is
        (\Store \times \Heap \times \PState) \dunion
        \set{\fault}
    \end{array}
  $
  \caption{Syntax of commands $\cmd\in\Cmd$ and basic semantic domains.}
  \label{fig:commands}
\end{mathfig}

\begin{definition}[Commands]
  The \emph{set of commands}, $\Cmd \ni \cmd$, is defined by the grammar
  in \cref{fig:commands} where
    $\pvar{x} \in \PVar$,
    $\pvars{x} \in \PVar^*$ is a list of pairwise distinct variables, and
    $\p{f} \in \FName$.
  The notation  \acode{var x1,x2...,xn in CMD} denotes
  $
    \acode{var x1=0 in var x2=0 in ... var xn=0 in CMD}
  $.
\end{definition}

We place some restrictions on these commands to simplify exposition.
We write $\progvars(\cmd)$ for the free program variables of a command.
The set $\mods(\cmd)$ is the set of free variables that are
potentially modified by a command,
i.e.~any free \p{x} of~$\cmd$ appearing in instructions
of the form \acode{x:=...}; in particular,
$
  \mods(\acode{var x = EXP in CMD}) =
    \mods(\cmd) \setminus \set{\pvar{x}}
$.
In a command $\cmd_1 \parallel \cmd_2$,
we apply the mild syntactic restriction that
${\mods(\cmd_{1}) = \mods(\cmd_{2}) = \emptyset}$.
Each individual thread is still able to modify variables
that are created locally and to modify shared heap cells,
but are not allowed to modify the free variables.\footnote{To lift this restriction,
  one can use the ``variables as resources'' technique~\cite{BornatCY06}.
  Our restriction simplifies the handling of the local state
  without sacrificing expressivity:
  any local variable in the scope common to both threads that needs to be modified
  can be instead implemented by using a shared memory cell.
}
In a function definition
$\acode{let f(x_1,...,x$_n$)=CMD_1 in CMD_2}$,
we use the natural restriction
$\progvars(\cmd_1) \subseteq \set{ \pvar{x}_1,\dots,\pvar{x}_n, \pvar{ret}}$.
Also for simplicity, we assume each function name is given at most one
definition.
The function $\funnames\from\Cmd \to \powerset(\FName)$
returns the function names occurring in~$\Cmd$
that are not bound by a \code{let}.
Although function definitions may be recursive,
we will disallow recursion in our logical rules to simplify the development.
In the programs we consider,
all potentially divergent behaviour stems from \code{while}.
It is straightforward to reformulate the \ref{rule:while} rule
into a \omittedrulelabel{Let}{rule:let} rule
that supports terminating recursion.

Commands manipulate heaps
  $h \in \Heap \is \Addr \finpto \Val$
(where $\Addr \is \Nat$ and~$\emptyset$ is the empty heap)
and \emph{local variable stores},
  $\store \in \Store \is \PVar \pto \Val$.
A command can contain free function names,
so we use a \emph{function implementation context}
$
  \functxt \in \FunCtxt \is
    \FName \pto (\PVar^{*}, \Cmd)
$,
to map function names to
pairs comprising a finite list of distinct variables (the formal arguments)
and a command (the body of the function).

A command induces transitions over \emph{program configurations}
$
  \pconf\in\PConf \is (\Store \times \Heap \times \PState) \dunion \set{\fault}
$
which keep track of
  the current variable store and global heap, and
  the program states $ C\in\PState $ (see~\cref{fig:commands})
    which represent the set of the active threads and their execution state.
The~$\fault$ program configuration represents a faulty configuration,
  e.g. the one reached after dereferencing an unallocated address.
For the details of program states we refer to \cref{app:command-semantics};
what is relevant is that~$\checkmark$ is the program state of a terminated thread,
and we can define a function
$\threads \from \PConf \to \powerset(\TId)$
that computes the set of thread identifiers ($t\in\TId$)
of the active threads of a program configuration
(details in Appendix).

To model fair traces of commands,
we use a small-step operational semantics, parametrised by a function
implementation context~$\functxt$
and 
defined by a relation
$
  {\redto} \subseteq \PConf \times \Type{Sched} \times \PConf
$.
In a transition $(\pconf[1], {\schpl}, \pconf[2])$,
the scheduling annotation, ${\schpl} \in \Type{Sched}$,
keeps track of who executed the step:
\[
  \Type{Sched} \is
    \set{\LocOf{t} | t \in \TId}
      \dunion
    \set{\env}
\]
that is, either a local active thread~$t$ or the environment.
Environment steps can have arbitrary effects on the heap
and can generate faults at any time:
\[
  \infer{
    h'\in \Heap
  }{
    \store, h, \pstate \envstep \store, h', \pstate
  }
\qquad
  \infer{
    \pconf \in \PConf
  }{
    \pconf \envstep \fault
  }
\]
The full definition of the transition semantics
is defined in \cref{fig:oper-semantics} and
\cref{fig:oper-semantics-fail} of the Appendix.

We call \emph{program traces} the infinite sequences of the form
$
  \pconf[0] \schpl[0] \pconf[1] \schpl[1] \cdots
  $
where,
for all $i\in\Nat$,
    $\pconf[i] \in \Type{PConf}$ and 
    ${\schpl[i]} \in \Type{Sched}$.
We use~$\ptrace$ to range  over infinite suffixes of program traces
and~$\PTrace$ for the set of all program traces.
We define the \emph{set of \pre\functxt-program traces}
\[
\PTrace_{\functxt} \is
\set{ \pconf[0]\schpl[0]\pconf[1]\schpl[1] \cdots
  | \forall i\in\Nat\st
  \pconf[i] \step{\schpl[i]} \pconf[i+1]
}.
\]

\begin{definition}[Fairness]
\label{def:fairness-main}
  A \pre\functxt-program trace
    $ (\pconf[0]\schpl[0]\pconf[1]\schpl[1]\cdots) \in \PTrace[\functxt] $
  is \emph{fair} if:
  \begin{gather}
    \forall i \in \Nat \st
      \forall t \in \threads(\pconf[i]) \st
        \exists j \ge i\st
          ({\schpl[j]} = {\LocOf{t}} \lor
          \pconf[j] = \fault)
    \\
    \forall i \in \Nat \st
      \exists j \ge i \st {\schpl[j]} = {\env}
  \end{gather}
  That is: a trace is fair if, at any point in time,
  every thread that can take a step (and the environment)
  will eventually be scheduled.
\end{definition}

The open-world program semantics defines the behaviour of a command
when run concurrently with an arbitrary environment.
It corresponds to the fair program traces of a command,
with the information about program states and thread identifiers removed.

\begin{definition}[Open World Semantics]
\label{def:prog-semantics-main}
  We call \emph{traces}
  the infinite sequences $
    {\conf[0] \pl[0] \conf[1] \pl[1] \cdots}
  $
  where,
for all $i\in\Nat$,
      $\conf[i] \in \Conf \is (\Store \times \Heap) \union
      \set{\fault}$ and 
      ${\pl[i]} \in \set{\loc,\env}$.
  We use~$\trace$ for ranging over infinite suffixes of traces
  and~$\Trace$ for the set of all traces.
  For a trace $ \trace = \conf[0] \pl[0] \conf[1] \pl[1] \cdots $,
  we define $\trAt{i} \is (\conf[i],\pl[i])$,
and $\trFrom{i} \is \conf[i] \pl[i] \conf[i+1] \pl[i+1] \cdots$.
The function $\stripTr{\hole} \from \PTrace \to \Trace$ is defined by
  $
    \stripTr{\pconf[0]\schpl[0]\pconf[1]\schpl[1]\cdots}
      \is
        \conf[0] \pl[0] \conf[1] \pl[1] \cdots
  $
  where
  \begin{align*}
    \conf[i] &\is
      \begin{cases}
        (\store, h) \CASE \pconf[i] = (\store,h,\wtv)\\
        \fault      \CASE \pconf[i] = \fault
      \end{cases}
    &
    \pl[i] &\is
      \begin{cases}
        \loc \CASE {\schpl[i]} \in \Type{Sched}\setminus \set{\env}\\
\env \CASE {\schpl[i]} = {\env} \\
      \end{cases}
  \end{align*}

  The \emph{open-world program semantics function},
  $
    \sem{\hole}_{\functxt} \from \Cmd \to \powerset(\Trace)
  $
  is defined by
  \[
    \sem{\cmd}_{\functxt}
      \is
      \Set{\strut
        \stripTr{\pconf[0] \ptrace} |
          (\pconf[0]\ptrace) \in \PTrace[\functxt],
          \fv(\cmd) \subseteq \dom(\store_0),
          \pconf[0] = (\store_0,\wtv,\cmd),
          \pconf[0]\ptrace \text{ is fair}
      }
  \]
  The notation $\sem{\cmd}$ is syntactic sugar for $\sem{\cmd}_{\emptyset}$.
\end{definition}

The goal of \tadalive\ is to prove termination of the local command.

\begin{definition}[Local termination]
\label{def:loc-termination}
  A trace $\trace \in \Trace$ is \emph{locally terminating},
  written $\locterm(\trace)$,
  if it contains finitely many occurrences of~$\loc$.
\end{definition}

It might seem odd that our program semantics only contains infinite traces,
since our goal is proving termination.
Traces that locally terminate simply have an infinite tail of environment steps.
To simulate a closed system one can select for the traces
where the environment steps are all identity steps.

\begin{remark}[On primitive blocking]
  It is important to remember  that the primitives of our programming language
  are non-blocking, in the sense that they can always take a step if scheduled:
  for all $h \in \Heap, C \in \PState $,
  for all~$\store$ with $\dom(\store) \supseteq \progvars(C)$, and
  every $t \in \threads(C)$,
  there is a $\pconf \in \PConf$ such that
  $(\store,h,C) \step{\LocOf{t}} \pconf$.
  Hence, a trace is locally terminating only if all the threads terminated.

  For languages which have blocking primitives
  (e.g. built-in locks/channels),
  traces may be locally terminating because a non-terminated thread
may not have a local successor (i.e.~it is not enabled)
  at any point in the future
  (e.g.~if a built-in lock remains locked forever,
  an acquire operation would not have local successors).
  With blocking primitives, fairness also comes in two variants: strong and weak.
  Strong fairness requires that if an operation is infinitely often enabled
  it is infinitely often executed.
  Strong and weak fairness coincide for languages
  like ours where every primitive is enabled at all times.
  
Notice that our lack of blocking primitives does not make our setting less general:
  blocking primitives can be implemented on top of non-blocking ones,
  both with weak and strong fairness assumptions for termination,
  as illustrated by our spin and ticket lock examples.
  In other words, blocking primitives can be given \tadalive\ specifications
  and be treated uniformly by the logic.
The addition of built-in blocking primitives to the language
  does not pose new challenges.
\end{remark}

\subsection{\tadalive\ Assertions and Worlds}
\label{sec:assertions}

We formally introduce the \tadalive\ assertion language,
and its semantics in terms of its models,
called \emph{worlds}.
The \tadalive\ assertions
are built from the standard \emph{classical}
connectives and quantifiers of separation logic,\!\footnote{\tada\ interprets the separating conjunction intuitionistically.
  With \tadalive, we interpret it classically in order to  not
  lose information about the  obligations.
}
\tada\ region and guard assertions,
and new \tadalive\ obligation and layer assertions.
To formalise the assertions, we assume a number of basic domains:
\begin{itemize}
  \item a set of \emph{logical variables}, denoted~$\LVar$,
        disjoint from~$\PVar$;
  \item an enumerable set of \emph{region types}, $\RType \ni \rt$;
  \item an enumerable set of \emph{region identifiers}, $\RId \ni \rid$;
  \item the set of \emph{levels}, $\Level \is \Nat \ni \lvl$,
        to stratify regions to avoid the problem of re-entrancy\footnote{In Iris, levels roughly correspond to masks.}
        (explained in \cref{rem:levels});
  \item a set of \emph{abstract states}, $\AState \ni a$,
        including sets and lists of values;
  \item a set of \emph{guards}, $\Guard \ni G$,
        which will offer the support for the \emph{guard algebras} 
        defined later;
  \item a well-founded partial order
          $(\Layer, \layleq, \layTop, \layBot)$
        of \emph{layers},
        which will be associated to special guards called \emph{obligations};
        and
  \item a set of ordinals,~$\Ord$.
\end{itemize}

For layers, we use the abbreviations
$k_1 \laylt k_2 \is (k_1 \layleq k_2 \land k_2 \not\layleq k_1)$
and
$k \laygeqq n \is (\forall k'\laygt k \st k' \laygeq n)$.
The \emph{set of abstract values} is
$\AVal \is \Val \union \AState \union \Guard \union \RId \union
\Layer$.

As is standard,
when used in assertions,
we extend numeric and Boolean expressions
to use logical variables and abstract values too.
A \emph{logical variable store},
$l \in \LStore \is \LVar \pto \AVal$,
assigns values to logical variables.
Given a logical and a program variable store $l,\store$,
the evaluation of expressions
  $\esem{\vexp}(l,\store) \in \AVal$
and of Boolean expressions
  $\bsem{\bexp}(l,\store) \in \Bool$
are standard.

Assertions and worlds are built using partial commutative monoids.

\begin{definition}[PCM]
\label{def:pcm}
  A (multi-unit) \emph{partial commutative monoid} (PCM) is a tuple
    $(X,\bullet, E)$
  comprising  a set~$X$, a binary \emph{partial} composition operation
  $\bullet\from X \times X \pto X$
  and a set of unit  elements~$E$,
  such that the following axioms are satisfied
  (where either both sides are defined and equal, or both sides are undefined):
  \begin{alignat*}{3}
  \forall x,y,z \in X. &\quad&
    (x \bullet y) \bullet z &= x \bullet (y \bullet z)
      &\quad& \text{(associativity)}
  \\
  \forall x,y \in X. &&
    x \bullet y &=  y \bullet x
      && \text{(commutativity)}
  \\
  \forall x \in X.\exists e \in E. &&
    \quad x \bullet e &= x
      && \text{(identity)}
  \end{alignat*}

  \noindent
  For $x,y \in X$,
    we write $ x \compat y $
    if $ x \bullet y \neq \bot $,
  and $x \resleq y$ if $ \exists x_1\st y = x \bullet x_1$.
A PCM is \emph{cancellative} when,
  for any $x,y_1,y_2 \in X$,
  if $x \guardOp y_1 = x \guardOp y_2$ then $y_1 = y_2$.
\end{definition}

The partial heaps form a PCM $(\Heap, \dunion,  \set{\emptyset})$, as
standard in separation logics.
We also use  \emph{guard algebras}  and \emph{obligation
  algebras} which are PCMs for describing auxilary  ghost state,  specified by the user of the logic.

\begin{definition}[Guard Algebras]
\label{def:guard-alg}
  A \emph{guard algebra} is a PCM
  $(\Type{Grd},\guardOp,\set{\guardZero})$ with  $\Type{Grd} \subseteq \Guard$.
  \tadalive\ is parametrised by a function $\GuardsOf(\hole)$
  mapping a region type~$\rt$ to a guard algebra
  $\GuardsOf(\rt) = (\GuardsOf[\rt],\guardOp[\rt],\set{\guardZero[\rt]})$.
  The~$\rt$ subscript is omitted from $\guardOp[\rt]$ and $\guardZero[\rt]$ when
  it is clear from the context.
\end{definition}

As discussed, the 
obligations represent ghost state
for describing liveness invariants.
They form an \emph{obligation algebra} which is 
little more complicated to define due to the  association of
obligations  with layers.

\begin{definition}[Obligation Algebras]
\label{def:guard-alg}
  \tadalive\ is parametrised by a
  set of \emph{atoms} $\AObl$ and a
  \emph{layered obligation structure}:
  that is, a pair $(\Oblig, \lay)$
  where
    $\Oblig = \powerset(\AObl) \subseteq \Guard$
    and
    $\lay\from\Oblig \to \Layer$
    such that
    $ \forall O \in \Oblig\st \layBot \laylt \lay(O) \layleq \layTop $.
We will implicitly coerce
  atoms~$a \in \AObl$ into obligations~$\set{a} \in \Oblig$.
  An \emph{obligation algebra} is a guard algebra
  $(\Type{Obl},\oblOp,\set{\oblZero})$
  where $\Type{Obl} \subseteq \Oblig$,
  $\oblZero = \emptyset$,
  $\oblOp$ is union of disjoint sets and
  $
  \forall O_1,O_2 \in \Type{Obl}\st
  O_1 \resleq O_2 \implies
  \lay(O_1) \geq \lay(O_2)
  $.

  \tadalive\ is parametrised by a function $\ObligsOf(\hole)$
  mapping a region type~$\rt$ to an obligation algebra
  $\ObligsOf(\rt) = (\ObligsOf[\rt],\oblOp[\rt],\set{\oblZero[\rt]})$.
  The~$\rt$ subscript is omitted from
    $\oblOp[\rt]$ and $\oblZero[\rt]$
  when its clear from the context.
\end{definition}

In \cref{sec:overview},  we have seen examples of obligation algebras.
For  instance, the $\cmd_1'\parallel\cmd_2''$ example  used two atoms
$\obl{u}_1$ and $\obl{u}_2$, giving rise to the obligation algebra
with elements
$ \set{\obl{u}_1} $, $\set{\obl{u}_2}$, and $\set{\obl{u}_1,\obl{u}_2}$.
As mentioned, we make no difference between an atom $\obl{u}_1$ and the obligation $\set{\obl{u}_1}$ using the symbol of the former for both.
For our examples, it is enough to assign layers to atoms,
e.g.~$\lay(\obl{u}_1) \laylt \lay(\obl{u}_2)$, 
and extend the layers to obligations
by taking the minimum layer of the composed atoms,
for example $\lay(\obl{u}_1\oblOp\obl{u}_2) = \lay(\obl{u}_1)$. 
Note that, by construction, each obligation is incompatible with itself:
$ \oblUndef{O\oblOp O} $.

\begin{definition}[\tadalive\ Assertions]
 The \emph{set of \tadalive\ assertions},
 $\Assrt \ni P, Q, \dots$, is defined by the grammar in \cref{fig:assrt-syntax}.
 The only binder is~$\exists$.
 The function $\freevars \from \Assrt \to (\PVar\dunion\LVar)$
 returns the free variables of an assertion and its definition is standard.
 We also define
   $\progvars(P) \is \freevars(P)\inters\PVar$ and
   $\logvars(P) \is \freevars(P)\inters\LVar$.
 We write $P(x_1,\dots,x_n)$ to indicate that
 $\logvars(P) \subseteq \set{x_1,\dots,x_n}$
 and,  for $v_1,\dots,v_n \in \AVal$, write 
 $P(v_1,\dots,v_n)$ for $P\subst{x_1->v_1,,x_n->v_n}$.
\end{definition}

\begin{mathfig}
\adjustfigure[\small]
  \begin{grammar}
    P \is
        \bexp
      | \exists x \st P
      | \vexp \in X
| P \implies Q
      | P \land Q
| \emp
      | P * Q
      | \vexp \mapsto \vexp
      | \region[\lvl]{\rt}{\rid}(\vexp)
      | \done{\rid}{d}
\breakhere
      | \guardA{\vexp}{\rid}
      | \locObl{\vexp}{\rid}
      | \envObl{\vexp}{\rid}
      | \empObl[R]
      | \empObl[\lvl]
      | \minLay{\rid}{m}
\\
    d \is \blacklozenge | \lozenge | (\vexp,\vexp)
    \qquad\qquad
\rt \in \RType,
      \lvl\in\Level,
      \rid \in \RId \union\LVar,
      R \subseteq \RId,
      m \in \Layer.
  \end{grammar}\caption{Syntax of Assertions.
    Logical expressions,~$\vexp$, and
    logical Boolean expressions,~$\bexp$,
    are standard.
  }
  \label{fig:assrt-syntax}
\end{mathfig}

We summarise the intuitive meaning of our assertions before giving their formal semantics.
\begin{itemize}

  \item \tada\ \emph{region assertion} $ \region[\lvl]{\rt}{\rid}(a) $
    asserts the existence of a shared region with type~$\rt$,  identity~$\rid$,
    level~$\lvl$ and abstract state~$a$.
    Region assertions represent shared resources and, hence, are duplicable.
    We have
    $ {\vdash \region[\lvl]{t}{\rid}(a) \iff
            \region[\lvl]{t}{\rid}(a) * \region[\lvl]{t}{\rid}(a) }$.

  \item \tada\ \emph{atomicity tracking assertion}
    $ \done{\rid}{\blacklozenge} $
    gives permission to perform a single atomic change of the state
    of region~$\rid$. Once the change is performed,
    the assertion becomes $ \done{\rid}{(a_1,a_2)} $
    recording the abstract states just before and after the
    change (the linearization point).
    The assertion $\done{\rid}{\lozenge}$ asserts that the
    environment has the permission to do the atomic update.
    We have
    $ \vdash
      \done{\rid}{\blacklozenge} * \done{\rid}{\blacklozenge}
        \implies \False
    $,
    and
    $ \vdash
      \done{\rid}{\blacklozenge}
        {\iff}
      (\done{\rid}{\blacklozenge} * \done{\rid}{\lozenge})
    $.

  \item \tada\ \emph{guard assertion}  $ \guardA{G}{\rid} $
    asserts that the guard~$G$ is held locally.
    Guard composition is reflected by separation:
    $ {\vdash \guardA{G_1 \guardOp  G_2}{\rid}
            {\iff} \guardA{G_1}{\rid}  * \guardA{G_2}{\rid}} $.

  \item \tadalive\ \emph{local obligation assertion} $ \locObl{O}{\rid} $
    asserts that obligation~$O$ is held locally. We have
    $ \vdash \locObl{O_1 \oblOp O_2}{r} \iff \locObl{O_1}{r} *
    \locObl{O_2}{r} $. Separating conjunction is interpreted
    classically precisely  so that we do
    not lose local obligation information: that is,
    $ \vdash \locObl{O}{r} \not\implies \emp $.
    It is often
    useful to use the same guard algebra for guards and obligations.
    We  write $\gOblA{O}{r} \is \guardA{O}{r} * \locObl{O}{r}$.

  \item \tadalive\ \emph{environment obligation assertion}
    $ \envObl{O}{\rid} $ asserts that~$O$ is held by the
    environment: $ \vdash \envObl{O_1 \oblOp O_2}{r} \iff \envObl{O_1}{r} *
    \envObl{O_2}{r} $. Unlike for
    local obligations, it is possible to lose this information,  $
    \vdash \envObl{O}{r} \implies \emp $,
    because we do not need to keep track of the full obligations
    held by the environment, just a lower bound.
    The composition of environment and local obligation assertions is
    subtle,
    inspired by the subjective separation of~\cite{scsl}.
    The existence of the local obligation can be recorded in a frame:
         $ \vdash \locObl{O}{r} \iff \locObl{O}{r} * \envObl{O}{r}$.
    We also have the derived law $
     \vdash  \locObl{O_1 \oblOp O_2}{\rid}
      \iff
      (\locObl{O_1}{\rid} * \envObl{O_2}{\rid})
      *
      (\envObl{O_1}{\rid} * \locObl{O_2}{\rid})
    $, giving knowledge to each thread
    of the obligations delegated to the other.

  \item \tadalive\ \emph{empty obligation assertion} $ \empObl[R] $
    (resp.~$ \empObl[\lvl] $)
    asserts that no obligation is locally held for regions with
    identifiers in~$R$ (resp.~regions of level~$<\lvl$).

  \item \tadalive\  \emph{layer assertion} $ \minLay{\rid}{m} $
    asserts that the layer of the obligations held locally for region
    with identifier~$\rid$ is greater or equal than~$m$. We often use
    notation such as $ \minLay{\rid}{m} \leq m'$ to denote
    $ \minLay{\rid}{m} \wedge m \leq m'$.

\end{itemize}

We introduce the \emph{worlds} of \tadalive,  which are instrumented
heaps  providing  the models
of the  assertions of \tadalive.
A world  is a \emph{local} model in the sense that it  reflects
the state as seen from the perspective of a single thread.
It is built from a local heap, 
and a set of shared regions with associated guards and obligations.
Worlds are parametrised by a set of region identifiers~$\R$ which,
intuitively, are the regions which the current operation is supposed
to update abstractly  exactly once.
We say the regions in~$\R$ are \emph{tracked} for proving atomicity,
using special ghost state  given by the \emph{atomicity tracking algebra}
that supports the semantics of the atomicity tracking assertions.

\begin{definition}[Atomicity Tracking Algebra]
\label{def:atrack-pcm}
  The \emph{atomicity tracking algebra} is a PCM defined by
  $
    \ATrack \is \bigl(
        (\AState \times \AState) \dunion \set{\blacklozenge, \lozenge},
        \atrackOp,
        \ATEmp
      \bigr)
  $,
  where the composition is
  $
   \blacklozenge \atrackOp \lozenge = \blacklozenge = \lozenge \atrackOp \blacklozenge
  $,
  $
   \lozenge \atrackOp \lozenge = \lozenge
  $ and
  $
    \forall a,b \in \AState\st
      (a, b) \atrackOp (a, b) = (a, b)
  $
  (undefined otherwise), and the set of unit elements is
  $
    \ATEmp \is (\AState \times \AState) \dunion \set{\lozenge}
  $.
  The expression evaluation function is extended to map
  expressions~$d$ in the atomicity tracking assertions  to the corresponding elements of~$\ATrack$:
  $\esem{\blacklozenge}(\varsigma) = \blacklozenge$,
  $\esem{\lozenge}(\varsigma) = \lozenge$,
  $\esem{(\vexp_1,\vexp_2)}(\varsigma) = (\esem{\vexp_1}(\varsigma),\esem{\vexp_2}(\varsigma))$.
\end{definition}

\begin{definition}[Worlds]
Let $\R \subseteq \RId$.
A \emph{world}, $w \in \World[\R]$, is a tuple
$
  w = ( h, \regMap, \guardMap, \atomMap, \oblMap, \envMap )
$
where
\begin{itemize}
  \item $h \in \Heap$ is the local heap, i.e.~the cells owned locally;
  \item $
      \regMap \in \RMap \is
        \RId \finpto (\RType \times \Level \times \AState)
    $
    describes the shared regions;
  \item $\guardMap \in \GMap \is \RId \finpto \Guard$
    describes the local guards;
\item $\atomMap \in \AMap{\R} \is
            \R \to \ATrack$
    describes the local atomicity tracking components;
\item $\oblMap \in \ObMap \is \RId \finpto \Oblig$
    describes the local obligations;
\item $\envMap \in \ObMap \is \RId \finpto \Oblig$
  describes the  environment  obligations,
  known to be held locally by the  environment;
\end{itemize}
satisfying the following well-formedness constraints:
\begin{itemize}
  \item
    $
      \dom(\regMap) = \dom(\guardMap) =
      \dom(\oblMap) = \dom(\envMap)
      \supseteq \R
    $,
  \item
    $\forall \rid \in \RId \st$ if $\regMap(\rid) = (\rt,\wtv[2])$ then
    $\guardMap(\rid) \in \GuardsOf[\rt]$,
    $\oblMap(\rid) \in \ObligsOf[\rt]$,
    $\envMap(\rid) \in \ObligsOf[\rt]$,
  \item
    $
      \forall \rid \in \dom(\oblMap) \st
        \oblMap(\rid) \compat \envMap(\rid)
    $.
\end{itemize}
\end{definition}

A shared region with identifier~$\rid$, given by
  $\regMap(\rid) = (\rt,\lambda, a)$,
has type~$\rt$  and abstract state~$a$.
For a world~$w$, we write~$h_w$ and~$\regMap[w]$ and so on,
for the corresponding components of~$w$.
We also define
  $\rtype[w](\rid) \is \rt$,
  $\level[w](\rid) \is \lvl$, and
  $\astate[w](\rid) \is a$,
if $\regMap[w](\rid) = (\rt,\lvl,a)$.

We define a PCM on worlds
(called \emph{world algebra}).
We define how worlds compose
by first definining composition on each component of a world.
Heap composition is disjoint union.
Region maps only compose if they are equal.
Given $\regMap \in \RMap$, the compositions
  $\gmapOp[\regMap] \from \GMap \times \GMap \pto \GMap$ and
  $ \amapOp[\R] \from \AMap{\R} \times \AMap{\R} \pto \AMap{\R} $
are:
\begin{alignat*}{2}
  \guardMap[1]\gmapOp[\regMap] \guardMap[1] &\is
\fun r  \in\dom(\regMap).
      \guardMap[1](r) \guardOp[\rt] \guardMap[2](r)
        &\quad&\text{if }\forall r \in \dom(\regMap)\st
            \regMap(r)=(\rt,\wtv[2])
            \land
            \guardMap[1](r) \guardOp[\rt] \guardMap[2](r) \neq \bot
\\
  \atomMap[1] \amapOp[\R] \atomMap[2] &\is
\fun r \in\R.
      \atomMap[1](r) \atrackOp \atomMap[2](r)
        &&\text{if }\forall r\in\R\st
            \atomMap[1](r) \atrackOp \atomMap[2](r) \neq \bot
\end{alignat*}
and undefined otherwise.
The composition $\obmapOp[\regMap]$ on~$\ObMap$ is defined analogously to
$\gmapOp[\regMap]$ on~$\GMap$.

The local and environment obligation maps compose in a subtle way
inspired by the subjective separation of~\cite{scsl}.
To express this interaction,
we define a composition on \emph{pairs} of local/environment
obligation maps.
Given $\oblMap[1],\oblMap[2],\envMap[1],\envMap[2] \in \ObMap$,
we define
\[
  (\oblMap[1],\envMap[1])\obssep[\regMap](\oblMap[2],\envMap[2]) \is
  \begin{cases}
    (\oblMap[1]\obmapOp[\regMap]\oblMap[2], \envMap)
      \CASE
\envMap = \min_{\resleq} \set{
          \envMap |
            \envMap[1] \resleq (\oblMap[2] \obmapOp[\regMap] \envMap)
            \land
            \envMap[2] \resleq (\oblMap[1] \obmapOp[\regMap] \envMap)
        }
        \AND
        (\oblMap[1]\obmapOp[\regMap]\oblMap[2]) \neq \bot
    \\
    \bot \OTHERWISE
  \end{cases}
\]
Note that,  for obligation algebras, the minimum taken by the definition
is always unique if it exists.
Indeed, in general one can set
$
  \envMap(r) =
    \envMap[1](r)\setminus \oblMap[2](r)
      \union
    \envMap[2](r)\setminus \oblMap[1](r)
$.
For example, assuming
$\obl{a}$, $\obl{b}$, $\obl{c}$, $\obl{d}$, $\obl{e}$, and $\obl{f}$
are distinct atoms,
we have
\[
  \pars[\big]{\map{r->\obl{a}\oblOp\obl{b}},\map{r->\obl{c}\oblOp\obl{e}}}
  \obssep[\regMap]
  \pars[\big]{\map{r->\obl{c}\oblOp\obl{d}},\map{r->\obl{a}\oblOp\obl{f}}}
  =
  \pars[\big]{\map{r->\obl{a}\oblOp\obl{b}\oblOp\obl{c}\oblOp\obl{d}},
    \map{r->\obl{e}\oblOp\obl{f}}
  }
\]
provided the composition $\obl{a}\oblOp\obl{b}\oblOp\obl{c}\oblOp\obl{d}$
is defined.
Furthermore, this definition supports the implication
$
  \locObl{O}{r} \implies \locObl{O}{r} * \envObl{O}{r}
$ since
$
  \pars[\big]{\map{r->O},\map{r->\oblZero}}
  \obssep[\regMap]
  \pars[\big]{\map{r->\oblZero},\map{r->O}}
  =
  \pars[\big]{\map{r->O},\map{r->\oblZero}}
$.

\begin{definition}[World Algebras]
The PCM of \emph{world algebras},
  $(\World[\R], \worldSsep, \wEmp[\R])$,
is defined by  the set of worlds $\World[\R]$,
\begin{itemize}[label=--]
\item the \emph{subjective world composition},~$\worldSsep$, given by:
  \[
    ( h_1, \regMap[1], \guardMap[1], \atomMap[1],
      \oblMap[1], \envMap[1] )
    \worldSsep
    ( h_2, \regMap[2], \guardMap[2], \atomMap[2],
      \oblMap[2], \envMap[2] )
    =
    ( h_1 \dunion h_2,
     \regMap,
     \guardMap[1] \gmapOp[\regMap] \guardMap[2],
     \atomMap[1] \amapOp[\R] \atomMap[2],
     \oblMap,
     \envMap
    )
  \]
  if
    $h_1 \compat h_2$,
    $\regMap = \regMap[1] = \regMap[2]$,
    $\guardMap[1] \gmapOp[\regMap] \guardMap[2] \neq \bot$,
    $\atomMap[1] \amapOp[\R] \atomMap[2] \neq \bot$, and
    $(\oblMap[1], \envMap[1]) \obssep[\regMap] (\oblMap[2],\envMap[2])
      = (\oblMap,\envMap)$,
    undefined otherwise; and

\item the set of unit elements given by:
  \[
    \wEmp[\R] \is
      \Set{(\emptyset, \regMap, \guardMap, \atomMap, \oblMap, \envMap)
        \in \World[\R]
        | \begin{array}{l}
            \forall \rid \st
              \regMap(\rid) = (\rt,\wtv[2])
              \implies
                \guardMap(\rid) = \guardZero[\rt] \land
                \oblMap(\rid) = \oblZero[\rt],
            \\
            \forall \rid \in \R \st
              \atomMap(\rid) \in \ATEmp
          \end{array}
        }
  \]
\end{itemize}
\end{definition}
Notice that the units are  worlds with arbitrary shared regions,
atomicity components from~$\ATEmp$, and arbitrary environment obligations.

Subjective composition of worlds~$(\worldSsep)$ is lifted to
composition of sets of worlds~$(*)$,
defined as
${p_1 * p_2 \is
  \set{ w_1 \worldSsep w_2 | w_1 \in p_1, w_2 \in p_2, w_1 \compat w_2 }}$.

We want the region and environment obligations assertions
to enjoy the elimination rule,
e.g.~$ {\region[\lvl]{t}{\rid}(a) * Q \implies Q} $.
Assertions therefore denote sets of worlds that are upward-closed
with respect to adding regions and adding environment obligations.
Formally, we define the \emph{world ordering} $ \wleq[\R] $ as the smallest
reflexive and transitive relation such that:
\begin{align*}
  (h,
    \regMap,
    \guardMap,
    \atomMap,
    \oblMap,
    \envMap)
  &\wleq[\R]
  (h,
    \regMap \map{ \rid -> (\rt, \lvl, a)},
    \guardMap \map{ \rid -> \guardZero[\rt]},
    \atomMap,
    \oblMap \map{ \rid -> \oblZero[\rt]},
    \envMap \map{ \rid -> \oblZero[\rt]})
  && \rid \not\in \dom(\regMap)
\\(h,
    \regMap,
    \guardMap,
    \atomMap,
    \oblMap,
    \envMap)
  &\wleq[\R]
  (h,
    \regMap,
    \guardMap,
    \atomMap,
    \oblMap,
    \envMap\map{r->\envMap(r)\oblOp O})
  && \oblMap(r) \compat O \compat \envMap(r)
\end{align*}
The upward-closed sets of worlds
$
  \UCWorld[\R] \is \Set{
    p \subseteq \World[\R] |
      \forall w,w'\st
        w \wleq[\R] w' \land w \in p
          \implies w' \in p
  }
$ are the semantic domain of our assertions.

\begin{mathfig}
  \adjustfigure[\small]
  \def\iff{\textit{\quad iff\quad }}
  \begin{align*}
    \SAT \varsigma,w |= \emp & \iff
      w \in \wEmp[\R]
    \\
    \SAT \varsigma, w |= \bexp & \iff
      \bsem{\bexp}(\varsigma)
    \\
    \SAT \varsigma, w |= \vexp_1 \mapsto \vexp_2 & \iff
      \heap[w] = \map{{\esem{\vexp_1}(\varsigma)} -> {\esem{\vexp_2}(\varsigma)}}
      \land
      (\emptyset, \regMap[w], \guardMap[w], \atomMap[w], \oblMap[w], \envMap[w]) \in \wEmp[\R]
\\
    \SAT \varsigma, w |= P \implies Q & \iff
    \forall w' \st w \wleq[\R] w' \land
      \SAT \varsigma, w' |= P \;\implies\; \SAT \varsigma, w' |= Q
    \\
    \SAT \varsigma, w |= \exists x \st P & \iff
      \exists v \in \AVal \st \SAT \varsigma[x \mapsto v], w |= P
    \\
    \SAT \varsigma, w |= \vexp \in X & \iff
      \esem{\vexp}(\varsigma) \in X
    \\
    \SAT \varsigma, w |= P_{1} \land P_{2} & \iff
      (\SAT \varsigma, w |= P_1) \land (\SAT \varsigma, w |= P_2)
    \\
    \SAT \varsigma, w |= P_{1} * P_{2} & \iff
      \exists w_{1}, w_{2}\st
        w = w_{1} \worldSep w_{2} \land
        (\SAT \varsigma, w_{1} |= P_1) \land (\SAT \varsigma, w_{2} |= P_2)
\\
    \SAT \varsigma, w |= \region[\lvl]{\rt}{\rid}(\vexp) & \iff
      \regMap[w](\esem{\rid}(\varsigma)) = (\rt, \lvl, \esem{\vexp}(\varsigma)) \land
      w \in \wEmp[\R]
    \\
    \SAT \varsigma, w |= \done{\rid}{d} & \iff
      \atomMap[w](\esem{\rid}(\varsigma)) = \esem{d}(\varsigma) \land
      (\heap[w], \regMap[w], \guardMap[w], \atomMap[w]\map{\esem{\rid}(\varsigma) -> \lozenge}, \oblMap[w], \envMap[w]) \in \wEmp[\R]
    \\
    \SAT \varsigma, w |= \guardA{\vexp}{\rid} & \iff
      \guardMap[w](\esem{\rid}(\varsigma)) = \esem{\vexp}(\varsigma) \land
      (\heap[w], \regMap[w], \guardMap[w]\map{\esem{\rid}(\varsigma) -> \guardZero}, \atomMap[w], \oblMap[w], \envMap[w]) \in \wEmp[\R]
    \\
    \SAT \varsigma, w |= \locObl{\vexp}{\rid} & \iff
      \oblMap[w](\esem{\rid}(\varsigma)) = \esem{\vexp}(\varsigma) \land
      (\heap[w], \regMap[w], \guardMap[w], \atomMap[w], \oblMap[w]\map{\esem{\rid}(\varsigma) -> \oblZero}, \envMap[w]) \in \wEmp[\R]
    \\
    \SAT \varsigma, w |= \envObl{\vexp}{\rid} & \iff
      \esem{\vexp}(\varsigma) \resleq \envMap[w](\esem{\rid}(\varsigma))
      \land w \in \wEmp[\R]
    \\
    \SAT \varsigma, w |= \empObl[R] & \iff
      \forall \rid\in R\st
        \regMap[w](r) = (\rt,\wtv[2]) \implies \oblMap[w](r) = \oblZero
    \\
    \SAT \varsigma,w |= \empObl[\lvl] & \iff
      \forall \rid,\lvl'<\lvl\st
        \regMap[w](r) = (\rt,\lvl',\wtv) \implies \oblMap[w](r) = \oblZero
    \\
    \SAT \varsigma, w |= \minLay{\rid}{m} & \iff
      \lay(\oblMap[w](\esem{\rid}(\varsigma))) \geq m
  \end{align*}
  \caption{Definition of assertion satisfaction.}
  \label{fig:assert-sat-def}
\end{mathfig}

\begin{definition}[Satisfaction Relation]
Let $\varsigma \from (\PVar \dunion \LVar) \pto \AVal$
be the union of a program and logic variable store.
For a world $w \in \World[\R]$ and an assertion~$P$,
the \emph{assertion satisfaction} relation,
  $\SAT \varsigma, w |= P$,
is defined in \cref{fig:assert-sat-def}.
\end{definition}
We write $ \VALID \R  |= P $ if,  for
$ \forall \varsigma \from (\PVar \dunion \LVar) \pto \AVal, w\in\World[R]$
we have
$
  \SAT[\R] \varsigma,  w |= P $, and write $
  \WorldSem{\R}{\varsigma}{P} \is
    \Set{ w | \SAT[\R] \varsigma , w |= P }
$ for any assertion~$P$.
It is easy to check that $\WorldSem{\R}{\varsigma}{P} \in \UCWorld[\R]$
for every~$P$ and~$\varsigma$.

\subsection{Protocols: Interference and World Rely}
\label{sec:protocols}

A world describes the state of the current thread, both the local
state owned by the thread (the heap, guards, local obligations and
atomity tracking components), the shared state (the regions) and the
environment obligations describing obligations owned locally by the
environment.  We define the \emph{world rely relation} which describes
how the world may change as a result of the ``well-behaved''
interference of the environment characterised by the region
interference relations, the atomicity tracking components and the
environment obligations.
To define the world rely,
we need to introduce two other components of \tadalive:
the region protocols, expressed by the \emph{region interference function}, and
atomicity contexts.

The type of each region is associated with a
\emph{region interference function}
which establishes which updates to a shared region
are allowed by the owner of which guards.

\begin{definition}[Region Interference]
  \label{def:interference}
  \tadalive\ is  parametrised by
  the \emph{region interference function},~$\regLTS$,
  which takes a region type $\rt \in \RType$ and returns a function
  $
    \regLTS[\rt] \from
      \GuardsOf[\rt] \to
      \powerset((\AState \times \ObligsOf[\rt]) \times
                (\AState \times \ObligsOf[\rt]))
  $.
  Every function $\regLTS[\rt]$ is required to satisfy three properties:
  \begin{itemize}
    \item reflexivity:
      $((a,\oblZero[\rt]),(a,\oblZero[\rt]))\in\regLTS[\rt](\guardZero[\rt])$,
      for all $a\in\AState$;
    \item monotonicity in the guards:
      $
        \forall G_1,G_2 \in \GuardsOf[\rt] \st
          G_1 \resleq G_2 \implies \regLTS[\rt](G_1) \subseteq \regLTS[\rt](G_2)
      $;
    \item closure under obligation frames:
      for all $O_1,O_2, O \in \ObligsOf[\rt]$,
        if $ ((a_1,O_1), (a_2,O_2)) \in \regLTS[\rt](G) $
           and $  O_1 \compat O $ and $  O_2 \compat O $,
        then $ ((a_1, O_1\oblOp[\rt]  O ), (a_2, O_2 \oblOp[\rt]  O)) \in \regLTS[\rt](G) $.
  \end{itemize}
  We write $ \regLTS[\rt](\wtv) $ for
  $ \Union_{G\in\GuardsOf[\rt]} \regLTS[\rt](G) $.
For any
    $T \subseteq (\AState \times \Oblig) \times (\AState \times \Oblig)$,
  we write
    $\iorel(T) \is \set{ (a,b) | ((a,\wtv),(b,\wtv)) \in T }$.
  \end{definition}

The final concept we need before  introducing  the  world rely relation  is the
\emph{atomicity context},~$\actxt$. 
In \tadalive\ proofs, we keep in the context of the judgment information
about which updates we are currently proving are abstractly atomic.
The rule driving this bookkeeping is the \ref{rule:make-atomic} rule.
Although we will properly explain the rule in \cref{sec:other-rules},
we sketch the main idea as a motivation for the atomicity context now.
The relevant ``skeleton'' of the rule is as follows:
\begin{inlineproofrule}
  \specsdelim{big}
  \infer*{
    \rid \notin \dom(\actxt)
    \\
    T \subseteq \regLTS_{\rt}( G )
    \\
    R =  \iorel(T)
    \\
    \dots
    \\\\
    \TRIPLE  m; \lvlp;
      {\actxt\map{r -> (X,k,X',T)}} |-
      {\exists x \in X \st
        \region[\lvl]{\rt}{\rid}(x)
        * \done{\rid}{\blacklozenge}}
      \cmd
      {\exists x, y \st
       R   (x,y) \land \done{\rid}{(x,y)}}
  }{
    \ATRIPLE m; \lvlp ;
    \actxt |-
      \A x \in X \eventually[k] X'.
        < \region[\lvl]{\rt}{\rid}(x) * \guardA{G}{\rid} >
          \cmd
        < \exists y\st \region[\lvl]{\rt}{\rid}(y) *
        \guardA{G}{\rid} \land R (x,y) >
  }
\end{inlineproofrule}
The judgments include the context information  such as the layer $m$,
the 
level $\lvlp$ and the atomicity context
    $\actxt $,
    and the pseudo-quantifier includes a layer~$k$. 
We formally introduce these details in \cref{sec:specs}.
Here, we focus on motivating  the  use of the atomicity context~$\actxt$.
This rule describes how an update to the state of a region~$\rid$
can be declared atomic even if it was realised through a series of steps.
It does this by converting a Hoare triple to an atomic triple,
provided the Hoare triple bears evidence
(through the atomicity tracking assertions of the premise)
that, although many steps might have been taken,
the abstract state was changed by the command exactly once.
The atomic triple may constrain the environment interference
with a non-trivial pseudo-quantifier.
The proof of the premise in general needs to make use of these assumptions
about the environment, but the conversion to a Hoare triple
means we cannot use pseudo-quantification to represent them.
These assumptions are instead made available
to the proof of the Hoare triple
using the atomicity context,
that records the $(X,k,X')$ information from the pseudo-quantifier
and the relation~$T$ which stores  the update that we are proving happens atomically.

\begin{definition}[Atomicity Context]
  An \emph{atomicity context}~$\actxt$ is a finite partial function
  from~$\RId$ to tuples of the form
  $ (X,k,X',T) $
  where
  $ X, X' \subseteq \AState $,
  $ k \in \Layer $, and
  $ T \subseteq (\AState \times \Oblig) \times (\AState \times \Oblig) $
is closed under obligation frames (as in Definition~\ref{def:interference}).
\end{definition}

\noindent Assuming $\actxt(\rid) = (X,k,X',T)$, we write
  $\safe(\actxt, \rid) \is X$,
  $\good(\actxt, \rid) \is X'$,
  $\live(\actxt, \rid) \is (X,k,X')$
    which we write ${X \eventually[k] X'}$, and
  ${\trrel(\actxt, \rid) \is T}$.
  For every $\rid\in\dom(\actxt)$, we require
  $ \set{ x | (x,\wtv) \in \iorel(T)} \subseteq \safe(\actxt, \rid)$.
The set $\dom(\actxt)$ declares  the
regions for which we are tracking atomicity:
for $\rid\in\dom(\actxt)$, the environment will only change the
abstract state within $\safe(\actxt,\rid)$ and will obey the liveness
condition given by $\live(\actxt, \rid)$ that the environment will
always eventually return to a good state in $\good(\actxt, \rid) \is X'$;
and the local thread will only change the abstract state
at most once according to the relation $\iorel(\trrel(\actxt,\rid))$.
We write
$\vDash_\actxt$ for $\vDash_{\dom(\actxt)}$, and similarly for
$ \VALID \actxt |= $,
$\WorldSem{\actxt}{\varsigma}{P}$,
$\World[\actxt]$
and $\wEmp[\actxt]$.

\begin{definition}[World Rely]
  The \emph{world rely relation},
  $
    {\rely[\actxt]} \subseteq \World[\actxt] \times \World[\actxt]
  $,
  is the smallest reflexive and transitive relation satisfying
  the rules in \cref{fig:world-rely}.
\end{definition}

\Cref{rule:rely-interf} describes  the case where the environment
can update the abstract state of a region according to the
interference relation $\glts_{\rt}$. Notice that, for this rule, when $  \atomMap(\rid)
\in \set{\blacklozenge, \lozenge}$, the environment can only change
the abstract state to something in $\safe(\actxt,\rid)$. When
$\atomMap(\rid)$ is undefined  or a pair of abstract states, then the environment does not have this
restriction and can do any update consistent with $\glts_{\rt}$. Also, notice how the environment obligations map~$\envMap$ is  affected by the
transition.
\Cref{rule:rely-linpt} describes  the case where
the atomic update given by~$\actxt$ has been delegated to the environment
($\atomMap \map{\rid -> \lozenge}$)
in which case the current thread can observe the abstract state change
corresponding to the update.

\begin{mathfig}
  \let\RightTirName\RuleName \begin{mathpar}
    \infer*[Right=wr$_1$]{
      \guardMap(\rid) \compat G
      \\
      ((a_1,O_1), (a_2,O_2)) \in \glts_{\rt}(G)
      \\
      \atomMap(\rid) \in \set{\blacklozenge, \lozenge}
        \implies a_2 \in \safe(\actxt,\rid)
      \\
      O_2 \compat \oblMap(r)
    }{
      (h,
        \regMap \map{ \rid -> (\rt, \lvl, a_1)},
        \guardMap,
        \atomMap,
        \oblMap,
        \envMap\map{ \rid -> O_1 })
      \rely[\actxt]
      (h,
        \regMap \map{ \rid -> (\rt, \lvl, a_2)},
        \guardMap,
        \atomMap,
        \oblMap,
        \envMap\map{ \rid -> O_2 })
    }
    \label{rule:rely-interf}
  \\\infer*[Right=wr$_2$]{
      ((a_1 ,O_1), (a_2,O_2)) \in \trrel(\actxt,\rid)
      \\
      O_2 \compat \oblMap(r)
    }{
      (h,
        \regMap \map{ \rid -> (\rt, \lvl, a_1)},
        \guardMap,
        \atomMap \map{ \rid -> \lozenge},
        \oblMap,
        \envMap \map{ \rid -> O_1 })
      \rely[\actxt]
      (h,
        \regMap \map{ \rid -> (\rt, \lvl, a_2)},
        \guardMap,
        \atomMap \map{ \rid -> (a_1, a_2)},
        \oblMap,
        \envMap \map{ \rid -> O_2 })
    }
    \label{rule:rely-linpt}
  \end{mathpar}
  \caption{World Rely rules}
  \label{fig:world-rely}
\end{mathfig}

\smallskip

So far, we have introduced assertions, and worlds as their models.
These structures express information mostly over \emph{ghost state},
that is, state that is purely logical and has no representation
in concrete executions.
For example, the notion that there is some shared region is purely fictional,
as in the concrete machine there is no special way to mark
a portion of the heap as shared.
We introduced interference protocols and the world rely,
as a way to specify the expected well-behaved transformations
shared resources may be subjected to.
Since well-behaved interference from the environment can change the state of
shared regions, a single world (describing a single state for each region)
cannot capture the logical state we may be in,
when interleaved with environment actions.
Views are the sets of worlds that can explain the logical state we may be in
after being suspended for an arbitrary number of environment steps.
Views represent information about the logical state,
that cannot be invalidated by a well-behaved environment.

\begin{definition}[Views, Stability]
\label{def:view}
\label{def:stable}
  An upward-closed set of worlds, $p \in \UCWorld[\actxt]$, 
  is an \emph{\pre\actxt-view}
  if it is closed under~$\rely[\actxt]$:
  that is, $
   \forall w \in p, w' \in \World_{\actxt} \st w \rely[\actxt]  w' \implies w' \in p
  $.
An assertion~$P$ is \emph{\pre\actxt-stable}, written
  $
    \STABLE \actxt |= {P}
  $,
  if and only if, for all $\varsigma \from (\PVar \dunion \LVar) \pto \AVal$,
  $ \WorldSem{\actxt}{\varsigma}{P} $ is an \pre\actxt-view.
\end{definition}

\noindent We write
  $\View[\actxt]$ for the set of all \pre\actxt-views and
  $\Stable[\actxt]$ for the set of all \pre\actxt-stable assertions.

\begin{definition}[View Algebra]
  The PCM of \emph{view algebras},
    $ (\View[\actxt],  *, \set{\wEmp[\actxt]}) $,
  is formed from the set $\View[\actxt]$,
  and the composition
  $p_1 * p_2 \is
    \set{ w_1 \worldSsep w_2 | w_1 \in p_1, w_2 \in p_2, w_1 \compat w_2 }$.
\end{definition}

Notice that the composition of views always gives rise to a view:
in the case where there are no compatible pairs of worlds in the views,
the result is the empty view (the denotation of~$\False$).

\paragraph{On checking stability.}
\tadalive\ proofs require checking
stability of assertions in some crucial steps.
The notion of stability of \cref{def:stable} is given in terms of the semantics of assertions, but it is possible, in principle,
to provide a set of lemmas to prove stability of common cases without reasoning at the level of the model.
For example, any traditional separation logic assertion
(such as~$\emp$, $x\mapsto v$, pure formulas)
is always stable; guard and local obligation assertions are also automatically stable; stability is preserved by~$*$,~$\land$,~$\lor$, and existential quantification.
The crucial sources of instability are
  region assertions, environment obligation assertions,
  and $\done{\rid}{\lozenge}$.
Stability of the first two can be established
by inspecting the protocol of regions.
A rule that would be expressive enough to prove most stability checks
for our examples is:
\begin{mathpar}
  \infer*{
    \forall x\in X, x',G',O'\!\st
      (G' \compat G(x)) \land
      ((x,O(x)),(x', O')) \in \regLTS[\rt](G')
        \implies
          x' \in X \land O' = O(x')
  }{
    \STABLE \actxt |= {
      \exists x\in X.
           \region[\lvl]{\rt}{\rid}(x)
         * \guardA{G(x)}{\rid}
         * \envObl{O(x)}{\rid}
    }
  }
\end{mathpar}
It is similarly easy to extract from \cref{fig:world-rely} rules
involving the atomicity context information:
\begin{mathpar}
  \infer*{
\safe(\actxt, \rid) = X
  }{
    \STABLE \actxt |= {
      \exists x\in X. \region[\lvl]{t}{\rid}(x)
        * \done{\rid}{\blacklozenge}
    }
  }
  \and
  \infer*{
    \rid\in\dom(\actxt)
  }{
    \STABLE \actxt |= {
      \done{\rid}{\lozenge} \lor
      \done{\rid}{(\wtv[2])}
    }
  }
\end{mathpar}

\subsection{Linking Levels of Abstraction: Interpretations and Reification}
\label{sec:interpr-reific}

As we mentioned, worlds and views represent ghost information about state.
Ultimately, however, we want to use this information to express properties
of concrete execution traces.
To do so, we need to formalise the link between
worlds with their logical instrumentation and  concrete states comprising variable stores and heaps.
The first component that contributes to this link is   a 
\emph{region interpretation},
which specifies  the implementation-dependent content of a shared
region: for example, for a shared spin lock, the interpretation of 
the abstract shared region $\region[\lvl]{spin}{\rid}(x,l)$ is the
view given by $x\mapsto l$, 
a single cell storing~$l$ at~$x$. 

\begin{definition}[Region Interpretation]
\label{def:interpretation}
  \tadalive\ is parametrised by
  a \emph{region interpretation function}
  $
    \rIntSem[\rt]{\hole} \from
      \RId \times \Level \times \AState \to \View_{\emptyset}
  $
  for each $\rt\in\RType$,
  such that, for every
    $\rid \in \RId$,
    $\lvl \in \Level$,
    $a \in \AState$,
  $
    \forall w \in \rIntSem[\rt]{\rid,\lvl,a}\st
      \forall \rid'\in\dom(\oblMap[w])\setminus\set{\rid}\st
        \oblMap[w](\rid') = \oblZero
  $.
  We also require the interpretation to be \pre\lvl-safe,
  a technical condition explained in \cref{sec:rules-atomic}
  that is usually immediate to check (see \cref{lemma:lsafe}).
  A region interpretation's companion is the \emph{syntactic region interpretation}
  $
    \rInt[\rt] = (r, l, a, P)
  $
  where $r,l,a \in \LVar$,
    $\freevars(P) \subseteq \set{r,l,a}$,
    $\STABLE \emptyset |= P $, and
    $ \VALID \emptyset |=
        P\subst{l->\lvl} \implies \empObl[\RId\setminus\set{r}] $.
  We write $ \rInt(\region[\lvl]{t}{\vexp_1}(\vexp_2)) $ for
    $ P\subst{\rid->\vexp_1,l->\lvl,a->\vexp_2} $.
  We require that
  $
    \rIntSem[\rt]{\rid,\lvl,a} =
      \WSem[\emptyset]{\rInt(\region[\lvl]{\rt}{\rid}(a))}
  $;
  in practice, we will define region interpretations
  by writing syntactic interpretations
  and using the previous equation
  as a definition for the corresponding region interpretation functions.
\end{definition}

It is important to understand that interpretations are not merely an indirect
way of writing assertions.
In our spin lock example, the crucial difference between the two assertions
$\region[\lvl]{spin}{\rid}(x,l,\alpha)$ and $x\mapsto l$
is that the first is subjected to interference,
while the latter expresses ownership of the cell at~$x$.
The requirement that the interpretation of some region with id~$r$
must imply $\empObl[\RId\setminus\set{r}]$ forbids an intepretation
to own local obligations of other regions.
This is necessary for soundness:
if we removed the restriction, we could fool ourselves into thinking that
we fulfilled an obligation $\locObl{O}{\rid}$ by
creating another region with the obligation in its interpretation.

\begin{remark}[On ``opening'' regions and levels]
\label{rem:levels}
As in \tada, the region interpretation is used to ``open'' a region:
that is,~import the region interpretation as local state
in order to do a single atomic update.
The idea is to obtain instantaneously the ownership
of the content of the region for the atomic update,
and to re-establish the region interpretation
for the updated abstract state,
before immediately relinquishing ownership by ``closing'' up the region.
As in \tada, this opening and closing mechanism depends on
the level of the region, which is a device to avoid inconsistencies.
With a specification at level~$\lvl$, the rules enable a region  to be
opened at level~$\lvl' < \lvl$ to obtain a resulting specification at
level~$\lvl'$.  This means that, although a region can be shared
($ \vdash \region[\lvl']{t}{r}(a) \Leftrightarrow
\region[\lvl']{t}{\rid}(a)*\region[\lvl']{t}{\rid}(a) $),
it cannot be \emph{opened} twice,
which would result in
$\rInt(\region[\lvl']{t}{\rid}(a))*\rInt(\region[\lvl']{t}{\rid}(a))$
with a potential contradictory duplication of non-duplicable resources.
\end{remark}

The second  component that expresses the link between the  instrumented 
worlds and   concrete states  is the \emph{reification} function.
Reification has two main purposes.
First, at level~$\lvl$, all the regions with level lower than~$\lvl$ are
closed, which means that the resources in their interpretation
do not exist as far as the world is concerned.
The concrete heap cells accounted for inside these interpretations
will however correspond to cells in the concrete heap.
To bridge this gap, the reification opens all closed regions
importing the resources in their interpretation as local resources,
obtaining a ``collapsed'' world.
Second, all the ``ghost'' components of collapsed worlds
(like regions, guards or obligations)
do not have any representation in memory so reification erases them.

\begin{definition}[Reification]
\label{def:reification}
  Let $\lvl\in\Level$ and let
  $
    \clreg(\lvl,w) \is
      \set{ \rid \in \RId | \level[w](\rid) < \lvl }
  $.
The \emph{region collapse function},
  $ \collapse{\lvl}{(\hole)}\from \World[\actxt] \to \powerset(\World[\actxt])$,
  is defined~by:
  \[
    \collapse{\lvl}{w_0} \is
      \Set{ w_{0} \worldSsep w_{1} \worldSsep \ldots \worldSsep w_{n} |
        \begin{array}{l}
          \clreg(\lvl,w_0) = \set{\rid_{1}, \dots, \rid_{n}},\\
\regMap[w_0](r_i) = (\rt_i, \lvl_i, a_i),
            w_i \in \rIntSem[\rt_i]{r_i,\lvl_i,a_i},\\
            \forall i\leq n \st
            \forall r\in\dom(\envMap[w_0\worldSsep\dots\worldSsep w_i]) \st \envMap[w_0\worldSsep\dots\worldSsep w_i](r) = \oblZero
        \end{array}
      }
  \]
  The function
  $
  \reify{w}{\lvl}{} \is
    \set{
      h \in \Heap |
        (h, \wtv[5]) \in \collapse{\lvl}{w}
    }
  $
  is called the \emph{world reification} of~$w$ at level~$\lvl$.
  For any ${p \in \UCWorld[\actxt]}$, the function
  $
    \sem{\,p\,}_{\lvl} \is
      \Union_{w\in p} \reify{w}{\lvl}{}
  $ is called \emph{reification} of~$p$ at level~$\lvl$.
\end{definition}

To understand if a world~$w_1$ can represent local resources
consistent with some global heap~$h$,
we need to identify if there is a world~$w_2$
representing the resources of the environment 
such that $h \in \reify{w_1 \worldSep w_2}{\lvl}{}$.
That would mean that it is possible to factor~$h$ as
$ h = h_{w_1} \dunion h' \dunion h_{w_2} $
where $h_{w_1}$ are the cells fully owned by the local thread,
$h_{w_2}$ are the ones fully owned by the environment,
and the cells in $h'$ are the ones that are shared
and come from opening the interpretations of closed regions
in the world collapse.
When collapsing, we are assuming, conceptually,
that we are collapsing a world that represents
\emph{every} resource in the system.
Correspondingly, the definitions that use reification
---crucially, \cref{def:frame-pres-upd,def:at-frame-pres-upd}---
always complete the local resources
with some ``global'' frame before applying reification.

In addition to opening shared regions,
the collapse function also checks that no environment obligations
are assumed.
To understand this, consider a world~$w_1$ representing local resources,
and a world~$w_2$ completing it to a world $w = w_1 \worldSep w_2$
representing the global resources.
The global world~$w$ cannot assert the existence
of obligations in the environment:
all those have been already accounted for in~$w_2$.
The definition of collapse explicitly enforces this constraint
by the condition on the environment obligation map.
We explain why this condition is important in \cref{sec:viewshift}.

\subsection{Frame preservation}
\label{sec:fpu}

Having established the link between worlds/views and concrete state,
we can move to establishing a link between concrete \emph{steps} in a trace
and their logical justification in terms of logical state.
The fundamental driver of this link is the notion of
  \emph{frame-preserving update},
inspired by the frame-preserving update from~\cite{iris},
which represents the essence of the Rely/Guarantee reasoning in \tadalive.
The frame-preserving update looks at a
specific concrete update from some~$h_1$ to~$h_2$,
and states under which conditions this logical update can be described
as an update from logical state~$p_1$ to logical state~$p_2$.
The~$p_1$ and~$p_2$ are sets of worlds describing local resource,
whereas the~$h_1$ and~$h_2$ are global concrete heaps.
We therefore need to complete~$p_1$ with some frame~$f$
and  use reification to relate this logical state to~$h_1$: that
is, $h_1 \in \sem{p_1 * f}_{\lvl}$.
There will usually be more than  one such~$f$. The frame-preserving update requires that any such~$f$ that is a view
should remain a valid frame after the update:
that is,~$h_2 \in \sem{p_2 * f}_{\lvl}$.

\begin{definition}[Frame-Preserving Update]
\label{def:frame-pres-upd}
  Given
    $h_1,h_2 \in \Heap$,
    $p_1,p_2 \in \UCWorld[\actxt]$ and
    $\lvl \in \Level$,
  we define
  $
    \update[\lvl;\actxt] h_1 -> h_2 |= p_1 ->* p_2
  $ to hold if and only if
  \[
    \forall f \in \View[\actxt]\st
      h_1 \in \sem{p_1 * f}_{\lvl}
      \implies
          h_2 \in \sem{p_2 * f}_{\lvl}.
  \]
\end{definition}

\tadalive\ implements the Rely/Guarantee proof principle by requiring
every update to be frame-preserving.
Views are resources that are preserved by protocol-compliant environment interference.
The idea of a Rely,
a set of allowed environment updates,
is represented by assuming environment steps are frame-preserving updates
on resources that are compatible with our current view.
By frame preservation, any such update would preserve our view.
Conversely, the idea of a Guarantee,
an over-approximation of the effects of local steps
under the assumption of Rely,
is encoded by requiring every local step to be a frame-preserving update,
and thus unable to disrupt any view held by the environment.

To see how this works more concretely,
let us consider an example.
We use the notation
$ \update*[\lvl] |= p ->* q $
to mean
$
  \forall h \in \sem{p * \True}\st
    \exists h'\st
      \update[\lvl] h -> h' |= p ->* q
$,
that is,
$ \update*[\lvl] |= p ->* q $
holds when~$p$ to~$q$ can be used to justify \emph{some} concrete update.

\begin{example}
  \label{ex:fpu}
  Assume we have a region type~$\rt$ with abstract states $a,b,c,d$,
  a single guard~$\gEx$ (with $\gEx \guardOp \gEx = \bot$)
  and interference protocol consisting of transitions
  $ \gEx : (a,\oblZero) \interfTo (b,\oblZero) $
  and
  $ \gEx : (b,\oblZero) \interfTo (c,\oblZero) $
  .
  We want to show that
  (for $\lvl < \lvl'$)
  $
    \update*[\lvl'] |=
      \region[\lvl]{t}{\rid}(a) * \guardA{\gEx}{\rid}
        ->* \region[\lvl]{t}{\rid}(c) * \guardA{\gEx}{\rid}
  $
  holds,
  but
  $
    \update*[\lvl'] |=
      \region[\lvl]{t}{\rid}(a) * \guardA{\gEx}{\rid}
        ->* \region[\lvl]{t}{\rid}(d) * \guardA{\gEx}{\rid}
  $
  and
  $
    \update*[\lvl'] |=
      \region[\lvl]{t}{\rid}(a)
        ->* \region[\lvl]{t}{\rid}(b)
  $
  do not.
  Consider any view~$f$ that is a frame of
  $\region[\lvl]{t}{\rid}(a) * \guardA{\gEx}{\rid}$. The~$f$ cannot hold $\guardA{\gEx}{\rid}$ because
 ~$\gEx$ is not compatible with itself.
  As a consequence, since~$f$ is a view,
  it needs to be closed under world rely,
  which means that it is closed under the interference,
  which can transform~$a$ into~$b$ and~$b$ into~$c$.
  For~$f$ to be compatible with $\region[\lvl]{t}{\rid}(a)$,
  it needs to contain some world associating~$a$ to~$\rid$;
  to be a view,~$f$ needs to contain some other world associating
 ~$c$ to~$\rid$, which makes it compatible with
  $\region[\lvl]{t}{\rid}(c) * \guardA{\gEx}{\rid}$.
  Therefore
  $
    \update*[\lvl'] |=
      \region[\lvl]{t}{\rid}(a) * \guardA{\gEx}{\rid}
        ->* \region[\lvl]{t}{\rid}(c) * \guardA{\gEx}{\rid}
  $
  holds.

  Now, the view~$f$ above is not required to contain any world
  associating~$d$ to~$\rid$.
  Such an~$f$ is a counterexample to
  $
    \update*[\lvl'] |=
      \region[\lvl]{t}{\rid}(a) * \guardA{\gEx}{\rid}
        ->* \region[\lvl]{t}{\rid}(d) * \guardA{\gEx}{\rid}
  $
  holding.

  Finally, consider $\region[\lvl]{t}{\rid}(a)$;
  we can construct a frame~$f_a$ in which all worlds
  associate~$a$ to~$\rid$ and own the guard~$\gEx$.
  Such set of worlds can be a view because owning~$\gEx$
  disables the transition from~$a$ to~$b$.
  However,~$f_a$ would be compatible with $\region[\lvl]{t}{\rid}(a)$ but
  not with $\region[\lvl]{t}{\rid}(b)$,
  which means
  $
    \update*[\lvl'] |=
      \region[\lvl]{t}{\rid}(a)
        ->* \region[\lvl]{t}{\rid}(b)
  $
  does not hold.
\end{example}

This definition of frame-preserving update simplifies drastically
the semantics of TaDA specifications.
For \tadalive, however,
we need to introduce the stronger notion of
  \emph{atomic} frame-preserving update.
To see the motivation behind the stronger condition,
consider the region interference relation
$
  \gEx : (a,\obl{k}) \interfTo (b,\oblZero)
$
and
$
  \gEx : (b,\oblZero) \interfTo (c,\obl{k})
$.
The update from~$a$ to~$c$ via~$b$ is very different from
a direct update  from~$a$ to~$c$.
The intermediate step to~$b$ fulfils the obligation $\obl{k}$,
which may be crucial information for the progress argument.
We therefore want to enforce that if we are justifying a step
as going from~$p$ to~$q$, all the allowed transitions between region states
need to match \emph{a single transition} in the interference protocol.

\begin{definition}[Atomic Frame-Preserving Update]
  \label{def:at-frame-pres-upd}
  Given
    $h_1,h_2 \in \Heap$,
    $p_1,p_2 \in \UCWorld[\actxt]$ and
    $\lvl \in \Level$,
  we define
  $
    \update[\lvl;\actxt] h_1 -> h_2 |= p_1 -> p_2
  $ to hold if and only if
  \[
  \forall f \in \UCWorld[\actxt] \st
    h_1 \in \sem{p_1 * f}_{\lvl}
    \implies
    h_2 \in \sem{p_2 * {\relyAt[\actxt]}(f)}_{\lvl}
  \]
  where the \emph{atomic world rely relation}, $\relyAt[\actxt]$,
  is defined to be the smallest reflexive relation closed under
  the rules of \cref{fig:world-rely},
  with the restriction that \cref{rule:rely-interf,rule:rely-linpt} can be
  applied at most once per region identifier. It is formally defined in
  Appendix~\ref{app:atomic-rely}.
\end{definition}

Intuitively, this says that if the environment
has some resource~$f$ compatible with~$p_1$,
it should expect that after a step,
the resource~$f$ might be transformed into ${\relyAt[\actxt]}(f)$.
When~$f$ is a view, one gets back \cref{def:frame-pres-upd},
as views are precisely the resources that cannot be invalidated by any
number of updates of the environment.
We will use atomic frame-preserving updates to check the safety of
logical traces with respect to some  specification
in \cref{def:spec-sem}.

\subsection{Viewshifts and ``classical'' resources}
\label{sec:viewshift}

Before moving to specifications,
we define \emph{viewshift},
a semantic generalisation of implication,
which is a prime example of application of frame-preserving update,
used in our \ref{rule:consequence} rule.
They correspond to ``purely logical'' updates
in that they update the ghost resources without affecting the concrete memory.

\begin{definition}[Viewshift]
\label{def:viewshift}
  Given $p_1,p_2 \in \UCWorld[\actxt]$, the judgement
  $
    \viewshift \lvl;\actxt |= p_1 => p_2
  $,
  holds if
  $
  \forall h\in \Heap\st
    \update[\lvl;\actxt] h -> h |= p_1 ->* p_2
  $.
  For two assertions $P, Q$, the assertion~$P$ \emph{viewshifts} to~$Q$, written
  $
    \viewshift \lvl;\actxt |= P => Q
  $,
  if and only if,
  $\forall \varsigma \from (\PVar \dunion \LVar) \pto \AVal$,
  $
    \viewshift \lvl;\actxt |=
      \WorldSem{\actxt}{\varsigma}{P} => \WorldSem{\actxt}{\varsigma}{Q}
  $.
\end{definition}

Viewshifts are typically employed to ``allocate'' a new region
by sharing some local resource (a form of weakening).
For example, assume
$
  \rInt(\region{t}{\rid}(x,v)) \is x \mapsto v * \locObl{\obl{a}}{\rid}
$.
We have that $P_0 = (x \mapsto 0)$ viewshifts to
$ \exists \rid\st \region{t}{\rid}(x,0) $:
the underlying reification does not change,
and any frame of~$P_0$ with non-empty reification,
must only have regions reifying to cells disjoint from~$x$;
moreover, such frame will only have finitely many regions allocated,
so it is always possible to draw a fresh~$\rid$ from the infinite set~$\RId$
to satisfy the existential quantification over~$\rid$.

Viewshifts also ensure that obligation information is not updated inconsistently.
For example, in the ``region allocation'' step above,
we cannot viewshift~$P_0$ to
$ P_1 = \exists \rid\st \region{t}{\rid}(x,0) * \envObl{\obl{k}}{\rid}$,
which would mean we are pretending
there is an obligation~$\obl{k}$ in the environment
without any evidence of that being true.
To show that the viewshift does not hold,
we can choose~$h=\map{x->v}$
and show that $ \update h -> h |= P_0 -> P_1 $ is false.
To see this, pick~$f=\emp$ as the global frame;
$h$ is in the reification of~$P_0 * \emp$
but the reification of
$ \exists \rid\st \region{t}{\rid}(x,0) * \envObl{\obl{k}}{\rid} * \emp $
is empty:
the frame~$\emp$ has empty local obligation map,
so every world~$w$ considered by
the region collapse of~$P_1 * \emp$
has $\envMap[w] \neq \oblZero$.
The idiomatic, and correct, pattern of creation of environment obligations
would viewshift~$P_0$ to, say,
$ \exists \rid\st \region{t}{\rid}(x,0) * \locObl{\obl{b}}{\rid} $
for some relevant obligation $\obl{b}$ compatible with $\obl{a}$,
and then with implication obtain
$ \exists \rid\st \region{t}{\rid}(x,0) * \locObl{\obl{b}}{\rid} * \envObl{\obl{b}}{\rid}$:
in this case the environment obligation has been created
from the evidence of the existence of a corresponding local obligation.

\smallskip

It is important to note that, in a logic with the ability to share assertions,
as regions in TaDA or invariants in Iris allow,
having classical resources does not have the expected effect.
By definition, a classical resource~$P$ cannot be ``forgotten'',
i.e.~$P * Q \not\implies Q$.
By using viewshift, however, it is possible to create a region
with interpretation defined so it contains~$P$,
and immediately discard it (regions are not classical resources),
obtaining $P * Q \vshift Q$.
This, for example, makes \tadalive\ incapable of proving absence of memory leaks
even if its heap assertions are classical.
We however manage to avoid this issue for
local obligation assertions~$\locObl{O}{\rid}$
because of their specific semantics.
First, $\locObl{O}{\rid}$ can only ever be part of the interpretation for the
region~$\rid$, as imposed by our restrictions on region interpretations.
Second, the very notion of fulfilling the obligation is defined as
transferring its ownership to the interpretation.
Moreover, the region protocol constrains the loss of an obligation
to happen only in correspondence with some region state change,
so the only way to get rid of a local obligation
is to induce the desired state change in the region
and transfer the obligation to the interpretation.
We need obligations to be classical resources
in order for this to be the \emph{only} way of losing them.
We made heaps and guards behave classically for the sake of uniformity,
but this is not essential.

The issue of having genuinely classical resources
in a logic with regions/invariants
has been tackled in~\cite{iron},
with the main use cases being proving absence of memory leaks.
The techniques presented there could provide the basis
for an alternative way of handling \tadalive-style obligations.

\subsection{Specification format}
\label{sec:specs}

With all these definitions in place, we can now proceed to
define \tadalive\ specifications 
and their trace semantics.
Most of the time, \tadalive\ proofs manipulate triples  of two forms:
\begin{subequations}
  \begin{gather}
    \ATRIPLE m;\lvl;\actxt |- \A x\in X\eventually[k] X'.<P(x)>\cmd<Q(x)>
    \label{specs:special-cases-atomic}
    \\
    \TRIPLE m;\lvl;\actxt |- {P}\cmd{Q}
    \label{specs:special-cases-hoare}
  \end{gather}
\end{subequations}
called \emph{atomic triples} and \emph{Hoare triples}, respectively.
It is however possible for 
a command to  manipulate some resources $\na{P}$ non-atomically,
and some other resources $\at{P}(x)$ atomically, at the same time.
In general, specifications use  \emph{hybrid triples}:
\begin{equation*}
  \ATRIPLE m;\lvl;\actxt |-
    \A x \in X \eventually[k] X'.
    <\na{P}|\at{P}(x)>
    \cmd
    \E y.
    <\na{Q}(x,y)|\at{Q}(x,y)>
\end{equation*}
a minor generalisation\footnote{The difference is the $\exists y$, which is  used  in the uncommon case
  when the linearization point is non-deterministic
  \emph{and} the Hoare postcondition depends on this
  non-deterministic choice.}
of hybrid triple discussed in~\cref{sec:proof-reuse}.
Intuitively, the Hoare precondition $\na{P}$ is a resource that is
owned by the command and, as such,
cannot be invalidated by actions of the environment.
The command is allowed to manipulate this owned resource non-atomically,
provided it satisfies the Hoare postcondition $\na{Q}$ upon termination.
The atomic precondition $\at{P}(x)$ represents the resource that can be
shared between the command and the environment.
The environment can update it, but only with the effect of going from
$\at{P}(x)$ for some $x \in X$ to $\at{P}(x')$ for some $x' \in X$.
The command is allowed to update it exactly once from $\at{P}(x)$
to perform its linearization point,
transforming it to a resource satisfying the atomic postcondition $\at{Q}(x)$.
The atomic postcondition only needs to be true
\emph{just after} the linearization point
as the environment is allowed to  update it immediately afterwards.
The pseudo-quantified variable~$x$ has two important uses:
it represents the ``surface'' of allowed interference by the environment;
it is bound in the postcondition to the value of the parameter
of the atomic precondition \emph{just before} the linearization point.

The atomic and Hoare triples in
\eqref{specs:special-cases-atomic} and \eqref{specs:special-cases-hoare}
are then special cases of the hybrid triple:\footnote{We use the standard notation $a\doteq b$ to mean $a=b \land \emp$.
}
\begin{subequations}
  \begin{gather}
    \forall \vec{v}_0\st
    \ATRIPLE m;\lvl;\actxt |-
      \A x \in X \eventually[k] X'.
      <\pvars{v}_0 \doteq \vec{v}_0 | P'(x)>
      \cmd
      \E \vec{v}_1.
      <\pvars{v}_0 \doteq \vec{v}_0 \land \pvars{v}_1 \doteq \vec{v}_1| Q'(x)>
    \\
    \ATRIPLE m;\lvl;\actxt |- <P|\emp> \cmd <Q|\emp>
  \end{gather}
\end{subequations}
resp.,
where
  $\pvars{v}_0 = \progvars(P(x))$,
  $\pvars{v}_1 = \progvars(Q(x)) \setminus \pvars{v}_0$,
$P'(x) = P(x)\subst{\pvars{v}_0->\vec{v}_0}$ and
$Q'(x) = Q(x)\subst{\pvars{v}_0->\vec{v}_0,\pvars{v}_1->\vec{v}_1}$
(for technical reasons the atomic pre/post-conditions in the general triples cannot contain program variables).
We omit the pseudo-quantifier from an atomic triple (as above)
when the pseudo-quantified variable does not occur in the triple,
and thus could be quantified as
$\PQ{x \in \set{1} \eventually[\layBot] \set{1}}$.
We also use the abbreviated form $\PQ{x\in X}$
when the liveness assumption is trivial, i.e.~$\PQ{x\in X\eventually[\layBot] X}$.

\begin{definition}[Specification]
\label{def:spec}
  \emph{Specifications},
    $\spec \in \Spec$,
  have the form:\begin{equation}
    \SPEC m;\lvl;\actxt |=
      \A x \in X \eventually[k] X'.<\na{P}|\at{P}(x)>
      \E y.<\na{Q}(x,y)|\at{Q}(x,y)>
    \tag{\maltese}
    \label{spec-format}
  \end{equation}
  where
  \begin{itemize}
  \item $m\in\Layer$,
  $\lvl\in\Level$ and
  $\actxt\in\AContext$;
  \item $x,y \in \LVar$;
  \item $X' \subseteq X \subseteq \AVal $ and $k \in \Layer$;
  \item $\na{P}, \na{Q}(v,v') \in \Stable[\actxt]$
    for all $v \in X$ and $v' \in \AVal$;
  \item $\at{P}(v), \at{Q}(v,v') \in \Assrt$
    for all $v \in X$ and $v' \in \AVal$, and
    $\progvars(\at{P})=\progvars(\at{Q})=\emptyset$.
  \item $ \forall x \in X\st
    \VALID \actxt |= \at{P}(x) \implies \empObl[\lvl] $.
  \item $ \forall x \in X, y\st
    \VALID \actxt |= \at{Q}(x,y) \implies \empObl[\lvl] $.
  \end{itemize}
\end{definition}

In addition to the atomicity context~$\actxt$,
the \emph{context of a specification}
$ \color{context} m,\lvl,\actxt $
consists also of
  a layer~$m$, and
  a level~$\lvl$.
These components record information about the proof context of the judgement.
The layer~$m$ indicates that we are in a context where we are forbidden
from assuming as live obligations with layers~$\geq m$ (or incomparable to~$m$).
The level~$\lvl$ indicates that the regions with level $\geq \lvl$
are open (and cannot be re-opened).

\subsection{Trace Semantics of Specifications}
\label{sec:specs-semantics}

Finally, we can define the semantics of a specification.
The idea of the semantics is to collect all traces that are deemed as acceptable
to a specification~$\spec$, so that we can later say a command satisfies~$\spec$
if its traces are all accepted by the semantics of~$\spec$.
The general principle in accepting a trace is the following:
the local steps are only expected to be correctly implementing the functionality
declared by~$\spec$ if the environment satisfies the assumptions
implied by the (safety and liveness) protocols and~$\spec$ itself.
If a trace has gone wrong as a consequence of the environment making moves
outside of the assumptions, that trace is accepted,
as the problem is not the responsibility of the local command itself.
If a trace has gone wrong as a consequence of local steps,
then the trace is rejected.

The semantics of a specification therefore traverses a trace
to decide whether to accept or reject it,
by determining who is to blame for failures.
We decouple the traversal needed for checking the safety constraints,
and the one checking the liveness ones.
In terms of safety, a specification like~\eqref{spec-format} (in \cref{def:spec})
expects that:
\begin{itemize}
  \item the precondition holds:
        the starting resource satisfies
        $\na{P} * \at{P}(x)$, for some $x\in X$;
  \item the interference precondition holds:
        every step of the environment, before the local linearization point
        takes place goes from a resource satisfying
        $\at{P}(x_1)$, for some $x_1\in X$,
        to a resource satisfying
        $\at{P}(x_2)$, for some $x_2\in X$.
\end{itemize}
If any of the above are violated, the blame is on the environment.
In return, the local steps are expected to:
\begin{itemize}
  \item respect atomicity:
        transform the resources of $\at{P}(x)$ exactly once
        to resources satisfying $\at{Q}(x,y)$, for any $x\in X$ and some~$y$;
  \item respect the pre-/postconditions:
        transform (in possibly many steps)
        the resources in $\na{P}$ to resources satisfying
        $\na{Q}(x,y)$ at the end of the execution.
\end{itemize}
In this sense, the resources in $\at{P}(x)$ should be understood as shared:
the environment can use them to change the value of~$x$,
and the local steps can use them atomically
to perform the linearization point.
Note that $\at{Q}(x,y)$ is only guaranteed to hold \emph{immediately after}
the linearization point.

\paragraph{Key idea of the liveness semantics.}
In terms of termination, the specification~\eqref{spec-format}
guarantees local termination
\emph{only if}
the environment is \emph{live},
i.e~it satisfies the layered liveness invariants
represented by the pseudo-quantifiers (of the specification and in~$\actxt$)
and the obligations.
The idea is again to identify when non termination is caused
by a bad environment or by bad local steps.
Consider the case of liveness invariants encoded by obligations.
Imagine we annotate each position of a trace
indicating which obligations are held at that point by the environment
and which are held locally.
Now suppose the environment always eventually fulfils \emph{every} obligation
(i.e.~for each obligation~$O$
  there are infinitely many positions where~$O$
  is not held by the environment).
This environment is certainly live, so it cannot be blamed for non termination.
The layer structure, however, allows the environment to fulfil obligations
of layer~$k$ by relying on eventual fulfilment of obligations at layer~$<k$.
Therefore, if there is an obligation~$O$ that is locally held forever,
the environment is still considered live if it never fulfills obligations
at layer $> \lay(O)$.
In this scenario, the local steps are blamed for non termination:
by holding~$O$ forever, the local steps are not allowed to rely on the environment
being live at higher layers than $\lay(O)$.

This scheme leads to the following semantic interpretation of layers.
The local steps can blame the environment for non termination,
by waiting for the fulfilment of some environment obligation~$O_1$ indefinitely;
in turn the environment can blame the local steps for the inability to fulfil~$O_1$
by claiming to be waiting for the fulfilment of some local obligation~$O_2$ with
$\lay(O_1)>\lay(O_2)$; the local step can justify the indefinite postponement of
the fulfilment of~$O_2$ by shifting blame on the environment again,
appealing to an environment obligation with even lower layer,
and so on.
This blame-shifting cannot be unbounded:
every time the blame is shifted,
the layers considered are strictly lower and,
by well-foundedness of layers,
this cannot happen ad libitum.
Ultimately, the blame is unambiguously placed on the environment or the local steps
and the trace is accepted or rejected accordingly.

This intuition about obligations extends to liveness assumptions
attached to pseudo-quantifications in the triple and in the atomicity context.
All these assumptions need to be layered to avoid unsound circularities,
which is why the pseudo-quantifier carries a layer~$k$.
The specifications mention another layer,~$m$,
which represents a (strict) upper bound on the layers
that we may consider live when proving some command satisfies the specification.
An environment is still considered live by the specification,
if it keeps an obligation of layer $\laygeq m$ forever unfulfilled.

\bigskip

We now define the formal semantics of specifications,
as set of concrete traces that satisfy the specification.
To check if a concrete trace~$\trace$ satisfies a specification,
the semantics first collects all the possible ``logical'' justifications
of the trace in a set~$\straces$.
To justify a trace means to instrument each step with sets of worlds that
show how the trace respects the
(safety) logical constraints of the specification.
The set~$\straces$ is then further analysed to check
that every instrumented trace
where the environment satisfies the liveness assumptions
is locally terminating.

We begin by defining the \emph{trace safety judgement}, of the form
  $\accept{\trace}{\na{p}}{\at{p}}{v} : \straces$,
the purpose of which is to check
the safety constraints implied by~$\spec$.
The judgement formalises the idea of a specification~$\spec$
as a trace acceptor, that is,
an automaton reading a trace step by step, and either accepting or rejecting it.
If we ignore~$\straces$ for a moment,
the trace safety judgement represents a snapshot of the state
of this imaginary automaton at a point when some prefix of the trace
has been already successfully processed, and~$\trace$ is
the suffix that remains to be processed.
The automaton traverses the trace producing a guess for an instrumentation,
i.e.~logical resources corresponding to the concrete memory contents
that explains why the trace is acceptable.
The instrumentation needs to describe, for instance,
when the linearization point is thought of taking place,
what portions of the state are  considered as  shared and which owned.
Let $(\sigma, h)$ be the current concrete state,
i.e.~$\trace = (\store,h)\;\trace'$.
The triple $(\na{p}, \at{p}, v)$ in the judgement
encode the automaton's current state,
representing the current guess for the instrumentation
of $(\sigma, h)$.
The resources currently considered as locally owned are represented
by the view $\na{p}$;
the $v$~can be either an abstract value, in which case the automaton
thinks that we are still in the interference phase,
or it can be a pair $\DONE(v_1,v_2)$ which means we are past the linearization point, which updated the abstract state from~$v_1$ to~$v_2$;
the $\at{p}(v)$ is a set of worlds parametric on~$v$ and corresponds to the shared atomic resources if we are before the linearization point,
or it is the empty resource if we are past it.
The judgement assumes that
$
  h \in \sem{\na{p} * \at{p}(v) * f}_{\lvl}
$
for some frame~$f$,
i.e.~the current concrete state is consistent
     with the current instrumentation guess.
The initial state of the automaton will be chosen so that this holds
at the beginning of the trace, and each transition of the automaton
will by construction preserve this correspondence.

As it walks down a trace,
the automaton updates $\na{p}$, $\at{p}$ and~$v$, 
trying to construct a consistent instrumentation for the whole trace.
Such a sequence of automaton states constitute its \emph{run} over the trace.
\emph{Specification traces} augment each state of a trace
with the instrumentation from a run of the automaton.

\begin{definition}[Specification Traces]
  \label{def:spec-trace}
  Define
  $
    \AVal' \is \AVal \dunion \set{\DONE(v_1,v_2) | v_1,v_2 \in \AVal}
  $,
  the set of \emph{specification states} to be
  $
    \SState[\actxt] \is
      \View[\actxt] \times
      (\AVal' \to \UCWorld[\actxt]) \times
      \AVal'
  $
  and the set of \emph{specification configurations} to be
  $
    \SConf[\actxt] \is
      \Store \times \Heap \times \SState[\actxt]
  $.
  The set of \emph{specification traces}, $\STrace[\actxt]$,
  is the set of infinite sequences of the form
  $
    \sconf_1 \pl[1] \sconf_2 \pl[2] \cdots
  $
  where
    $\sconf_i \in \SConf[\actxt]$ and
    $\pl[i] \in \set{\loc, \env}$.
Given a set of specification traces
    $ \straces \subseteq \STrace $,
  we write $ \sconf \pl \straces $
  for the set
  $ \set{ \sconf \pl \strace | \strace \in \straces} $.
\end{definition}

The trace safety judgement accumulates,
as it traverses a trace,
all the successful instrumentations of the trace in~$\straces$,
which we can later check against liveness properties.
Let us define the judgement formally, and then explain it in detail.

\begin{definition}[Trace Safety]
\label{def:trace-safety}
  Let $\spec \in \Spec$ with components named as~\eqref{spec-format},
    $\trace \in \Trace$,
    $ \straces \subseteq \STrace $, and
    $(\na{p},\at{p},v) \in \SState$ such that
  \[
    \at{p}(x) =
    \atWSem{\at{P}}(x) \is
      \begin{cases}
        \WSem[\actxt]{\at{P}(x) \land x \in X} \CASE x \in \AVal \\
        \wEmp[\actxt] \OTHERWISE
      \end{cases}
  \]
  The \emph{trace safety judgement} is the relation
    $\accept{\trace}{\na{p}}{\at{p}}{v} : \straces$
  defined coinductively in \cref{fig:spec-sem}.\footnote{Here~$\trace$ ranges over subsequences of traces.}
  We write $ \term(\trace) $, if the trace~$\trace$ contains no local steps.
  \begin{mathfig}[ht]
    \begin{proofrules}
  \small
  \infer*[right=Stutter]{
    \label{rule:stutter}
\update
    h_1 -> h_2 |= \na{p} \ssep \at{p}(v) -> \na{p'} \ssep \at{p}(v)
    \\
    \accept{(\store_2, h_2)\;\trace}{\na{p'}}{\at{p}}{v} : \straces
    \\
    \term(\trace)
    \implies
    v = \DONE(v_1,v_2) \land
    \na{p}' = \WorldSem{\actxt}{\store_2}{\na{Q}(v_1,v_2)}
  }{\accept{(\store_1,h_1)\loc(\store_2,h_2)\;\trace}{\na{p}}{\at{p}}{v}
    : ((\store_1,h_1, \na{p},\at{p},v)\loc\straces)
  }
  \and \infer*[right=LinPt]{
    \label{rule:linpt}
\update
    h_1 -> h_2 |=
    \na{p} \ssep \at{p}(v)
    ->
    \na{q'} \ssep \WorldSem{\actxt}{}{\at{Q}(v,v')}
    \\
    \term(\trace)
    \implies
    \na{q}' = \WorldSem{\actxt}{\store_2}{\na{Q}(v,v')}
    \\
    \accept{(\store_2, h_2)\;\trace}{\na{q'}}{\at{p}}{\DONE(v,v')} : \straces
  }{\accept{(\store_1,h_1)\loc(\store_2,h_2)\;\trace}{\na{p}}{\at{p}}{v}
    : ((\store_1,h_1, \na{p},\at{p},v)\loc\straces)
  }
  \and \infer*[right=Env]{
    \label{rule:env}
    \straces =
    \textstyle \Union \Set{
      (\store,h_1, \na{p},\at{p},v) \env \straces[v'] | v' \in X, E(v') }
    \\\\
    \forall v'\in X\st
    E(v')
    \implies
    \accept{(\store, h_2)\;\trace}{\na{p}}{\at{p}}{v'} : \straces[v']
    \\
v\in\AVal
    \\\\
    E(v') \is
    \left(
    \exists \pe,\pe' \st
    h_1 \in \sem{ \na{p} \ssep \at{p}(v) \ssep \pe }_\lvl
    \land
    \update
    h_1 -> h_2 |=
    \at{p}(v) \ssep \pe
    ->
    \at{p}(v') \ssep \pe'
    \right)
  }{\accept{(\store,h_1)\env(\store,h_2)\;\trace}{\na{p}}{\at{p}}{v}
    : \straces
  }
  \and \infer*[right=Env']{
    \label{rule:env2}
\textbf{if }
    \exists \en{p},\en{p}'\st
    h_1 \in \sem{\na{p} \ssep \en{p}}_\lvl \land
    \update h_1 -> h_2 |= \en{p} -> \en{p}'
    \textbf{ then }
\accept{(\store, h_2)\;\trace}{\na{p}}{\at{p}}{\DONE(v,v')} : \straces
    \textbf{ else }
    \straces = \emptyset
  }{\accept{(\store,h_1)\env(\store,h_2)\;\trace}{\na{p}}{\at{p}}{\DONE(v,v')}
    : ((\store_1,h_1, \na{p},\at{p},\DONE(v,v')) \env \straces)
  }
  \and \infer*[right=Env$_\fault$]{ }{\label{rule:env-fault}
    \accept{(\store,h)\env \fault\:\trace}{\na{p}}{\at{p}}{v}
    : \emptyset
  }
\end{proofrules}
     \caption{Safety Specification Semantics}
    \label{fig:spec-sem}
  \end{mathfig}

\end{definition}

The judgement $\accept{\trace}{\na{p}}{\at{p}}{v} : \straces$
assumes the initial configuration $(\store_0,h_0)$ of the trace~$\trace$
satisfies $ h_0 \in \sem{\na{p} * \at{p}(v) * \True}_{\lvl} $.
\Cref{rule:stutter}
  checks that any local step other than the linearization point
  updates the local Hoare view (to some $\na{p}'$) in a frame-preserving manner;
  this implies that, before the linearization point,
  the abstract state~$v$ needs to be preserved by such step.
\Cref{rule:linpt}
  checks that the linearization point is frame-preserving and consistent
  with the atomic postcondition $\at{Q}$.
Both \cref{rule:stutter,rule:linpt} check that the Hoare postcondition
is satisfied if we are considering the last local step of the trace
(i.e.~if $\term(\trace)$ holds).
\Cref{rule:env}
  checks whether the current environment step,
  assumed to happen before the linearization point,
  can be seen as a transition
  changing the abstract state from~$v$ to $v'$ in a way that does not
  disrupt any frame (including~$\na{p}$).
  If that is the case, the rest of the trace is checked for safety.
\Cref{rule:env2}
  performs the same check but after the linearization point.
  In both cases, if the environment step cannot be seen as frame-preserving,
  then the trace is accepted since the environment did not respect the assumptions.
  Similarly,
\Cref{rule:env-fault}
  accepts the trace after a fault caused by the environment.

As we briefly mentioned,
\cref{def:trace-safety} is inspired by alternating automata~\cite{Vardi95}.
The ``alternation'' aspect is necessary because of the
angelic/demonic duality between local and environment steps:
when processing an environment step,
we need to be prepared to handle \emph{every} possible interpretation
of the update that took place;
for local steps, we are allowed to pick \emph{any} interpretation of the update.
Note that these ambiguities arise purely from the fact that
we are instrumenting the trace with ``ghost'' logical state:
at each step there is no ambiguity in a trace
about how the \emph{concrete} state has been updated.
This dual interpretation gives rise to
the two kinds of transitions in an alternating automaton.
An automata-based presentation of the trace safety judgement
would use existentially branching transitions for local steps,
and universally branching transitions for environment steps.
We further mimic alternating automata in the way we factor
safety and liveness constraints.
Alternating automata impose safety constraints
by constructing sets of runs that linearise the choices
for the existential transitions and the branching due to universal transitions.
In our setting these sets of runs correspond to~$\straces$.
The liveness constraints can be then checked by, for example,
requiring each run in the set to visit final states infinitely often,
the usual B\"uchi-style acceptance condition.
Here we also examine the instrumented traces of~$\straces$ individually
and impose a liveness acceptance condition;
the condition in our case is more complex as it has to take into account
layers, pseudo-quantifiers and obligations.
One key simplification introduced by this approach is
that we can cleanly separate the branching (safety) aspect
---the quantifier alternation due to duality environment/local steps---
from the linear-time liveness aspect.

\medskip

Building on trace safety, we can now define the semantics of a specification
$\sem{\spec}$ as the set of traces that are safe and that additionally satisfy
the liveness constraints implied by the obligations and the liveness assumptions
of~$\spec$.
Conceptually, we want to require local termination,
if the environment satisfies the layered liveness invariants
represented by pseudo-quantifiers and obligations.
To harmonise the pseudo-quantification and obligation-related
liveness assumptions of a specification,~$\spec$,
we collect all of them in a set of so-called
\emph{pseudo-obligations}:
\[
  \PObl^{\spec} \is
    \Set{(\rid, O) | \rid \in \RId, O \in \AObl} \dunion
    \Set{(\rid,\live(\actxt,\rid)) | \rid \in \dom(\actxt)} \dunion
    \set{X \eventually[k] X'}
\]
where $\actxt, X, X'$ and~$k$ are taken from the specification.

We extend the layer function to $\lay\from\PObl^{\spec} \to \Layer$ by setting
$\lay(a) = k$ if
  $a = (r, O)$ and $\lay(O) = k$, or
  $a = (r, X\eventually[k] X')$, or
  $a =(X\eventually[k] X')$.
Furthermore, define
\begin{align*}
\POLay{k}^{\spec} &\is \set{\pobl \in \PObl^{\spec} | \lay(\pobl) \laylt k}
&
\ObLay{k} &\is \set{O \in \AObl | \lay(O) \laylt k}
\end{align*}

Now we want to understand, for each position of a specification trace,
which pseudo-obligations we are holding locally and which are held by the environment.
This information is contained in the single worlds so as a first step we
extract, from a specification trace, the set of traces of worlds that it represents.

\begin{definition}[World Traces]
\label{def:world-traces}
  Given an atomicity context,~$\actxt$,
  we call \emph{world traces}, $\WTrace_{\actxt}$,
  ranged over by $\wtrace, \wtrace', \ldots$,
  the infinite sequences of the form
  $
    (h_0, \na{w}^0, \at{w}^0, \en{w}^0, v^0)
    \pl[0]
    (h_1, \na{w}^1, \at{w}^1, \en{w}^1, v^1)
    \pl[1] \ldots
  $,
  where, for all $i \in \Nat$,
   $h_i \in \Heap$,
   $\na{w}^i, \at{w}^i, \en{w}^i \in \World[\actxt]$ and
   $v^i \in \AVal'$.
We define the function:
  \[
    \worldify{\lvl}(\store,h,\na{p},\at{p},v) \is
      \Set{ (h, \na{w}, \at{w}, \en{w}, v) |
        \na{w} \in \na{p}, \at{w} \in \at{p}(v),
        h \in \reify{\na{w} \worldSsep \at{w} \worldSsep \en{w}}{\lvl}{}{}
      }
  \]
  which we extend to specification traces by
  $
    \worldify{\lvl}(\sconf[0] \pl[0]\sconf[1] \pl[1] \dots) \is
      \set{ c_0\pl[0] c_1 \pl[1] \dots |
        \forall i\st c_i \in \worldify{\lvl}(\sconf[i])
      }
  $.
  A world trace $
    (h_0, \na{w}^0, \at{w}^0, \en{w}^0, v^0) \pl[0]
    (h_1, \na{w}^1, \at{w}^1, \en{w}^1, v^1) \pl[1] \ldots
  $ is \emph{\pre\relyAt[\actxt]-respecting}
  if for all $i\in\Nat$:
  \[
    {\pl[i]} = {\env} \implies \na{w}^i \relyAt[\actxt] \na{w}^{i+1}
    \quad\land\quad
    {\pl[i]} = {\loc} \implies \en{w}^i \relyAt[\actxt] \en{w}^{i+1}
  \]
  Given specification trace  $ \strace \in \STrace[\actxt] $,    the set $\wtr{\strace}{\lvl;\actxt}$
  is the \emph{set of world traces of $ \strace $},
  defined by
  \[
    \wtr{\strace}{\lvl;\actxt} \is
      \set{ \wtrace \in \worldify{\lvl}(\strace) |
        \wtrace \text{ is \pre\relyAt[\actxt]-respecting}
      }
  \]
  We lift $\wtr{\hole}{\lvl;\actxt}$ to apply to sets of specification traces
  in the obvious way.
\end{definition}

We can now define two predicates indicating when a pseudo-obligation
is considered to be held by the environment ($\envheld_\lvl$)
or locally ($\locheld_\lvl$) in a position of a world trace.
\begin{align*}
  \envheld_\lvl(\pobl,(\_,\na{w},\_,\en{w},v)) &\is
  \begin{cases}
    \oblMap[\en{w}](\rid) \resgeq \obl{O} \land
    \level[\na{w}](\rid) < \lvl
    \CASE \pobl = (\rid, \obl{O})
    \\
    \astate[\na{w}](\rid) \not\in X_2 \land
    \level[\na{w}](\rid) < \lvl
    \CASE \pobl = (\rid, X_1 \eventually[k] X_2)
    \\
    v \in X_1 \setminus X_2
    \CASE \pobl = (X_1 {\eventually[k]} X_2)
  \end{cases}
\\
  \locheld_\lvl((\rid,\obl{O}),(\_,\na{w},\wtv[3])) & \is
    \oblMap[\na{w}](\rid) \resgeq \obl{O} \land
    \level[\na{w}](\rid) < \lvl
\end{align*}

Equipped with these definitions, we can state the liveness constraints
associated with a specification.
The idea is that one can assign the ``blame'' for local non-termination either
to the environment or to the local behaviour.
If we deem the environment responsible for non-termination
then the specification will classify the trace as acceptable,
otherwise it will reject it. 
The idea behind this ``blame'' assignment
is to examine the world traces justifying the safety of a trace
and consider, for each position,  which obligations are held by the environment
and which are held locally.
To understand the intuition, consider the case of liveness invariants
encoded by obligations.
Suppose the environment always eventually fulfills \emph{every} obligation,
  i.e.~for each obligation~$O$
  there are infinitely many positions where~$O$
  is not held by the environment.
This environment is certainly \emph{live},
i.e.~it respects the liveness assumptions
and the local code is responsible for any non-terminating behaviour.
But what if the environment itself is blocking on some locally held obligation,
and as a consequence is not able to fulfill some~$O$?
Whether the environment or the local code is to blame depends on the layers.
The environment is to blame if, from some point in the trace,
it never fulfills some~$O$ but the local steps always eventually fulfill
every obligation of layer strictly lower than $\lay(O)$.
Conversely, an environment which keeps~$O$ unfulfilled because of some
forever-unfulfilled obligation $O'$ held locally with $\lay(O) > \lay(O')$
cannot be blamed for local non-termination.

This intuition about obligations extends to liveness assumptions
attached to pseudo-quantifications in the triple and in the atomicity context.
The~$\goodenv$ predicate given in \cref{def:spec-sem}  formalises the above blame-assigning mechanism.
A world trace which satisfies~$\goodenv$ is one where
the environment cannot be blamed for local non-termination.
The specification semantics then is the set of safe traces where,
if~$\goodenv$ is satisfied, then the trace is locally terminating.

\begin{definition}[Specification~Semantics]
\label{def:spec-sem}
Fix a specification $\spec \in \Spec$
with components named as in~\eqref{spec-format}.
The $\goodenv_\spec(\strace)$ predicate checks whether the environment
  is satisfying the liveness assumptions of the specification:
  \begin{align*}
  \goodenv_\spec(\wtrace) \is
    \forall \pobl \in \POLay{m}^{\spec}\st
      \;&\textbf{if }
        \forall \rid, O \in \AObl_{\le\lay(\pobl)}\st
        \forall i\in \Nat\st
          \exists j \geq i\st
            \neg \locheld_\lvl((\rid, O),\trAt[\wtrace]{j})
      \\
      &\textbf{then }
        \forall i\in \Nat\st
          \exists j \geq i\st
            \neg \envheld_\lvl(\smash{\pobl},\trAt[\wtrace]{j})
  \end{align*}

  Let
  $ \na{p} = \WorldSem{\actxt}{\store_0}{\na{P}} $, and
  $ \at{p} = \atWSem{\at{P}} $.
  We define the trace semantics $\sem{\spec} \subseteq \Trace$ of specification~$\spec$
  as the set:
  \begin{equation*}
  \sem{\spec} \is
    \Set{
      (\store_0, h_0)\;\trace |
        \begin{aligned}
        \forall v_0\in X\st\;
        &\textbf{if }
          h_0 \in \sem{\na{p} * \at{p}(v_0) * \True}_\lvl
        \\&\textbf{then }
              \exists \straces\st
              \accept{(\store_0, h_0)\;\trace}{\na{p}}{\at{p}}{v_0} : \straces
              \\&\qquad\qquad\land
              \forall \wtrace\in\wtr{\straces}{\lvl;\actxt}\st \goodenv_\spec(\wtrace) \implies
              \locterm((\store_0, h_0)\;\trace)
        \end{aligned}
    }
  \end{equation*}
  where~$\lvl$ and~$\actxt$ are the level and atomicity context
  from the specification~$\spec$.
\end{definition}

The more precise intuition behind the specification semantics is as follows.
Once it has been established that there is a way to instrument the trace
to justify why the local steps satisfy the safety constraints of~$\spec$,
we consider the set of valid instrumentations~$\straces$.
First,  we extract the set of world traces represented
by the traces of~$\straces$.
Each such world trace should either be locally terminating,
in which case the trace is accepted
or, if it is non-terminating, the non-termination
should be  due to the environment not satisfying the liveness assumptions of~$\spec$.
The predicate $\goodenv_\spec(\strace)$ holds
for a specification trace~$\strace$
if the environment always eventually fulfills any pseudo-obligation
with layer~$k$ and if 
\emph{no obligation of layer~$\laylt k$ is constantly held by the local thread}.
Blame for non-termination can be unambiguously assigned thanks to
well-foundedness of layers:
if there is a forever-unfulfilled local obligation~$O_0$, we can try to blame the environment by identifying a lower-layer obligation~$O_1$ that is forever-unfulfilled
by the environment; the environment can shift the blame back to the local steps
if one can find a lower-layer local obligation~$O_2$ that is forever-unfulfilled.
Well-foundedness implies this blame-shifting game must be bounded in length
and the ultimate culprit can always be identified.
This effectively encodes the acyclicity of the layered termination argument.

\subsection{The Semantic Judgement}
\label{sec:sem-judge}

We are now ready to define the semantic version of our judgements,
$\JUDGE[\funspec] |=  \cmd:\spec$, indexed by   a function
specification context,  $\funspec$, which, for each function,
provides the arguments of  the function and the specification
of the function body.

\begin{definition}[Function Specification Context]
  A \emph{function specification context},~$\funspec$, is a partial function
  $
    \funspec \in \FunSpec \is
      \FName \pto (\PVar^{*}, \Spec)
  $.
\end{definition}

\begin{definition}[Semantic Triple]
\label{def:judgement-sem}
  Given $\functxt \in \FunCtxt$ and $\funspec\in\FunSpec$, 
 a function implementation context~$\functxt$ is a \emph{correct implementation} of~$\funspec$,
  written $ \JUDGE |= \functxt : \funspec $,
  if and only if 
  $
    \forall \pvar{f},\pvars{x},\spec \st
      \funspec(\pvar{f}) = (\pvars{x},\spec)
        \implies
          \exists \cmd \st
            \functxt(\pvar{f}) = (\pvars{x}, \cmd) \land
              \sem{\cmd}_{\functxt} \subseteq \sem{\spec}.
  $
 The semantic triple
  $ \JUDGE[\funspec] |= \cmd : \spec $,
stating  that command~$\cmd$
  satisfies  specification~$\spec$
  under any correct implementation of the functions specified
  in~$\funspec$, is defined by:
  \[
    \JUDGE[\funspec] |= \cmd : \spec
      \quad\text{if and only if}\quad
      \forall \functxt \st \JUDGE |= \functxt : \funspec
      \implies
      \sem{\cmd}_{\functxt} \subseteq \Sem{\spec}
  \]
  Note that when~$\cmd$ has no free function names,
  the judgement
    $\JUDGE[\funspec] |= \cmd : \spec$
  is equivalent to
    $ {\sem{\cmd} \subseteq \Sem{\spec}} $.
\end{definition}

Since the semantics of our triples is a complex conditional termination statement, it is useful to show when it corresponds to unconditional termination.
Intuitively, to state facts about the behaviour of a command in an ``empty'' environment, we should be using a Hoare triple (no resource needs to be shared, no interference experienced) and there should be no assumption of obligations being owned by the environment.
We characterise the preconditions that ensure this as the ones that always admit $\empObl[\lvl]$ as a global frame.

\begin{definition}[Grounded view]
  \label{def:grounded}
  Fix an arbitrary level~$\lvl$.
  We say $p\in\View[\emptyset]$ is \emph{\pre\lvl-grounded} if
  \[
    \forall h \st
      h \in \sem{p * \True}_{\lvl}
        \implies h \in \sem{p * \empObl[\lvl]}_{\lvl}
  \]
  We say a stable assertion~$P$ is \pre\lvl-grounded if,
  for all $\store\in\Store$,
  the view $\WorldSem{\emptyset}{\store}{P}$ is \pre\lvl-grounded.
\end{definition}

Examples of grounded assertions are standard separation logic assertions like
$\emp$ or $x\mapsto v$.

Unconditional termination applies to programs running in isolation.
Note that, technically, we cannot consider traces without environment steps
as the fairness constraint requires infinitely many of those.
We therefore model the isolated executions of a command as the executions
where the environment steps do not modify the current state.
It is easy to check that,
  for each finite or infinite sequence of local steps of~$\cmd$,
  there is a corresponding \emph{fair} trace of~$\cmd$
  with only identity environment steps,
and vice versa.

\begin{definition}[Closed-World Semantics]
  \label{def:closed-world-sem}
  Given a command~$\cmd$, its \emph{closed-world semantics}
  $\Sem[C]{\cmd} \subseteq \Trace$
  is the subset of the open-world semantics $\sem{\cmd}$
  of the traces where every environment step is an identity step,
  i.e.~of the form $\conf\envstep[]\conf$.
  Additionally, for an assertion~$P$, we define
  $
    \Sem[C]{\cmd}(P)_\lvl \is
      \Set{ (\store,h)\trace \in \Sem[C]{\cmd} |
        h \in \sem{\WorldSem{\emptyset}{\store}{P} * \True}_\lvl
      }
  $, that is the closed-world traces of~$\cmd$ that
  start from a state satisfying the precondition~$P$.
\end{definition}

\begin{theorem}[Adequacy]
  \label{th:adequacy}
  For every \pre\lvl-grounded assertion~$P$,
  if $ \TRIPLE m;\lvl;\emptyset |= {P}\cmd{Q} $
  then all traces in $\Sem[C]{\cmd}(P)_\lvl$ are locally terminating.
\end{theorem}

\begin{proof}
  Take a trace $(\store,h)\trace \in \Sem[C]{\cmd}(P)_\lvl$ and
  let $p=\WorldSem{\emptyset}{\store}{P}$.
  From the semantic triple we know that
  $
    \Sem[C]{\cmd}(P)_\lvl \subseteq
    \sem{\cmd} \subseteq
    \sem{\HSPEC |= {P} {Q}}
  $.
  We have $ h\in \sem{p * \True}_\lvl $ and,
  since~$P$ is \pre\lvl-grounded,
  $ h\in \sem{p * \empObl[\lvl]}_\lvl $.
  By the definition of the specification trace semantics,
  we therefore know that
  $ \accept{(\store, h)\;\trace}{p}{\emp}{1} : \straces $
  for some~$\straces$.
  Note that a frame-preserving update on a grounded view keeps it grounded,
  and that identity environment steps can always be justified
  as a frame-preserving update that does not update the resources.
  In particular, from these facts we can deduce that there is some
  $\strace \in \straces$ such that at each point in time the global frame
  is $\empObl[\lvl]$.
  From this we can extract a world-trace $\wtrace\in\wtr{\straces}{\lvl;\emptyset}$ such that
  $
    \forall \pobl \in \POLay{m}\st
        \forall i\in \Nat\st
            \neg \envheld_\lvl(\smash{\pobl},\trAt[\wtrace]{i})
  $ which implies $\goodenv(\wtrace)$.
  By the definition of the specification's semantics
  this implies local termination of our concrete trace $(\store,h)\trace$.
\end{proof}

As a corollary, we have that if
$ \TRIPLE m;\lvl;\emptyset |= {\emp}\cmd{\True} $ holds,
then every isolated execution of~$\cmd$ from the empty heap and arbitrary store terminates.
 \section{\tadalive\ Rules}
\label{sec:rules}

We now introduce the rules of \tadalive,
summarised in~\cref{fig:proof-rules},
using a simple but tricky running example
to motivate and explain them.

\begin{example}[Distinguishing client]
\label{ex:distinguishing-client}
  Consider the following client of a lock module:
\begin{small}\begin{center}
\begin{tabular}{r@{\quad}||@{\quad}l}\begin{sourcecode*}[boxpos=c]
lock(x);
[done] := true;
unlock(x);
\end{sourcecode*}
&\begin{sourcecode*}[boxpos=c]
var d = false in
while(!d){
  lock(x); d := [done]; unlock(x);
}
\end{sourcecode*}
\end{tabular}
\end{center}
   \end{small}The code is interesting in that it can distinguish whether
  the lock implementation is a  spin or CLH lock.
  Under weak fairness, when \p{x} is a spin lock,
  this client program does not always terminate.
  It is possible for the lock invocation of the left thread
  to be scheduled infinitely often
  but always in a state in which  the lock is locked.
  As a result, \code{done} will never be set to \code{true},
  making the while loop spin forever.
  The spin lock has been \emph{starved} by the other thread.
In contrast, when \p{x} is a CLH lock,
  this client program is guaranteed to terminate:
  a fair scheduler will eventually allow the left thread to enqueue itself
  in the internal queue of the lock;
  from then on, the thread on the right can only acquire the lock at most once;
  after unlocking, the next \code{lock(x)} call of the right thread
  would enqueue it after the left thread, which is now the only
  unblocked thread. The CLH lock is \emph{starvation free}.

  It is worth noting that none of the proof principles of~\cite{HoffmannMS13,GotsmanCPV09,GotsmanY11,CookPR07,LiangFS14,total-tada}
  are powerful enough to handle this example due to the blocking behaviour
  it displays.
  Even replacing locks with primitive locks,
  due to the mix of busy-waiting blocking and locks,
  the example cannot be handled by any of the proof systems of~\cite{Kobayashi00,Leino10,BostromM15,Jacobs18toplas}.
  Since LiLi does not have a rule for parallel,
  this client cannot be proven within the LiLi logic.
\end{example}

We show that the distinguishing client terminates with the CLH lock,
by proving the Hoare triple
$
  \TRIPLE \layTop |-
    {\ap{L}(\p{x},0) * \pvar{done}\mapsto \p{false}}
      {\cmd_\ell \parallel \cmd_r}
    {\True}
$,
where~$\cmd_\ell$ and~$\cmd_r$
are the left and right threads of the example, respectively.
Since our triples are total,
this triple immediately guarantees termination of the program.
Our overall argument is as follows.
The CLH specification guarantees termination of a call to \code{lock(x)}
if the lock is always eventually unlocked by the environment.
This is intuitively true for both threads:
they always unlock the lock after having acquired it.
The call to \code{lock(x)} will therefore terminate in both threads.
The only other potentially non-terminating operation is
the while loop in the right thread.
The loop is implementing a busy-wait pattern on \code{done},
and needs the help of the left thread to terminate.
We will be able to prove that since \code{done} is going to be eventually
set to true (and never reset to false), the loop will terminate.

\subsection{The Basics: Regions}
Let us formalise the argument in \tadalive,
introducing the proof rules as they are needed.
Recall the specifications of CLH lock:
\begin{align*}
  &\ATRIPLE \topOf |-
    \A l \in \set{0, 1} \eventually[\botOf] \set{0}.
      <\ap{L}(\pvar{x}, l)>
        \code{lock(x)}
      <\ap{L}(\pvar{x}, 1) \land l = 0>
  \\
  &\ATRIPLE \botOf |-
    <\ap{L}(\pvar{x}, 1)>
      \code{unlock(x)}
    <\ap{L}(\pvar{x}, 0)>
\end{align*}where we make explicit the previously omitted layers $\topOf \laygt \botOf$
(we justify the choice of layers in the proof of CLH lock).
These specifications will be available in the proof as ``axioms''
stored in a function specification context $\funspec$ parametrising every triple
of the client proof; we omit the parametrisation to aid readability.
The predicate $ \ap{L}(\p{x},l) $ is given a definition in the proof of the lock module, and, in the spirit of CAP, the client proof should not be relying
on the \emph{definition} of the predicate but in its abstract properties.
Here we rely on the fact that
$ \ap{L}(\p{x},\wtv) * \ap{L}(\p{x},\wtv) $ is false,
expressing that a lock is an exclusively owned resource
---see \cref{sec:abs-pred} for the other properties of $\ap{L}$ exposed to the client.

The two threads of the distinguishing client both access
the lock \p{x} and the heap cell \p{done}.
Consider the precondition
$\ap{L}(\p{x},0) * \pvar{done}\mapsto \p{false}$.
Both resources in the precondition
are non-duplicable, so if we give them to $\cmd_\ell$,
the other thread would not be able to also have them.
As we anticipated, shared state in \tada\ is handled using regions.
We therefore introduce a new region type $\rt[dc]$
(for $\rt[d]$istinguishing $\rt[c]$lient)
which encapsulates the resources in the precondition:
$\region{dc}{\rid}(x,\lvar{done},l,d)$
is the shared resource encapsulating a lock at~$x$ with state~$l$
and a cell at $\lvar{done}$ storing the Boolean~$d$.
Although in this case the abstract state is not hiding any detail,
since both~$l$ and~$d$ are visible,
in general the abstraction of the contents is an essential mechanism
for reasoning about \emph{abstract} atomicity.
In the proof of CLH lock, for example, to be able to see the operations
as abstractly atomic, it is essential to hide the queue from the abstract state.

Assuming the lock region encapsulated by the $\ap L$ predicate has level $\lvl$,
the lock specifications will have level $\lvl+1$ in the context,
indicating they consider the lock region closed.
To allow $\rt[dc]$ to encapsulate the lock region
and use the lock specifications to derive updates to its own state,
we let it have level $\lvl+1$.
The top-level triples for the distinguishing client
have level $\lvl+2$ as a consequence. 
We will elide all details about levels as they can be mechanically inferred
from the applications of the
\ref{rule:lift-atomic} and \ref{rule:update-region} rules.

We now design the protocol of the region, with the intent of encoding
the following safety invariants:
\begin{enumerate}[label={(I\arabic*)}]
  \item
      the addresses of the lock and the flag never change;
      \label{inv:dc-fixed}
  \item
      only the thread which acquired the lock can unlock it;
      \label{inv:dc-safety-lock}
  \item
      only the left thread will ever modify \p{done},
      and at most once from \p{false} to \p{true};
      \label{inv:dc-safety-done}
\end{enumerate}
and the following liveness invariants:
\begin{enumerate}[resume*]
  \item
      the lock will always eventually be unlocked;
      \label{inv:dc-live-lock}
  \item
      the value at \p{done} will always eventually be \p{true}.
      \label{inv:dc-live-done}
\end{enumerate}
Note that, together,
invariants~\ref{inv:dc-safety-done} and~\ref{inv:dc-live-done}
imply that eventually the value at \p{done} will always be true.
To encode invariants~\ref{inv:dc-safety-lock} and~\ref{inv:dc-safety-done}
we introduce a guard algebra for~$\rt[dc]$ generated from two guards
$\guard{k}$ (for $\guard{k}$ey) and 
$\guard{d}$ (for $\guard{d}$one),
with
  $\guardUndef{\guard{k}\guardOp\guard{k}}$
and
  $\guardUndef{\guard{d}\guardOp\guard{d}}$
to represent exclusivity of the permissions they give
on the lock and flag respectively.
We can reuse the same guard algebra for the obligation algebra
associated with~$\rt[dc]$:
the atom obligations $\obl{k}$ and $\obl{d}$ will represent
liveness invariants~\ref{inv:dc-live-lock} and~\ref{inv:dc-live-done}
respectively.
The protocol $\regLTS[dc]$ formalises all the invariants:
\begin{align}
\guardZero &:
  ((x,\lvar{done},0, \ghost{\p{false}}{d}), \oblZero)
      \interfTo ((x,\lvar{done},1, \ghost{\p{true}}{d}), \obl{k})
  \label{eq:dist-client-interf-lock}
\\
\guard{k} &:
  ((x,\lvar{done},1, \ghost{\p{false}}{d}), \obl{k})
      \interfTo ((x,\lvar{done}, 0, \ghost{\p{true}}{d}), \oblZero)
  \label{eq:dist-client-interf-unlock}
\\
\guard{d} &:
  ((x,\lvar{done},\ghost{0}{l}, \p{false}), \obl{d})
      \interfTo ((x,\lvar{done},\ghost{0}{l}, \p{true}), \oblZero)
  \label{eq:dist-client-interf-d}
\end{align}
The fact that no transition can change $x$ and~$\lvar{done}$
encodes~\ref{inv:dc-fixed};
this is such a common pattern that we adopt the convention to declare
which are the \emph{fixed} components of the abstract state and
omit them from the protocol transitions completely.
The choice for the guards
of~\eqref{eq:dist-client-interf-unlock} and~\eqref{eq:dist-client-interf-d}
reflects~\ref{inv:dc-safety-lock} and~\ref{inv:dc-safety-done},
respectively;
we will give $\guard{d}$ to the left-hand thread,
and $\guard{k}$ will be obtained by locking the lock.
The obligation~$\obl{k}$ is obtained by locking and fulfilled by unlocking;
the obligation~$\obl{d}$ is fulfilled by setting \p{done} to \p{true};
these facts encode~\ref{inv:dc-live-lock} and~\ref{inv:dc-live-done}.

We assign layers to the obligations to reflect the intuitive dependency:
the lock needs to be assumed live in the process of fulfilling the obligation on the flag.
We therefore set
$
  \layBot \laylt
  \botOf = \lay(\obl{k}) \laylt
  \lay(\obl{d}) = \topOf \laylt
  \layTop
$.

To complete the definition of shared region~$\rt[dc]$,
we  link its abstract state
to the actual heap content that it encapsulates
using the \emph{region interpretation}:
\[
  \rInt(\region{dc}{\rid}(x, \lvar{done}, l, d)) \is
    \ap{L}(x, l) * \lvar{done} \mapsto d *
    \bigl( l=0 \dotimplies \gOblA{\obl{k}}{\rid} \bigr) *
    \bigl( d   \dotimplies \gOblA{\obl{d}}{\rid} \bigr)
\]
This  assertion describes a portion of the heap being shared
(the lock at~$x$ and the cell at $\lvar{done}$) and
the linking of the ghost state (the guards and obligations)
with the abstract state.
The assertion $\gOblA{\obl{k}}{\rid}$ is an abbreviation for
$ \guardA{\guard{k}}{\rid} * \locObl{\obl{k}}{\rid} $,
which indicates local ownership of
the guard~$\guard{k}$ and obligation~$\obl{k}$.
The interpretation of a region establishes the invariant that,
when~$l=0$, the guard and obligation~$\obl{k}$
will be ``owned'' by the region (and by no thread as a consequence).
Similar links are established between the value of~$d$ and~$\obl{d}$.

\begin{figure}[tb]
  \adjustfigure \resizebox{\textwidth}{!}{$\begin{proofoutline}
\CTXT \layTop |- \\
  \PREC{\ap{L}(\pvar{x}, 0) * \pvar{done} \mapsto \p{false}} \\
  \ASSR{
    \exists \rid, l \st
    \region{dc}{\rid}(\pvar{x}, \pvar{done}, l, \p{false}) *
    \gOblA{\obl{d}}{\rid} *
    \pars[\big]{l = 1 \dotimplies \envObl{\obl{k}}{\rid}}
  } \\
  \begin{proofjump}[rule:exists-elim]
\begin{proofparall}\PARPREC{
        (\exists l \st
          \region{dc}{\rid}(\pvar{x}, \pvar{done}, l, \p{false}) *
          \gOblA{\obl{d}}{\rid} * \pars[\big]{l = 1 \dotimplies \envObl{\obl{k}}{\rid}}
        ) 
      }{
        (\exists l, d \st
          \region{dc}{\rid}(\pvar{x}, \pvar{done}, l, d) *
          \pars[\big]{l = 1 \dotimplies \envObl{\obl{k}}{\rid}} *
          \pars[\big]{\neg d \dotimplies \envObl{\obl{d}}{\rid}}
        )
      }
      \begin{thread}
        \CTXT \layTop |- \\
        \PREC{\exists l \st
          \region{dc}{\rid}(\pvar{x}, \pvar{done}, l, \p{false}) *
          \gOblA{\obl{d}}{\rid} *
          \pars[\big]{l = 1 \dotimplies \envObl{\obl{k}}{\rid}}
        } \\
        \begin{proofjump*}[rule:atomicity-weak,rule:atomic-exists-elim,rule:frame-atomic]
          \label{dist-client:left-lock-1}
\A l \in \set{0,1}. \\
          \PREC<
            L |
            \exists d \st
            \region{dc}{\rid}(\pvar{x}, \pvar{done}, l, d)
          > \\
          \begin{proofjump}[rule:liveness-check]
\A l \in \set{0,1} \eventually \set{0}. \\
            \PREC<
              \emp |
              \exists d \st
              \region{dc}{\rid}(\pvar{x}, \pvar{done}, l, d)
            > \\
            \begin{proofjump*}[rule:lift-atomic,rule:frame]\label{dist-client:left-lock-2}
              \A l \in \set{0,1} \eventually \set{0}. \\
              \PREC<\ap{L}(\pvar{x}, l)> \\
              \CODE{lock(x);} \\
              \POST<\ap{L}(\pvar{x}, 1) \land l = 0> \\
            \end{proofjump*} \\
            \POST<
              \gOblA{\obl{k}}{\rid} |
              \exists d \st
              \region{dc}{\rid}(\pvar{x}, \pvar{done}, 1, d)
            >
          \end{proofjump} \\
          \POST<
            L * \gOblA{\obl{k}}{\rid} |
            \exists d \st
              \region{dc}{\rid}(\pvar{x}, \pvar{done}, 1, d)
          >
        \end{proofjump*} \\
        \ASSR{
          \region{dc}{\rid}(\pvar{x}, \pvar{done}, 1, \p{false}) *
          \gOblA{\obl{d}}{\rid} *
          \gOblA{\obl{k}}{\rid}
        } \\
        \CODE{[done] := true;} \\
        \ASSR{
          \region{dc}{\rid}(\pvar{x}, \pvar{done}, 1, \p{true}) *
          \gOblA{\obl{k}}{\rid}
        } \\
        \CODE{unlock(x);} \\
        \POST{
          \exists l \st
          \region{dc}{\rid}(\pvar{x}, \pvar{done}, l, \p{true})
        }
      \end{thread}
      \PARALLEL
      \begin{thread}
        \CTXT \layTop |- \\
        \PREC{
          \exists l, d \st
          \region{dc}{\rid}(\pvar{x}, \pvar{done}, l, d) *
          \pars[\big]{l = 1 \dotimplies \envObl{\obl{k}}{\rid}} *
          \pars[\big]{\neg d \dotimplies \envObl{\obl{d}}{\rid}}
        } \\
        \CODE{var d = false in} \\
\ASSR{
          \exists \beta_0, l, d \st
          \region{dc}{\rid}(\pvar{x}, \pvar{done}, l, d) *
          \pars[\big]{l = 1 \dotimplies \envObl{\obl{k}}{\rid}} \land {} \\
          \pvar{d} \implies d \land
          \beta_0 = \ifte{\pvar{d}}{0}{1} *
          \pars[\big]{\neg d \dotimplies \envObl{\obl{d}}{\rid}}
        } \\
        \begin{proofjump}[rule:exists-elim]
          \ASSR{
            \exists l, d \st
            \region{dc}{\rid}(\pvar{x}, \pvar{done}, l, d) *
            \pars[\big]{l = 1 \dotimplies \envObl{\obl{k}}{\rid}} \land {} \\
            \pvar{d} \implies d \land
            \beta_0 = \ifte{\pvar{d}}{0}{1} *
                 \pars[\big]{\neg d \dotimplies \envObl{\obl{d}}{\rid}}
          } \\
          \begin{proofjump}[rule:while]
            \CODE{while($\neg\pvar{d}$) \{} \META{\forall \beta, b \st} \\
            \PREC{
              \exists l, d \st
              \region{dc}{\rid}(\pvar{x}, \pvar{done}, l, d) *
              \pars[\big]{l = 1 \dotimplies \envObl{\obl{k}}{\rid}} \land {} \\
              \pvar{d} \implies d \land
              \beta = \ifte{\pvar{d}}{0}{1} \land
              b \implies d \land
              \neg \pvar{d}
            } \\
\CODE{lock(x);} \\
\CODE{d := [done];} \\
\CODE{unlock(x);} \\
            \POST{
              \exists \gamma, l, d \st
              \region{dc}{\rid}(\pvar{x}, \pvar{done}, l, d) *
              \pars[\big]{l = 1 \dotimplies \envObl{\obl{k}}{\rid}} \land {} \\
              \pvar{d} \implies d \land
              \gamma = \ifte{\pvar{d}}{0}{1} \land
              \gamma \le \beta \land
              b \implies \gamma < \beta
            } \\
            \CODE{\}} \\
          \end{proofjump} \\
          \POST{
            \exists \gamma, l, d \st
            \region{dc}{\rid}(\pvar{x}, \pvar{done}, l, d) *
            \pars[\big]{l = 1 \dotimplies \envObl{\obl{k}}{\rid}} \land {} \\
            \pvar{d} \implies d \land
            \gamma = \ifte{\pvar{d}}{0}{1} *
            \pars[\big]{\neg d \dotimplies \envObl{\obl{d}}{\rid}} \land
            \pvar{d} \land      
            \gamma \le \beta_0
          } \\
        \end{proofjump} \\
        \POST{
          \exists l \st
          \region{dc}{\rid}(\pvar{x}, \pvar{done}, l, \p{true})
        }
      \end{thread}
      \\
      \PARPOST{\exists l \st
        \region{dc}{\rid}(\pvar{x}, \pvar{done}, l, \p{true})}{\exists l \st
        \region{dc}{\rid}(\pvar{x}, \pvar{done}, l, \p{true})}
    \end{proofparall}
  \end{proofjump} \\
  \POST{\exists \rid, l \st \region{dc}{\rid}(\pvar{x}, \pvar{done}, l, \p{true})}
\end{proofoutline}
 $}
  \caption{Proof Sketch of the Distinguishing Client.
    Here $
      L = \bigl(\exists l, d \st
            \region{dc}{\rid}(\pvar{x}, \pvar{done}, l, d) *
            \pars[\big]{l = 1 \dotimplies \envObl{\obl{k}}{\rid}}\bigr)
    $.}
  \label{fig:dist-client-outline}
\end{figure}

Now that we set the scene, we can proceed with the proof,
outlined in \cref{fig:dist-client-outline}.
The first operation to do is to transform the precondition
$ \ap{L}(\p{x},0) * \pvar{done}\mapsto \p{false} $
to an assertion about the $\region{dc}{\rid}(\p{x},\p{done},l,d)$ region.
We can do that by using the consequence rule:
\begin{inlineproofrule}
  \infer*[right={ConsH}]{
  \STABLE \actxt |= P
  \\
  \STABLE \actxt |= Q
  \\\\
  \viewshift \levl; \actxt |= P => P'
  \\
  \TRIPLE[\funspec] m; \levl; \actxt |-
      {P'} \cmd {Q'}
  \\
  \viewshift \levl; \actxt |= Q' => Q
}{
   \TRIPLE[\funspec] m; \levl; \actxt |-
      {P} \cmd {Q}
}
 \end{inlineproofrule}
which allows the use of \emph{viewshift} to logically manipulate the assertions.
Since triples are only well-defined if the Hoare pre-/postconditions
are stable,
the rule asks to check stability of the assertions of the conclusion
as this does not follow from stability
of the assertions of the triple in the premise.
An analogous rule, called~\ref{rule:consequence},
holds for hybrid triples
---with no stability checks on the atomic pre-/postconditions---
so viewshifting is available at any point in a derivation.

In our example, we want to create the guards and obligations needed to
match the interpretation of $\region{dc}{\rid}(\p{x},\p{done},l,d)$ and 
create the region, replacing its interpretation.
Here we might be tempted to match the interpretation
with $l=0$ and $d=\p{false}$,
$
  \ap{L}(\p{x}, 0) * \p{done} \mapsto \p{false} *
  \gOblA{\obl{k}}{\rid} * \gOblA{\obl{d}}{\rid}
$
to viewshift to
  $\region{dc}{\rid}(\p{x},\p{done},0,\p{false}) * \locObl{\obl{d}}{\rid}$.
While this viewshift holds, there is an issue:
in \tadalive, all the assertions in Hoare triples
(or in Hoare position in hybrid triples)
need to be stable for the triple to have well-defined semantics.
The proof system enforces the stability of these assertions,
by inserting stability checks in crucial rules.
This means that if we viewshift now to a non-stable assertion,
then we would fail at some point in the derivation to satisfy a stability check.
While
$\ap{L}(\p{x},0) * \pvar{done}\mapsto \p{false}$
is stable, as we own these resources, the assertion
  $\region{dc}{\rid}(\p{x},\p{done},0,\p{false}) * \locObl{\obl{d}}{\rid}$
is not stable:
  a region is subjected to the transitions of the protocol.
Since we have the guard $\obl{d}$ (from $\gOblA{\obl{d}}{\rid}$)
the environment cannot own it, hence cannot fire the transitions
guarded by $\obl{d}$; this makes $d=\p{false}$ stable.
The transitions changing the state of the lock, however, can affect the region.
This leads us to\footnote{Recall that $\bexp \dotimplies Q$ stands for
  $ (\bexp \land Q) \lor (\neg \bexp \land \emp) $.}
$
  \exists \rid, l \st
  \region{dc}{\rid}(\pvar{x}, \pvar{done}, l, \p{false}) *
  \gOblA{\obl{d}}{\rid} *
  \pars[\big]{l = 1 \dotimplies \envObl{\obl{k}}{\rid}}
$
where we also add $\envObl{\obl{k}}{\rid}$ when $l=1$,
a stable fact.

\subsection{The Parallel Rule}

We now want to proceed with an application of the parallel rule:
\begin{inlineproofrule}
  \infer*[right={Par}]{
  \TRIPLE[\funspec] m_1; \levl; \actxt |- {P_1}{\cmd_1}{Q_1}
  \\
  \VALID \actxt |= \minLay{Q_1}{m_2} \layleq m
  \\\\
  \TRIPLE[\funspec] m_2; \levl; \actxt |- {P_2}{\cmd_2}{Q_2}
  \\
  \VALID \actxt |= \minLay{Q_2}{m_1} \layleq m
}{\TRIPLE[\funspec] m; \levl; \acontext |-
    {P_1 \ssep P_2}
      {\cmd_1 \parallel \cmd_2}
    {Q_1 \ssep Q_2}
}
 \end{inlineproofrule}
The abbreviation
$ \VALID \actxt |= \minLay{P}{k} $
means
$ \forall r\in\RId \st
    \VALID \actxt |= P \implies \minLay{r}{k} $,
that is, all the obligations owned by~$P$ have layer~$\laygeq k$.
$ \VALID \actxt |= \minLay{P}{k} \layleq k'$ means
$ \VALID \actxt |= \minLay{P}{k} $ and $ {k} \layleq k' $.
The intuition behind these constraints is as follows.
The layer in the context of the triples indicates a strict upper bound
on the layers that can be assumed live in the proofs of the triples.
If thread~2 needs layers lower than~$m_2$,
then if thread~1 has unfulfilled obligations by the end of its execution,
these cannot conflict with the assumptions made by the proof of thread~2.
It is not however sound to leave
an obligation~$O_2$ of layer~$\laylt m_2$ unfulfilled by thread~1:
if thread~1 terminates first, leaving~$O_2$ unfulfilled in its postcondition,
thread~2 may be assuming~$O_2$ live in a situation where it will never be fulfilled.

In our example we need to apply consequence again to massage the viewshifted precondition into an assertion of the form $P_1 * P_2$.
The region assertion is duplicable, as is $l=1 \dotimplies \envObl{\obl{k}}{\rid}$,
but we want to give the non-duplicable resource $\gOblA{\obl{d}}{\rid}$
to the thread on the left, as it is the one that will fulfill the $\obl{d}$
obligation.
This has a side-effect though:
since the $\guard{d}$ guard is given to the left thread,
the value of $d$ from the right thread's perspective is not stably \p{false}.
Moreover, we want to know that the left thread has the obligation $\obl{d}$.
So we use
$\gOblA{\obl{d}}{\rid} \iff \gOblA{\obl{d}}{\rid} * \envObl{\obl{d}}{\rid}$
to give $\envObl{\obl{d}}{\rid}$ to the thread on the right, and we stabilise
the assertion to $\neg d \dotimplies \envObl{\obl{d}}{\rid}$.
All in all we obtain:
\begin{gather*}
  \bigl(
    \exists \rid, l \st
    \region{dc}{\rid}(\pvar{x}, \pvar{done}, l, \p{false}) *
    \gOblA{\obl{d}}{\rid} *
    \pars[\big]{l = 1 \dotimplies \envObl{\obl{k}}{\rid}}
  \bigr)
  \;\implies\; \exists \rid\st P_1 * P_2
  \\
  \begin{aligned}
    P_1 &\is \exists l \st
      \region{dc}{\rid}(\pvar{x}, \pvar{done}, l, \p{false}) *
      \gOblA{\obl{d}}{\rid} * \pars[\big]{l = 1 \dotimplies \envObl{\obl{k}}{\rid}}
    \\
    P_2 &\is \exists l, d \st
      \region{dc}{\rid}(\pvar{x}, \pvar{done}, l, d) *
      \pars[\big]{l = 1 \dotimplies \envObl{\obl{k}}{\rid}} *
      \pars[\big]{\neg d \dotimplies \envObl{\obl{d}}{\rid}}
  \end{aligned}
\end{gather*}
Which allows us to apply consequence and the standard \ref{rule:exists-elim}
rule to obtain a precondition in the form required by the \ref{rule:parallel}
rule.
For both threads we are aiming at postcondition
$
  \exists l \st
  \region{dc}{\rid}(\pvar{x}, \pvar{done}, l, \p{true})
$
which has no obligation so it satisfies the layer conditions trivially.

To see why the layer conditions are important for soundness,
imagine we forgot to unlock \p{x} in the left thread,
obtaining a non-terminating program.
We would obtain~$\locObl{\obl{k}}{\rid}$ in the postcondition of~$\cmd_\ell$,
but the check would fail as $\botOf=\lay(\obl{k}) \not\laygeq \layTop$.
Choosing $\botOf$ for the context layer of the triple
for~$\cmd_r$ would not work:
in its proof we need to assume~$\obl{d}$ live, and $\lay(\obl{d}) = \topOf$.

\subsection{Handling a call to \code{lock}}
\label{sec:rules-atomic}

Let us focus on the proof of the left-hand thread first.
The difficult step is the execution of the first instruction, since 
this is the only potentially non-terminating instruction of the thread.
If we let $
  L = \exists l, d \st
        \region{dc}{\rid}(\pvar{x}, \pvar{done}, l, d) *
        \pars[\big]{l = 1 \dotimplies \envObl{\obl{k}}{\rid}}
$,
\Cref{dist-client:left-lock-1} can be derived as follows:
\begin{equation*}
\footnotesize \infer*[Right=\ref{rule:layer-weak},vcenter]{
    \infer*[Right=\ref{rule:frame-hoare}]{
      \infer*[Right=\ref{rule:consequence}]{
        \infer*[Right=\ref{rule:atomicity-weak}]{
          \infer*[Right=\ref{rule:atomic-exists-elim}]{
            \ATRIPLE \topOf |-
              \A l \in \set{0,1}.
              < L | \exists d \st
                  \region{dc}{\rid}(\pvar{x}, \pvar{done}, l, d) >
              \code{lock(x)}
              < L * \gOblA{\obl{k}}{\rid} |
                \exists d \st \region{dc}{\rid}(\pvar{x}, \pvar{done}, 1, d) >
          }{
            \ATRIPLE \topOf |-
              < L | \exists l, d\st
                  \region{dc}{\rid}(\pvar{x}, \pvar{done}, l, d) >
              {\code{lock(x)}}
              < L * \gOblA{\obl{k}}{\rid} |
                \exists d\st \region{dc}{\rid}(\pvar{x}, \pvar{done}, 1, d) >
          }
        }{
          \TRIPLE \topOf |-
            { L * 
              \exists l, d \st
                \region{dc}{\rid}(\pvar{x}, \pvar{done}, l, d) }
            {\code{lock(x)}}
            { L * \gOblA{\obl{k}}{\rid} *
              \exists d\st
                \region{dc}{\rid}(\pvar{x}, \pvar{done}, 1, d) }
        }
      }{
        \TRIPLE \topOf |-
          { \exists l, d \st
              \region{dc}{\rid}(\pvar{x}, \pvar{done}, l, d) 
              * \pars[\big]{l = 1 \dotimplies \envObl{\obl{k}}{\rid}} }
          {\code{lock(x)}}
          { \exists d\st \region{dc}{\rid}(\pvar{x}, \pvar{done}, 1, d) *
              \gOblA{\obl{k}}{\rid} }
      }
    }{
      \TRIPLE \ghost[c]{\layTop}{\topOf} |-
        { \exists l \st
            \region{dc}{\rid}(\pvar{x}, \pvar{done}, l, \p{false}) *
            \gOblA{\obl{d}}{\rid} *
            \pars[\big]{l = 1 \dotimplies \envObl{\obl{k}}{\rid}} }
        {\code{lock(x)}}
        { \region{dc}{\rid}(\pvar{x}, \pvar{done}, 1, \p{false}) *
            \gOblA{\obl{d}}{\rid} *
            \gOblA{\obl{k}}{\rid} }
    }
  }{
    \TRIPLE \layTop |-
      { \exists l \st
          \region{dc}{\rid}(\pvar{x}, \pvar{done}, l, \p{false}) *
          \gOblA{\obl{d}}{\rid} *
          \pars[\big]{l = 1 \dotimplies \envObl{\obl{k}}{\rid}} }
      {\code{lock(x)}}
      { \region{dc}{\rid}(\pvar{x}, \pvar{done}, 1, \p{false}) *
          \gOblA{\obl{d}}{\rid} *
          \gOblA{\obl{k}}{\rid} }
  }
\end{equation*}

Let us unpack the derivation.
As a first step, we would like to frame the irrelevant resources,
in this case $\locObl{\obl{d}}{\rid}$.
In \tadalive, this step is more subtle and interesting than usual,
because of the layer-related side-conditions of \cref{rule:frame-hoare}
(a special case of \cref{rule:frame}):
\begin{inlineproofrule}
  \infer*[right={FrameH}]{
  \;\freevars(R) \inters \modvars(\cmd) = \emptyset
  \\\\
  \VALID \actxt |= \minLay{R}{m}\quad
  \\
  \TRIPLE[\funspec] m; \lvl; \actxt |-
    {P\phantom{{} * R}}
    \cmd
    {Q\phantom{{} * R}}
  \\
  \STABLE \actxt |= { R }
}{
  \TRIPLE[\funspec] m; \lvl; \actxt |-
    {P * R}
    \cmd
    {Q * R}
}
   \label{rule:frame-hoare}
\end{inlineproofrule}
With this rule, one can only frame obligations if they are of layer
greater or equal the context layer.
Here we can use consequence (omitted) to obtain a stable frame
$
  R =
    \exists l\st
    \region{dc}{\rid}(\pvar{x}, \pvar{done}, l, \p{false}) *
    \gOblA{\obl{d}}{\rid}
$
of the pre- and postconditions.
We have
  $\VALID |= \minLay{R}{\topOf}$
but
  $\not\VALID |= \minLay{R}{\layTop}$;
since the layer in the context of the goal is~$\layTop$,
before we can apply \ref{rule:frame-hoare} we need to
artificially lower the layer using \cref{rule:layer-weak-hoare}
before applying frame:
\begin{inlineproofrule}
  \infer*[right={LayWH}]{
  m_1 \layleq m_2
  \\\\
  \TRIPLE[\funspec] m_1; \lvl ; \actxt |- {P} \cmd {Q}
}{
  \TRIPLE[\funspec] m_2; \lvl ; \actxt |- {P} \cmd {Q}
}
   \label{rule:layer-weak-hoare}
\end{inlineproofrule}
Notice that lowering the layer in the context is always sound
(even for hybrid triples):
if we can prove the triple assuming live only layers $\laylt k_1 \layleq k_2$,
then the proof is valid in contexts where layers up to~$k_2$
can be assumed live.

The layer constraint of \ref{rule:frame} is crucial for soundness.
Suppose we remove the constraint.
Then we would be able to frame a locally held obligation~$O$
with layer $k \laylt m$,
i.e.~one of the layers that might be assumed live by the proof.
This would allow the proof to assume live environment obligations 
that have layer~$\laygeq k$, the eventual fulfilment of which might depend
on the eventual fulfilment of~$O$.
But since~$O$ is in the frame, it is constantly held and not fulfilled
for the whole duration of the execution of the command we are proving.
The frame condition forces us to record the ``minimum'' layer of the frame
in the context, ruling out unsound circular reasoning.

After framing, we use the rule of consequence to massage the assertions
to prepare them to the form required for
the later application of \ref{rule:liveness-check}.

The rest of the derivation does not involve liveness reasoning,
and follows a standard \tada\ proof pattern.
We use \ref{rule:atomic-exists-elim} and \ref{rule:atomicity-weak}
to turn the Hoare triple into an atomic triple:
\begin{inlineproofrule}
\infer*[right={A$\exists$Elim}]{
  \ATRIPLE[\funspec] m; \lvl; \actxt |-
    \A x \in \pqsets{X}, z \in Z.
    <\na{P} | \at{P}(x,z)>
      \cmd
    <\na{Q} | \at{Q}(x,z)>
  \phantom{, z \in Z}}{
  \ATRIPLE[\funspec] m; \lvl; \actxt |-
    \A x \in \pqsets{X}.
    <\na{P} | \exists z \in Z \st \at{P}(x,z)>
      \cmd
    <\na{Q} | \exists z \in Z \st \at{Q}(x,z)>
} \end{inlineproofrule}
\begin{inlineproofrule}
\infer*[right={AtomW}]{
  \STABLE \actxt |= {P'}
  \\
  \ATRIPLE[\funspec] m; \levl; \actxt |-
    <P {\,}|{\,} P'> \cmd <Q {\,}|{\,} Q'>
  \\
  \STABLE \actxt |= {Q'}
}{
  \TRIPLE[\funspec] m; \levl; \actxt |-
    {P * P'} \cmd {Q * Q'}
}
 \end{inlineproofrule}
\Cref{rule:atomic-exists-elim} says that if one can prove the command
is resilient to interference on~$z$ and does not affect the resource on~$z$
until its atomic update, then we can relax the specification to state
that the command allows changes to~$z$ and might also affect~$z$ during
the interference phase.
\Cref{rule:atomicity-weak} says that if you prove a command is atomic,
you can relax the specification not to insist on atomicity;
this can be done provided the atomic pre- and postcondition are stable,
as required for the Hoare triple to be well-defined.

The combination of \ref{rule:atomic-exists-elim} and
\ref{rule:atomicity-weak} simply states that if we can prove a
command performs an update atomically,
and the pre- and postconditions are stable,
then we can prove that the command also performs the update non-atomically.

\medskip

\omittedrulelabel*{SubPq}{rule:subst}Now let us consider the derivation for \cref{dist-client:left-lock-2},
which lifts the specification of CLH lock to the context of the client:
\begin{equation*}
\footnotesize \infer*[Right=\ref{rule:atomic-exists-elim},vcenter]{
    \infer*[Right=\ref{rule:lift-atomic}]{
      \infer*[Right=\ref{rule:frame}]{
        \infer*[Right=\ref{rule:subst-atomic}]{
          \ATRIPLE \topOf |-
            \A l\in\set{0,1} \eventually \set{0}.
              <\ap{L}(\p{x},l)>\code{lock(x)}<\ap{L}(\p{x},1) \land l=0>
        }{
          \ATRIPLE \topOf |-
            \A l\in\set{0,1} \eventually \set{0}, d\in\Bool.
              <\ap{L}(\p{x},l)>\code{lock(x)}<\ap{L}(\p{x},1) \land l=0>
        }
      }{
        {\ATRIPLE \topOf |-
          \A l\in\set{0,1} \eventually \set{0}, d\in\Bool.
            <\rInt(\region{dc}{\rid}(\pvar{x}, \pvar{done}, l, d)) >
              \code{lock(x)}
            <\rInt(\region{dc}{\rid}(\pvar{x}, \pvar{done}, 1, d))
            * \gOblA{\obl{k}}{\rid}>}
      }
    }{
      {\ATRIPLE \topOf |-
        \A l\in\set{0,1} \eventually \set{0}, d\in\Bool.
          < \emp |
            \region{dc}{\rid}(\pvar{x}, \pvar{done}, l, d) >
            \code{lock(x)}
          < \gOblA{\obl{k}}{\rid} |
            \region{dc}{\rid}(\pvar{x}, \pvar{done}, 1, d) >}
    }
  }{
    {\ATRIPLE \topOf |-
      \A l\in\set{0,1} \eventually \set{0}.
        < \emp |
          \exists d \st
            \region{dc}{\rid}(\pvar{x}, \pvar{done}, l, d) >
          \code{lock(x)}
        < \gOblA{\obl{k}}{\rid} |
          \exists d \st
            \region{dc}{\rid}(\pvar{x}, \pvar{done}, 1, d) >}
  }
\end{equation*}\Cref{rule:subst-atomic} simply gives a way to
manipulate the pseudo-quantified variable and its domain:
\begin{inlineproofrule}
  \infer*[right={SubPqA}]{
  \label{rule:subst-atomic}
    f \from X \to Y
    \\
    Y' = f(X')
    \\\\
    \forall x\in X\st
      \VALID \actxt |= P'(x) {\iff} P(f(x))
    \\
    \forall x\in X\st
      \VALID \actxt |= Q(f(x)) {\iff} Q'(x)
    \\\\
    \ATRIPLE[\funspec] m; \levl ; \actxt |-
      \A y \in Y \eventually[k] Y'.
      < P(y) >
        \cmd
      < Q(y) >
  }{
    \ATRIPLE[\funspec] m; \levl ; \actxt |-
      \A x \in X \eventually[k] X'.
      < P'(x) >
        \cmd
      < Q'(x) >
  }
\end{inlineproofrule}
These are manipulations that would normally be carried out using consequence, but need to be done specially since the pseudo-quantifier is a component
of triples and not of assertions.
In our example we simply use it to remove the unused variable $d$,
choosing $f\from \set{0,1}\times\Bool \to \set{0,1}$
to be the first projection.

The interesting step of the derivation of \cref{dist-client:left-lock-2} is
the application of \cref{rule:lift-atomic}:
\begin{inlineproofrule}
\footnotesize
  \specsdelim{big}
\infer*[right={LiftA}]{
\rid \in \dom(\actxt) \implies R = \mathop{id}
  \\
  \lsafe \lvl;\actxt |= {P(x), Q(x,z)}
  \\\\
  \oblfree \lvl{+}1;\actxt |= {P(x), Q(x,z)}
  \\
  \Set{
    ((x, O_1) , (z, O_2)) |
      x\in X \land R(x,z)
  }
  \subseteq
    \regLTS[\rt]( G )
  \\\\
  \phantom{ {+} 1}
  \ATRIPLE[\funspec] m; \lvl; \actxt |-
    \A x \in \pqsets{X}.
      < \rInt(\region[\lvl]{\rt}{\rid}(x))
        \ssep \guardA{G}{\rid} * P(x) * \locObl{O_1}{\rid} >
      \cmd
      < \exists z\st
          \rInt(\region[\lvl]{\rt}{\rid}(z)) \ssep
          Q(x,z) \land R(x,z) * \locObl{O_2}{\rid} >
}{
  \ATRIPLE[\funspec] m; \lvl {+} 1; \actxt |-
    \A x \in \pqsets{X}.
      < \locObl{O_1}{\rid} |
\region[\lvl]{\rt}{\rid}(x)
* \guardA{G}{\rid} * P(x) >
      \cmd
      < \locObl{O_2}{\rid} |
        \exists z\st
\region[\lvl]{\rt}{\rid}(z)
* Q(x,z) >
      \phantom{{}\land R(x,z)}
}
 \end{inlineproofrule}
Let us unpack it.
The purpose of the rule is to take an atomic specification that
applies to some resource, and lift it to the effect the atomic update
has on some region that contains that resource in its interpretation.
In our example, it says:
you have proven that the command locks the lock;
the lock is part of the interpretation of
$\region{dc}{\rid}(s, \pvar{x}, \pvar{done}, l, d)$
and the update to the lock
amounts to going from the interpretation with $l=0$ to the interpretation with $l=1$.
The rule needs to make sure that the region update is among the ones permitted by the associated protocol.
It does so by checking:
\begin{enumerate}
\item that there is a transition in the protocol matching the update;
\item that such transition is guarded by a guard that is owned;
\item that the local obligations are updated as the protocol mandates.
\end{enumerate}
To check the first condition, the rule uses a relation~$R$ between abstract states of the region;
by the fourth premise,~$R$ can only include updates
that the owner of~$G$ is allowed to perform.
The second condition is enforced by requiring the precondition to own~$G$.
The third condition is ensured by going from owning~$O_1$ to owning~$O_2$,
which, according to the fourth premise, is the expected update of obligations.
In our example we have $O_1 = \oblZero$ and $O_2 = \obl{k}$
and the update matches transition~\eqref{eq:dist-client-interf-lock}.
The third premise uses the abbreviation 
`$\oblfree \lvl;\actxt |= P$' which stands for
$\VALID \actxt |= P \implies \empObl[\lvl]$.
This implies that~$P$ and~$Q$ cannot own obligations of $\rid$, and so
$O_1$ and $O_2$ capture the whole of the updated obligations.
Note that because of the well-formedness restrictions on triples,
in the conclusion of the rule the obligations are transferred to the Hoare
pre/post-conditions: there they belong to a closed region.
The first premise says: if the region we are updating is tracked by
the atomicity context, this needs to be a trivial update,
or else it would count as a linearisation point 
(which is instead handled using \cref{rule:update-region}).
In our example $\actxt=\emptyset$ as we are not proving atomicity of the client,
so we are allowed any protocol compliant update.

Finally, the second premise $\lsafe \lvl;\actxt |= {Q(x,z)}$ requires~$Q$
to preserve its meaning at level $\lvl{+}1$.
The formal definition of \pre\lvl-safety is given in~\cref{app:lsafe};
all the \pre\lvl-safety conditions in our proofs
can be immediately discharged by applications of the following lemma.

\begin{lemma}
\label{lemma:lsafe}
  The properties below hold, for arbitrary~$\lvl\in\Level$:
  \begin{enumerate}
    \item $\emp$, $\vexp_1\mapsto \vexp_2$ and $\bexp$ are \pre\lvl-safe.
    \item $\guardA{G}{\rid}$ and $\locObl{O}{\rid}$ are both \pre\lvl-safe.
    \item If\/ $\lvl' < \lvl$,
          then $
              \region[\lvl']{\rt}{\rid}(a) *
              \envObl{O}{\rid}
          $ is \pre\lvl-safe.
          \label{lm:lsafe-closed}
    \item If\/ $P,Q$ are both \pre\lvl-safe,
          then so are
            $P \land Q$, $P \lor Q$, and $P*Q$.
    \item If $P(v)$ is \pre\lvl-safe for all $v\in\AVal$,
          then $\exists x\st P(x)$ is \pre\lvl-safe.
  \end{enumerate}
\end{lemma}

\omittedrulelabel*{$\exists$Elim}{rule:exists-elim}
\omittedrulelabel*{Cons}{rule:consequence}

\begin{figure}[p]
\centering
  \expandpqsets \begin{proofrule*}
    \infer*[right={LiveC}]{
  \ENVLIVE n; \lvl; \actxt |- L -M->> T
  \\
  m \laygeq n \quad k \laygeqq n
  \\
  \forall x \in X \st \VALID \lvl; \actxt |= \at{P}(x) * T \implies x \in X'
  \\\\ \ATRIPLE[\funspec] m; \lvl; \actxt |-
    \A x \in X \eventually[k] X'.
      <\na{P}| \at{P}(x)> \cmd \E y. <\na{Q}(x, y)|\at{Q}(x)>
}{\phantom{{}\eventually[k]X'}
  \ATRIPLE[\funspec] m; \lvl; \actxt |-
    \A x \in X.
      <\na{P} * L | \at{P}(x)> \cmd \E y. <\na{Q}(x, y) * L | \at{Q}(x)>
}
     \label{rule:liveness-check}
  \end{proofrule*}
  \par\medskip
  \hspace{-1em}
  \begin{proofrule*}
    \infer*[right=While]{\forall \beta \le \beta_0 \st
    \ENVLIVE m(\beta); \lvl; \actxt |- L - M ->> T(\beta)
  \\\forall \beta \le \beta_0 \st
    \VALID \actxt |= \minLay{P(\beta)}{m(\beta)} \layleq m
  \\\\\forall \alpha\st
    \STABLE \actxt |=
      {\exists \alpha'\st L * M(\alpha') \land \alpha' \leq \alpha}
  \\\progvars(T,L,M) \cap \modvars(\cmd) = \emptyset
  \\\\\forall \beta \le \beta_0 \st
    \forall b \in \Bool \st
    \TRIPLE[\funspec] m; \lvl; \actxt |-
        {P(\beta) *  (b \dotimplies T(\beta)) \land \bexp}
        \cmd
        {\exists \gamma \st P(\gamma)
          \land \gamma \leq \beta
          * (b \dotimplies \gamma < \beta)}
}{\TRIPLE[\funspec] m; \lvl; \actxt |-
    {P(\beta_0) * L}
      {\acode{while(BEXP)\{CMD\}}}
    {\exists \gamma \st  P(\gamma) * L
      \land \neg \bexp
      \land \gamma \leq \beta_0}
}
     \label{rule:while}
  \end{proofrule*}
  \par\medskip
  \begin{proofrule*}[.6\textwidth]
    \infer*[right={Par}]{
  \TRIPLE[\funspec] m_1; \levl; \actxt |- {P_1}{\cmd_1}{Q_1}
  \\
  \VALID \actxt |= \minLay{Q_1}{m_2} \layleq m
  \\\\
  \TRIPLE[\funspec] m_2; \levl; \actxt |- {P_2}{\cmd_2}{Q_2}
  \\
  \VALID \actxt |= \minLay{Q_2}{m_1} \layleq m
}{\TRIPLE[\funspec] m; \levl; \acontext |-
    {P_1 \ssep P_2}
      {\cmd_1 \parallel \cmd_2}
    {Q_1 \ssep Q_2}
}
     \label{rule:parallel}
  \end{proofrule*}\begin{proofrule*}[.4\textwidth]
    \infer*[right={LayWH}]{
  m_1 \layleq m_2
  \\\\
  \TRIPLE[\funspec] m_1; \lvl ; \actxt |- {P} \cmd {Q}
}{
  \TRIPLE[\funspec] m_2; \lvl ; \actxt |- {P} \cmd {Q}
}
     \label{rule:layer-weak}
  \end{proofrule*}
  \par\medskip
  \begin{proofrule*}
    \infer*[right={Frame}]{
  \forall x \in X \st
    \VALID \actxt |= \minLay{\na{R} * \at{R}(x)}{m}
  \\
  \progvars(\na{R}) \inters \modvars(\cmd) = \emptyset
  ,
  \progvars(\at{R}(x)) = \emptyset
  \\
  \STABLE \actxt |= { \na{R} }
  \\
  \forall x \in X\st
    \STABLE \actxt |= { \at{R}(x) }
  \\
  \oblfree \lvl;\actxt |= {\at{R}(x)}
  \\\\
  \ATRIPLE[\funspec] m; \lvl; \actxt |-
    \A x \in \pqsets{X}.
    <\na{P}\phantom{{} * \na{R}} | \at{P}(x)\phantom{{} * \at{R}(x)}>
      \cmd
    \E y.<\na{Q}(x,y)\phantom{{} * \na{R}} | \at{Q}(x,y)\phantom{{} * \at{R}(x)}>
}{
  \ATRIPLE[\funspec] m; \lvl; \actxt |-
    \A x \in \pqsets{X}.
    <\na{P} * \na{R} | \at{P}(x) * \at{R}(x)>
    \cmd
    \E y.<\na{Q}(x,y) * \na{R} | \at{Q}(x,y) * \at{R}(x)>
}
     \label{rule:frame}
  \end{proofrule*}
\par\medskip
\begin{proofrule*}\infer*[right={AtomW}]{
  \STABLE \actxt |= {P'}
  \\
  \ATRIPLE[\funspec] m; \levl; \actxt |-
    <P {\,}|{\,} P'> \cmd <Q {\,}|{\,} Q'>
  \\
  \STABLE \actxt |= {Q'}
}{
  \TRIPLE[\funspec] m; \levl; \actxt |-
    {P * P'} \cmd {Q * Q'}
}
     \label{rule:atomicity-weak}
  \end{proofrule*}
  \par\medskip
  \begin{proofrule*}
    \begingroup
\specsdelim{big}
\infer*[right={MkAtom}]{
  \lvl < \lvlp
  \\
  \rid \notin \dom(\actxt)
  \\
  \actxt' = \actxt\map{r -> (X,k,X',T)}
  \\\\
  T \subseteq \regLTS_{\rt}( G )
  \\
  R = \iorel(T)
  \\
  \forall x\in X \st
  \STABLE \actxt |= {
    \region[\lvl]{\rt}{\rid}(x) * \guardA{G}{\rid}
  }
  \\\\
  \TRIPLE[\funspec] m; \lvlp; \actxt' |-
    {\exists x \in X \st
      \region[\lvl]{\rt}{\rid}(x)
      * \done{\rid}{\blacklozenge}}
    \cmd
    {\exists x, y \st
        R(x,y) \land \done{\rid}{(x,y)}}
}{
  \ATRIPLE[\funspec] m; \lvlp ; \actxt |-
    \A x \in X \eventually[k] X'.
      < \region[\lvl]{\rt}{\rid}(x) * \guardA{G}{\rid} >
        \cmd
      < \exists y\st \region[\lvl]{\rt}{\rid}(y) * \guardA{G}{\rid} \land R(x,y) >
}
\endgroup
     \label{rule:make-atomic}
  \end{proofrule*}
  \par\medskip
  \let\proofrulesize\footnotesize
\begin{proofrule*}
    \specsdelim{Big}
\infer*[right={UpdReg}]{
  r \in \dom(\actxt)
  \\
  \actxt' = \actxt\map{r->\bot}
  \\
  \lsafe \lvl;\actxt |= {P(x), Q_1(x,y), Q_2(x,y)}
  \\\\
  \oblfree \lvl{+}1;\actxt |= {P(x), Q_1(x,y), Q_2(x,y)}
  \\
  \Set{((x, O_1) , (z, O_2(x))) | x\in X }
    \subseteq \trrel(\actxt, r)
  \\\\
  \hspace{3em}{\ATRIPLE[\funspec] m; \levl; \actxt' |-
    \A x \in \pqsets{X}.
    < \rInt(\region[\levl]{\rt}{\rid}(x)) * P(x) * \locObl{O_1}{\rid} >
      \cmd
    < \exists z\st
      \rInt(\region[\levl]{\rt}{\rid}(z)) *
      \locObl{O_2(x)}{\rid} *
      \left(
        \begin{array}{@{}r@{}l@{}}
             & R(x,z) \land Q_1(x, z) \\
        \lor & \ghost{R(x,z)}{x = z} \land Q_2(x)
        \end{array}
      \right)
    >}
}{
{\ATRIPLE[\funspec] m; \levl {+} 1;\actxt |-
    \A x \in \pqsets{X}.
      < \locObl{O_1}{\rid} |
        \region[\levl]{\rt}{\rid}(x) * P(x) * \done{\rid}{\blacklozenge} >
      \cmd
      < \locObl{O_2(x)}{\rid} |
        \exists z\st
        \region[\levl]{\rt}{\rid}(z) *
        \left(
          \begin{array}{@{}r@{}l@{}}
               & Q_1(x, z) * \done{\rid}{(x,z)} \\
          \lor & \ghost{Q_1(x, z)}{Q_2(x)}
                  * \done{\rid}{\blacklozenge}
          \end{array}
        \right)
      >}
}
     \label{rule:update-region}
  \end{proofrule*}
  \let\proofrulesize\small
  \par\medskip
  \begin{proofrule*}
    \specsdelim{big}
\infer*[right={LiftA}]{
\rid \in \dom(\actxt) \implies R = \mathop{id}
  \\
  \lsafe \lvl;\actxt |= {P(x), Q(x,z)}
  \\\\
  \oblfree \lvl{+}1;\actxt |= {P(x), Q(x,z)}
  \\
  \Set{
    ((x, O_1) , (z, O_2)) |
      x\in X \land R(x,z)
  }
  \subseteq
    \regLTS[\rt]( G )
  \\\\
  \phantom{ {+} 1}
  \ATRIPLE[\funspec] m; \lvl; \actxt |-
    \A x \in \pqsets{X}.
      < \rInt(\region[\lvl]{\rt}{\rid}(x))
        \ssep \guardA{G}{\rid} * P(x) * \locObl{O_1}{\rid} >
      \cmd
      < \exists z\st
          \rInt(\region[\lvl]{\rt}{\rid}(z)) \ssep
          Q(x,z) \land R(x,z) * \locObl{O_2}{\rid} >
}{
  \ATRIPLE[\funspec] m; \lvl {+} 1; \actxt |-
    \A x \in \pqsets{X}.
      < \locObl{O_1}{\rid} |
\region[\lvl]{\rt}{\rid}(x)
* \guardA{G}{\rid} * P(x) >
      \cmd
      < \locObl{O_2}{\rid} |
        \exists z\st
\region[\lvl]{\rt}{\rid}(z)
* Q(x,z) >
      \phantom{{}\land R(x,z)}
}
     \label{rule:lift-atomic}
  \end{proofrule*}
  \par\medskip
  \begin{proofrule*}
    \infer*[right={A$\exists$Elim}]{
  \ATRIPLE[\funspec] m; \lvl; \actxt |-
    \A x \in \pqsets{X}, z \in Z.
    <\na{P} | \at{P}(x,z)>
      \cmd
    <\na{Q} | \at{Q}(x,z)>
  \phantom{, z \in Z}}{
  \ATRIPLE[\funspec] m; \lvl; \actxt |-
    \A x \in \pqsets{X}.
    <\na{P} | \exists z \in Z \st \at{P}(x,z)>
      \cmd
    <\na{Q} | \exists z \in Z \st \at{Q}(x,z)>
}     \label{rule:atomic-exists-elim}
  \end{proofrule*}
\caption{Key \tadalive\ rules.
    Abbreviations:\\
    $ \VALID \actxt |= \minLay{P}{k} $
    means
    $ \forall r\in\RId \st
        \VALID \actxt |= P \implies \minLay{r}{k} $;
    \\
    $\lsafe \lvl;\actxt |= P$ can be discharged using \cref{lemma:lsafe};
    \\
    $ \oblfree \lvl;\actxt |= P $
    means
    $\VALID \actxt |= P \implies \empObl[\lvl]$;\\
    $k \laygeqq n$ means $(\forall k'\laygt k \st k' \laygeq n)$.
}
  \label{fig:proof-rules}
\end{figure}

\subsection{The \ref{rule:liveness-check} rule}

In a \tada\ safety proof, the derivations of
\cref{dist-client:left-lock-1}
and
\cref{dist-client:left-lock-2}
could be plugged together: the safety specification of the \code{lock} operation
does not contain the $\set{0,1}\eventually \set{0}$ component,
and the premise of 
\cref{dist-client:left-lock-1}
matches exactly the conclusion of
\cref{dist-client:left-lock-2}
(modulo framing~$L$, which would anyway not be used in a safety proof).
In \tadalive, the discrepancy between the two steps expresses the need
for a termination argument for this call.
What needs to be proven is the fact that,
in the current context of the $\rt[dc]$ protocol,
during the interference phase of this call to \code{lock(x)},
the environment will always eventually unlock the lock.
The \ref{rule:liveness-check} rule allows to remove the liveness condition
of the specification, in a context where the corresponding liveness invariant
can be proven to hold:
\begin{inlineproofrule}
  \infer*[right={LiveC}]{
  \ENVLIVE n; \lvl; \actxt |- L -M->> T
  \\
  m \laygeq n \quad k \laygeqq n
  \\
  \forall x \in X \st \VALID \lvl; \actxt |= \at{P}(x) * T \implies x \in X'
  \\\\ \ATRIPLE[\funspec] m; \lvl; \actxt |-
    \A x \in X \eventually[k] X'.
      <\na{P}| \at{P}(x)> \cmd \E y. <\na{Q}(x, y)|\at{Q}(x)>
}{\phantom{{}\eventually[k]X'}
  \ATRIPLE[\funspec] m; \lvl; \actxt |-
    \A x \in X.
      <\na{P} * L | \at{P}(x)> \cmd \E y. <\na{Q}(x, y) * L | \at{Q}(x)>
}
 \end{inlineproofrule}

The first premise 
$
 \ENVLIVE n;\lvl,\actxt |- L -M->> T 
$
is called the \emph{environment liveness condition},
and it roughly corresponds to checking
$
  \Always\,L \implies \AlwaysEv\,T
$
(with $M$ acting as a certificate of the property holding, explained later).
Here, we pick:
\begin{align}
  L &\is \exists l\in\set{0,1}, d \st
            \region{dc}{\rid}(\pvar{x}, \pvar{done}, l, d)
          * \pars[\big]{l = 1 \dotimplies \envObl{\obl{k}}{\rid}}
  \label{dist-client:livec-L}
  \\
  T &\is \exists d \st
            \region{dc}{\rid}(\pvar{x}, \pvar{done}, 0, d)
  \label{dist-client:livec-T}
\end{align}
and we can conclude $ \Always\,L \implies \AlwaysEv\,T $ because
when~$T$ does not hold, i.e.~$l=1$,
we know, from~$L$, that $\envObl{\obl{k}}{\rid}$ holds;
if we can show $\obl{k}$ is live,
the protocol says that if the environment holds~$\obl{k}$,
it will eventually fulfil it;
under the protocol the transition fulfilling it is setting~$l=0$
which brings us to~$T$.
The environment can always set~$l=1$ again after that,
but the same argument then applies.

To show $\obl{k}$ is live we have to look at the layers.
Here we have $m=\topOf$ and $k=\botOf=\lay(\obl{k})$.
Recall that $k \laygeqq n$ holds if
$\forall k'\laygt k \st k' \laygeq n$.
We can therefore set $n=\topOf$:
$\botOf \laygeqq \topOf$ holds since $\forall k'\laygt \botOf \st k' \laygeq \topOf$.
Since we don't own any obligation locally ($\obl{d}$ has been framed, recording this fact in the context layer $m$) we can consider $\obl{k}$ live when proving the environment liveness condition.

The environment liveness condition is the central component of both
\ref{rule:liveness-check} and \ref{rule:while};
we explain it in depth now,
and then resume our proof of the distinguishing client.

\subsection{The Environment Liveness Condition}
\label{sec:env-live}

\begin{figure}[tbp]
\begin{proofrule*}
  \infer*[right={EnvLive}]{
  \STABLE \lvl; \actxt |= L
  \\
  \VALID \lvl; \actxt |=
    L \implies L * \exists \alpha\st M(\alpha)
  \\\\
  \ENVLIVE m; \lvl; \actxt |- L * M(\alpha) : L * M(\alpha) -->> T
}{
  \ENVLIVE m; \lvl; \actxt |- L - M ->> T
}
   \label{rule:envlive}
  \end{proofrule*}
  \begin{proofrules}
    \infer*[right=ECase,vcenter]{
  \ENVLIVE m; \levl; \actxt |-
    L(\alpha) : L_1(\alpha) -->> T
  \\\\
  \ENVLIVE m; \levl; \actxt |-
    L(\alpha) : L_2(\alpha) -->> T
}{
  \ENVLIVE m; \levl; \actxt |-
    L(\alpha) : L_1(\alpha) \lor L_2(\alpha) -->> T
}
     \label{rule:envlive-case}
  \and
    \infer*[right=EQuant,vcenter]{
  \forall x \in X \st
  \ENVLIVE m; \levl; \actxt |-
    L(\alpha) : L(x, \alpha) -->> T
}{
  \ENVLIVE m; \levl; \actxt |-
    L(\alpha) : \exists x \in X \st L(x, \alpha) -->> T
}
     \label{rule:envlive-quant}
  \end{proofrules}
  \begin{proofrule*}
    \infer*[Right=LiveT,vcenter]{
  \forall \alpha. \VALID \actxt |= T'(\alpha) \implies T
}{
  \ENVLIVE m; \lvl; \actxt |-
    L(\alpha) : T'(\alpha) -->> T
}
     \label{rule:envlive-target}
  \end{proofrule*}
  \begin{proofrule*}
  \infer*[right=LiveO]{
  \decr[\actxt](L', L, T)
  \\
  \forall \alpha\st
    \VALID \actxt |= \minLayStrict{L'(\alpha)}{\lay(O(x))}
  \\\\
  \lvl < \lvlp
  \\
  \forall \alpha \st
    \SAT[\actxt] |- L'(\alpha) \implies
    \exists x\st \region[\lvl]{\rt}{\rid}(x) * \envObl{O(x)}{\rid} * \True
    \land m \laygt \lay(O(x))
}{ 
\ENVLIVE m; \lvlp; \actxt |-
  L(\alpha) :
  L'(\alpha)
  -->> T
}
   \label{rule:envlive-obl}
  \end{proofrule*}
  \begin{proofrule*}
  \infer*[right=LiveA]{
  \decr[\actxt](L', L, T)
  \\
  m \laygt k
  \\
  \forall \alpha\st
    \VALID \actxt |= \minLayStrict{L'(\alpha)}{k}
  \\\\
  (X \eventually[k] X') = \live(\actxt,r)
\\
  \lvl < \lvlp \\
  \SAT[\actxt] |- L'(\alpha) \implies
    \exists x \in X \setminus X' \st
    \region[\lvl]{\rt}{\rid}(x)
    * \done{r}{\lozenge}
    * \True
}{
  \ENVLIVE m; \lvlp; \actxt |-
    L(\alpha) :
    L'(\alpha)
    -->> T
}
   \label{rule:envlive-pq}
  \end{proofrule*}
  \caption{Environment Liveness Condition Rules}
  \label{fig:envlive}
\end{figure}

The essence of the termination argument is captured
in \labelcref{rule:liveness-check,,rule:while}
by the conditions of the form
$ \ENVLIVE m; \lvl; \actxt |- L - M ->> T $.
They establish ``always eventually~$T$ holds'' facts.
The condition is parametrised by~$L$,
an assertion that holds at any point in the traces
we are considering, an assertion $T$, characterising the so-called \emph{target}
states, and an assertion~$M(\alpha)$ parametric on some ordinal $\alpha$,
which represents the environment progress measure.
Intuitively, the condition states that,
from any state satisfying $L * M(\alpha)$, for some $\alpha$,
we can find an environment transition that \emph{must eventually} happen
that would take us either to $T$,
or to some state satisfying $L * M(\alpha')$ with $\alpha' < \alpha$.
Additionally, any transition from $L$~to~$L$ that \emph{may} happen
does not strictly increase the progress measure,
unless they end in a target state.
The transitions that must happen are characterised by being those
that either:
\begin{enumerate*}[label=(\arabic*)]
  \item fulfil some obligation known to be in the environment and with layer 
    lower than the ones we may hold locally, or
  \item fulfil some environment liveness assumption stored in $\actxt$
    with layer lower than the ones we may hold locally.
\end{enumerate*}
This entails that, under an environment that always eventually fulfils the obligations we are assuming live, $\Always\,L \implies \AlwaysEv\,T$ holds,
as desired.

In the \ref{rule:while} rule, an environment liveness condition
is combined with the condition
\[
  \forall \alpha\st
    \STABLE \actxt |=
      {\exists \alpha'\st L * M(\alpha') \land \alpha' \leq \alpha}
\]
which requires us to prove that any protocol compliant step from a state
satisfying $L * M(\alpha_0)$ for some $\alpha_0$ will take us to a state
satisfying $L * M(\alpha_1)$ for some $\alpha_1 \leq \alpha_0$.
In other words, in traces satisfying $\Always\,L$ the progress measure
never increases.
This, in conjunction with 
$ \ENVLIVE m; \lvl; \actxt |- L - M ->> T $,
entails $\Always\,L \implies \Finally\,T$,
as needed for soundness of \cref{rule:while}.

Take the environment liveness condition required by the application
of \ref{rule:liveness-check} in the proof of the distinguishing client.
Given $m = \topOf$ and $M(\alpha) = (\alpha = 0)$ we have to prove:
\[
  \ENVLIVE m; \lvl; \actxt |- 
    \exists l\in \set{0,1} \st
      \region{dc}{\rid}(\pvar{x}, \pvar{done}, l, \wtv) *
      \pars[\big]{l = 1 \dotimplies \envObl{\obl{k}}{\rid}}
  - M ->>
      \region{dc}{\rid}(\pvar{x}, \pvar{done}, 0, \wtv)
\]
That is, during the interference phase, we know that at any point in time
the lock will be in some state $l \in \set{0,1}$;
we want to prove that the environment will always eventually set~$l$~to~$0$.
Here this is particularly easy to show:
$L$ states that when $l = 1$ the obligation $\obl{k}$
is held by the environment;
since $\lay(\obl{k}) = \botOf \laylt \topOf = m$
(and $L$ does not hold obligations)
we can assume the obligation will be eventually fulfilled;
the only transition that can fulfil it, is the one that sets $l=0$,
so in exactly one such step we reach $T$.
This justifies the trivial definition of $M$:
we do not need to keep track of progress towards~$T$ as we reach it
in exactly one of the steps that \emph{must eventually} happen.

\begin{figure}[tb]
  \centering \begin{tikzpicture}[
  invariant/.style = {
    draw,thick,rounded corners=2pt,
  },
  target/.style = {
    invariant,
    fill = green!10
  },
  state/.style = {
    circle, fill=black,
    minimum size = 3pt,
    inner sep=0pt,
    outer sep=1pt,
  },
  alpha/.style args = {#1 #2}{
label={[font=\small,budget,inner sep=0pt,outer sep=0pt]#2:{\strut$\alpha_#1$}}
  },
  potential/.style = {
    dotted,oblig
  },
  obl/.style = {
    fill=white,circle,
    inner sep=.5pt,
  },
  pq/.style = {
    obl,inner sep=0pt,pos=.4
  },
  over/.style={
    preaction={draw=white,opacity=.8,line width=3pt,-, shorten <=3pt},
  },
  every edge/.style={
    over, draw
  },
]

\node[invariant,thick,fit={(0,0.5) (7.5,6)}] (L) {};
\useasboundingbox;
\node[below left] at (L.north west) {$L$};
\node[target,circle,minimum size=8em] (T) at (5,3) {};
\node[above right=.7] at (T) {$T$};

\draw[invariant]
  (L.182) -- (T.182) node[pos=0,above right]{$L_1$}
  (L.-70) -- (T.270) node[at=(L.south west),above right]{$L_2$}
  (L.20) -- (T.30) node[at=(L.south east),above left]{$L_3$}
;

\path (L.north west)
  ++ (-50:2.3) node[state,label=left:$w_0$,alpha=1 below] (x1) {}
++ (10:1)    node[state,alpha=1 below] (x2) {}
  +  (40:1)    node[state,alpha=1' above right] (x2p) {}
  ++ (-15:1)   node[state,alpha=1 182] (x3) {}
  +  (-30:1)   node[state] (x3p) {}
  ++ (230:1)   node[state,alpha=1 175] (x4) {}
  ++ (-100:1)  node[state,alpha=2 190] (x5) {}
  +  (130:1.1)  node[state,alpha=2'\; left] (x5p) {}
  ++ (-40:1)   node[state,alpha=2 below] (x6) {}
  ++ (20:1)    node[state,alpha=3 above] (x7) {}
  ++ (-80:1)  node[state,label=left:$w_1$] (x8) {}
  ++ (-10:1)  coordinate (x9)
;

\draw[->,bend left,thick]
  (x1) edge (x2)
  (x2) edge[potential] (x2p)
edge (x3)
  (x3) edge[potential] node[obl]{$\obl{a}$} (x3p)
       edge (x4)
  (x4) edge[oblig,bend right] node[obl]{$\obl{a}$} (x5)
  (x5) edge[potential,bend left] node[obl]{$\obl{b}$} (x5p)
       edge[bend right] (x6)
  (x6) edge[oblig,bend right] node[obl]{$\obl{b}$} (x7)
  (x7) edge (x8)
;
\draw[dashed,over,thick,->] (x9)
  (x8) to[bend left] (x9)
  to[out=-40,in=-30,looseness=3] ([xshift=.5cm,yshift=-.3cm]T)
;

\node[left=.5em,align=right,budget] at (L.30) {
$\alpha_1 > \alpha_2, \alpha_1'$\\
  $\alpha_3 > \alpha_2 > \alpha_2'$
};

\end{tikzpicture}   \caption{Illustration of \cref{rule:envlive} and the $\decr$ condition.}
  \label{fig:envlive-pic}
\end{figure}

The environment liveness condition can be proven
using the rules in \cref{fig:envlive}.
The only rule that applies directly is \ref{rule:envlive},
which checks that in a state satisfying~$L$ one can always measure progress
(second premise),
and then asks to discharge an auxiliary judgement of the form
$ {\ENVLIVE m; \lvl; \actxt |- L(\alpha):L(\alpha) - M ->> T} $
which is best explained with the help
of the illustration in \cref{fig:envlive-pic}.
The condition works under the hypothesis that the assertion~$L$ holds
at any point of the traces under consideration,
so in the picture we are considering infinite sequences of steps
within the outer rectangle.
The target states~$T$ describe some subset of~$L$,
which we want to show is visited infinitely often\footnote{Notice that ``$T$ is visited infinitely often''
  is equivalent to~$\AlwaysEv{T}$.}
by any infinite trace that
complies with the liveness rely as
specified by the region protocols and pseudo-quantifiers.
\Cref{rule:envlive-case} allows the splitting of the invariant~$L$ into
a disjunction of, say, $L_1, L_2,L_3$ and~$T$ as in the picture.
We need to prove there is going to be
eventual progress towards reaching~$T$
from each of these cases.
If we start from~$T$ we already reached the target,
and this case can be discharged by \cref{rule:envlive-target}.
The other cases are covered by
\Cref{rule:envlive-obl} which justifies progress
by appealing to an environment-owned atomic obligation~$O$
which is live~(premises two and three);
and
\Cref{rule:envlive-pq} which justifies progress
by appealing to a liveness assumption stored in the atomicity context.
The \ref{rule:envlive-quant} rule
is a generalisation of \cref{rule:envlive-case}.

To see how progress is justified,
consider the trace of \cref{fig:envlive-pic} starting from~$w_0$.
Assume the progress measure at~$w_0$ is~$\alpha_1$
(i.e.~$L*M(\alpha_1)$ holds in~$w_0$).
Each case $L_i$ can be discharged with either
\cref{rule:envlive-obl} or \cref{rule:envlive-pq}.
Imagine $L_1$ is discharged using \ref{rule:envlive-obl}:
the rule requires us to find an obligation $\obl{a}$ which,
in every state of $L_1$,
is necessarily owned by the environment ($\envObl{\obl{a}}{\rid}$)
for some region $\region{t}{\rid}(\wtv)$.
Then the $\decr$ condition checks that
the progress measure will \emph{improve}
when the environment will fulfil~$\obl{a}$; formally:

\begin{definition}[\/$\decr$]
\label{def:impr}\label{def:decr}
  Given assertions $L(\alpha)$, $L'(\alpha)$ and $T$,
the condition
  $\decr[\actxt](L', L, T)$
holds if and only if, for arbitrary $\store \in \Store$, letting
\begin{align*}
        l(\alpha) &= \WorldSem{\actxt}{\store}{L(\alpha)} &
        l'(\alpha) &= \WorldSem{\actxt}{\store}{L'(\alpha)} &
        t &= \WorldSem{\actxt}{\store}{T * \True}
\end{align*}
the following holds:
\[
\forall \alpha_1, \alpha_2 \ge \alpha_1 \st
{\relyAt[\actxt]}(l'(\alpha_1)) \inters l(\alpha_2) \subseteq l'(\alpha_1) \union t
\]
 \end{definition}

Intuitively, the $\decr$ condition
considers an arbitrary transition $(w_1, w_2)$
from the current case~$L'$ to~$L$, obeying the atomic rely
(thus allowed by the safety constraints of the protocols).
It then compares the progress measure~$\alpha_1$ and~$\alpha_2$,
taken before and after the transition,
checking that:
\begin{enumerate}[itemsep=0pt,parsep=0pt]
  \item the measure strictly improved ($\alpha_1 > \alpha_2$); or
  \item the measure stalled ($\alpha_1 = \alpha_2$)
        but we remained within case $L'$,
        and thus the pending obligation~$O$/liveness assumption
        are still pending; or
  \item we reached~$T$ (allowing the measure to vary arbitrarily)
\end{enumerate}

Examine the trace from~$w_0$ in \cref{fig:envlive-pic}.
While the trace stays within~$L_1$
the environment obligation~$\obl{a}$ stays unfulfilled,
(steps are labelled with the obligation they fulfil, if any)
and $\decr$ requires the measure~$\alpha_1$ to decrease,
or in the worst case stay constant.
Since~$\obl{a}$ is live, the environment will eventually fulfill it,
thus taking us outside of~$L_1$.
If such transition takes us to another case,
$L_2$~in the illustration,
$\decr$ requires the measure to strictly decrease to some $\alpha_2 < \alpha_1$.
This process cannot repeat \emph{ad libitum}:
the progress measure is an ordinal and hence well-founded.
The effect is that eventually, the only option is to reach the target.
Note that transitions that end in the target are allowed by $\decr$ to
increase the progress measure:
in the picture the transition reaching~$T$ increases the measure
from~$\alpha_2$ to~$\alpha_3$.
This allows the ``reset'' of the measure so that
the trace can go outside of~$T$ and
the whole process of reaching~$T$ again can be repeated
an unbounded number of times.

The idea behind \Cref{rule:envlive-pq} is analogous to the above description,
but progress is justified by appealing to
an environment liveness assumption stored in~$\actxt$.
The layer of the assumption needs to be lower than any layer we may be holding.
Since the environment liveness assumptions only hold
in the interference phase of an update,
the rule needs evidence that the linearisation point on~$r$
has not occurred yet, which is provided by $\done{r}{\lozenge}$.

In the proof of the distinguishing client,
the environment liveness condition
for the application of \cref{rule:liveness-check} between \cref{dist-client:left-lock-1} and \cref{dist-client:left-lock-2},
is proved by:\begin{small}\begin{equation*}
  \infer*[right={\ref{rule:envlive}},vcenter]{
    \infer*[right={\ref{rule:envlive-case}}]{
      \infer*[right={\ref{rule:envlive-target}}]{
        \forall \alpha\st\VALID \emptyset |= L_0(\alpha) \implies T
      }{
        \ENVLIVE \topOf; \emptyset |- L(\alpha) : L_0(\alpha) -->> T
      }
      \and
      \infer*[right={\ref{rule:envlive-obl}}]{
        \decr[\emptyset](L_1, L, T)
      }{
        \ENVLIVE \topOf; \emptyset |-
          L(\alpha) : L_1(\alpha) -->> T
      }
    }{
      \ENVLIVE \topOf; \emptyset |-
        L(\alpha) : L_0(\alpha) \lor L_1(\alpha) -->> T
    }
  }{
    \ENVLIVE \topOf; \emptyset |- L - M ->> T
  }
  \label{deriv:dist-client-envlive}
\end{equation*}\end{small}where $L$ and $T$ are defined in~\eqref{dist-client:livec-L} and~\eqref{dist-client:livec-T}, and~$M(\alpha)=(\alpha=0)$.
Since~$L$ trivially implies $L * \exists \alpha\st M(\alpha)$,
we can apply \ref{rule:envlive},
setting
$
  L(\alpha) = (L \land \alpha = 0)
$.
Then we apply \ref{rule:envlive-case} to
split on the value of~$l$:
$L(\alpha) = L_0(\alpha) \lor L_1(\alpha)$ where
$
  L_0(\alpha) = \region{dc}{\rid}(\pvar{x}, \pvar{done}, 0, \wtv) \land \alpha=0
$ and $
  L_1(\alpha) =
    \region{dc}{\rid}(\pvar{x}, \pvar{done}, 1, \wtv) *
    \envObl{\obl{k}}{\rid} \land \alpha=0
$.
If $l=0$ we can apply \ref{rule:envlive-target} as we are already in~$T$;
if $l=1$, $L_1$ entails~$\envObl{\obl{k}}{\rid}$
so we can apply \ref{rule:envlive-obl} with
$\region{t}{\rid} = \region{dc}{\rid}$ and $O = \obl{k}$.
The atomic obligation $\obl{k}$ is live:
$\topOf \laygt \lay(\obl{k}) = \botOf$, and~$L_1$ holds no obligations.
To check that the $\decr$ condition is satisfied,
we need to consider the transitions allowed by the protocol $\rt[dc]$:
\begin{itemize}
  \item $l=1$ to $l=1$ keeps the measure constant but keeps us in $L_1$,
  \item $l=1$ to $l=0$ brings us directly in~$T$.
\end{itemize}
Although in this case the progress measure is trivial
and the proof of the environment liveness condition simple,
the generality provided by non-trivial progress measures
is needed for more interesting examples.
For instance,
  our proofs of spin lock (\cref{examples:spin_lock})
  and CLH lock (\cref{examples:clh_lock})
do make use of this added generality.

We chose to state the $\decr$ condition as a semantic check;
while this achieves full generality,
in typical proofs the environment liveness condition
only involves a single region,
and $\decr$ can be checked by examining the region's protocol.

\subsection{The while rule}
\label{sec:while-rule}

By using the rules we described so far,
one can justify most of the proof outline of the distinguishing client
in \cref{fig:dist-client-outline}.
For instance, the proof of \code{lock(x)} for the left thread
can be reused as is to prove the \code{lock(x)} call
in the body of the loop of the right-hand thread.

The main missing step is the application of the \ref{rule:while} rule:
\begin{inlineproofrule}
  \infer*[right=While]{\forall \beta \le \beta_0 \st
    \ENVLIVE m(\beta); \lvl; \actxt |- L - M ->> T(\beta)
  \\\forall \beta \le \beta_0 \st
    \VALID \actxt |= \minLay{P(\beta)}{m(\beta)} \layleq m
  \\\\\forall \alpha\st
    \STABLE \actxt |=
      {\exists \alpha'\st L * M(\alpha') \land \alpha' \leq \alpha}
  \\\progvars(T,L,M) \cap \modvars(\cmd) = \emptyset
  \\\\\forall \beta \le \beta_0 \st
    \forall b \in \Bool \st
    \TRIPLE[\funspec] m; \lvl; \actxt |-
        {P(\beta) *  (b \dotimplies T(\beta)) \land \bexp}
        \cmd
        {\exists \gamma \st P(\gamma)
          \land \gamma \leq \beta
          * (b \dotimplies \gamma < \beta)}
}{\TRIPLE[\funspec] m; \lvl; \actxt |-
    {P(\beta_0) * L}
      {\acode{while(BEXP)\{CMD\}}}
    {\exists \gamma \st  P(\gamma) * L
      \land \neg \bexp
      \land \gamma \leq \beta_0}
}
 \end{inlineproofrule}
Let us review the main differences with the simplified
\ref{rule:while-blocking} rule presented in \cref{sec:overview}.
First, the two triples in the premises of \ref{rule:while-blocking}
(corresponding to the blocked and unblocked case)
are compressed here in a single triple:
this is convenient in proofs as the proof of the two triples
only differs on the treatment of the variant.
When $b=\p{false}$, $(b\dotimplies T) = \emp = (b\dotimplies \gamma < \beta)$,
obtaining the first triple of \ref{rule:while-blocking},
when $b=\p{true}$ we obtain the other triple.
Second, the target states~$T$ are parametrised over the variant~$\beta$:
each value of the variant may represent a different ``phase''
of the local progress of the while loop;
in each of these phases the loop may be blocked waiting for a different
set of target states to be reached.
Third, as anticipated,
the $\Always\,L \implies \Finally\,T$ condition is expressed as
the conjunction of the first and third premise.

There are two additional side-conditions.
Since $T,L$ and $M$ assert facts about arbitrary intermediate states
of an iteration, they cannot refer to any local variable that may be modified
by the body of the loop, hence the fourth premise.

The most important addition is the layer condition of the second premise.
The idea is that we should be forbidden from constantly owning
obligations of layers that we might assume live in the proof of the
environment liveness condition.
By requiring $ \minLay{P(\beta)}{m(\beta)} $ we make sure that
the loop invariant only owns obligations of layer higher than $m(\beta)$,
and the $m(\beta)$ in the context of the environment liveness condition
indicates that only layers lower than that may be assumed live.
The layer~$m$ in the context of the triple in the conclusion
is an upper bound for any layer that may be assumed live in the proof
of the loop.

Consider the application of \ref{rule:while} in the proof
of the distinguishing client.
The while loop of the right-hand thread is busy-waiting until \p{done}
is set to true.
The target states are therefore 
$
  T \is \region{dc}{\rid}(\pvar{x}, \pvar{done}, \wtv, \p{true})
$.
In this example, the target states do not depend on the variant $\beta$,
which itself is quite trivial:
when the loop is finally unblocked, it needs at most one iteration to
terminate.
The local variant can simply be $\beta = \ifte{\pvar{d}}{0}{1}$,
i.e.~when \p{d} is \p{false}
  there needs to be one unblocked iteration to terminate,
and when~\p{d} is finally \p{true}
  the loop will take no more iterations.
The loop invariant is
\[
  P(\beta) \is
    \exists l, d \st
      \region{dc}{\rid}(\pvar{x}, \pvar{done}, l, d) *
      \pars[\big]{l = 1 \dotimplies \envObl{\obl{k}}{\rid}} \land
      \pvar{d} \implies d \land
      \beta = \ifte{\pvar{d}}{0}{1}
\]
on which we can frame the stable assertion
\[
  L \is
    T \lor
    \region{dc}{\rid}(\pvar{x}, \pvar{done}, \wtv, \p{false}) 
    * \envObl{\obl{d}}{\rid}
\]
Since the loop invariant owns no obligations,
we can set $m(\beta)=\layTop=m$,
and we need to prove the environment liveness condition
$
  \ENVLIVE \layTop; \emptyset |- L - M ->> T
$;
here, as for the application of \ref{rule:liveness-check},
with the fulfilment of the environment obligation $\obl{d}$
we immediately reach the target, so $M$ can be trivial ($M(\alpha)=(\alpha=0)$).
The derivation is as follows:
\[
\small
\infer*[right={\ref{rule:envlive}}]{
  \infer*[right={\ref{rule:envlive-case}}]{
    \infer*[right={\ref{rule:envlive-target}}]{
    }{
      \ENVLIVE \layTop; \emptyset |- L(\alpha) : T -->> T
    }
    \quad
    \infer*[right={\ref{rule:envlive-obl}}]{
      \decr[\actxt](\region{dc}{\rid}(\pvar{x}, \pvar{done}, \wtv, \p{false}) * \envObl{\obl{d}}{\rid} , L, T)
    }{
      \ENVLIVE \layTop; \emptyset |- L(\alpha) : \region{dc}{\rid}(\pvar{x}, \pvar{done}, \wtv, \p{false}) * \envObl{\obl{d}}{\rid} -->> T }
  }{
  \ENVLIVE \layTop; \emptyset |- L(\alpha) : L(\alpha) -->> T }
}{
  \ENVLIVE \layTop; \emptyset |- L - M ->> T
}
\]
where
$L(\alpha) = L \land \alpha=0$.
We split~$L$ into two cases using \ref{rule:envlive-case}.
In the first case~$T$ holds, so \ref{rule:envlive-target} applies.
In the second, $d = \p{false}$ and since $\lay(\obl{d}) = \topOf < \top$,
$\obl{d}$ is a live obligation.
The $\decr$ condition is satisfied:
the allowed transitions either keep $d$ constant or set it to $\p{false}$,
taking us directly to $T$.

The stability of
$ \exists \alpha'\st L * M(\alpha') \land \alpha' \leq \alpha $
holds trivially as $\alpha$ is constantly~0.
The condition is this trivial in this case because
it checks that transitions to and from~$T$
are not resetting the progress measure;
here, once \p{done} is set to true,
it will not be set to \p{false} any more,
so once~$T$ is reached there is no way to leave it.

\paragraph{Non termination of distinguishing client with spin lock.}
If the lock at \p{x} is implemented as a spin lock,
the distinguishing client may not terminate.
Indeed, there is no \tadalive\ proof for the distinguishing client
if one assumes the spin lock specifications:
in the precondition we need to specify an impedance budget $\alpha$
for the lock $\ap{L}(\p{x}, 0, \alpha)$;
whatever ordinal we may choose for $\alpha$,
there is no way to consume some budget at every potential iteration
of the loop of $\cmd_r$ and never exhaust the budget,
as the number of iterations is effectively unbounded.

\subsection{Other rules}
\label{sec:other-rules}

The rules in \cref{fig:proof-rules} are the most important
\tadalive-specific rules.
We have omitted standard rules like the
axioms for primitive atomic commands, the
rules handling sequencing, function calls
(recall that for simplicity we restrict to non-recursive function definitions)
and structural manipulations.
They are reproduced in full in the Appendix.

Let us conclude with an explanation of the two \tadalive-specific rules
of \cref{fig:proof-rules} which are not illustrated by the proof
of distinguishing client:
\cref{rule:update-region,rule:make-atomic}.
While the goal of \ref{rule:lift-atomic} is to
lift an atomic update on a resource inside the interpretation of a region
to the corresponding update on the region itself,
\ref{rule:update-region} obtains the same effect but on a region $r$ that
is supposed to be updated once atomically (i.e.~$r\in \dom(\actxt)$).
While \ref{rule:lift-atomic} applies to regions with $r\in \dom(\actxt)$,
the update allowed in that case needs to be an identity step
from the point of view of the abstract state of the region.
A genuine update to the region needs to be recorded as the unique linearization
point on that region; this is precisely the purpose of \ref{rule:update-region}.
Most of the premises of \ref{rule:update-region} have the same function as
in \ref{rule:lift-atomic}: checking that the update of the abstract state
\emph{and} the obligations comply with the protocol.
The difference is that here
the update expected to take place in the linearization point
is recorded in $\trrel(\actxt, r)$
(i.e.~the components of $\actxt(r)$ recording the
expected update to abstract state and obligations of $r$).
To be able to claim the linearization point took place exactly once,
the precondition of the triple requires the $\notdone{r}$ resource
which represents the unique permission to perform the linearization point.
The postcondition allows for two cases:
either the update was successful,
in which case the atomicity tracking component is recording the update $(x,z)$;
or the update was not performed ($x=z$) and the $\notdone{r}$ resource
is still available for future updates.

\Cref{rule:make-atomic} is another crucial rule for proving abstract atomicity:
it states that a Hoare triple can be promoted to an atomic triple
if it contains a ``certificate'',
in the form of  $\notdone{r}$ being updated to $\done{r}{(x,y)}$,
that the region in question was updated atomically exactly once,
with the expected update.
The expected update, and the additional interference assumptions
given by the pseudo-quantifiers need to be stored in the atomicity context
so that the triple in the premise can make use of the interference precondition
assumptions, and certificate the right update took place.
Any expected update must be protocol compliant ($T\subseteq \regLTS[\rt](G)$).
Notice how the atomicity context records the liveness assumption expressed
by the pseudo-quantifier, so that it is available for termination proofs
in the proof of the triple in the premise;
in particular they can be used by applications of \ref{rule:envlive-pq}.
The proof of spin lock and CLH lock in \cref{sec:evaluation} illustrate applications of \ref{rule:make-atomic} and \ref{rule:update-region}.

\subsection{Abstract Predicates}
\label{sec:abs-pred}

In the spirit of CAP, abstract resources provided by a library should
be presented to clients by only exposing their abstract properties,
and not their definition.

In our example, the $\ap{L}(x,l)$ predicate is defined internally to the 
proof of the lock module, say using internal regions
(of some maximum level~$\lvl$) expressing the internal protocols of the module.

The proof of our distinguishing client only relies on the following
abstract properties:
\begin{enumerate}[label=(L\arabic*)]
  \item $ \ap{L}(\p{x},\wtv) * \ap{L}(\p{x},\wtv) $ is false,
        expressing that a lock is an exclusively owned resource;
        \label{ax:lock-excl}
  \item $\ap{L}(\p{x},l)$ is stable for all~$l$;
        \label{ax:lock-stable}
  \item $\ap{L}(\p{x},l)$ is \pre\lvl-safe for all~$l$;
        \label{ax:lock-safe}
  \item $\ap{L}(\p{x},l)$ is obligation-free,
    i.e.~$ \ap{L}(\p{x},l) \implies \empObl[\RId] $\\
        \label{ax:lock-obl-free}
        (which also entails~$ \ap{L}(\p{x},l) \implies \minLay{\rid}{m} $
        for all~$r\in\RId$ and~$m\in\Layer$).
\end{enumerate}
For instance,
the interpretation $ \rInt(\region{dc}{\rid}(x, \lvar{done}, l, d)) $
is well-formed thanks to properties~\ref{ax:lock-stable}
and~\ref{ax:lock-obl-free}.
The proof also involves side conditions on layers,
stability and \pre\lvl'-safety
which can be discharged by appealing to \ref{ax:lock-stable}
\ref{ax:lock-safe} and
\ref{ax:lock-obl-free}.

More generally, a module would typically expose to clients viewshifts representing separation properties of the abstract predicates (e.g.~duplicability),
stability properties,
\pre\lvl-safety, obligation freedom and relevant $ \minLay{P}{m} $ facts.

\subsection{What is leaked by \tadalive\ specifications?}
\label{sec:spec-encapsulation}

\tadalive's triples are rather expressive:
they support strong specifications of updates via logical atomicity,
and conditional termination properties via liveness assumptions.
It is natural to ask whether our triples force the leak of
any unnecessary detail about the implementation.
In particular, there are three components of the proof system that have a ``global'' flavour:
the level and layer in the context of the judgment,
and the layer decorating the liveness assumption of the pseudo-quantification.
Although necessary for soundness,
the management of levels is tedious but relatively straightforward.
Iris introduced namespaces for invariants to ease the management of so-called masks, which serve essentially the same function as levels in \tada.
A similar construction could be used to ease management of layers.
Here we keep it simple and require proofs of clients to use layers high enough to be able to reuse the libraries specifications.

The layers decorating a triple, on the other hand, are a more delicate matter.
The main complication arises from the choice of parametrising \tadalive\ with
a global layer structure.
If a specification insists on the use of a specific subset of layers,
that could seem like an unnecessary leak of implementation details.
For example, there could be two valid implementations of a module
that use wildly different internal layer structures to justify their internal blocking behaviour.
Should the abstract specification of the module insist on a specific layer structure for the internal layers, that would rule out valid implementations for no good reason.
In \tadalive\ modularity of the layers can be achieved
by exploiting a crucial property of derivations:
their validity is invariant under a strict-ordering-preserving remapping of layers.
This allows a style of specification which generalises the one we have seen in our examples until now, where the layer structure relevant for the proof of an implementation is parametrised over a client-provided remapping of layers.
To avoid cluttering the proofs we do not explicitly parametrise the proofs in
\cref{sec:evaluation}.
In \cref{sec:lock-coupling-set},
where the construction becomes relevant and used in a non-trivial way,
we explain how to convert a proof so that it is parametric on the layer remapping.

In terms of behaviour,
\tadalive's specifications are able to hide internal blocking,
as showed by the blocking counter example
of \cref{sec:proof-reuse} (formalised in \cref{app:ex-blocking-counter}).
There is, in fact, one progress property leaked by the specification layers
that is currently not exploited by \tadalive.
In the special case when the layer in the context
is the \emph{globally} smallest layer $\layBot$,
the proof of the triple cannot rely on any liveness assumption at all.
This can be used to differentiate, for example,
a wait-free counter implemented as an hardware-atomic fetch-and-add
(which admits a proof with~$\layBot$ in the context)
and a blocking counter (which only admits proofs with layer~$> \layBot$).
This is a useful distinction:
wait-freedom is an important progress property,
asserting termination without assumptions on liveness of other threads and
\emph{without fairness assumptions on the scheduler}~\cite{Herlihy11}.
Currently, however, \tadalive's semantics does not support
deriving wait-freedom as a consequence of $\layBot$ as the context layer:
the current triple semantics only implies
termination of the \emph{fair} traces.
Extending \tadalive's semantics to encompass wait-freedom is left as future work.

\subsection{Soundness}
\label{sec:soundness}

We have proven soundness of \tadalive\ rules
against the semantic judgement of \cref{def:judgement-sem}.

\begin{theorem}[Soundness]
\label{th:soundness}
  If
  $\JUDGE[\funspec] |- \cmd : \spec$
  then
  $\JUDGE[\funspec] |= \cmd : \spec$.
\end{theorem}

The detailed proofs of the liveness-related rules
are produced in \cref{app:soundness}.
The soundness of most rules is an adaptation
of the soundness arguments of the corresponding \tada\ rules.
The rules that drive the liveness argument are
\cref{rule:parallel}, \cref{rule:liveness-check} and \cref{rule:while}.

The soundness of the parallel rule follows from 
the layered liveness invariants semantics
explained in \cref{sec:specs-semantics}.
The argument is roughly as follows.
There are two possible ways the parallel composition
$\cmd_1 \parallel \cmd_2$
may fail to terminate:
either one thread terminates and the other does not,
or they both do not terminate.
In the first case, when the terminating thread, say $\cmd_1$, terminated,
we are in a state where thread 1 does not own any obligation of layers that
may be assumed live by $\cmd_2$ (this is from the conditions on the layers of the postcondition of $\cmd_1$).
By the triple about $\cmd_2$ in the premises,
$\cmd_2$ is only allowed not to terminate if the environment is constantly owning an obligation $O$ of layer lower than $m_2$.
Since $\cmd_1$ cannot do that, we obtain that said $O$ must be owned by the overall environment of the parallel composition. In such case the triple
of the conclusion allows the program to diverge.

In the second case, both threads are not terminating.
Each of the threads, say~1,
is allowed to keep an obligation constantly unfulfilled,
as long as it can blame thread~2 by showing
an obligation of strictly lower layer
that is kept constantly unfulfilled by~2.
Since layers are well-founded there needs to be some thread
that will not be justified in not fulfilling some of its obligations.
This cannot be as we were able to prove the two triples in the premises.

The soundness of \cref{rule:while} considers the worst-case scenario for progress: an infinite sequence of iterations, all of which do not start from a target state in~$T(\beta)$, and therefore do not decrease the variant~$\beta$.
In such case we know that the assertion~$L$ holds at every point of the trace:
it has been framed so the local steps and the environment steps must preserve it, and it is stable (as checked by \ref{rule:envlive}).
We are thus within the hypothesis of the environment liveness condition,
which proves, together with the premise asking the progress measure~$\alpha$ to never increase, that eventually the target states will be reached.
Although they may be reached in the middle of an iteration, instead of at the beginning as it would be required to invoke the triple that decreases the variant~$\beta$,
in the worst case this can happen boundedly many times (the progress measure is well-founded and must always decrease).
Therefore we eventually reach $T(\beta)$ and not leave it until the next iteration
starts from a state satisfying $P(\beta) * T(\beta) \land \bexp$,
which matches the premise that ensures the variant decreases.
This can only happen boundedly many times as the variant is well-founded.

\Cref{rule:liveness-check}'s soundness argument is a variation of the one for
\ref{rule:while}.

\medskip

As a simple corollary of soundness and \cref{th:adequacy},
if we can prove
${\TRIPLE m;\lvl;\emptyset |- {\emp}\cmd{\True}}$
then $\cmd$ run in isolation terminates from the empty heap.
For our distinguishing client (\cref{ex:distinguishing-client})
for instance, we can wrap up the proof by initialising the state
and prove
\[
  \TRIPLE \layTop |-
    {\emp}
    {
      \code{var done=false, x in x:=makeLock();}
      (\cmd_\ell \parallel \cmd_r)
    }
    {\True}
\]
which implies termination of the program.
 \section{Evaluation}
\label{sec:evaluation}

In the previous section, we introduced the \tadalive\ proof system,
explaining the rules on the distinguishing client,
which showcases in a simple setting the proof mechanics of the logic.

In this section we consider more challenging case studies,
to demonstrate how \tadalive\ achieves proof scalability and reuse in practice.

We start by proving correctness of the spin and CLH lock implementations
against the specifications we discussed in \cref{sec:overview}.
The proof of spin lock highlights the use of the liveness assumption
of a pseudo-quantifier in a proof,
and the handling of impedance through the impedance budget.
The proof of CLH has a number of interesting features.
The CLH code exhibits both \emph{internal blocking},
i.e.~blocking that is resolved internally and does not leak to the client,
and \emph{external blocking},
i.e.~blocking that has to be resolved by the client
and thus leaks in the liveness assumption of the pseudo-quantifier.
As a consequence,
the termination argument requires using a combination of
obligations (for internal blocking) and
the liveness assumption of the pseudo-quantifier (for external blocking).
Moreover, the obligations (and their layers)
are not simple tokens like the ones for the simple
examples of \cref{sec:overview,sec:rules},
but form an infinite set.
This reflects the unboundedness of the internal queue of threads.

The two lock examples demonstrate \tadalive's ability to
abstract from implementation details,
and only leak to the client the parts of the termination argument
which depend on the choices of the clients.
In the same vein, we will follow this with
a counter module using a spin lock
to protect access to a cell holding the value of the counter.
Interestingly, since the blocking due to the use of a lock is internal,
the specification of the counter will not be blocking.
The impedance suffered by the internal spin lock does however leak
to the interface for the counter:
the counter will have its own impedance budget which will be internally
spent to call operations on the lock.

To exhibit \tadalive's ability to reason about liveness locally,
we will verify a double blocking counter, 
showing that for simple common programming patterns, 
the layer system leads to natural and modular client proofs.

Finally, we comment on a proof of a lock-coupling set,
produced in full in \cref{appendix:lock-coupling-set}.
The example considers a data structure implemented as a linked list
with CLH locks guarding the single cells.
The example is challenging for the presence of a dynamic number of locks.
At first sight it might seem it is impossible to represent this using
the static association of layers to obligations of \tadalive.

Obligations however, as demonstrated in this case study,
are a very general form of ghost state and can easily represent dynamic
properties of state.

\paragraph{Other case studies.}
Ticket lock and MCS lock~\cite{ArtBook} are
alternative implementations of starvation-free locks;
they can be given the same specification as the CLH lock,
and their liveness argument can be carried out in the same way
as the one we present for CLH.\footnote{The proof of ticket lock requires some minor ghost code to side-step the lack of support for helping.
}
A paradigmatic example of fine-grained data structure
is the Treiber stack~\cite{ArtBook} which, in its standard form, is non-blocking
and has been proven in Total~\tada\ already.
It is easy to adapt the code to have a \code{pop} operation which blocks
on an empty stack.
Such operation would be blocking and suffers impedance.
Its specification and proof mirrors closely the proof of the spin lock.
Challenging variants of the lock-coupling set are
the ``optimistic'' and ``lazy'' sets.
The proof of optimistic set uses a combination of
the proof of the lock-coupling set and the impedance budget technique
(optimisic set operations impede each other).

\medskip

These case studies cover all the proof patterns needed
to prove all the examples of the LiLi papers~\mbox{\cite{LiangF16,LiangF18}}.
Notably, proofs in LiLi involving modules that use locks,
require in-lining some non-atomic implementation of the lock operations
in the client, resulting in non-modular proofs and
unnecessarily intertwined termination arguments.

\subsection{Spin lock}
\label{examples:spin_lock}

\paragraph{Code}
The spin lock module implements a lock by storing a single bit in
a heap cell; locking is implemented by trying to
CAS the heap cell from~0 to~1 until the CAS succeeded;
unlocking simply sets the cell back to~0.
In \cref{fig:spin-code} we give all the operations of a spin lock module.

\begin{figure}
  \centering\small \begin{tabular}{c@{\hspace{3em}}c@{\hspace{3em}}c}
    \codefromfile{examples/spin-lock/code-make} &
    \codefromfile{examples/spin-lock/code-lock} &
    \codefromfile{examples/spin-lock/code-unlock}
  \end{tabular}
  \caption{Code of spin lock operations.}
  \label{fig:spin-code}
\end{figure}

\paragraph{Specifications}
We will prove the module satisfies the following specifications:
\begin{align*}
  \forall \alpha \st\;
  &\TRIPLE \botOf |-
    {\emp}
      {\code{makeLock()}}
    {\exists \rid\st \ap{L}(\rid, \pvar{ret}, 0, \alpha)}
  \\
  \forall \ordfun \st\;
  &\ATRIPLE \topOf |-
    \A l \in \set{0, 1} \eventually[\botOf] \set{0}, \alpha.
      <\ap{L}(\rid, \pvar{x}, l, \alpha) \land {\alpha > \ordfun(\alpha)}>
        \code{lock(x)}
      <\ap{L}(\rid, \pvar{x}, 1, \ordfun(\alpha)) \land l = 0>
  \\
  &\ATRIPLE \botOf |-
    <\ap{L}(\rid, \pvar{x}, 1, \alpha)>
      \code{unlock(x)}
    <\ap{L}(\rid, \pvar{x}, 0, \alpha)>
\end{align*}
where $\ap{L}(\rid, \pvar{x}, l, \alpha)$ abstractly represents
the lock resource
  at abstract location~$\rid$
  (omitted for readability in \cref{sec:overview}) and
  concrete address~$\pvar{x}$,
  with abstract state~$l \in \set{0,1}$ and
  impedance budget~$\alpha$ (an ordinal).
The purpose of the impedance budget, as described in \cref{sec:overview},
is to prevent the environment from taking possession of the lock
an unbounded number of times. Without this bound, the \p{CAS} operation
in the implementation of \p{lock} could be indefinitely preempted
by the environment locking the lock, preventing it from ever 
taking its possession and terminating, even if the environment
always unlocks the lock when it is locked.
This is enforced by requiring the lock operation to strictly decrease
the impedance budget using
$\phi \from \Ord \to \Ord$, a function that can be freely instantiated
by the client upon usage of the specification,
which indicates precisely how much the budget will decrease after this call
(which is client dependent information).
The specification of \code{makeLock} then allows the client
to pick an arbitrary ordinal as the initial budget.

\paragraph{Shared Region}
The abstract shared lock resource will be represented by a region
$ \region{spin}{\rid}(x,l,\alpha) $
where
  $ x \in \Addr $,
  $ l \in \set{0,1} $,
  $ \alpha \in \Ord $.
Here $x$ is a fixed parameter of the region.

\begin{convention}
  \label{conv:exc-guard}
  An exclusive guard, $\gEx$, is very commonly used
  to express some exclusive permission on some shared resource,
  which cannot be composed with itself:
  i.e.~$\guardUndef{\gEx\guardOp\gEx}$.
  Local ownership of $\gEx$ is exclusive in that
  no other thread can at the same time assert ownership of $\gEx$.
  A ubiquitous use of this guard is in representing the resource offered by a module.

  Take for example the current spin lock module.
  Since this is a concurrent module it uses internally shared resources.
  We therefore have a region
  $ \region{spin}{\rid}(x,l,\alpha) $
  encapsulating the shared internal resources of the counter.
  From the perspective of the client, however,
  at the moment of creation of a lock by, say,
  a \code{makeLock()} operation,
  the lock is exclusively owned by the client.
  This, for example, is reflected in the fact that, until the client
  shares the lock or invokes operations on it,
  it remains unlocked.
  To represent this fact, one typically defines an exclusive guard $\gEx$
  guarding each transition of the region interference:
  e.g.~$ \gEx: (0,O_1) \interfTo (1,O_2),  \gEx: (1,O_1) \interfTo (0,O_2)$.
  Then the \code{makeLock()} operation can be given the specification above,
  which gives to the client the stable assertion
  $ \region{spin}{\rid}(\ret,0,\alpha) * \guardA{\gEx}{\rid} $, wrapped in the
  predicate $\ap{L}(\rid,\ret,0,\alpha)$.
  (Note how $\region{spin}{\rid}(\ret,0,\alpha)$ is not stable on its own.)
  To re-share the lock, the client will create its own region
  encoding the invariants governing the interaction over the lock
  (and the other resources of the client)
  the interpretation of which will contain
  the guard $\guardA{\gEx}{\rid}$.

  Note that assertions
  have very different meanings if occurring in the \emph{atomic} precondition
  of a triple, as opposed to the \emph{Hoare} precondition:
  the resources in the atomic precondition are not owned by the local thread,
  but only acquired instantaneously at the linearization point.
  For example, in the triple
  \[\small
  \forall \ordfun \st \ATRIPLE |-
    \A l \in \set{0,1} \eventually[\botOf] \set{0}, \alpha.
      <\region{spin}{\rid}(\pvar{x},l,\alpha)
      * \guardA{\gEx}{\rid} \land {\alpha > \ordfun(\alpha)}>
      {\code{lock(x)}}
      <\region{spin}{\rid}(\pvar{x},1,\ordfun(\alpha))
      * \guardA{\gEx}{\rid} \land l {=} 0>
  \]
  the exclusivity of $\gEx$ is only granted \emph{instantaneously} to the thread
  acting on it atomically, i.e.~either the environment during the interference phase
  as allowed by the pseudo-quantifier, or the local thread at the linearization point.

  Since this pattern is ubiquitous,
  we reserve the $\gEx$ guard constructor for this use, and
  will omit the $\guardUndef{\gEx\guardOp\gEx}$
  axiom when specifying guard algebras.
\end{convention}

\paragraph{Guards and Obligations}
For the $\rt[spin]$ region we only have the exclusive guard $\gEx$,
and no obligation constructors, as the implementation has no internal blocking.
All the blocking behaviour is represented by the liveness assumption
in the pseudo-quantifier of the specification of \code{lock}.
Note that without the exclusive guard, the specification of
\p{makeLock} would not hold as the lock would not be stably
unlocked.

\paragraph{Region protocol}
The interference protocol for $\rt[spin]$ is very simple:
\begin{align*}\gEx & \from
    ((0, \alpha), \oblZero)
      \interfTo
    ((1, \beta), \oblZero)
  \text{ only if }\beta < \alpha
\\
\gEx &\from
    ((1, \alpha), \oblZero)
      \interfTo
    ((0, \alpha), \oblZero)
\end{align*}
It states that whoever owns $\gEx$ can freely acquire or release the lock,
provided that at each acquisition, some budget is spent ($\beta < \alpha$),
preventing the lock from being locked an unbounded number of times.

\paragraph{Region interpretation}
The implementation uses a single cell stored in the heap, and we have no non-trivial guards/obligations; the interpretation is thus straightforward:
\[
  \rInt(\region{spin}{\rid}(x, l, \alpha)) \is x \mapsto l
\]
Note how $\alpha$ is pure ghost state in that it is not linked
to any information in the concrete memory.

\begin{mathfig}
  \begin{proofoutline}\llap{\META{\forall \phi\st\,}}\CTXT \topOf; \emptyset |-
  \A l \in \set{0, 1} \eventually[\botOf] \set{0}, \alpha. \\
  \PREC<\ap{L}(\rid, \pvar{x}, l, \alpha) \land {\alpha > \ordfun(\alpha)}>
  \\
  \begin{proofjump}
    [rule:consequence]
    \PREC<\region{spin}{\rid}(\pvar{x}, l, \alpha) * \guardA{\gEx}{\rid}
          \land {\alpha > \ordfun(\alpha)}
         >
    \\
    \begin{proofjump}
      [rule:make-atomic]
      \CTXT \topOf;
        \map{r -> \bigl(
          \set{0,1} \times \Ord, \botOf, \set{0} \times \Ord,
((0,\alpha), \oblZero) \interfTo ((1,\ordfun(\alpha)), \oblZero)
        \bigr)}
      |- \\
      \PREC{\exists l, \alpha \st
              \region{spin}{\rid}(\pvar{x}, l, \alpha)
              \land \alpha > \ordfun(\alpha)
              * \notdone{\rid}
           } \\
      \CODE{var d = 0 in}\\
      \ASSR{\exists l, \alpha \st
              \region{spin}{\rid}(\pvar{x}, l, \alpha)
              * (\pvar{d}=0 \land
                 \notdone{\rid}
                 \land \alpha > \ordfun(\alpha)
                )
           } \\
      \begin{proofjump}
        [rule:consequence,rule:exists-elim]
        \PREC{P(\beta_0)} \\
        \begin{proofjump}[rule:while]
\CODE{while(d=0)\{}\\
          \begin{proofindent}
            \META{\forall b \in \Bool, \beta\st}\\
            \PREC{P(\beta) *
              b \dotimplies T(\beta)
              \land \pvar{d} = 0
            }\\
            \CODE{d := CAS(x,0,1);}\\
            \POST{\exists \gamma \st P(\gamma) \land
              \beta \ge \gamma \land
              b \implies \gamma < \beta
            }
          \end{proofindent} \\
          \CODE{\}}
        \end{proofjump} \\
        \POST{\exists \gamma \st P(\gamma) \land
          \beta_0 \ge \gamma \land \pvar{d} \neq 0
        }
      \end{proofjump}
      \\
      \POST{\exists l,\alpha\st
              \done{\rid}{((l, \alpha), (1, \ordfun(\alpha)))} \land l = 0
           }
    \end{proofjump}
    \\
    \POST<\region{spin}{\rid}(\pvar{x}, 1, \ordfun(\alpha)) *
          \guardA{\gEx}{\rid} \land l = 0>
  \end{proofjump}
  \\
  \POST<\ap{L}(\rid, \pvar{x}, 1, \ordfun(\alpha)) \land l = 0>
\end{proofoutline}
   \caption{Spin lock: proof of \protect{\code{lock}}.}
  \label{fig:spin-lock-outline}
\end{mathfig}

\paragraph{Predicates}
The lock resource is abstractly represented by the predicate
\[
  \ap{L}(\rid, x, l, \alpha) \is
    \region{spin}{\rid}(x, l, \alpha) * \guardA{\gEx}{\rid}
\]

\paragraph{Verification of \code{lock}}
\Cref{fig:spin-lock-outline} is the proof of the
\code{lock} operation.
The only step that involves reasoning about liveness is
the application of the \ref{rule:while} rule.
To apply this rule, we must first define the loop invariant,
$P(\beta)$, the target states, $T(\beta)$, the persistent loop invariant, $L$,
$m(\beta)$, and the environmental progress measure, $M(\alpha)$.

\noindent
The loop invariant is
\begin{equation*}
  P(\beta) \is
    \exists l, \alpha \st
      \region{spin}{\rid}(\pvar{x}, l, \alpha) \land
      \beta \ge \alpha
* \left(
        \begin{array}{@{}r@{\;}l@{}}
          &\bigl(
            \pvar{d} = 0 \land
            \notdone{\rid} \land
            \alpha > \ordfun(\alpha)
          \bigr)
          \\\lor&
          \bigl(\exists l', \alpha' \st
            \pvar{d} = 1 \land
            \done{\rid}{((l', \alpha'), (1, \ordfun(\alpha')))}  \land
            l' = 0
          \bigr)
        \end{array}
        \right)
\end{equation*}
which contains:
\begin{itemize}
  \item the safety information to prove the uniqueness of the linearization point, namely that
    if the \code{CAS} failed, i.e. $\pvar{d}=0$, then we have not touched the resource yet and we still have permission to perform the linearization point ($\notdone{\rid}$);
    whereas if the \code{CAS} succeeded, i.e. $\pvar{d}=1$, then
    we did perform the linearization point with the expected effect.
  \item the definition of the local variant $\beta$
    as an upper bound on the impedance budget $\alpha$.
\end{itemize}
Indeed, whenever some budget is spent,
the loop approaches termination
as eventually, the exhaustion of the budget prevents further interference,
allowing the \p{CAS} operation to succeed and the loop to terminate.
Therefore decreasing the upper bound to the interference budget
corresponds to progress for the \p{while} operation.
Without additional information however, we cannot show the local variant
must eventually strictly decrease, indeed,
in the case $l=1$ we cannot exit the loop and the environment is not forced
to spend budget.
Therefore, the termination argument will need the 
assumption that the environment always eventually unlocks the lock
to allow the termination of the \p{while} loop or further
decrease of the variant due to the environment locking the lock.
This guarantee is given by the atomicity context
$
  \actxt =
    \map{r -> (
      \set{0,1} \times \Ord, \botOf, \set{0} \times \Ord,
      R
)}
$
with $R = ((0,\alpha), \oblZero) \interfTo ((1,\ordfun(\alpha)), \oblZero)$.
  
The target states, $T$, must clearly include unlocked states, where $l = 0$,
but, as it must eventually be stable, this is insufficient, since once the
lock is unlocked, the environment can lock it again.
However, when the lock is unlocked, if the environment takes possession
of it, the environment must also simultaneously decrease the impedance budget,
i.e. $\beta > \alpha$.

The argument that~$T$ is always eventually true, relies on the assumption
from the atomicity context that the environment will always eventually unlock
the lock. However, this assumption only holds before the linearization point.
In particular, as the loop variant must contain $\notdone{\rid}$, since the loop
body may perform the linearization point, the persistent loop invariant cannot,
and therefore~$T$ must also contain a disjunct where the linearization point
has occurred and~$T$ holds a witness $\done{\rid}{(\wtv[2])}$.

We therefore declare the target states as the ones where,
either the linearization point has been performed,
or the lock is unlocked, or some budget was spent:
\[
T(\beta) \is
\exists l,\alpha\st
\region{spin}{\rid}(\pvar{x}, l, \alpha) \land
(\done{\rid}{(\wtv[2])} \lor l=0 \lor \beta > \alpha)
\]

The persistent loop invariant here is simply
$
  L = \region{spin}{\rid}(\pvar{x}, \wtv[2])
$, which is a valid stable frame of the loop.

To apply \ref{rule:while} we also need to specify $m(\beta)$,
which in this case is simply $\topOf$,
which satisfies the layer constraints of the rule;
and the environment progress measure $M$:\[
  M(\en{\alpha}) \is
    \exists l, \alpha \st
      \region{spin}{\rid}(\pvar{x}, l, \alpha) \land \en{\alpha} = 2\alpha + l
\]
(here we use the variable $\en{\alpha}$ for the environment progress measure variable, to avoid clashes with the impedance budget $\alpha$.)
This environmental progress measure is decreased by both the environment locking and unlocking the lock:
\begin{itemize}
\item Unlocking the lock decreases $l$ from $1$ to $0$, so as $2\alpha + 1 > 2\alpha + 0$,
  the environmental progress measure decreases. \\
\item Locking the lock decreases the impedance budget from $\alpha$ to $\alpha' < \alpha$,
  while also increasing $l$ from $0$ to $1$. Since $\alpha' < \alpha$ implies $\alpha' + 1 \le \alpha$,
  $2\alpha + 0 \ge 2\alpha' + 2 > 2\alpha' + 1$, the environmental progress measure decreases.
\end{itemize}

Given these parameters, the proof first establishes
the loop invariant holds at the beginning for some $\beta_0$,
by applying \ref{rule:consequence}:
\begin{gather*}
  \exists l, \alpha \st
    \region{spin}{\rid}(\pvar{x}, l, \alpha)
    \land \alpha > \ordfun(\alpha)
    * \notdone{\rid}
    \land \pvar{d}=0
  \;\Longrightarrow\;
  \exists \beta_0\st P(\beta_0) * L
  \\
  \exists \beta_0,\beta\st
    P(\beta) * L \land \pvar{d} \neq 0 \land \beta_0 \geq \beta
  \;\Longrightarrow\;
  \exists \alpha\st
    \region{spin}{\rid}(\pvar{x}, \wtv[2]) *
    \done{\rid}{((0, \alpha), (1, \ordfun(\alpha)))} \land l = 0
\end{gather*}
Note that we will often implicitly apply the \ref{rule:consequence} rule in proofs,
only detailing the application when emphasis is desired.
Next, \ref{rule:exists-elim} on $\beta_0$ gets rid of the existential quantification, so we are ready to apply \ref{rule:while}.

To complete the application of the rule we need to show
\begin{gather}
  \ENVLIVE \topOf; \actxt |- L - M ->> T(\beta)
  \label{cond:spin-lock-envlive} \\
\forall \alpha\st
  \STABLE \actxt |= {\exists \alpha'\st L * M(\alpha') \land \alpha' \leq \alpha}
          \label{cond:spin-lock-M} \\
  \progvars(T,L,M) \cap \modvars(\cmd) = \emptyset
  \label{cond:spin-lock-mods}
\end{gather}

Condition~\eqref{cond:spin-lock-M} is easily seen to hold,
as we showed above, all possible environmental interference
on the region decreases the environmental progress metric,
which is sufficient for this to hold.

Condition~\eqref{cond:spin-lock-mods} is also easily seen to hold as
the only program variable predicated over in $T$, $L$ and $M$
is $\p{x}$, which is not modified by the body of the loop.

Finally, condition~\eqref{cond:spin-lock-envlive} is proven as follows.
We observe that:
\begin{align}
  L(\en{\alpha}) = L * M(\en{\alpha}) &\equiv
    \bigl(
      \exists l,\alpha\st
      \region{spin}{\rid}(\pvar{x}, \ghost{1}{l}, \alpha) *
      (\done{\rid}{(\wtv[2])} \lor l \doteq 0)
      \land \en{\alpha} = 2\alpha + \ghost{1}{l}
    \bigr)
    \tag{$L_1(\en{\alpha})$}
    \\
    &\,\lor
    (
      \ghost{\exists l,\alpha\st{}}{\exists \alpha\st{}}
      \region{spin}{\rid}(\pvar{x}, 1, \alpha) *
      \notdone*{\rid}
      \land \en{\alpha} = 2\alpha + 1
    )
    \tag{$L_2(\en{\alpha})$}
\end{align}
We can then derive the environment liveness condition:
\[
\infer*[right={\ref{rule:envlive}}]{
  \infer*[right={\ref{rule:envlive-case}}]{
    \infer*[right={\ref{rule:envlive-target}}]{
      \forall \en{\alpha}\st
        \VALID \actxt |= L_1(\en{\alpha}) \implies T(\beta)
    }{
      \forall \en{\alpha} \st \ENVLIVE \topOf; \actxt |- L(\en{\alpha}) : L_1(\en{\alpha}) -->> T(\beta)
    }
  \and
    \infer*[right={\ref{rule:envlive-pq}}]{
\decr[\actxt](L_2, L, T(\beta))
    }{
      \forall \en{\alpha} \st \ENVLIVE \topOf; \actxt |-
        L(\en{\alpha}) : L_2(\en{\alpha}) -->> T(\beta) }
  }{
    \forall \en{\alpha} \st \ENVLIVE \topOf; \actxt |- L(\en{\alpha}) : L(\en{\alpha}) -->> T(\beta)
  }
}{
  \ENVLIVE \topOf; \actxt |- L - M ->> T(\beta)
}
\]

Formally, the application of \ref{rule:envlive} requires
us to prove
$
\VALID \actxt |=
L \implies L * \exists \en{\alpha}\st M(\en{\alpha})
$
which  is trivial.
An application of the \ref{rule:envlive-case} rule then splits
between the cases where $L_1$ and $L_2$ hold.
Intuitively,
  $L_1$ represents the case where
    we performed the linearization point or the lock is unlocked, while
  $L_2$ the case where
    we still have not performed the linearization point
    and the lock is locked.
If $L_1$ holds, then $T$ holds, so no progress of the environment is required,
therefore, this case can be discharged via an application of \cref{rule:envlive-target}.
In the case where $L_2$ holds we can apply \cref{rule:envlive-pq} to invoke the
liveness assumption stored in $\actxt$: if the lock is unlocked, the metric strictly decreases.

To show that the liveness assumption encoded in the atomicity context
for the region $\region{spin}{\rid}$,
$\live(\actxt,r) = \set{0,1} \times \Ord \eventually[k] \set{0} \times \Ord$
is active, the \ref{rule:envlive-pq} rule requires that in the current case:
\begin{itemize}
\item The abstractly atomic update being tracked on $\rid$ has yet to occur:
\[
\forall \alpha_e \st \SAT[\actxt] |- L_2(\alpha_e) \implies
\exists (l, \alpha) \in (\set{0,1} \times \Ord) \setminus (\set{0} \times \Ord) \st
\region{spin}{\rid}(\pvar{x}, l, \alpha) * \done{r}{\lozenge} * \True
\]
\item No obligations of layer less than or equal to $k$
  is continuously held locally:
  \begin{align*}
    & m \laygt k \\
    & \forall \alpha\st
    \VALID \actxt |= \minLayStrict{L'(\alpha)}{k}
  \end{align*}
\end{itemize}
If these hold, then the $\decr[\actxt](L_2, L, T(\beta))$
predicate shows that discharging the liveness
invariant will strictly decrease $\alpha_e$.
To show this holds, taking $\store \in \Store$ arbitrary and letting 
\begin{align*}
  l(\alpha) &= \WorldSem{\actxt}{\store}{L(\alpha)} &
  l'(\alpha) &= \WorldSem{\actxt}{\store}{L_2(\alpha)} &
  t &= \WorldSem{\actxt}{\store}{T(\beta) * \True}
\end{align*}
we need to show
\[
\forall \alpha_1, \alpha_2 \ge \alpha_1 \st
{\relyAt[\actxt]}(l'(\alpha_1)) \inters l(\alpha_2) \subseteq l'(\alpha_1) \union t
\]
This holds, as, given an arbitrary $\alpha_1 \in \Ord$,
any step taken from $l'(\alpha_1)$ by the atomic world rely relation
either leaves the state of the region $\region{spin}{\rid}$
unchanged, preserving the state $l'(\alpha_1)$,
or releases the lock, decreasing the metric.
Therefore, for any $\alpha_2 \ge \alpha_1$,
${\relyAt[\actxt]}(l'(\alpha_1)) \inters l(\alpha_2) \subseteq l'(\alpha_1)$
holds, which implies the goal.

To conclude the argument,
we briefly comment on the proof of the body of the \p{while} loop.
The full proof of the body can be found in figure \ref{fig:spin-lock-body-outline}.
The applications of \cref{rule:update-region,rule:frame}
lift the concrete atomic \code{CAS} to a (potential) update
to the $\rt[spin]_{\rid}$ region.
An application of \ref{rule:consequence} allows us to
introduce $\gamma$ as an upper bound to the impedance budget, initially
$\delta$ after the linearization point.

Then, we apply \cref{rule:atomic-exists-elim} to remove
the pseudo-quantification on $l$ and $\alpha$.
At this point, the abstract state $l,\alpha$ of the region
$\region{spin}{\rid}$ in the postcondition
is weakened to any state that might be reached before
or after the linearization point.
However, we keep record of what happened exactly at the linearization point
because of the $\done{\rid}{\wtv}$ assertions.
The later application of \ref{rule:make-atomic} will be able to fetch
the atomic update witness $\done{\rid}{((l, \alpha), (1, \ordfun(\alpha)))}$
and declare the appropriate atomic update in the overall specification.
Note that the overall Hoare postcondition after the application
of \ref{rule:atomicity-weak} is stable.

\begin{mathfig}
  \begin{proofoutline}\META{\forall b \in \set{0,1}, \beta\st}\\
  \PREC{\exists l, \alpha \st
    \region{spin}{\rid}(\pvar{x}, l, \alpha)
    * \notdone{\rid}
    \land {\alpha > \ordfun(\alpha)}
    \land \beta \ge \alpha
    \\\quad{}
    \land b \implies (l = 0 \lor \beta {{} > {}} \alpha)
    \land \pvar{d} = 0
  }\\
  \begin{proofjump}[rule:atomicity-weak,rule:atomic-exists-elim]
\label{step:spin-lock-linpt}
    \A l \in \set{0,1}, \alpha. \\
    \PREC<
    \region{spin}{\rid}(\pvar{x}, l, \alpha)
    * \notdone{\rid}
    \land {\alpha > \ordfun(\alpha)}
    \land \beta \ge \alpha \land {} \\
    b \implies (l = 0 \lor \beta {{} > {}} \alpha)
    \land \pvar{d} = 0
    >\\
    \begin{proofjump}[rule:consequence]
      \PREC<
      \region{spin}{\rid}(\pvar{x}, l, \alpha)
      * \notdone{\rid}
      \land {\alpha > \ordfun(\alpha)}
      \land \beta \ge \alpha \land {} \\
      b \implies (l = 0 \lor \beta {{} > {}} \alpha)
      \land \pvar{d} = 0
      >\\
      \begin{proofjump}[rule:update-region]
        \PREC<
        \pvar{x} \mapsto l \land
             {\alpha > \ordfun(\alpha)} \land
             \beta \ge \alpha \land {} \\
             b \implies (l = 0 \lor \beta {{} > {}} \alpha)
             >\\
             \begin{proofjump}[rule:frame]
               \PREC<\pvar{x} \mapsto l>\\
               \CODE{d := CAS(x,0,1);}\\
               \POST<\pvar{x} \mapsto 1 \land
               ((\pvar{d} = 0 \land l = 1) \lor (\pvar{d} = 1 \land l = 0))
               >
             \end{proofjump} \\
             \POST< \exists \delta \st
             \pvar{x} \mapsto 1 \land
\left(
                  {\begin{array}{r@{\;}l}
                      &(
                      \pvar{d} = 1 \land
                      l = 0 \land
                      \delta = \ordfun(\alpha) \land
                             {\beta > \ordfun(\alpha)}
                             ) \\ \lor &
                             (
                             \pvar{d} = 0 \land
                             l = 1 \land
                             \delta = \alpha \land
                                    {\alpha > \ordfun(\alpha)} \land
                                    b \implies {\beta > \alpha}
                                    )
                  \end{array}}
                  \right) \land \beta \ge \alpha
                  >\\
      \end{proofjump} \\
      \POST< \exists \delta \st
      \region{spin}{\rid}(\pvar{x}, 1, \delta)
      * \left(
      {\begin{array}{r@{\;}l}
          &(
          \pvar{d} = 1 \land
          l = 0 \land
          {\beta > \delta}
          \land
          \done{\rid}{((l,\alpha),(1,\ordfun(\alpha)))}
          ) \\ \lor &
          (
          \pvar{d} = 0 \land
          l = 1 \land
          {\delta > \ordfun(\delta)} \land
          b \implies {\beta > \delta}
          \land
          \notdone{\rid}
          )
      \end{array}}
      \right)
      \land \beta \ge \delta
      >\\
    \end{proofjump} \\
    \POST<\exists \gamma, \delta \st
    \region{spin}{\rid}(\pvar{x}, 1, \delta) \land
    \beta \ge \gamma \ge \delta \land
    b \implies \beta {{} > {}} \gamma
    \\ \quad{}
    * \left(
    {\begin{array}{r@{\;}l}
        & (
        \pvar{d} = 0 \land
             {\delta >  \ordfun(\delta)} \land
             \notdone{\rid}
             )
             \\ \lor &
             (
             \pvar{d} = 1 \land l = 0 \land
             \done{\rid}{((l, \alpha), (1, \ordfun(\alpha)))}
             )
    \end{array}}
    \right)
    >
  \end{proofjump} \\
  \POST{\exists l, \alpha, \gamma \st
    \region{spin}{\rid}(\pvar{x}, l, \alpha) \land
    \beta \ge \gamma \ge \alpha \land
    b \implies \beta > \gamma
    \\ \quad{}
    * \left(
    \begin{array}{r@{\;}l}
      & 
      \bigl(
      \pvar{d} = 0 \land
      \alpha > \ordfun(\alpha) \land
      \notdone{\rid}
      \bigr)
      \\ \lor &
      \bigl(\exists l', \alpha' \st
      \pvar{d} = 1 \land
      \done{\rid}{((l', \alpha'), (1, \ordfun(\alpha')))}  \land
      l' = 0
      \bigr)
    \end{array}
    \right)
  }
\end{proofoutline}
   \caption{Spin lock: Proof of \protect{\code{while}} loop body.}
  \label{fig:spin-lock-body-outline} 
\end{mathfig}

Finally, \cref{fig:spin-unlock-outline} shows the proof outlines for the
\code{makeLock} and \code{unlock} operations.
The only notable step of the proof of \code{makeLock} is the last application
of \ref{rule:consequence} to viewshift the postcondition
from $ \pvar{ret} \mapsto 0$ to
$\exists \rid\st \region{spin}{\rid}(x, 0, \alpha) * \guardA{\gEx}{\rid}$,
which is possible because the interpretation of the region matches with this
resource, so the reifications of the two assertions coincide.

The proof of \code{unlock} is a straightforward lifting of the atomic
reset of the cell at \code{x} to the region $\rt[spin]_{\rid}$.
Neither proof involves a liveness argument.

\begin{mathfig}
  \begin{proofoutline}
  \TITLE{Proof of \code{makeLock()}:}
  \CTXT \botOf; \emptyset |- \\
  \PREC{\emp} \\
  \begin{proofjump}[rule:consequence]
    \PREC{\emp} \\
    \CODE{ret := alloc(1);} \\
    \CODE{[ret] := 0;} \\
    \POST{\ret \mapsto 0} \\
  \end{proofjump} \\
  \POST{\exists \rid\st \ap{L}(\rid, \ret, 0, \alpha)}
\end{proofoutline}
   \hspace{6em}
  \begin{proofoutline}
  \TITLE{Proof of \code{unlock(x)}:}
  \CTXT \botOf; \emptyset |- \\
  \PREC<\ap{L}(\rid, \pvar{x}, 1, \alpha)> \\
  \begin{proofjump}[rule:consequence]
    \PREC<\region{spin}{\rid}(\pvar{x}, 1, \alpha) * \guardA{\gEx}{\rid}> \\
    \begin{proofjump*}
      [rule:lift-atomic,rule:frame,rule:subst]\label{step:spin-unlock}\PREC<\pvar{x} \mapsto 1> \\
      \CODE{[x] := 0;} \\
      \POST<\pvar{x} \mapsto 0> \\
    \end{proofjump*} \\
    \POST<\region{spin}{\rid}(\pvar{x}, 0, \alpha) * \guardA{\gEx}{\rid}>
  \end{proofjump} \\
  \POST<\ap{L}(\rid, \pvar{x}, 0, \alpha)>
\end{proofoutline}
   \caption{Spin lock: proof of \protect{\code{makeLock}} and \protect{\code{unlock}}. 
    Here \cref{step:spin-unlock} is \explainproofjump{step:spin-unlock}.
  }
  \label{fig:spin-unlock-outline}
  \label{fig:spin-make-outline}
\end{mathfig}
 
\subsection{CLH lock}
\label{examples:clh_lock}

\begin{figure}
  \centering\small \begin{tikzpicture}[thick]

  \node (x) at (-2,0) {\p{x}};
  \node (head) {head};
  \node (tail) at (1.1,0) {tail};
  \node[draw=gray, rounded corners, fit=(head)] (headbox) {} ;
  \node[draw=gray, rounded corners, fit=(tail)] (tailbox) {} ;
  \draw[->] (x) edge (headbox);

  \node[minimum size=1.5em] (cell0) at ($ (head) + (-5.5,2) $) {$l$};
  \node[draw=gray, rounded corners, inner sep=0, fit=(cell0), label={[label distance=0cm]-50:$cell_0$}] (cell0box) {};
  \draw[->] (headbox) edge[in=-90,out=140] (cell0box);

  \node[minimum size=1.5em] (cell1) at ($ (cell0) + (2,0) $) {$1$};
  \node[draw=gray, rounded corners, inner sep=0, fit=(cell1), label={[label distance=0cm]-50:$cell_1$}] (cell1box) {};

  \node[minimum size=1.5em] (cell2) at ($ (cell1) + (2,0) $) {$1$};
  \node[draw=gray, rounded corners, inner sep=0, fit=(cell2), label={[label distance=0cm]-50:$cell_2$}] (cell2box) {};

  \node (dots) at ($ (cell2) + (2,0) $) {$\cdots$};

  \node[minimum size=1.5em] (cellnm1) at ($ (dots) + (2,0) $) {$1$};
  \node[draw=gray, rounded corners, inner sep=0, fit=(cellnm1), label={[label distance=0cm]-50:$cell_{n-1}$}] (cellnm1box) {};

  \node[minimum size=1.5em] (celln) at ($ (cellnm1) + (2,0) $) {$1$};
  \node[draw=gray, rounded corners, inner sep=0, fit=(celln), label={[label distance=0cm]-50:$cell_n$}] (cellnbox) {};
  \draw[->] (tailbox) edge[in=-90,out=45] (cellnbox);

  \node (p0) at ($(cell0) + (0.7,1)$) {\strut\p{p}};
  \node (c0) at ($(p0) + (0.5,0)$) {\strut\p{c}};
  \node[draw=gray, rounded corners, dotted, fit=(p0)(c0), label={[label distance=0cm]130:$t_1$}] (t1box) {};
  \draw[->] (p0) edge (cell0);
  \draw[->] (c0) edge (cell1);

  \node (p1) at ($(cell1) + (0.7,1)$) {\strut\p{p}};
  \node (c1) at ($(p1) + (0.5,0)$) {\strut\p{c}};
  \node[draw=gray, rounded corners, dotted, fit=(p1)(c1), label={[label distance=0cm]130:$t_2$}] (t2box) {};
  \draw[->] (p1) edge (cell1);
  \draw[->] (c1) edge (cell2);

  \node (adots) at ($(dots) + (0,1)$) {$\cdots$};
  \draw[->] (adots) edge (cell2box);
  \draw[->] (adots) edge (cellnm1box);
  
  \node (pn) at ($(cellnm1) + (0.7,1)$) {\strut\p{p}};
  \node (cn) at ($(pn) + (0.5,0)$) {\strut\p{c}};
  \node[draw=gray, rounded corners, dotted, fit=(pn)(cn), label={[label distance=0cm]130:$t_n$}] (tnbox) {};
  \draw[->] (pn) edge (cellnm1);
  \draw[->] (cn) edge (celln);

\end{tikzpicture}   \caption{Illustration of the memory layout of CLH lock.}
  \label{fig:clh-functioning}
\end{figure}

\paragraph{Code} A CLH lock is an implementation of a fair lock module that
guarantees fairness by queuing the threads that are waiting to take its possession.
Its implementation is shown in \cref{fig:clh-code}.
\begin{figure}
\begin{tabular}{c@{\hspace{3em}}c@{\hspace{3em}}c}
  {\codefromfile[numbers=left]{examples/clh-lock/code-make}} &
  {\codefromfile[numbers=left]{examples/clh-lock/code-lock}} &
  {\codefromfile[numbers=left]{examples/clh-lock/code-unlock}}
\end{tabular}
\caption{Code of CLH lock operation.}
\label{fig:clh-code}
\end{figure}

The diagram in \cref{fig:clh-functioning} describes the state of the queued threads,
$t_1, t_2, \ldots, t_n$, waiting to take possession of the lock, as well as the
module's head and tail pointers into the queue.

As described in \cref{sec:overview}, this queue is represented by associating
each of the $n$ threads queuing on the lock
with the heap cells $cell_1, cell_2, \ldots, cell_{n-1}, cell_n$ in memory.
Each thread executing the \p{lock} operation to take possession of the lock then
holds in its local state the address of its cell and that of its predecessor's cell.
These are held in the program variables \p{c} and \p{p} respectively in
the implementation of \p{lock}.
The local instance of these program variables for each queued threads
and the cells they are pointing to can be seen in \cref{fig:clh-functioning}.

The thread associated with the cell
at the head of the queue is said to hold the lock,
and the value stored in its cell determines the state of the lock, $l$.
When a thread first takes possession of the lock, the lock will be locked.
Therefore, the initial value in these cells, when the associated
threads join the queue, is $1$.
This can be seen in the implementation of the \p{lock} operation which
allocates and sets its associated cell to value $1$ on line \ref{line:clh-lock-alloc}
before enqueuing itself.
Once the thread holding the lock wishes
to release it, it can do so by setting the value of its cell to $0$,
unlocking the lock and signalling to the next thread in the queue that
it can now take possession of the lock. This can be
seen in the implementation of the \p{unlock} operation
which fetches the address of the cell associated with
the lock's owner from the queue's head pointer
and then sets its value to $0$.

In \cref{fig:clh-functioning} the thread $t_1$ is at the head of the queue,
waiting for the lock to be released. If the lock is released by its owner,
$t_1$ then gains the exclusive permission to take possession of the lock by setting
the value of the module's head pointer to the address of its associated cell.
$t_1$ detects the lock has been
released by repeatedly reading the value of its predecessor's cell in the \p{while}
loop on line \ref{line:clh-lock-wait} and then sets the head pointer to the address
of its cell, \p{c}, on line \ref{line:clh-lock-lin}.

Once the lock is released, only the thread at the head of the queue (if any) has
the permission to take possession of the lock next. Due to this, if the owners
of the lock continuously eventually release it, the threads waiting on the lock
take possession of it in the order they are enqueued.

To enqueue itself, the \p{lock} operation performs a \p{FAS} operation on the tail pointer,
placing the cell it has allocated with value $1$ at the tail of the queue, and writting the
address of its predecessor to the \p{p} program variable. The order in which the \p{lock}
operations are enqueued is then the order in which they executed line \ref{line:clh-lock-enqueue}.
Any weakly fair scheduler will eventually give each thread executing the \p{lock} operation
the opportunity to execute this \p{FAS} operation, allowing it to enqueue itself.

As long as the client then guarantees that every thread holding the lock eventually releases it,
the thread will eventually take possession of the lock once it reaches the front of the queue and
the \p{lock} operation will terminate, guaranteeing fairness.

To be able to provide the same guarantee,
that every thread requesting the lock will eventually be able to take its possession
as long as the lock is always eventually released, the spin lock requires that its client
only call the \p{lock} operation concurrently a finite number of times.
This is exposed in the spin lock specification via ordinals bounding the impedance on the lock.

An interesting aspect of this example is that it features a combination
of internal and external blocking:
the client needs to always eventually unlock the lock
---external blocking, requiring the client to provide a guarantee---
and the \code{lock} operation needs to eventually take possession of the lock
once the previous thread signals its release ---internal blocking,
guaranteed by the implementation.
This second guarantee will be enforced using obligations not exposed in the specification.
The proof will therefore involve an environment liveness condition
discharged using both \ref{rule:envlive-obl} and \ref{rule:envlive-pq}.

\paragraph{Specifications}
We will prove the following fair lock module specifications:
\begin{align*}
  &\ATRIPLE \topOf |-
    \A l \in \set{0, 1} \eventually[{\botOf}] \set{0}.
      <\ap{L}(s, \pvar{x}, l)>
        \code{lock(x)}
      <\ap{L}(s, \pvar{x}, 1) \land l = 0>
  \\
  &\ATRIPLE \botOf |-
    <\ap{L}(s, \pvar{x}, 1)>
      \code{unlock(x)}
    <\ap{L}(s, \pvar{x}, 0)>
\end{align*}where $\ap{L}(s, \pvar{x}, l)$ abstractly represents the lock resource
at abstract location~$s$
  (omitted for readability in \cref{sec:overview}) and
concrete address~$\pvar{x}$,
with abstract state~$l \in \set{0,1}$.

To abstract the representation of a thread's position in the queue,
we will associate, through ghost state, to each thread requesting the lock,
a \emph{ticket number}~$t \in \Nat$ which corresponds to the order of arrival
of the \p{lock} implementation at line \ref{line:clh-lock-enqueue}.
Every time a thread joins the queue, it gets assigned the next available ticket.

This example shows a common proof pattern of \tadalive:
there is an inner region that exposes all the information needed for the
termination argument (here the value of the next ticket to
be handed out,~$t$, so that individual threads can reason
about the threads queuing on the lock)
and an outer one that hides enough information to make the operation
abstractly atomic.
This pattern nicely separates the concerns in the proof:
proving atomicity is done via the outer region, termination via the inner one.
Because of this, the abstract location of the lock~$s$ will consist of the
pair of inner and outer region identifiers.
This is not a concern for modularity however:
the type of~$s$ can be made opaque to the client,
which just threads it through the proof unmodified.

\paragraph{Shared Regions}
The abstract shared lock resource will be represented by a region
$ \region{clh}{\rid}(\rid', x, h, l, o) $
where
  $\rid' \in \RId$,
  $ x, h \in \Addr $,
  $ l \in \set{0,1} $,
  $ o \in \Nat $.
Here $\rid'$, the region identifier of the inner region and
$x$, the address of the lock, are the fixed parameters of the region.
The abstract state of the region includes~$l$, which represents the lock's state,
$o$, which is the ticket number of the thread holding the lock,
and~$h$ is the address of the cell associated with the owner.

Once a \p{lock} operation has enqueued itself, the difference between the
ticket of the lock's owner,~$o$ and the operation's ticket, $t$, $t - o$,
corresponds to the thread's current position in the queue.

The internal region $ \region{lclh}{\rid'}(x, h, l, o, t) $ also
exposes the next ticket to be handed to the next thread
queuing on the lock, $t \in \Nat$.

\paragraph{Notation}
Lists will frequently be used in the ghost state for the proof
of the CLH lock. We introduce notation to manipulate lists so
as to simplify the exposition of the reasoning.
Given $n\in X$ and $ ns,ns' \in X^* $ lists of elements of $X$,
we write
  $n \lstPlus ns$,
  $ns \lstPlus n$, and
  $ns \lstPlus ns'$
for prepend, append, and concatenation, respectively;
$\lstLen{ns}$ is the length of~$ns$,
and $ \lstAt{ns}{i} = n $ states
that the \pre i-th element (from~0) in~$ns$ is~$n$ and $i<\lstLen{ns}$;
$\lfun{fst}(ns)$ and $\lfun{last}(ns)$
are the first and the last element of~$ns$, respectively
and $\lfun{tail}(ns)$ represents the list~$ns$ without the first element
when~$ns$ is non empty.

\paragraph{Guard algebra:}
Take $p, c \in \Addr, ns \in \Addr^{*}, o, t \in \Nat$ arbitrary.
For this proof, two guards will be necessary. First $\guard{t}(p, c, t)$,
which encodes the current thread's ticket, $t$, once it has joined the queue,
as well as $p, c \in \Addr$, pointers to the thread's predecessor's cell
in the queue and its own respectively.
The second guard we require is $\guard{q}(ns, o)$,
which is used to track the overall queue, by tracking
the cells associated with enqueued threads, $ns \in \Addr^{*}$,
and the ticket number of the current owner, $o \in \Nat$.

To use this as intended, a few axioms on the guard algebra will be required.
First, an axiom to create new tickets, adding a new cell to the queue and associating
a new, unique ticket number to the thread:
\[
  \guard{q}(ns \lstPlus \lst{p}, o) =
    \guard{q}(ns \lstPlus \lst{p, c}, o)
      \guardOp
    \guard{t}(p, c, o + \lstLen{ns} + 1)
\]
This will be used to create the relevant guard resources $\guard{t}$, when a lock
operation enqueues itself on line \ref{line:clh-lock-enqueue}.
Similarly, an axiom to remove a thread's predecessor from the queue once it can take possession of the lock:
\[
  \guard{q}(\lst{p,c} \lstPlus ns, o)
    \guardOp
  \guard{t}(p, c, o + 1)
  =
  \guard{q}(\lst{c} \lstPlus ns, o + 1)
\]
This will be used to update the relevant guard resources $\guard{q}$ with the relevant $\guard{t}$,
when a lock operation takes possession of the lock on line \ref{line:clh-lock-lin},
placing its associated cell, $c$, at the head of the queue.
Finally, an axiom to guarantee that a ticket guard, $\guard{t}$ is well-formed
with respect to the queue in a guard $\guard{q}$:
\[
  \guardDef{\guard{q}(ns, o) \guardOp \guard{t}(p, c, t)}
  \iff
  \lstAt{ns}{t - o - 1} = p \land \lstAt{ns}{t - o} = c
\]

\paragraph{Obligation algebra:}
Take $o, o', t, t' \in \mathbb{N}$ arbitrary.
As mentioned above, to verify the totality of the CLH lock operation, once a thread is enqueued,
if its predecessor relinquishes possession of the lock, it must eventually take its possession.
Otherwise, although the lock will be permanently unlocked, no other thread waiting for the lock
can take its possession, as they are not at the head of the queue.

To encode this liveness invariant which must be fulfilled, we associate an atom obligation $\obl{p}(t)$
with the ownership of the ticket $t \in \Nat$.
The CLH lock's transition system will then require that this obligation be discharged by
taking possession of the lock once it is unlocked by the thread with ticket $t - 1$.

The layer associated with $\obl{p}(t)$ is then $t$, so that these obligations are resolved
in the order the associated threads are enqueued. Finally, as with the guard algebra, we have
an obligation $\obl{o}(o, t)$, which will remain in the shared region's state and
track the owner's ticket, $o$, and the next ticket to be handed out, $t$,
associated with the obligation $\obl{p}$ via the obvious axioms.
\begin{gather*}
  \obl{o}(o, t)     = \obl{o}(o, t+1) \oblOp \obl{p}(t) 
  \qquad\qquad
  \obl{o}(o + 1, t) = \obl{o}(o, t) \oblOp \obl{p}(o + 1) 
  \\
  \oblDef{\obl{o}(o, t) \oblOp \obl{p}(t')}
    \iff
    o \le t' < t
  \\
  \Layer \is \Nat\union\set{\topOf,\botOf}
  \quad
  \forall i \in \Nat \st
    \topOf \laygt i \laygt \botOf
  \qquad
  \lay(\obl{o}(o,t)) = 0
  \quad
  \lay(\obl{p}(t)) = t
\end{gather*}

\paragraph{Region protocols}
The interference protocol for the $\region{lclh}{}$ region is as follows:
\begin{align*}
  \guard{e} &: ((h, l, o, t), \oblZero) \interfTo ((h, l, o, t + 1), \obl{p}(t)) \\
  \guard{e} &: ((h, 0, o, t), \obl{p}(o+1)) \interfTo ((h', 1, o + 1, t), \oblZero) \\
  \guard{e} &: ((h, 1, o, t), \oblZero) \interfTo ((h, 0, o, t), \oblZero)
\end{align*}

The first transition allows a thread to place itself in the queue waiting to
obtain the CLH lock, updating the next ticket to be handed out from $t$ to
$t + 1$.
While doing so, the threads acquires an obligation,
$\obl{p}(t)$, requiring it to eventually take possession of the lock
once it is at the head of the queue.
The second, allows the thread at the head of the queue to take possession of the lock,
by changing the state, $l$, incrementing the owner ticket, $o$, to its own
(tracked by the thread's obligation) and changing the owner pointer of the lock to
that of its own associated cell.
This discharges the obligation $\obl{p}(o + 1)$, as the thread then leaves the queue,
to take possession of the lock.
Finally, the third transition allows the lock to be unlocked.

The interference protocol for the $\region{clh}{}$ region is then:
\begin{align*}
  \guard{e} &: ((h, l, o), \oblZero) \interfTo ((h, l, o), \oblZero) \\
  \guard{e} &: ((h, 0, o), \oblZero) \interfTo ((h', 1, o + 1), \oblZero) \\
  \guard{e} &: ((h, 1, o), \oblZero) \interfTo ((h, 0, o), \oblZero)
\end{align*}

This hides the enqueuing step of the \p{lock} operation, allowing the operation
to appear atomic.

\paragraph{Region interpretation}
As explained above, the CLH lock associates a cell with each thread
queuing on it, as well as its owner.
The list of each of these cells in the order in which
the associated threads are queued, with the owner's cell as the head,
will be denoted $\lvar{ns}$. $\lfun{tail}(ns)$ is then the list of cells
queueing on the lock.
While threads are queuing, the associated cells must have value~$1$;
this is represented using the predicate $\lfun{ones}$:
\[
\lfun{ones}(ns) \is \lstAt{ns}{1} \mapsto 1 * \dots * \lstAt{ns}{\lstLen{ns}-1} \mapsto 1
\]
The inner shared region, $\region{lclh}{}$, holds the cells associated with each queued thread, this is represented by
the resource $\lfun{ones}(\lvar{ns})$ in the region interpretatio. 

The shared region also holds a pointer to the tail of the queue,~$\lvar{ns}$,
as well as a pointer to its owner's cell,
whose value is the state of the lock,~$l$, as described above.
This is represented by the resource:
\[
x \mapsto h, \lfun{last}(ns) * h \mapsto l
\]
The shared region's ghost state is then comprised of:
\begin{itemize}
\item $\guardA{\guard{q}(ns, o)}{\rid'}$ the guard keeping track of the list of cells,
  $ns \in \Addr^{*}$ and the current owner of the lock, $o \in \Nat$.
\item $\locObl{\obl{o}(o, t)}{\rid'}$ the obligation keeping track of the next ticket
  to hand out, $t \in \Nat$, and the current owner's ticket, $o \in \Nat$.
\end{itemize}
Finally, the invariant $t - o = \lstLen{\lvar{ns}}$ is used to guarantee that
each thread that holds a ticket is associated with a cell in the queue $\lvar{ns}$ and
$h = ns(0)$, associates the head of $ns$ and the address of the owner's cell.
All of this ties together to give the following region interpretation:
\begin{multline*}
  \rInt(\region{lclh}{\rid'}(x, h, l, o, t)) \is
  \exists ns \in \Addr^{*} \st \;
  x \mapsto h, \lfun{last}(ns) * {} \\
  h \mapsto l * \lfun{ones}(ns) *
  \guardA{\guard{q}(ns, o)}{\rid'} *
  \locObl{\obl{o}(o, t)}{\rid'} \land
  t - o = \lstLen{ns} \land
  ns(0) = h
\end{multline*}
The outer shared region then holds full permission to update the inner
region, $\guardA{\guard{e}}{\rid'}$, and asserts that each
thread queuing on the lock, with tickets $o + 1$ to $t - 1$,
holds an obligation to take possession of the lock once
their predecessor releases it,
$\smash{\Sep*^{t - 1}_{i=o+1} \envObl{\obl{p}(i)}{\rid'}}$,
where $\rid'$ is the identifier of the inner region:
\[
  \rInt(\region{clh}{\rid}(\rid', x, h, l, o)) \is
  \exists t \in \Nat \st
  \region{lclh}{\rid'}(x, h, l, o, t) *
  \guardA{\guard{e}}{\rid'} *
  \smash{\Sep*^{t - 1}_{i=o+1} \envObl{\obl{p}(i)}{\rid'}}
\]

\paragraph{Predicates}
The lock resource is then abstractly represented by the predicate:
\[
\ap{L}(s, x, l) \is
  \exists \rid,\rid'\st s=(\rid,\rid') \land
  \exists o \in \Nat \st
  \exists h \in \Addr \st
  \region{clh}{\rid}(\rid', x, h, l, o) * \guardA{\gEx}{\rid'}
\]
which abstracts away the CLH lock's implementation
details: the ticket and cell address
associated with the lock's current owner.

\paragraph{Proof of \code{lock}}
\begin{mathfig}[tb]
\resizebox{.6\textwidth}{!}{
  \begin{proofoutline}
    \CTXT \topOf; \emptyset |-
    \A l \in \set{0, 1} \eventually[{\botOf}] \set{0}. \\
    \PREC<\ap{L}(s, \pvar{x}, l)> \quad\META{s = (\rid,\rid')}\\
    \begin{proofjump}[rule:atomic-exists-elim]
      \CTXT \topOf; \emptyset |- \\
      \A l \in \set{0, 1} \eventually[{\botOf}] \set{0}, o, h. \\
      \PREC<\region{clh}{\rid}(\rid', x, h, l, o) * \guardA{\guard{e}}{\rid'}> \\
      \begin{proofjump}[rule:make-atomic]
        \CTXT \topOf; \map{\rid -> (X_1, {\botOf}, X_2, R)} |- \\
        \PREC{\exists l, o, h \st \region{clh}{\rid}(\rid', \pvar{x}, h, l, o) *
          \done{\rid}{\blacklozenge}} \\
        \CODE{c := alloc(1); [c] := 1;} \\
        \PREC{\exists l, o, h \st \region{clh}{\rid}(\rid', \pvar{x}, h, l, o) *
          \done{\rid}{\blacklozenge} * \pvar{c} \mapsto 1} \\
        \CODE{p := FAS(x + 1, c);} \\
\CODE{v := [p];} \\
        \ASSR{\exists \beta, t' \st P(\beta) * L} \\
        \CODE{while(v != 0) \{ v := [p]; \}}\\
        \ASSR{\exists t' \st P(0) * L \land \p{v}=0}\\
        \CODE{[x] := c;} \\
        \ASSR{\exists o \in \Nat, h, h'\st \done{\rid}{((h, 0, o), (h', 1, o + 1))} * \pvar{p} \mapsto 0 \land l = 0} \\
        \CODE{dealloc(p)} \\
        \POST{\exists o \st \done{\rid}{((\wtv, 0, o), (\wtv, 1, o + 1))} \land l = 0} \\
      \end{proofjump} \\
      \POST<\exists h' \st\region{clh}{\rid}(\rid', \pvar{x}, h', 1, o + 1) * \guardA{\guard{e}}{\rid'} \land l = 0> \\
    \end{proofjump} \\
    \POST<\ap{L}(s, \pvar{x}, 1) \land l = 0>
  \end{proofoutline}
}\hspace*{-3.5em}\begin{minipage}[b]{.4\textwidth}
\begin{align*}
X_1 &\is \set{ (h,l,o) | h \in \Addr, l \in \set{0,1}, o \in \Nat }
  \\
  X_2 &\is \set{ (h,0,o) | h \in \Addr, o \in \Nat }
  \\
  R &\is \set{((h,0,o),(h',1,o)) | h,h' \in \Addr, o\in \Nat}
  \\[1ex]
  P(\beta) &\is
    \exists l, o, h \st
      \region{clh}{\rid}(\rid', \pvar{x}, h, l, o) * \guardA{\guard{t}(\pvar{p}, \pvar{c}, t')}{\rid'} \land {} \\ & \qquad 
        (\pvar{v} = 0 \Rightarrow (t' = o + 1 \land l = 0 \land h = \pvar{p})) \land {}
      \\ & \qquad 
        o < t'\land \beta = \pvar{v} \\[1ex]
  L &=
    \exists l \in \set{0,1}, o, h \st
      \region{clh}{\rid}(\rid', \pvar{x}, h, l, o) * \locObl{\obl{p}(t')}{\rid'} * {} \\ & \qquad
      \Sep*^{t' - 1}_{i=o+1} \envObl{\obl{p}(i)}{\rid'} * \done{\rid}{\blacklozenge} \land o < t'
\end{align*}
\end{minipage}
\caption{Outline of CLH \texttt{lock} proof.}
\label{fig:clh-outline}
\end{mathfig}
Figure \ref{fig:clh-outline} gives an outline of the proof of the clh $\code{lock}$
operation implementation, the definition of the loop invariant $P(\beta)$ will be given later.
The steps involving liveness are the $\code{FAS}$ operation,
which enqueues the thread, hence obtaining the obligation to take possesion of the lock
once the previous thread relinquishes possession of it,
the while loop which waits for the previous thread to release the lock,
whose liveness depends on the previous threads in the queue taking possession
and then releasing the lock in turn and the write operation at line 
\ref{line:clh-lock-lin} which takes possession of the lock.
We begin with the details of the $\code{FAS}$ operation's proof,
shown in \cref{fig:clh-fas}.

\begin{mathfig}[tb]
\begin{proofoutline}
  \CTXT \topOf; \map{\rid -> (X_1, {\botOf}, X_2, R)} |- \\
  \ASSR{
         \exists l \in \set{0,1}, o \in \Nat, h \in \Addr \st
         \region{clh}{\rid}(\rid', \pvar{x}, h, l, o) *
         \done{\rid}{\blacklozenge} *
         \pvar{c} \mapsto 1
  } \\
  \begin{proofjump*}[rule:frame-hoare,rule:atomicity-weak,rule:atomic-exists-elim,rule:lift-atomic,rule:atomic-exists-elim,rule:lift-atomic,rule:atomic-exists-elim]
    \label{proof:clh-lock-enqueue}
    \CTXT \topOf; \map{\rid -> (X_1, {\botOf}, X_2, R)} |- \\
    \A l \in \set{0,1}, o, t \in \Nat, h \in \Addr, ns \in \Addr^{*}. \\
    \PREC<
    \pvar{x} \mapsto h, \lfun{last}(ns) *
    h \mapsto l * \lfun{ones}(ns) *
    \guardA{\guard{q}(ns, o)}{\rid'} *
    \locObl{\obl{o}(o, t)}{\rid'} * {}\\
    \Sep*^{t - 1}_{i=o+1} \envObl{\obl{p}(i)}{\rid'} \land
    t - o = \lstLen{ns} \land
    ns(0) = h \ssep
    \pvar{c} \mapsto 1
    > \\
    \begin{proofjump}[rule:consequence]
      \PREC<
      \pvar{x} \mapsto h, \lfun{last}(ns) *
      h \mapsto l * \lfun{ones}(ns) *
      \guardA{\guard{q}(ns, o)}{\rid'} *
      \locObl{\obl{o}(o, t)}{\rid'} * {}\\
      \Sep*^{t - 1}_{i=o+1} \envObl{\obl{p}(i)}{\rid'} \land
      t - o = \lstLen{ns} \land
      ns(0) = h \ssep
      \pvar{c} \mapsto 1
      > \\
      \begin{proofjump}[rule:layer-weak,rule:frame]
        \CTXT \botOf; \map{\rid -> (X_1, {\botOf}, X_2, R)} |- \\
        \PREC<\pvar{x} + 1 \mapsto \lfun{last}(ns)> \\
        \CODE{p := FAS(x + 1, c);} \\
        \POST<\pvar{x} + 1 \mapsto \pvar{c} \land \pvar{p} = \lfun{last}(ns)>
      \end{proofjump} \\
       \ASSR<
       \pvar{x} \mapsto h, \pvar{c} *
       h \mapsto l * \lfun{ones}(ns) *
       \pvar{c} \mapsto 1 * \guardA{\guard{q}(ns, o)}{\rid'} *
       \locObl{\obl{o}(o, t)}{\rid'} * {}\\
       \Sep*^{t - 1}_{i=o+1} \envObl{\obl{p}(i)}{\rid'} \land
       t - o = \lstLen{ns} \land
       ns(0) = h
       > \\
    \end{proofjump} \\
    \ASSR<
    \exists ns' \in \Addr^{*} \st
    \pvar{x} \mapsto h, \lfun{last}(ns') *
    h \mapsto l * \lfun{ones}(ns') *
    \guardA{\guard{q}(ns', o)}{\rid'} *
    \locObl{\obl{o}(o, t + 1)}{\rid'} * {} \\
    \Sep*^{t}_{i=o+1} \envObl{\obl{p}(i)}{\rid'} \land
    (t + 1) - o = \lstLen{ns'}  \land
    ns(0) = h \land
    ns'=ns \oplus \pvar{c}
    \ssep (\guardA{\guard{t}(\pvar{p}, \pvar{c}, t)}{\rid'} * {}\\
    \locObl{\obl{p}(t)}{\rid'} * 
    \Sep*^{t - 1}_{i=o+1} \envObl{\obl{p}(i)}{\rid'} \land o < t)
    > \\
  \end{proofjump*} \\
  \ASSR{
         \exists l \in \set{0,1}, o, t' \in \Nat, h \in \Addr \st
         \region{clh}{\rid}(\rid', \pvar{x}, h, l, o) *
         \done{\rid}{\blacklozenge} *
         \guardA{\guard{t}(\pvar{p}, \pvar{c}, t')}{\rid'} *
         \locObl{\obl{p}(t')}{\rid'} *
         \Sep*^{t' - 1}_{i=o+1} \envObl{\obl{p}(i)}{\rid'} \land
         o < t'
  }
\end{proofoutline}
\caption{Proof outline of the \texttt{FAS} call of CLH \texttt{lock}.}
\label{fig:clh-fas}
\end{mathfig}

There, \cref{proof:clh-lock-enqueue} is composed of the rules:
\explainproofjump{proof:clh-lock-enqueue}.
The application of the \ref{rule:frame-hoare} rule frames off the view $\done{\rid}{\blacklozenge}$,
the \ref{rule:atomicity-weak} rule transfers all the remaining resources to the atomic
precondition and postcondition, the \ref{rule:atomic-exists-elim} rule pseudo-quantifies~$l$, $o$ and $h$,
\ref{rule:lift-atomic} then opens up the region $\region{clh}{\rid}$,
the applications of \ref{rule:atomic-exists-elim} and \ref{rule:lift-atomic} then pseudo-quantify $t$ and open the
region $\region{lclh}{}$ and the final application of 
\ref{rule:atomic-exists-elim} rule pseudo-quantifies~$\lvar{ns}$.

After using \ref{rule:layer-weak} to decrease the level of the assertion to $\botOf$
and \ref{rule:frame} to frame off everything except the region interpretation's
tail pointer, the \code{FAS} operation atomically updates it.
After everything is framed back on,
the consequence rule is then applied to the postcondition
so as to re-establish the invariant.
The axioms
\begin{align*}
  \guard{q}(ns \lstPlus \lst{p}, o) & = \guard{q}(ns \lstPlus \lst{p, c}, o)
        \guardOp \guard{t}(p, c, o + \lstLen{ns} + 1) \\
        \obl{o}(o, t) &= \obl{o}(o, t+1) \oblOp \obl{p}(t)
\end{align*}
are used to update the queue $ns$, by enqueuing \p{c}
---the local thread's cell---
at its tail, and updating the next ticket to $t' + 1$.
While doing so,
the thread acquires the guard $\guard{t}(\pvar{p}, \pvar{c}, t')$,
the obligation, $\obl{p}(t')$, which represent the thread's position in the
queue and its obligation to take possession of the lock once its predecessor
reliquishes it respectively.

As environmental obligations can always be duplicated, 
the thread also obtains
$\smash{\Sep*^{t' - 1}_{i=o+1} \envObl{\obl{p}(i)}{\rid'}}$ locally.
These environmental assertions will be necessary for the application of the \ref{rule:while} rule.
To finish reestablishing the invariant, as the thread is retaining $\oblA{\obl{p}(t')}{\rid}$ locally,
it can leave $\envObl{\obl{p}(t')}{\rid'}$ in the region invariant.
Finally, using the axiom
\[
\oblDef{\obl{o}(o, t) \oblOp \obl{p}(t')} \iff o \le t' < t
\]
as we hold $\obl{p}(t')$ locally, the assertion $o < t'$ holds stably.

Next, consider the proof of the \code{while} loop. The loop invariant is:
\begin{align*}
P(\beta) \is
  \exists l, o, t', h \st \; &
    \region{clh}{\rid}(\rid', \pvar{x}, h, l, o) *
    \guardA{\guard{t}(\pvar{p}, \pvar{c}, t')}{\rid'}
    \land o < t' \\ {} \land {} &
    (\pvar{v} = 0 \implies (t' = o + 1 \land l = 0 \land h = \pvar{p})) \land
    \beta = \pvar{v}
\end{align*}
which asserts that:
\begin{itemize}
\item $\guardA{\guard{t}(\pvar{p}, \pvar{c}, t')}{\rid'}$, the local thread is queueing for the lock
  with ticket $t'$ and with the address of the predecessor's cell and the current thread's cell in
  $\pvar{p}$ and $\pvar{c}$ respectively.
\item $o < t'$, the current owner must come before the local thread with ticket $t'$.
  This is stable due to the $\guard{t}$ guard. 
\item $\pvar{v} = 0 \implies (t' = o + 1 \land l = 0 \land h = \pvar{p})$,
  if $\pvar{v}$, the last read of the value of the predecessor cell, is $0$,
  then the owner is the predecessor of the current thread has unlocked
  the lock, as only then can it set its cell to $0$.
  Therefore $t' = o + 1$, and, consequently the lock is unlocked, $l = 0$.
  The owner's cell, $h$, will also take the value of that of the predecessor. 
\item $\beta = \pvar{v}$ which asserts that $\beta = 0$ once the thread has observed that its predecessor has
  taken possession of and then unlocked the lock (by reading the cell at address $\pvar{p}$ into $\pvar{v}$).
  $\beta$ will have value $1$ otherwise.
\end{itemize}
A thread with ticket $t'$ can take possession of a CLH lock
once its predecessor has taken possession of and relinquished the lock.
Once the lock reaches this state, $o = t' - 1$ and $l = 0$ hold stabily as
all transitions from this state would set $o \ge t'$, however we know that, $o < t'$.

The intent of this loop is to wait till this occurs,
allowing the thread to safely take possession of the
lock once the loop terminates. Hence, the goal state is:
\[
T =
\exists l \in \set{0,1}, o \in \Nat, h \in \Addr \st
\region{clh}{\rid}(\rid', \pvar{x}, h, l, o) \land
t' = o + 1 \land l = 0 \land h = \pvar{p}
\]
Once the lock reaches this state, a subsequent iteration of this \code{while} loop will terminate with
$\pvar{v} = 0$, breaking the loop. To reach the goal state, threads that come before the current thread
must both take possession and then unlock the lock. The first is guaranteed due to obligations $\obl{p}(t')$
for $t' < t$ and the second due to the pseudo-quantifier, guaranteeing that the lock must always eventually
be released. The progress measure
\[
M(\alpha) =
\exists l \in \set{0,1}, o \in \Nat, h \in \Addr \st
\region{clh}{\rid}(\rid', \pvar{x}, h, l, o) \land
\alpha = 2 (t' - o - 1) + l
\]
is decreased by both of these actions, and as $t' > o$ implies $2 (t' - o - 1) + l \ge 0$, the progress measure, $\alpha$, is a natural number, and therefore well-founded.

The use of the difference between $t'$, the local thread's ticket
and the owner's ticket, $o$, to bound the number of threads that
can take possession of the lock before the local thread
removes the necessity for the impedance bound, $\alpha$,
required in the proof of the spin lock module, and
that must leak in the associated specification
(as it imposses a restriction on any client).

To support this argument, the persistent loop invariant, $L$, must contain the resource $\done{\rid}{\blacklozenge}$
to make use of the liveness assumptions of the pseudo-quantifier,
guaranteeing that the lock is always eventually unlocked, and
the relevant environmental liveness assertions guaranteeing the threads queued
before the current thread will take possession of it once their predecessor
relinquishes it:
\[
L = \exists l \in \set{0,1}, o, t' \in \Nat, h \in \Addr \st
\region{clh}{\rid}(\rid', \pvar{x}, h, l, o) *
\locObl{\obl{p}(t')}{\rid'} * \Sep*^{t' - 1}_{i=o+1} \envObl{\obl{p}(i)}{\rid'} * \done{\rid}{\blacklozenge} \land o < t'
\]

\begin{mathfig}[htb]
\begin{minipage}[t]{.7\textwidth}
\[
\begin{proofoutline}
  \CTXT \topOf; \map{\rid -> (X_1, {\botOf}, X_2, R)} |- \\
  \PREC{\exists t' \in \Nat \st \exists \beta_0 \st
    P(\beta_0) * L
  } \\
  \begin{proofjump}[rule:consequence,rule:exists-elim]
    \META{\forall \beta_0, t \in \Nat \st } \\
    \CTXT \topOf; \map{\rid -> (X_1, {\botOf}, X_2, R)} |- \\
    \PREC{P(\beta_0) * L} \\
    \begin{proofjump}[rule:while]
      \CODE{while(v != 0) \{} \\
      \begin{proofindent}
        \META{\forall \beta \le \beta_0, b \in \mathbb{B} \st } \\
        \PREC{P(\beta) * b \dotimplies T(\beta) \land \pvar{v} \neq 0} \\
\CODE{v := [p];}\\
        \PREC{\exists \gamma \st P(\gamma) \land \gamma \le \beta \land b \dotimplies \gamma < \beta}
\end{proofindent} \\
      \CODE{\}}
    \end{proofjump} \\
    \PREC{\exists \gamma \st P(\gamma) * L \land \gamma \le \beta_0 \land \pvar{v} = 0}
  \end{proofjump} \\
  \POST{
    \exists o \in \Nat \st
    \region{clh}{\rid}(\rid', \pvar{x}, \pvar{p}, 0, o) *
    \done{\rid}{\blacklozenge} *
    \guardA{\guard{t}(\pvar{p}, \pvar{c}, o + 1)}{\rid'} *
    \locObl{\obl{p}(o + 1)}{\rid'}
  }
\end{proofoutline}\]
\end{minipage}\hspace{-.3\textwidth}
\begin{minipage}[t]{.5\textwidth}
\begin{align*}
  X_1 &\is \set{ (h,l,o) | h \in \Addr, l \in \set{0,1}, o \in \Nat }
  \\
  X_2 &\is \set{ (h,0,o) | h \in \Addr, o \in \Nat }
  \\
  R &\is \set{((h,0,o),(h',1,o)) | h,h' \in \Addr, o\in \Nat}
\end{align*}
\end{minipage}
\caption{Application of \ref{rule:while} in the CLH \texttt{lock} proof.}
\label{fig:clh-while}
\end{mathfig}

The \ref{rule:while} rule is applied as in \cref{fig:clh-while}.
The rule \ref{rule:exists-elim} is applied to quantify $t$ and $\beta_0$ over the antecedent.
To complete the application of the rule we need to show
\begin{gather}
  \ENVLIVE \topOf; \actxt |- L - M ->> T(\beta)
  \label{cond:clh-lock-envlive} \\
  \forall \alpha\st
    \STABLE \actxt |=
      {\exists \alpha'\st L * M(\alpha') \land \alpha' \leq \alpha}
  \label{cond:clh-lock-M}
\end{gather}
Condition~\eqref{cond:clh-lock-M} holds trivially,
as seen above all the possible operations on the module
decrease the environmental metric.

\begin{figure}
\adjustfigure[\small]
\begin{proofoutline}
  \META{\forall \beta_0, t \in \Nat, \beta, b \in \mathbb{B} \st } \\
  \CTXT \topOf; \map{\rid -> (X_1, {\botOf}, X_2, R)} |- \\
  \PREC{(\exists l \in \set{0,1}, o \in \Nat, h \in \Addr \st
    \region{clh}{\rid}(\rid', \pvar{x}, h, l, o) * \guardA{\guard{t}(\pvar{p}, \pvar{c}, t)}{\rid'} \land {} \\
    o < t \land (\pvar{v} = 0 \implies (t = o + 1 \land l = 0 \land h = \pvar{p})) \land
    \beta = \pvar{v} \land b \implies
    (t = o + 1 \land l = 0 \land h = \pvar{p}) \land (\pvar{v} \neq 0)} \\
  \begin{proofjump}[rule:atomicity-weak,rule:lift-atomic,rule:atomic-exists-elim,rule:lift-atomic,rule:atomic-exists-elim,rule:consequence]
    \CTXT \topOf; \map{\rid -> (X_1, {\botOf}, X_2, R)} |- \\
    \A l \in \set{0,1}, h \in \Addr, ns \in \Addr^{*}, o, nt \in \Nat. \\
    \PREC<\pvar{x} \mapsto h, \lfun{last}(ns) *  h \mapsto l *
    \lfun{ones}(ns) *
    \guardA{\guard{q}(ns, o)}{\rid} * \locObl{\obl{o}(o, nt)}{\rid'} * \Sep*^{nt - 1}_{i=o+1} \envObl{\obl{p}(i)}{\rid'}
    \land nt - o = \lstLen{ns} \land {} \\
    ns(0) = h \ssep
    (\guardA{\guard{t}(\pvar{p}, \pvar{c}, t)}{\rid'} \land
    o < t \land \beta \ge 1 \land b \implies (t = o + 1 \land l = 0 \land h = \pvar{p}) \land \pvar{p} \in ns)
    > \\
    \begin{proofjump}[rule:layer-weak,rule:frame]
      \CTXT \botOf; \map{\rid -> (X_1, {\botOf}, X_2, R)} |- \\
      \A v \in \set{0,1}. \\
      \PREC<\pvar{p} \mapsto v> \\
      \CODE{v := [p];}\\
      \POST<\pvar{p} \mapsto v \land \pvar{v} = v> 
    \end{proofjump} \\
    \POST<
    \pvar{x} \mapsto h, \lfun{last}(ns) *  h \mapsto l *
    \lfun{ones}(ns) *
    \guardA{\guard{q}(ns, o)}{\rid} * \locObl{\obl{o}(o, nt)}{\rid'} * \Sep*^{nt - 1}_{i=o+1} \envObl{\obl{p}(i)}{\rid'} \land nt - o = \lstLen{ns} \land {} \\
    ns(0) = h \ssep
    (\guardA{\guard{t}(\pvar{p}, \pvar{c}, t)}{\rid'} \land
    \beta = 1 \land {} \\
    \exists v \in \set{0,1} \st \pvar{v} = v \land
    b \implies \pvar{v} = 0 \land
    (\pvar{v} = 0 \implies (t = o + 1 \land l = 0 \land h = \pvar{p})))
    >
  \end{proofjump} \\
  \POST{\exists l \in \set{0,1}, o \in \Nat, h \in \Addr, \gamma \st
    \region{clh}{\rid}(\rid', \pvar{x}, h, l, o) * \guardA{\guard{t}(\pvar{p}, \pvar{c}, t)}{\rid'} \land {} \\
    o < t \land (\pvar{v} = 0 \implies (t = o + 1 \land l = 0 \land h = \pvar{p})) \land
    \gamma = \pvar{v} \land
    \gamma \le \beta \land b \implies \gamma = 0}
\end{proofoutline}
\caption{Proof outline of the CLH \code{lock}'s loop body.}
\label{fig:clh-loop-body}
\end{figure}

To prove~\eqref{cond:clh-lock-envlive}, take
\begin{align*}
L''_o(\alpha) &=
      \left(
      \begin{array}{c}
        \exists l \in \set{0,1}, h \st
          \region{clh}{\rid}(\rid',\pvar{x}, h, l, o) * \locObl{\obl{p}(t')}{\rid'} * {} \\
        \Sep*^{t' - 1}_{i=o+2} \envObl{\obl{p}(i)}{\rid'} * \done{\rid}{\blacklozenge} \land l = 0 \land o + 1 < t'
      \end{array}
      \right) * M(\alpha) \\
  L'_0(\alpha) &=
      \left(
      \begin{array}{c}
        \exists l \in \set{0,1}, o, h \st
          \region{clh}{\rid}(\rid',\pvar{x}, h, l, o) * \locObl{\obl{p}(t')}{\rid'} * {} \\
        \Sep*^{t' - 1}_{i=o+2} \envObl{\obl{p}(i)}{\rid'} * \done{\rid}{\blacklozenge} \land l = 0 \land o + 1 < t'
      \end{array}
      \right) * M(\alpha) \\
  L'_1(\alpha) &=
      \left(
      \begin{array}{c}
        \exists l \in \set{0,1}, o, h \st
          \region{clh}{\rid}(\rid',\pvar{x}, h, l, o) * \locObl{\obl{p}(t')}{\rid'} * {} \\
        \Sep*^{t' - 1}_{i=o+1} \envObl{\obl{p}(i)}{\rid'} * \done{\rid}{\blacklozenge} \land l = 1
      \end{array}
      \right) * M(\alpha) \\
  L(\alpha) &= L * M(\alpha) \\
\end{align*}
First split on $\alpha = 0 \lor \alpha > 0$:
\begin{mathpar}
\small
\infer*[right={\ref{rule:envlive}}]{
  \infer*[right={\ref{rule:envlive-case}}]{
    \infer*[right={\ref{rule:envlive-target}}]{
      \forall \alpha. \VALID \actxt |= L(\alpha) \land \alpha = 0 \implies T
    }{
      \ENVLIVE \topOf; \actxt |- L(\alpha) : L(\alpha) \land \alpha = 0 -->> T
    }
    \and
    \infer*[right={\ref{rule:envlive-case}}]{
      (\ref{clh:live-o}) \and (\ref{clh:live-a})
    }{
      \ENVLIVE \topOf; \actxt |-
      L(\alpha) : (L'_0(\alpha) \lor L'_1(\alpha) ) \land \alpha > 0 -->> T
    }
  }{
    \ENVLIVE \topOf; \actxt |-
      L(\alpha) : L(\alpha) -->> T
  }
}{
\ENVLIVE \topOf; \actxt |- L - M ->> T
}
\end{mathpar}
In the case $\alpha = 0$, the rule \ref{rule:envlive-target} applies directly. To show
$\ENVLIVE \topOf; \actxt |- L(\alpha) : L'(\alpha) \land \alpha > 0 -M->> T$ holds,
split on the state of the lock, $l = 0 \lor l = 1$.

In the case $l = 0$, for each $o \in \Nat$, the ticket of the current owner of the lock,
the environment is guaranteed to eventually take possession of the lock due to the environmental obligation
assertion $\envObl{\obl{p}(o+1)}{\rid'}$. To consider each case for $o \in \Nat$, we first apply the
rule \ref{rule:envlive-quant} and then the \ref{rule:envlive-obl} rule:
\begin{equation}
\infer*[right={\ref{rule:envlive-quant}}]{
  \infer*[right={\ref{rule:envlive-obl}}]{
    \decr[\actxt](L''_o, L, T) \\
         {\forall \alpha \st
           \VALID \actxt |=
           L''_o(\alpha) \implies
           \region{clh}{\rid}(\wtv[3], o) * \envObl{\obl{p}(o+1)}{\rid'} * \True
         } \\
  \forall \alpha\st
  \VALID \actxt |= \minLayStrict{L''_o(\alpha)}{\lay(\obl{p}(o+1))}\\
  \topOf \laygt \lay(\obl{p}(o+1))
  }{
    \forall o \in \Nat \st \ENVLIVE \topOf; \actxt |- L(\alpha) : L''_o(\alpha) -->> T
  }
}{
  \ENVLIVE \topOf; \actxt |- L(\alpha) : \exists o \in \Nat\st L''_o(\alpha) -->> T
} \label{clh:live-o}
\end{equation}

With the exception of $\decr[\actxt](L'_0(o), L, T)$, all of these conditions hold trivially.
This last condition holds as, given $\alpha_0 \in \Ord$,
all possible transitions either preserve $L'_0(\alpha)$
or decrease the metric.

In the case $l = 1$, progress is guaranteed due to the assumptions in the atomicity context, $\actxt$,
that eventually, the lock must be released, so the \ref{rule:envlive-pq} rule is applied:
\begin{align}
\infer*[right={\ref{rule:envlive-pq}}]{
  \decr[\actxt](L'_1, L, T) \\
  \topOf \laygt \botOf
  \\
  \forall \alpha\st
    \VALID \actxt |= \minLayStrict{L'_1(\alpha)}{k}
  \\
  (X_1 \eventually[\botOf] X_2) = \live(\actxt,r) \\
  \SAT[\actxt] |- L'_1(\alpha) \implies
    \exists x \in X_1 \setminus X_2 \st
    \region{clh}{\rid}(\rid', \pvar{x},x)
    * \done{r}{\lozenge}
    * \True
}{
  \ENVLIVE \topOf; \actxt |- L(\alpha) : L'_1(\alpha) -->> T
}
\label{clh:live-a}
\end{align}
Once again, with the exception of $\decr[\actxt](L'_1, L, T)$, all of these conditions hold trivially.
This last condition holds as, given $\alpha_0 \in \Ord$,
all possible transitions either preserve $L'_1(\alpha)$ or decrease the metric.

To conclude the proof of \p{lock}, the argument for the body of the \code{while} loop's proof is purely a safety argument, the full proof is in \cref{fig:clh-loop-body}.

The key step uses the axiom
\[
\guardDef{\guard{q}(ns, o) \guardOp \guard{t}(p, c, t)}
\iff
\lstAt{ns}{t - o - 1} = p \land \lstAt{ns}{t - o} = c
\]
Since we hold the guard $\guard{t}(\pvar{p}, \pvar{c}, t)$,
we can infer $\pvar{p} \in ns$. Then, after the value of the cell
at \p{p} has been read, if the value, \p{v}, is $0$, then,
since only the thread holding the lock can change the value of
their associated cell to $0$, then, $t = o + 1 \land l = 0 \land h = \pvar{p}$.
As a consequence, if $b$ holds initially, then $\pvar{v} = 0$ after the body of the loop
is executed, therefore the loop variant in the postcondition, $\gamma = 0$.
As initially, we know $\pvar{v} \neq 0$ from the loop condition, $\beta = 1$,
therefore $\gamma < \beta$.

\begin{figure}
\centering\small \begin{proofoutline}
  \CTXT \topOf; \map{\rid -> (X_1, {\botOf}, X_2, R)} |- \\
  \PREC{\exists o \in \Nat \st \region{clh}{\rid}(\rid', \pvar{x}, \pvar{p}, 0, o) * \done{\rid}{\blacklozenge} * \guardA{\guard{t}(\pvar{p}, \pvar{c}, o + 1)}{\rid} * \locObl{\obl{p}(o + 1)}{\rid'}} \\
  \begin{proofjump*}[rule:consequence,rule:exists-elim,rule:atomicity-weak,rule:update-region,rule:atomic-exists-elim,rule:lift-atomic,rule:atomic-exists-elim,rule:consequence]
    \label{proof:clh-lock-wait}
    \CTXT \topOf; \emptyset |- \\
    \A t \in \Nat, ns \in \Addr^{*}. \\
    \PREC<
    \pvar{x} \mapsto \pvar{p}, \lfun{last}(ns) *
    \pvar{p} \mapsto 0 *
\lfun{ones}(ns) *
    \guardA{\guard{q}(ns, o)}{\rid'} *
    \locObl{\obl{o}(o, t)}{\rid'} *
    \Sep*^{t - 1}_{i=o+1} \envObl{\obl{p}(i)}{\rid'}
    \land t - o = \lstLen{ns} \land {} \\
    ns(0) = \pvar{p} \ssep
    (\guardA{\guard{t}(\pvar{p}, \pvar{c}, o + 1)}{\rid'} *
    \locObl{\obl{p}(o + 1)}{\rid'} \land
    \lstAt{ns}{1} = \pvar{c})
    > \\
    \begin{proofjump}[rule:consequence]
      \PREC<
      \pvar{x} \mapsto \pvar{p}, \lfun{last}(ns) *
      \pvar{p} \mapsto 0 *
\lfun{ones}(ns) *
      \guardA{\guard{q}(ns, o)}{\rid'} *
      \locObl{\obl{o}(o, t)}{\rid'} *
      \Sep*^{t - 1}_{i=o+1} \envObl{\obl{p}(i)}{\rid'}
      \land t - o = \lstLen{ns} \land {} \\
      ns(0) = \pvar{p} \ssep
      (\guardA{\guard{t}(\pvar{p}, \pvar{c}, o + 1)}{\rid'} *
      \locObl{\obl{p}(o + 1)}{\rid'} \land
      \lstAt{ns}{1} = \pvar{c})
      > \\
      \begin{proofjump}[rule:layer-weak,rule:frame]
        \CTXT \botOf; \emptyset |- \\
        \PREC<\pvar{x} \mapsto \pvar{p}> \\
        \CODE{[x] := c;} \\
        \POST<\pvar{x} \mapsto \pvar{c}>
      \end{proofjump} \\
      \POST<
      x \mapsto \pvar{c}, \lfun{last}(ns) *
      \pvar{p} \mapsto 0 *
\lfun{ones}(ns) *
      \guardA{\guard{q}(ns, o)}{\rid'} *
      \locObl{\obl{o}(o, t)}{\rid'} * {} \\
      \Sep*^{t - 1}_{i=o+1} \envObl{\obl{p}(i)}{\rid'} \land
      t - o = \lstLen{ns} \land ns(0) = \pvar{p} \ssep
      (\guardA{\guard{t}(\pvar{p}, \pvar{c}, o + 1)}{\rid'} *
      \locObl{\obl{p}(o + 1)}{\rid'} \land
      \lstAt{ns}{1} = \pvar{c})
      >
    \end{proofjump} \\
    \POST<
    \exists ns' \in \Addr^{*} \st
    x \mapsto \pvar{c}, \lfun{last}(ns') *
    \pvar{c} \mapsto 1 *
\lfun{ones}(ns') *
    \guardA{\guard{q}(ns', o + 1)}{\rid'} *
    \locObl{\obl{o}(o + 1, t)}{\rid'} * {} \\
    \Sep*^{t - 1}_{i=o+2} \envObl{\obl{p}(i)}{\rid'} \land
    t - (o + 1) = \lstLen{ns'} \land ns'(0) = \pvar{c} \ssep
    (\pvar{p} \mapsto 0 * ns = \pvar{p} \lstPlus ns')
    >
  \end{proofjump*} \\
  \POST{
    \exists o \in \Nat, h, h' \in \Addr \st
    \done{\rid}{((h, 0, o), (h', 1, o + 1))} *
    \pvar{p} \mapsto 0 \land
    l = 0
  }
\end{proofoutline}
\caption{Proof outline for the linearization point of CLH \code{lock}.\\
  \cref{proof:clh-lock-wait} is \explainproofjump{proof:clh-lock-wait}
}
\label{fig:clh-lock-linpt-proof}
\end{figure}

Finally, in \cref{fig:clh-lock-linpt-proof}
we consider the details of the linearization point,
when the \code{lock} operation takes possession of the lock.
First \ref{rule:exists-elim} rule is applied to quantify the ticket of the current owner, $o$, (the predecessor of the current thread) over the antecedent.
Then the \ref{rule:atomicity-weak} and \ref{rule:update-region} rules are applied to atomically update the region state by acting on its interpretation.
The rules \ref{rule:atomic-exists-elim}, \ref{rule:lift-atomic} and \ref{rule:atomic-exists-elim} are then applied to pseudo-quantify $t$ and $ns$, the two variables
that are existentially quantified within the region invariants and open the region $\region{lclh}{}$.
Finally the \ref{rule:consequence} rule is applied to re-establish the invariant in the postcondition by adjusting the ghost state.
Specifically, the guard $\guard{T}$ and the obligation $\obl{P}$ are reabsorbed into $\guard{Q}$ and~$\obl{O}$ respectively, to update the list of threads waiting on the lock and increment the owner.
This is done using the axioms:
\begin{align*}
  \guard{q}(\lst{p,c} \lstPlus ns, o)
    \guardOp
  \guard{t}(p, c, o + 1)
  &=
  \guard{q}(c \lstPlus ns, o + 1) \\
  \obl{o}(o, t) \oblOp \obl{p}(o + 1) &= \obl{o}(o + 1, t)
\end{align*}

The inner part of the proof then decreases the layer and frames off unecessary resources to apply the update.
Note that this step of the proof discharges the obligation $\obl{p}(t')$.
This concludes the verification of the \p{lock} operation.

The CLH lock proof is able to internally encode the impedance bound enforced by
thread queueing using ghost state: the local ticket numbers of each thread
queueing for the lock and the owner's ticket number which is visible in the
abstract state of the region $\region{clh}{}$, but hidden from the client.

\paragraph{Proof of\/\code{unlock}}
Let
  $X = \set{ (h,1,o) | h \in \Addr, o \in \Nat } $ and
  $R = \set{ ((h,1,o), (h,0,o)) | h \in \Addr, o \in \Nat }$.
The proof of the \code{unlock} operation is as follows:
\[
  \begin{proofoutline}
  \CTXT \botOf; \emptyset |- \\
  \PREC<\ap{L}(r, \pvar{x}, 1)> \\
  \begin{proofjump}[rule:consequence,rule:atomic-exists-elim,rule:make-atomic]
    \label{proof:clh-unlock}
    \CTXT \botOf; \map{\rid -> (X, {\botOf}, X, R)} |- \\
    \PREC{\exists o \in \mathbb{N}, h \in \Addr \ldotp \region{clh}{\rid}(x, h, 1, o) *\done{\rid}{\blacklozenge}} \\
    \CODE{h := [x];} \\
    \ASSR{\exists o \in \mathbb{N} \ldotp \region{clh}{\rid}(x, \pvar{h}, 1, o) *\done{\rid}{\blacklozenge}} \\
    \CODE{[h] := 0;} \\
    \POST{\exists o \in \mathbb{N} \ldotp \done{\rid}{((\pvar{h}, 1, o),(\pvar{h}, 0, o))}}
  \end{proofjump} \\
  \POST<\ap{L}(r, \pvar{x}, 0)>
\end{proofoutline}
 \]
 
\subsection{Blocking Counter}
\label{app:ex-blocking-counter}
We sketch the proof of a blocking counter module:
a single cell storing a natural number that can be incremented,
guarded by a non-fair lock for concurrent access.
The example illustrates how the \tadalive\ specifications and proofs
neatly support hiding blocking when it is unobservable by the client,
while still leaking the requirement of bounded impedance from the lock.
This requires any client to only call operations making use of the lock
(in this case the \p{incr} operation) a bounded number of times.

\paragraph{Code}
The implementation of the module's operations is:
\begin{center}\small
\begin{tabular}{c@{\hspace{5em}}c@{\hspace{5em}}c}
  {\codefromfile[numbers=left]{examples/blocking-counter/code-make}} &
  {\codefromfile[numbers=left]{examples/blocking-counter/code-incr}} &
  {\codefromfile[numbers=left]{examples/blocking-counter/code-read}}
\end{tabular}
\end{center}

\paragraph{Specifications}
The abstract predicate $\ap{C}(s, x, n, \alpha)$ represents a blocking counter at address $x$ with value $n$ and
impedance bound $\alpha$. 

\begin{align*}
  &\forall \alpha \st \TRIPLE \topOf |- {\emp} {\code{makeCounter()}} {\exists s \st \ap{C}(s, \ret, 0, \alpha)} \\
  &\forall \ordfun \st \ATRIPLE \topOf |-
  \A n \in \Nat, \alpha.
  <\emp | \ap{C}(s, \pvar{x}, n, \alpha) \land {\alpha > \ordfun(\alpha)}>
  \code{incr(x)}
  <\ret = n | \ap{C}(s, \pvar{x}, n + 1, \ordfun(\alpha))> \\
  &\ATRIPLE \topOf |-
  \A n \in \Nat, \alpha.
  <\emp | \ap{C}(s, \pvar{x}, n, \alpha)>
  \code{read(x)}
  <\ret = n | \ap{C}(s, \pvar{x}, n, \alpha)> \\
\end{align*}

\paragraph{Shared Regions}
This proof will use two region types: $\region{cnt}{\rid}(\rid', x, s, la, n, \alpha)$
and $\region{lcnt}{\rid'}(x, s, la, l, n, \alpha)$ where
$\rid, \rid' \in \RId$,
$x, la \in \Addr$,
$l \in \set{0,1}$,
$n \in \Nat$,
$\alpha \in \Ord$
and~$s$ is the abstract location
of the lock guarding the counter resource.
Here $\rid'$, $x$, $s$ and $la$ are the fixed parameters of the regions,
representing respectively the region identifier of the inner region,
the address of the blocking counter and the abstract location and address
of the associated lock.

As in the CLH lock example, we will use two nested regions.
The region type $\region{lcnt}{}$ will be used as an inner
region revealing sufficient information to prove desired liveness properties,
in particular, exposing the state of the lock,~$l$.
The region type $\region{cnt}{}$ will be used to prove linearizability of our operations; to this end, it only exposes the value
of the blocking counter~$n$, and the lock's impedence bound~$\alpha$.

\paragraph{Guards and Obligations}

We associate the exclusive guard $\gEx$ with both $\region{cnt}{}$ and $\region{lcnt}{}$.
Besides this, this proof will also require the guards $\guard{u}$, $\guard{l}(n, n')$
and $\guard{k}(n, n')$, where $n, n' \in \Nat$, for the latter region.
These guards will be used to record the update to the value of the counter that will
occur at the moment the module's lock is locked in the proof of \p{incr}.
Since other threads cannot observe the value of the counter without first
holding the lock, performing this abstract update on the state of the outer region,
$\region{cnt}{}$, and then updating the concrete state of the counter before
releasing the lock results in a linearizable implementation.

To allow this, once the lock is locked, the concrete value of the counter, $n' \in \Nat$, and the
updated value of the counter, $n \in \Nat$, are stored in the guard $\guard{l}(n, n')$
within the region $\region{cnt}{}$. The thread holding the lock then holds the guard $\guard{k}(n, n')$,
which keeps a local record of the concrete and updated counter values; the values are required to match
with those stored in $\guard{l}(n, n')$ within the region by the axiom:
\[
\guard{l}(n, n') \guardOp \guard{k}(m, m') \text{ is defined} \iff n = m \land n' = m'
\]
When the lock is unlocked, the guard $\guard{u}$ is stored within the region $\region{cnt}{}$.
When a thread takes possession of the lock, it can be split into the guards
$\guard{l}(n, n')$ and $\guard{k}(n, n')$ using the axiom:
\[
\guard{u} = \guard{l}(n, n') \guardOp \guard{k}(n, n')
\]
Finally, if a thread holds the guard $\guard{k}(n, n')$,
it holds the lock, which can be inferred from the axiom:
\[
\guard{u} \guardOp \guard{k}(n, n') \text{ is undefined.}
\]

This pattern of three guards is often used as a \tada\ pattern
to encode mutual exclusion on some resource when a thread has
possession of a shared lock.

We also associate a single atom obligation~$\obl{k}$ with the region type $\region{lcnt}{}$.
This obligation encodes ownership of the blocking counter's lock, as well
as the obligation to unlock it. We set $\lay(\obl{k}) = \botOf$.

\paragraph{Region Protocols}
The guard-labelled transition system of the region $\region{cnt}{}$ is:
\begin{align*}
  \gEx &: ((n, \alpha), \oblZero) \interfTo ((n + 1, \beta), \oblZero) & \alpha > \beta
\end{align*}
and the guard-labelled transition system of the region $\region{lcnt}{}$ is:
\begin{align*}
  \gEx &: ((0, n, \alpha), \oblZero) \interfTo ((1, n, \beta), \obl{k}) & \alpha > \beta \\
  \gEx &: ((1, n, \alpha), \oblZero) \interfTo ((1, n + 1, \alpha), \oblZero) \\
  \gEx &: ((1, n, \alpha), \obl{k}) \interfTo ((0, n, \alpha), \oblZero)
\end{align*}

\paragraph{Region Intepretations}
The interpretation of the locked counter region $\rt[lcnt]$
links the state of the lock and counter to the abstract state of the region
and the ownership of $\obl{k}$.

The region $\rt[cnt]$ is a wrapper around the $\rt[lcnt]$ region
that hides the state of the lock and allows the counter value of
the region $\region{lcnt}{}$ to be disconnected from that of the outer
region when the lock is locked.
\begin{align*}
\rInt(\region{lcnt}{\rid}(x, s, la, l, n, \alpha)) &\is
  x \mapsto la, n * \ap{L}(s, la, l, \alpha) * (l = 0 \dotimplies \oblA{\obl{k}}{\rid})
\\
\rInt(\region{cnt}{\rid}(\rid', x, s, la, n, \alpha)) &\is
\exists n' \in \Nat, l \in \set{0,1} \st \region{lcnt}{\rid'}(x, s, la, l, n', \alpha) * \guardA{\gEx}{\rid'} \\ &\qquad {} *
\bigl(
(l = 0 \land n = n' \land \guardA{\guard{u}}{\rid}) \lor (l = 1 \land \guardA{\guard{l}(n,n')}{\rid} * \envObl{\obl{k}}{\rid'})
\bigr)
\end{align*}

\paragraph{Predicates}
The counter resource is abstractly represented by the predicate
\[
\ap{C}((\rid, \rid', s, la), x, n, \alpha) \is
  \region{cnt}{\rid}(\rid', x, s, la, n, \alpha) * \guardA{\gEx}{\rid}
\]

\paragraph{Verification of \code{incr}}
The proof of \code{incr} can be found in \cref{fig:bcounter-incr-outline}.
The only step requiring liveness reasoning is the call \code{lock(x)},
which is handled very similarly to the same call in
the left thread of the distinguishing client
where the environment liveness condition of the \ref{rule:liveness-check} rule
application is discharged using the fact that when $l=1$ holds, then $\envObl{\obl{k}}{\rid}$,
which, in this case, is obtained from the interpretation of the outer region, $\region{cnt}{}$.
The details of the proof of the \p{lock} operation cab be found in
\cref{fig:bcounter-incr-lock-details}.

\paragraph{Verification of the \code{makeCounter} and \code{read} operations}
The proof of \code{makeCounter} proceeds,
using standard steps on Hoare triples,
by establishing the postcondition
$ \exists x,la,lr,\alpha \st x \mapsto la, 0 * \ap{L}(lr, la, 0, \alpha) $
which can be viewshifted to
$ \exists x,la,\rid,\rid',lr,\alpha \st \ap{C}((\rid, \rid', lr, la), x, 0, \alpha) $.

The proof of \code{read} is almost identical to the proof in \cref{fig:bcounter-incr-outline}.
The reader might wonder if the lock acquisition in the code is strictly necessary.
Indeed, it is not given the current set of operations available to the client.
To prove the version where \code{read} does not acquire the lock,
however, we would need to change the region's protocol to encode the fact that
while holding a lock a single write to it is possible.
Since one would conceivably want to extend the module with other operations that write to the counter multiple times while holding the lock,
we formalised the more general protocol.
In the presence of such additional operations, \code{read}
would need to acquire the lock to be correct.

\begin{mathfig}
  \begin{proofoutline}
    \PREC{\exists n, \alpha, l \st \region{cnt}{\rid}(\rid', \pvar{x}, s', \pvar{l}, n, \alpha) *
      \done{\rid}{\blacklozenge} * \region{lcnt}{\rid'}(\pvar{x}, s', \pvar{l}, l, \wtv[2])
      * l = 1 \dotimplies \envObl{\obl{K}}{\rid'}
      \land {\alpha > \ordfun(\alpha)}
    } \\
    \begin{proofjump}[rule:atomicity-weak,rule:atomic-exists-elim]
      \A n \in \Nat, \alpha. \\
      \PREC<\exists l \st \region{lcnt}{\rid'}(\pvar{x}, s', \pvar{l}, l, \wtv[2])
      * l = 1 \dotimplies \envObl{\obl{K}}{\rid'}
      | \region{cnt}{\rid}(\rid', \pvar{x}, s', \pvar{l}, n, \alpha) *
      \done{\rid}{\blacklozenge} \land
           {\alpha > \ordfun(\alpha)}
           > \\
           \begin{proofjump}[rule:update-region]
             \PREC<\exists l \st \region{lcnt}{\rid'}(\pvar{x}, s', \pvar{l}, l, \wtv[2])
             * l = 1 \dotimplies \envObl{\obl{K}}{\rid'}
             | \exists n', l \st
             \region{lcnt}{\rid'}(\pvar{x}, s', \pvar{l}, l, n', \alpha) * \guardA{\gEx}{\rid'} * {} \\
             \left(
             \begin{array}{c}
               (l = 0 \land n = n' \land \guardA{\guard{u}}{\rid}) \lor {} \\
               (l = 1 \land \guardA{\guard{l}(n,n')}{\rid} * \envObl{\obl{k}}{\rid'})
             \end{array}
             \right) \land {\alpha > \ordfun(\alpha)}
             > \\
             \begin{proofjump}[rule:atomic-exists-elim]
               \A n, n' \in \Nat, l \in \set{0,1}, \alpha. \\
               \PREC<\exists l \st \region{lcnt}{\rid'}(\pvar{x}, s', \pvar{l}, l, \wtv[2])
               * l = 1 \dotimplies \envObl{\obl{K}}{\rid'}
               | \region{lcnt}{\rid'}(\pvar{x}, s', \pvar{l}, l, n', \alpha) * \guardA{\gEx}{\rid'} * {} \\
               \left(
               \begin{array}{c}
                 (l = 0 \land n = n' \land \guardA{\guard{u}}{\rid}) \lor {} \\
                 (l = 1 \land \guardA{\guard{l}(n,n')}{\rid} * \envObl{\obl{k}}{\rid'})
               \end{array}
               \right) \land {\alpha > \ordfun(\alpha)}
               > \\
               \begin{proofjump}[rule:liveness-check]
                 \A n, n' \in \Nat, l \in \set{0,1} \eventually[\botOf] \set{0}, \alpha. \\
                 \PREC<
                 \region{lcnt}{\rid'}(\pvar{x}, s', \pvar{l}, l, n', \alpha) * \guardA{\gEx}{\rid'} *
                 \left(
                 \begin{array}{c}
                   (l = 0 \land n = n' \land \guardA{\guard{u}}{\rid}) \lor {} \\
                   (l = 1 \land \guardA{\guard{l}(n,n')}{\rid} \land \envObl{\obl{k}}{\rid'})
                 \end{array}
                 \right) \land {\alpha > \ordfun(\alpha)}
                 > \\
                 \begin{proofjump*}[rule:lift-atomic,rule:frame]
                   \label{step:incr-lock-inner}
                   \A l \in \set{0,1} \eventually[\botOf] \set{0}, \alpha. \\
                   \PREC<\ap{L}(s', \pvar{l}, l, \alpha) \land {\alpha > \ordfun(\alpha)}> \\
                   \CODE{lock(l);} \\
                   \POST<\ap{L}(s', \pvar{l}, 1, \alpha) \land l = 0>
                 \end{proofjump*} \\
                 \POST< \locObl{\obl{k}}{\rid'} |
                 \region{lcnt}{\rid'}(\pvar{x}, s', \pvar{l}, 1, n', \alpha) * \guardA{\gEx}{\rid'} *
                 \left(
                 \begin{array}{c}
                   (l = 0 \land n = n' \land \guardA{\guard{u}}{\rid}) \lor {} \\
                   (l = 1 \land \guardA{\guard{l}(n,n')}{\rid} \land \envObl{\obl{k}}{\rid'})
                 \end{array}
                 \right) \land
                 l = 0
                 > \\
               \end{proofjump} \\
               \POST< \locObl{\obl{k}}{\rid'} |
               \region{lcnt}{\rid'}(\pvar{x}, s', \pvar{l}, 1, n', \alpha) * \guardA{\gEx}{\rid'} *
               \guardA{\guard{u}}{\rid} \land n = n'
               > \\
             \end{proofjump} \\
             \POST< \locObl{\obl{k}}{\rid'} |
             \exists n', l \st
             \region{lcnt}{\rid'}(\pvar{x}, s', \pvar{l}, 1, n', \alpha) * \guardA{\gEx}{\rid'} *
             \left(
             \begin{array}{c}
               (l = 0 \land n + 1 = n' \land \guardA{\guard{u}}{\rid}) \lor {} \\
               (l = 1 \land \guardA{\guard{l}(n + 1,n')}{\rid} \land \envObl{\obl{k}}{\rid'})
             \end{array}
             \right) *
             \guardA{\guard{k}(n + 1, n)}{\rid}
             > \\
           \end{proofjump} \\
           \ASSR<\locObl{\obl{k}}{\rid'} | \region{cnt}{\rid}(\rid', \pvar{x}, s', \pvar{l}, n + 1, \alpha) *
           \done{\rid}{\DONE({(n, \alpha), (n + 1, \ordfun(\alpha))})} *
           \guardA{\guard{k}(n + 1, n)}{\rid}
           >
    \end{proofjump} \\
    \ASSR{\exists n, \alpha \st \region{cnt}{\rid}(\rid', \pvar{x}, s', \pvar{l}, n + 1, \ordfun(\alpha)) *
      \done{\rid}{\DONE({(n, \alpha), (n + 1, \ordfun(\alpha))})} * \oblA{\obl{k}}{\rid'} *
      \guardA{\guard{k}(n + 1, n)}{\rid}
    }
  \end{proofoutline}
  \caption{Details of the proof of the \code{lock(l)} call of \code{incr}.
    \Cref{step:incr-lock-inner} is \explainproofjump{step:incr-lock-inner}.}
  \label{fig:bcounter-incr-lock-details}
\end{mathfig}

\begin{mathfig}
  \begin{proofoutline}
    \TITLE{Proof of \code{incr(x)}:}
    \llap{\META{\forall \ordfun\st\,}}\CTXT \topOf; \emptyset |-
    \A n \in \Nat, \alpha. \\
    \PREC<\emp | \ap{C}(s, \pvar{x}, n, \alpha) \land {\alpha > \ordfun(\alpha)}> \\
    \begin{proofjump}[rule:consequence,"$s=(\rid\text{,}\rid'\text{,}s'\text{,}la)$"]
      \PREC<\emp | \region{cnt}{\rid}(\rid', \pvar{x}, s', la, n, \alpha) *
      \guardA{\gEx}{\rid} \land {\alpha > \ordfun(\alpha)}> \\
      \begin{proofjump}[rule:make-atomic]
        \CTXT \topOf; \map{\rid ->
          (\Nat \times \Ord,
          {\botOf},
          \Nat \times \Ord,
          \set{((n, \alpha), (n + 1, \beta))| n \in \Nat, \alpha, \beta \in \Ord, \alpha > \beta})
        } |- \\
        \PREC{\exists n, \alpha \st \region{cnt}{\rid}(\rid', \pvar{x}, s', la, n, \alpha) *
          \done{\rid}{\blacklozenge} \land {\alpha > \ordfun(\alpha)}} \\
        \CODE{l := [x];} \\
        \ASSR{\exists n, \alpha, l \st \region{cnt}{\rid}(\rid', \pvar{x}, s', \pvar{l}, n, \alpha) *
          \done{\rid}{\blacklozenge} * \region{lcnt}{\rid'}(\pvar{x}, s', \pvar{l}, l, \wtv[2])
          * l = 1 \dotimplies \envObl{\obl{K}}{\rid'}
          \land {\alpha > \ordfun(\alpha)}
        } \\
        \CODE{lock(l);} \\
        \ASSR{\exists n, \alpha \st \region{cnt}{\rid}(\rid', \pvar{x}, s', \pvar{l}, n + 1, \ordfun(\alpha)) *
          \done{\rid}{\DONE({(n, \alpha), (n + 1, \ordfun(\alpha))})} * \oblA{\obl{k}}{\rid'} *
          \guardA{\guard{k}(n + 1, n)}{\rid}
        } \\
        \CODE{v := [x + 1];} \\
        \ASSR{\exists n, \alpha \st \region{cnt}{\rid}(\rid', \pvar{x}, s', \pvar{l}, n + 1, \ordfun(\alpha)) *
          \done{\rid}{\DONE({(n, \alpha), (n + 1, \ordfun(\alpha))})} * \oblA{\obl{k}}{\rid'} *
          \guardA{\guard{k}(n + 1, n)}{\rid} \land
          \pvar{v} = n} \\
        \CODE{[x + 1] := v + 1;} \\
        \ASSR{\exists n, \alpha \st \region{cnt}{\rid}(\rid', \pvar{x}, s', \pvar{l}, n + 1, \ordfun(\alpha)) *
          \done{\rid}{\DONE({(n, \alpha), (n + 1, \ordfun(\alpha))})} *
          \oblA{\obl{k}}{\rid'} *
          \guardA{\guard{k}(n + 1, n + 1)}{\rid} \land
          \pvar{v} = n} \\
        \CODE{unlock(l);} \\
        \ASSR{\exists n, \alpha \st \done{\rid}{\DONE({(n, \alpha), (n+1, \ordfun(\alpha))})} \land \pvar{v} = n} \\
        \CODE{ret := v;} \\
        \POST{\exists n, \alpha \st \done{\rid}{\DONE({(n, \alpha), (n+1, \ordfun(\alpha))})} \land \ret = n}
      \end{proofjump} \\
      \POST<\ret = n | \region{cnt}{\rid}(\rid', \pvar{x}, s', la, n + 1, \ordfun(\alpha)) * \guardA{\gEx}{\rid}>
    \end{proofjump} \\
    \POST<\ret = n | \ap{C}(s, \pvar{x}, n + 1, \ordfun(\alpha))>
  \end{proofoutline}
   \caption{Blocking counter: proof of \protect{\code{incr}}.}
  \label{fig:bcounter-incr-outline}
\end{mathfig}
 
\subsection{Double Blocking Counter}

We now develop the proof of a double blocking counter module,
that is, a module encapsulating two integers each protected by a fair lock.
The module offers linearizable operations to
increment/read each counter in isolation and
an \code{incrBoth} operation to atomically increment both.
The implementation of \code{incrBoth} needs to deal with the
ubiquitous pattern of locking multiple locks in a nested fashion,
which is one of the most common sources of deadlocks in coarse-grained
concurrent programs.
The example illustrates how the specification format and layer system of \tadalive\ allow for modular proofs of deadlock-freedom.
In particular, verifying the example in LiLi would require:
\begin{enumerate*}[label=(\roman*)]
  \item replacing the calls to the lock operations with some non-atomic abstract code
  \item building a termination argument that talks about the queues of the two fair locks; in particular the variant argument would need to consider both queues at the same time and argue about all the possible ways the threads in the environment may enter and exit both queues.
\end{enumerate*}
We avoid these complications by:
\begin{enumerate*}[label=(\roman*)]
  \item reusing the (fair) lock specifications which are truly atomic and properly hide the queues
  \item arguing about termination by means of two obligations with layers
  the order of which reflect the order of acquisition of locks.
  These obligations only represent the liveness invariant that each lock is always eventually released, the layers represent the dependency between the two locks. The proof requires no detail about why, thanks to the internal queues, this is sufficient to ensure global progress: that part of the argument has already been made in proving the lock specifications!
\end{enumerate*}

\paragraph{Code}
The implementation of the module's operations is in \cref{fig:dbcnt-code}
using the following abbreviations for readability:
\begin{align*}
  \p{x.lock1} &\is \p{[x]}   &
  \p{x.lock2} &\is \p{[x+1]} &
  \p{x.cnt1}  &\is \p{[x+2]} &
  \p{x.cnt2}  &\is \p{[x+3]}
\end{align*}

\begin{figure}
\centering\small \begin{tabular}{l@{\hspace{3em}}l@{\hspace{3em}}l@{\hspace{3em}}l}
  {\codefromfile[numbers=left]{examples/double-blocking-counter/code-make}} &
  {\codefromfile[numbers=left]{examples/double-blocking-counter/code-incr1}} &
  {\codefromfile[numbers=left]{examples/double-blocking-counter/code-incr2}} &
  {\codefromfile[numbers=left]{examples/double-blocking-counter/code-incrboth}}
\end{tabular}
\caption{Code of the double blocking counter operations.}
\label{fig:dbcnt-code}
\end{figure}

\paragraph{Specifications}
The fair lock module specifications assumed in this example are
\begin{align*}
  &\ATRIPLE \topOf[\rid] |-
    \A l \in \set{0, 1} \eventually[{\botOf[\rid]}] \set{0}.
      <\ap{L}(\rid, \pvar{x}, l)>
        \code{lock(x)}
      <\ap{L}(\rid, \pvar{x}, 1) \land l = 0>
  \\
  &\ATRIPLE \botOf[\rid] |-
    <\ap{L}(\rid, \pvar{x}, 1)>
      \code{unlock(x)}
    <\ap{L}(\rid, \pvar{x}, 0)>
\end{align*}where $\topOf[\rid]$ and $\botOf[\rid]$ are layers parametrised on the region
identifier $\rid$ of the shared lock.
It is a common \tadalive\ pattern to parametrise the layers of specifications
so that they can be instantiated with differently for each instance of the module.
In~\cref{sec:lock-coupling-set} we explain this parametrisation in general,
and how to parametrise the implementation proof accordingly.

The abstract predicate $\ap{DC}(t, x, n, m)$ represents a
double counter at address~$x$ with abstract location~$t$ and
values~$n$ and~$m$ respectively.
We wish to the show the implementations of the module's
operations satisfy the following specifications:
\begin{align*}
  &\TRIPLE \topOf |-
    {\emp}
    {\code{makeDCounter()}}
    {\exists t \st \ap{DC}(t, \ret, 0, 0)}
  \\
  &\ATRIPLE \topOf |-
      \A n, m \in \Nat.
      <\ap{DC}(t, \pvar{x}, n, m)>
      \code{incrBoth(x)}
      <\ap{DC}(t, \pvar{x}, n + 1, m + 1)>
  \\
  &\ATRIPLE \topOf |-
      \A n, m \in \Nat.
      <\emp | \ap{DC}(t, \pvar{x}, n, m)>
      \code{incr1(x)}
      <\ret = n | \ap{DC}(t, \pvar{x}, n + 1, m)>
  \\
  &\ATRIPLE \topOf |-
      \A n, m \in \Nat.
      <\emp | \ap{DC}(t, \pvar{x}, n, m)>
      \code{incr2(x)}
      <\ret = m | \ap{DC}(t, \pvar{x}, n, m + 1)>
\end{align*}

It is important to note here that we are making explicit the parametrisation
of the layers in the region identifiers~$s$,
because we will need to associate
different layers with the two instances of the lock.
As we will see later,
we will have two region identifiers~$s_1$ and~$s_2$, one per lock,
with associated layers $\topOf[s_1], \botOf[s_1], \topOf[s_2], \botOf[s_2]$.
The lock specifications themselves only require
$\topOf[s_1] \laygt \botOf[s_1]$
and
$\topOf[s_2] \laygt \botOf[s_2]$
but we will additionally impose, for this client proof,
$\botOf[s_1] \laygt \topOf[s_2]$.
This represents the fact that, in this client,
the release of lock~1 will depend on the acquisition of lock~2.

\paragraph{Shared Regions}
Like for the single counter example, we need two nested regions,
one to prove the atomicity of the operation ($\rt[dcnt]$)
and an inner one to prove termination ($\rt[ldcnt]$).
They differ in that $\rt[dcnt]$ only records the abstract states of the counters, while $\rt[ldcnt]$ includes the abstract states of the locks.
Formally:
$\region{dcnt}{\rid_1}((\rid_0, t_0), x, n, m)$ and
$\region{ldcnt}{\rid_0}(t_0, x, l_1, l_2, n, m)$ where
$\rid_0, \rid_1 \in \RId$,
$x \in \Addr$,
$l_1, l_2 \in \set{0,1}$ and
$n, m \in \Nat$,
and $t_0$ is a tuple $(la_1,la_2,s_1,s_2)$
with $la_1,la_2 \in \Addr$ and $s_1,s_2 \in \RId$.
Here $(\rid_0, t_0), x$, and $ t_0, x $ are the fixed parameters of the
two regions respectively.
The double blocking counter resource is abstractly represented by the predicate
$
  \ap{DC}((\rid_1, t_1), x, n, m) \is
    \region{dcnt}{\rid_1}(t_1, x, n, m) * \guardA{\guard{e}}{\rid_1}
$.

\paragraph{Guards and Obligations}
We introduce the guard constructors
$\guard{b}_i$,
$\guard{c}_i$, and
$\guard{w}_i$, for $i\in\set{1,2}$,
for bookkeeping of the value of the counters.
We need this ghost state because in \code{incrBoth} there is an intermediate state where one counter has been updated but the other hasn't;
we cannot update the abstract state in two steps because we are proving atomicity of the operation, so we need to update both counter values
in the abstract state in one go.
We record the intermediate concrete state in these guards so the information is there locally without affecting the shared abstract state prematurely.
The guard composition satisfies the axioms
\begin{align*}
  \guard{b}_1 &= \guard{c}_1(n, n') \guardOp \guard{w}_1(n, n') &
  \guard{b}_2 &= \guard{c}_2(n, n') \guardOp \guard{w}_2(n, n')
\end{align*}
Here $\guard{c}_i(n,n')$
  tracks the reference value (left in the region interpretation)
  for the \pre i-th counter's abstract ($n$) and concrete ($n'$) value and
$\guard{w}_i$ is a local ``witness'' for the same information
  about the \pre i-th counter,
  which can only be obtained when locking the \pre i-th lock
  (otherwise it would not be stable information).
This is enforced by the interpretation given later.

We associate two atom obligations $\obl{k}_1$ and $\obl{k}_2$
with the region type $\rt[ldcnt]$,
encoding ownership of the double counter's locks respectively, 
as well as the obligation to unlock them.

As anticipated, we choose the layers of the lock specifications
in a way that represents the dependency between the two locks.
We have a (double-counter-local) top ($\topOf$) and a bottom ($\botOf$) layer,
and intermediate layers for the locks:\footnote{The proof works with $ \topOf[s_2] = \botOf[s_1] $ too,
  but the ordered version better emphasizes the dependency between the locks.
}
\[
  \botOf =
  \botOf[s_2] = \lay(\obl{k}_2) \laylt
  \topOf[s_2] \laylt
  \botOf[s_1] =
  \lay(\obl{k}_1) \laylt
  \topOf[s_1] =
  \topOf
\]

\paragraph{Region Protocols}
The interference protocol of the region $\rt[dcnt]$ trivially allows for any change to the counter values:
\begin{align*}
  \guard{e} &: ((n, m), \oblZero) \interfTo ((n', m'), \oblZero)
\end{align*}
The interference protocol of the region $\rt[ldcnt]$ encodes the constraint
that we can update a counter only by holding the corresponding lock:
\begin{align*}
  \guard{e} &: ((0, l, n, m), \oblZero) \interfTo ((1, l, n, m), \obl{k}_1) &
  \guard{e} &: ((l, 0, n, m), \oblZero) \interfTo ((l, 1, n, m), \obl{k}_2) \\
  \guard{e} &: ((1, l, n, m), \obl{k}_1) \interfTo ((0, l, n, m), \oblZero) &
  \guard{e} &: ((l, 1, n, m), \obl{k}_2) \interfTo ((l, 0, n, m), \oblZero) \\
  \guard{e} &: ((1, l, n, m), \obl{k}_1) \interfTo ((1, l, n', m), \obl{k}_1) &
  \guard{e} &: ((l, 1, n, m), \obl{k}_2) \interfTo ((l, 1, n, m'), \obl{k}_2) 
\end{align*}\paragraph{Region Intepretations}
The interpretation of $\rt[dcnt]$ formalises the fact that the outer region
simply hides the state of the locks for the atomicity argument,
while the actual internal protocol of the module is encoded in the
interpretation of the inner region $\rt[ldcnt]$:
\begin{align*}
  \rInt(\region{dcnt}{\rid_1}((\rid_0,t_0), x, n, m)) & \is
    \exists l_1, l_2 \in \set{0,1} \st
      \region{ldcnt}{\rid_0}(t_0,x, l_1, l_2, n, m) *
      \guardA{\guard{e}}{\rid_0} * {} \\ & \quad
      l_1 = 1 \dotimplies \envObl{\obl{k}_1}{\rid_0} *
      l_2 = 1 \dotimplies \envObl{\obl{k}_2}{\rid_0}
\\
  \rInt(\region{ldcnt}{\rid_0}((la_1, la_2, s_1, s_2), x, l_1, l_2, n, m)) & \is
    \exists n', m' \in \Nat \st
    \\&\qquad
      x \mapsto la_1, la_2, n', m' *
      \ap{L}(s_1, la_1, l_1) *
      \ap{L}(s_2, la_2, l_2) 
    \\&\quad*
      \left(
      \begin{array}{@{\;}r@{\;}l}
        &( l_1 = 0 \land \oblA{\obl{k}_1}{\rid_0}
        * \guardA{\guard{b}_1}{\rid_0} \land n = n' )
        \\[3pt]\lor&
        ( l_1 = 1 \land \guardA{\guard{c}_1(n, n')}{\rid_0})
      \end{array}
      \right)
    \\&\quad*
      \left(
      \begin{array}{@{\;}r@{\;}l}
        &( l2 = 0 \land \oblA{\obl{k}_2}{\rid_0}
        * \guardA{\guard{b}_2}{\rid_0} \land m = m')
        \\[3pt]\lor&
        ( l2 = 1 \land \guardA{\guard{c}_2(m, m')}{\rid_0})
      \end{array}
      \right)
\end{align*}

\paragraph{Proof of \code{incrBoth}}
The proof outline of \code{incrBoth} is reproduced in
\cref{fig:dbcounter-dincr-outline}.
Most of the proof is routine;
the derivation for the acquisition of the first lock
follows closely the pattern
we already explained in \cref{sec:rules,app:ex-blocking-counter}.
We show the proof of the acquisition of the second lock in more detail,
to show the interplay between the layers.
At that point we are continuously holding the obligation of the first lock,
with layer greater than~$\topOf[s_2]$, so apply \ref{rule:layer-weak}
to lower the layer to $\topOf[s_2]$
enabling the application of \ref{rule:frame} to frame $
\notdone{\rid_1} *
\oblA{\obl{k}_1}{\rid_0} *
\guardA{\guard{w}_1(n,n)}{\rid_0}
$.
The obligation $\obl{k}_2$ has layer lower than $\topOf[s_2]$ so we are allowed
to invoke it to discharge the environment liveness condition
of the \ref{rule:liveness-check} application,
in a way that is analogous to the derivations
of the distinguishing client and \cref{app:ex-blocking-counter}.

\begin{mathfig}[p]
  \begin{proofoutline}
  \TITLE{Proof of \code{incrBoth(x)}:}
  \CTXT \topOf; \emptyset |-
  \A n, m \in \Nat. \\
  \PREC<\emp | \ap{DC}(t, \pvar{x}, n, m)> \\
  \begin{proofjump}[rule:consequence,"Sub ${t=(\rid_1,t_1), t_1=(r_0, t_0), t_0=(la_1, la_2, s_1, s_2)}$"]
    \PREC<\region{dcnt}{\rid_1}(t_1, \pvar{x}, n, m) * \guardA{\guard{e}}{\rid_1}> \\
    \begin{proofjump}[rule:make-atomic]
      \CTXT \topOf; \actxt \is \map{\rid_1 -> (\Nat^2, {\botOf}, \Nat^2, \set{((n, m), (n+1,m+1))| n, m \in \Nat})} |- \\
      \PREC{\exists n, m \st \region{dcnt}{\rid_1}(t_1, \pvar{x}, n, m) * \notdone{\rid_1}} \\
      \CODE{l1 := [x];} \\
\CODE{l2 := [x + 1];} \\
      \CODE{//\ }$\color{codecomment}
        t_1' \is (r_0,t_0'), \;
        t_0' \is (\pvar{l1}, \pvar{l2}, s_1, s_2) $\\
      \ASSR{\exists n, m \st \region{dcnt}{\rid_1}(t_1', \pvar{x}, n, m) *
        \notdone{\rid_1} * {} \\
        \exists l_1, l_2 \st
        \region{ldcnt}{\rid_0}(t_0', \pvar{x}, l_1, l_2, \_, \_) *
        l_1 = 1 \dotimplies \envObl{\obl{k}_1}{\rid_0} *
        l_2 = 1 \dotimplies \envObl{\obl{k}_2}{\rid_0}} \\
      \CODE{lock(l1);} \\
      \ASSR{\exists n, m \st \region{dcnt}{\rid_1}(t_1', \pvar{x}, n, m) *
        \notdone{\rid_1} *
        \oblA{\obl{k}_1}{\rid_0} *
        \guardA{\guard{w}_1(n,n)}{\rid_0} * {} \\
        \exists l_2 \st
        \region{ldcnt}{\rid_0}(t_0', \pvar{x}, \_, l_2, \_, \_) *
        l_2 = 1 \dotimplies \envObl{\obl{k}_2}{\rid_0}} \\
      \begin{proofjump}[rule:layer-weak,rule:frame,rule:atomicity-weak,rule:atomic-exists-elim]
        \CTXT \topOf[s_2]; \actxt |- \\
        \A n, m \in \Nat. \\
        \ASSR<\exists l_2 \st
        \region{ldcnt}{\rid_0}(t_0', \pvar{x}, \_, l_2, \_, \_) *
        l_2 = 1 \dotimplies \envObl{\obl{k}_2}{\rid_0}
        | \region{dcnt}{\rid_1}(t_1', \pvar{x}, n, m)> \\
        \begin{proofjump}[rule:lift-atomic,rule:atomic-exists-elim,rule:frame]
          \A l_1, l_2 \in \set{0,1}. \\
          \ASSR<\exists l_2 \st
          \region{ldcnt}{\rid_0}(t_0', \pvar{x}, \_, l_2, \_, \_) *
          l_2 = 1 \dotimplies \envObl{\obl{k}_2}{\rid_0}
          |
          \region{ldcnt}{\rid_0}(t_0', \pvar{x}, l_1, l_2, n, m) *
          \guardA{\guard{e}}{\rid_0}
          > \\
          \begin{proofjump}[rule:liveness-check]
            \CTXT \topOf[s_2]; \actxt |- \\
            \A n, m \in \Nat, l_1 \in \set{0,1}, l_2 \in \set{0,1} \eventually[{\botOf[s_2]}] \set{0}. \\
            \ASSR<
              \region{ldcnt}{\rid_0}(t_0', \pvar{x}, l_1, l_2, n, m) *
              \guardA{\guard{e}}{\rid_0} *
              l_2 = 1 \dotimplies \envObl{\obl{k}_2}{\rid_0}
            > \\
            \begin{proofjump*}[rule:lift-atomic,rule:frame]
              \label{step:dc-lock2}
              \CTXT \topOf[s_2]; \actxt |- \\
              \A l_2 \in \set{0,1} \eventually[{\botOf[s_2]}] \set{0}. \\
              \PREC<\ap{L}(s_2, \pvar{l2}, l_2)> \\
              \CODE{lock(l2);} \\
              \POST<\ap{L}(s_2, \pvar{l2}, 1) \land l_2 = 0>
            \end{proofjump*} \\
            \ASSR< \oblA{\obl{k}_2}{\rid_0} |
              \region{ldcnt}{\rid_0}(t_0', \pvar{x}, l_1, 1, n, m) *
              \guardA{\guard{e}}{\rid_0} *
              \guardA{\guard{w}_2(m,m)}{\rid_0}
            >
          \end{proofjump}\\
          \ASSR< \oblA{\obl{k}_2}{\rid_0} |
            \region{ldcnt}{\rid_0}(t_0', \pvar{x}, l_1, 1, n, m) *
            \guardA{\guard{e}}{\rid_0} *
            \guardA{\guard{w}_2(m,m)}{\rid_0}
          >
        \end{proofjump} \\
        \ASSR<\oblA{\obl{k}_2}{\rid_0} |
          \region{dcnt}{\rid_1}(t_1', \pvar{x}, n, m) *
          \guardA{\guard{w}_2(m,m)}{\rid_0}
        >
      \end{proofjump} \\
      \ASSR{\exists n, m \st \region{dcnt}{\rid_1}(t_1', \pvar{x}, n, m) *
        \notdone{\rid_1} * {} \\
        \oblA{\obl{k}_1}{\rid_0} *
        \guardA{\guard{w}_1(n,n)}{\rid_0} *
        \oblA{\obl{k}_2}{\rid_0} *
        \guardA{\guard{w}_2(m,m)}{\rid_0}} \\
      \CODE{v := x.cnt1;} \CODE{x.cnt1 := v + 1;} \CODE{v := x.cnt2;} \CODE{x.cnt2 := v + 1;} \\
      \ASSR{\exists n, m \st \region{dcnt}{\rid_1}(t_1', \pvar{x}, n, m) *
        \notdone{\rid_1} * {} \\
        \oblA{\obl{k}_1}{\rid_0} *
        \guardA{\guard{w}_1(n,n + 1)}{\rid_0} *
        \oblA{\obl{k}_2}{\rid_0} *
        \guardA{\guard{w}_2(m,m + 1)}{\rid_0}} \\
      \CODE{unlock(l2);} \\
      \ASSR{\exists n, m \st \region{dcnt}{\rid_1}(t_1', \pvar{x}, n, m) *
        \done{\rid_1}{((n, m), (n + 1, m + 1))} * {} \\
        \oblA{\obl{k}_1}{\rid_0} *
        \guardA{\guard{w}_1(n + 1,n + 1)}{\rid_0}} \\
      \CODE{unlock(l1);} \\
      \POST{\exists n, m \st \done{\rid_1}{\DONE((n, m), (n+1, m+1))}}
    \end{proofjump} \\
    \POST<\region{dcnt}{\rid_1}(t_1, \pvar{x}, n + 1, m + 1) * \guardA{\guard{e}}{\rid_1}>
  \end{proofjump} \\
  \POST<\ap{DC}(t, \pvar{x}, n + 1, m + 1)>
\end{proofoutline}
   \caption{Double blocking counter: proof of \code{incrBoth}.\\
\cref{step:dc-lock2} is \explainproofjump{step:dc-lock2}.
  }
  \label{fig:dbcounter-dincr-outline}
\end{mathfig}

\paragraph{A comparison with LiLi}
\label{sec:compare-lili-proof}
As we have seen in \cref{sec:overview} (Innovation~3),
the call of a CLH lock in LiLi involves two distinct atomic actions:
requesting the lock, and acquiring it.
Requesting a lock \p{x} is a non-blocking action as it just enqueues the current thread in the (concrete) queue for \p{x},
but the acquisition is represented with a (primitive) blocking operation
that waits until the current thread is at the head of the lock's queue,
and the lock is unlocked.
When proving the call to \code{lock(l1)} in \code{incrBoth},
the LiLi proof would require arguing about termination of acquisition
by appealing to progress of the threads in the environment.

To do so, in the LiLi methodology,
one has to identify the threads in the environment that will be able to make progress, and show how this progress is bringing us closer to acquiring lock~\p{l1}.
Consider the case when there are~$n_1>0$ threads ahead of us
in the queue for~\p{l1}.
Assume thread~$t_1$ is the head of the queue for~\p{l1}.
It can make progress in three ways:
\begin{itemize}
  \item if~\p{l1} is unlocked it can acquire it;
  \item if~\p{l1} is locked it can unlock it;
  \item if~\p{l1} is locked it can request~\p{l2}.
\end{itemize}
How do these actions represent progress for us?
The first case makes progress by moving to the second or third case.
The second case removes~$t_1$ from the queue of~\p{l1}
bringing us closer to the front of the queue.
The third case complicates matters:
in this case~$t_1$ is enqueued in the queue of~\p{l2}
with a non-deterministic number~$n_2$ of threads ahead of it.
The thread~$t_1$ is now blocked, and to track progress we need to consider
the head of the queue for~\p{l2}, which can only make progress by
acquiring the lock when unlocked, or releasing the lock when locked.
What progress had been made towards us acquiring~\p{l1}?
The measure of progress needs to consider the contents of the queues for both threads: the measure before~$t_1$ requests~\p{l2} needs to be $(n_1, \omega)$
(ordered lexicographically)
so that we can lower the measure to $(n_1,n_2)$ once~$t_1$ joined the queue of~\p{l2}.
Whenever~$t_1$ reaches, finally, the head of the queue of~\p{l2},
the measure of progress would become $(n_1,0)$,
and the only option for~$t_1$ is to release~\p{l2}.
Now thread~$t_1$ is back to the three options as above.
This is a problem because nothing would prevent~$t_1$ from requesting~\p{l2}
again. This could repeat ad libitum, leaving us to starve on~\p{l1}.
To rule this out, the argument needs to place a bound $b$ on the number of times
\p{l2} can be acquired while holding~\p{l1}; in our example this bound can be 1.
By mixing this bound in the measure $(n_1,b,n_2)$, the action of~$t_1$ releasing~\p{l2} brings real progress by taking $b$ from 1 to 0.
When that happens, the only option for~$t_1$ is to release the lock.
This brings down~$n_1$, the number of threads ahead of us;
at the same time we want to reset~$n_2$ to~$\omega$ and~$b$ to~$1$
to allow the new head of the queue of~\p{l1} to request~\p{l2}.

This substantiates our claim that LiLi's rely/guarantee reasoning
lacks in scalability; the key reason for this is that the progress argument
is forced to walk through all the possible ways the environment could be
implementing progress. This in turn requires to expose the internal state of both locks (their queues) to be used in the client's proof. In other words, the abstraction of the environment is not abstract enough.
By comparison, \tadalive's atomic specifications allow for the termination
of the~\p{lock} calls in the double blocking counter to be reasoned about individually,
without direct reference to the termination of the other, nor to internal state, using layers to prevent circular reasoning.
The appeal to obligation $\obl{k}_1$ being live to justify why the call to~\p{lock(l1)} terminates, abstracts away \emph{how} the environment may be keeping it live. The layers capture the essential information: the only thing that is important is that to keep $\obl{k}_2$ live, the environment does not assume $\obl{k}_1$ live.
 
\subsection{Lock-Coupling Set}
\label{sec:lock-coupling-set}

To conclude this series of examples,
we present a challenging fine-grained lock-based
implementation of a linearizable finite set.
A lock-coupling set implements a set by
maintaining an ordered linked list of the elements
with fair locks (here CLH locks) guarding each individual element.
The module exposes an \p{add} and \p{remove}
operation to add and remove elements from the
abstract set it represents.
To make modifications to the nodes of the linked list,
the operations traverse the list using a lock-coupling pattern.
In this pattern, all threads start the traversal at the head of the list.
To be at position~$i$ a thread must acquire the lock at that position.
To move to position~$i+1$ the thread would first acquire the lock at~$i+1$,
and then release the lock at position~$i$.
This way, the threads cannot overtake each other, and owning a lock
allows the owner to safely perform modifications at that position.
We sketch here the main points of interest of our proof,
the full details can be found in~\cref{appendix:lock-coupling-set}.

This example is challenging because it makes use of a dynamically changing
list of locks with non-trivial liveness dependencies between them.
In particular, the termination of the acquisition of each lock
depends on the usage of the locks further down the list.
Although these dependencies are acyclic,
they change over time as the list grows or shrinks.
At first sight, it is unclear how the seemingly static layer structure 
of \tadalive, and the fixed layers decorating the specifications
of lock operations can cope with this complexity,
without breaking modularity.

The \tadalive\ proof of this example relies on solving two key challenges:
\begin{itemize}
  \item How can we modularly coordinate the choice of layers needed for the proof of a module and the ones needed for the proofs of its clients?
  \item How can we dynamically reassign layers to resources?
\end{itemize}

We solve the first challenge by introducing a style of specification that allows
the client to ``remap'' the layers of the implementation into a larger layer structure, and the implementation to prove correctness with respect to a ``local'' layer structure which is opaque to the client.
The key observation is that a \tadalive\ derivation's validity is preserved by transformations of the layer structure that preserve the strict order between layers.
This leads to the following proof style.
Given two partial orders
  $(\Layer_1, \layleq_1, \layTop_1, \layBot_1)$
and
  $(\Layer_2, \layleq_2, \layTop_2, \layBot_2)$,
a function $ \laymap \from \Layer_1 \to \Layer_2  $
is \emph{strictly monotone} if
$
  \forall m,n\in \Layer_1 \st
    m \laylt_1 n \implies \laymap(m) \laylt_2 \laymap(n)
$.
A \emph{layer map} ${\laymap \from \Layer_1 \layto \Layer_2}$ is a strictly monotone function between the two partial orders.
Using this notion, we generalise the client-facing
CLH~lock specifications as follows:
\begin{align*}
  \exists (\LayerOf{clh}&, \layleqOf{clh}, \layTopOf{clh}, \layBotOf{clh}) \st
    \forall \laymap \from \LayerOf{clh} \layto \Layer \st
    \\
      &\ATRIPLE \laymap(\layTopOf{clh}) |-
        \A l \in \set{0, 1} \eventually[{\laymap(\layBotOf{clh})}] \set{0}.
          <\ap{L}_{\laymap}(s, \pvar{x}, l)>
            \code{lock(x)}
          <\ap{L}_{\laymap}(s, \pvar{x}, 1) \land l = 0>
      \\
      &\ATRIPLE \laymap(\layBotOf{clh}) |-
        <\ap{L}_{\laymap}(s, \pvar{x}, 1)>
          \code{unlock(x)}
        <\ap{L}_{\laymap}(s, \pvar{x}, 0)>
\end{align*}

From the perspective of the implementation,
a proof of correctness would start by defining the partial order of the
``internal'' layers.
In the case of CLH, as we have seen in \cref{examples:clh_lock},
we would let $\LayerOf{clh} = \Nat \union \set{\topOf,\botOf}$
with $\layTopOf{clh} = \topOf$
and $\layBotOf{clh} = \botOf$.
Then, to be able to prove the triples with the layers remapped by the arbitrary
layer map~$\laymap$ we would reproduce
the derivation presented in \cref{examples:clh_lock}
but with~$\laymap$ applied to every occurrence of an internal layer.
For example, the $\rt[lclh]$ region type
would also be parametrised by the layer map,~$\region{lclh}{\rid'}(\laymap,x, h, l, o, t)$,
so that its associated obligations and their layers can depend on~$\laymap$,
e.g.~$\lay(\obl{p}_{\laymap}(t)) = \laymap(t)$.
Since the map preserves the strict order of~$\LayerOf{clh}$,
the proof goes through exactly as in the un-parametrised case.

From the perspective of the client,
to use these specifications one would first obtain the arbitrary~$\LayerOf{clh}$
from the existential quantification.
Then the client would be able to choose a layer map from $\LayerOf{clh}$ to $\Layer$. Here $\Layer$ could be the global layer structure, in the case of a closed proof, or itself being the internal layer structure of a module using the lock module internally.
Note that the client needs to define~$\laymap$ parametrically on~$\LayerOf{clh}$,
since it has no control on the inner structure of $\LayerOf{clh}$.
For example, in the case of a client with a static list of locks,
one would use as~$\Layer$ the lexicographically ordered set of pairs from
$(\Nat \union \set{\top,\bot}) \times \LayerOf{clh}$
where the first component corresponds to the position of the lock from the end of the list.
Then, for the lock at position $i\in\Nat$, the client would instantiate
the specifications choosing $\laymap[i](k) \is (i,k)$.

The second challenge is also solved by a slight generalisation of
the lock specifications,
following a proof pattern that, if adopted, always increases the generality of module specifications: adding some fractional permissions to control the update of ghost parameters of the resource.
The idea is that the layer map is ghost state, and as such we should be able to update it using a viewshift.
To do this without invalidating the other thread's information about the region we are updating, we add standard fractional permissions to the lock specifications.
We introduce the abstract predicate $\ap{P}(s,\pi)$ representing ownership of
the fraction $0 \leq \pi \leq 1 $ of permissions for a lock at abstract location~$s$.
To split permissions, the predicate satisfies,
for $0 \leq \pi_1+\pi_2,\pi_1,\pi_2 \leq 1$,
$ \ap{P}(s, \pi_1+\pi_2) \iff \ap{P}(s, \pi_1) * \ap{P}(s, \pi_2) $.
The generalised lock specifications would then be:
\begin{align*}
  \exists (\LayerOf{clh}&, \layleqOf{clh}, \layTopOf{clh}, \layBotOf{clh}) \st
    \forall \laymap \from \LayerOf{clh} \layto \Layer \st
    \\
  & \TRIPLE \laymap(\layBotOf{clh}) |-
      {\emp}
        {\p{makeLock()}}
      {\exists s \st \ap{L}_{\laymap}(s, \pvar{x}, 0) * \ap{P}(s, 1)}
  \\
  \forall \pi {>} 0\st
  &
  \ATRIPLE \laymap(\layTopOf{clh}) |-
    \A l \in \set{0, 1} \eventually[{\laymap(\layBotOf{clh})}] \set{0}.
      <{\ap{P}(s, \pi)} | \ap{L}_{\laymap}(s, \pvar{x}, l)>
        \code{lock(x)}
      <\ap{P}(s, \pi) | \ap{L}_{\laymap}(s, \pvar{x}, 1) \land l = 0>
  \\
  &\ATRIPLE \laymap(\layBotOf{clh}) |-
    <\ap{L}_{\laymap}(s, \pvar{x}, 1)>
      \code{unlock(x)}
    <\ap{L}_{\laymap}(s, \pvar{x}, 0)>
\end{align*}
When creating a new lock, one gets a local resource representing an unlocked lock and full permissions.
Typically then permissions are distributed to the threads by splitting the full permission into smaller fractions.
A non-trivial fraction of permission is now needed to perform the \p{lock}
operation.
We can then provide the viewshift
$
  \ap{L}_{\laymap}(s, l) * \ap{P}(s,1)
  \vshift 
  \ap{L}_{\laymap'}(s, l) * \ap{P}(s,1)
$
which allows to change the layer map without invalidating
the knowledge about it in any other thread:
if we own~$\ap{P}(s,1)$ then no other thread can race on the lock.
Adapting the proof of CLH to support permissions and the viewshift above
follows standard (safety) proof patterns which we explain in \cref{appendix:lock-coupling-set}.

Let us briefly explain how we can use this viewshift in the lock-coupling set
example.
Conceptually, we want to organise the layers of the lock-coupling set module
as for a static list of locks: they go in decreasing order from the head of the list to the tail.
A thread holding a lock at position~$i$ will be able to eventually acquire the lock at position~$i+1$ because the release of such lock is associated with an obligation of strictly lower layer than the one associated with the lock at~$i$.
Each operation of the module inserts at most one element to the set per traversal of the list.
We therefore arrange the proof invariants so that each thread traversing the list will shift up the layer of the lock at the thread's current position by one.
This way, when the thread finally finds the position where the new element has to be inserted, there is already a gap of 1 between the layers associated with the positions being altered by the thread.
The layer sitting at the gap will be the one we associate with the lock of the new element.
The layer-map-altering viewshift we explained above is used at each step of the
traversal, to shift up the layer of the current lock.
This is possible without breaking the information owned by other threads
because when the current thread holds the lock at position~$i$
and the lock at~$i+1$ finally becomes available,
the current thread is the only thread with access to the reference (and the associated resources) of the lock at~$i+1$.
Formally, this means that when we obtain the lock at~$i+1$
we are able to obtain full permissions for it until we unlock the lock at~$i$.
With the full permissions we can apply the viewshift and effectively shift up
the layers associated with the lock at~$i+1$.

The only exception to this scheme is the lock at the head of the queue:
this is the only lock which does not need a remapping of layers as its associated layer can be $(\top, \layBotOf{clh})$ which is always bigger
than any layer ever associated with the locks at the other positions.

It is worth noting that the LiLi proof of the same example does not use the specifications of the fair locks modularly, but instead inlines the code of the lock operations, allowing for a non-modular handling of the internal state.

Interestingly, the same lock-coupling set specifications can be implemented
by using spin locks instead of CLH locks,
for each element except the one at the head.
In fact, the locks in the tail of the list do not experience any impedance.
At first sight, it seems impossible to represent this fact using our specifications for spin lock: the \code{lock} operation needs to consume non-trivial budget, but there is no bound on the number of calls to it.
The \tadalive\ way of expressing the absence of impedance in this example
uses a viewshift similar to the one we introduced above,
which allows us to reset the budget (and the layer map) when we own full permissions.
The proof in LiLi of this variant of the lock-coupling set again inlines
the lock code, with the effect of being able to redefine which internal steps are susceptible of impedance and which do not, breaking modularity.

\subsection{Limitations}
\label{sec:limitations}

\paragraph{Non-local linearization points.}
As  with other total program logics, \tadalive\ 
does not support helping/speculation.
Such patterns are challenging
for the identification of the linearization point,
which is entirely a safety property.
Extensions to TaDA that could support such patterns
are discussed in~\cite{PedroThesis}.
Such extensions are orthogonal to the termination argument.
We therefore choose, in line with the related literature,
to explore termination in a simpler logic.

\paragraph{Non-structural thread creation.}
\tadalive\ currently supports only structural parallel composition.
We believe the support of non-structural fork/join
would not require substantial new ideas.
For comparison, LiLi does not support parallel nor fork/join.

\paragraph{Scheduling non-determinism.}
A more interesting limitation comes from our approach to specifying impedance.
For non-blocking programs, the ordinal-based approach is complete.
It is not complete for blocking programs.
Consider $\cmd_2 \is (\cmd_1\parallel \code{[done]:=true})$
where $\cmd_1$ is the distinguishing client 
\emph{with a spin lock}.
Scheduler fairness guarantees the right-hand thread of $\cmd_2$ will be eventually executed.
The specification of spin lock, however,
states that every call to \code{lock} needs to consume  budget,
forcing the client to provide an  upper bound
for the total number of calls  to initialise the budget.
Unfortunately, $\cmd_2$ will call \code{lock} an  arbitrary unbounded number
of times, determined only by the choices of the scheduler.
It is, thus, not possible to  provide the initial budget,
and \tadalive\ cannot prove that the program terminates.
The impedance on the lock is only relevant when the client is unblocked
(i.e.~\p{done} is true)
but the specifications do not allow for the distinction.
To accommodate this behaviour,
we could introduce $ \alpha(\obl{d}) $
to represent a prophecy of the number of steps it will take
to fulfil \emph{live} obligation~$\obl{d}$.
This would solve the problem for $\cmd_2$,
because $\alpha(\obl{d})+1$
(where $\obl{d}$ is fulfilled by setting \p{done} to true)
would be the required budget.
How to introduce this extension soundly is future work.
To the best of our knowledge,
none of the approaches in the literature can handle this example.

\paragraph{Loop body specifications.}
Consider a loop invariant asserting
the possession of  obligation $\obl{k}$.
We cannot distinguish, by only looking at the specification of the loop body,
the case where  $\obl{k}$ is continuously held
throughout the execution of the body,
from the case where $\obl{k}$ is fulfilled and then reacquired before
the end of an iteration.
The current \ref{rule:while} rule conservatively rules out the use of
assumptions with layer higher than or equal to $\lay(\obl{k})$;
doing otherwise would be unsound in the case when $\obl{k}$ is held continuously.
A solution would be to introduce an assertion $\cond{live}(\obl{k})$,
certifying  that an obligation is fulfilled
at some point in a block of code.
It would allow the \ref{rule:while} rule
to only forbid layers which may depend on obligations
one holds in the loop invariant and
for which it was not possible to prove $\cond{live}(\obl{k})$.

\paragraph{More Expressive Layers.}
Advanced examples like the lock-coupling set of \cref{sec:lock-coupling-set}
need powerful parametric specifications in order to work around
the fact that the $\lay$ function is statically specified.
We are not aware of any example that cannot be proved using static layers
and critically requires more expressive layers.
Even for current proofs, however, being able to constrain layers
through assertions and allowing them to change as result of interference
would allow for more concise and intuitive proofs.
The $\lay$ function could in principle be encoded as ``regular'' ghost state
and the crucial relative order between layers be enforced through invariants.
It is however not clear how to ensure soundness if interference on layers is allowed.
We leave this exploration as future work.
 \section{Related Work}
\label{sec:relwork}

\paragraph{Primitive Blocking.}
There has been work on termination and deadlock-freedom
of concurrent programs with primitive blocking constructs.
Starting from the seminal work of~\cite{Kobayashi06},
the idea of tracking dependencies between blocking actions and ensuring
their acyclicity has been used to prove dead\-lock-freedom
of shared-memory concurrent programs
using primitive locks and (synchronous) channels~\cite{Leino10,BostromM15}.
Similar techniques have been used in~\cite{Jacobs18}
to prove global deadlock-freedom
(a \emph{safety} property requiring that at least some thread can take a step),
and~\cite{Jacobs18toplas} to prove termination.
This entire line of work assumes the invocation
of lock/channel \emph{primitives} as the only source of blocking.
As a consequence, this methodology provides no insight
on the issue of understanding abstract blocking patterns
arising from busy waiting and shared memory interference.
Moreover, the specifications for blocking built-ins
(hardcoded in the logic as ad hoc axioms)
impose a usage protocol in the client, instead of just capturing the abstract
effect of the operation: for instance, a call to \code{lock(x)} always
entails an obligation to unlock the lock, regardless of how the client
intends to use the lock.
This has had the side effect of requiring ad hoc extensions of the reasoning principles to increase the expressivity of this hard-coded protocol,
to allow, for example, for delegation of obligations~\cite{Hamin019}.
Our solution uniformly handles programs
that mix blocking primitives and ad-hoc synchronisation patterns,
and is not imposing any specific protocol on the client.

The notion of ``obligations'' found in~\cite{Leino10,BostromM15,Jacobs18,Jacobs18toplas}
is only superficially related to our obligations.
First, obligations found in the literature represent primitive
blocking events (like the acquisition of a lock).
They are also typically equipped with a structure to represent causal dependencies between these events, to detect deadlocks.
Our layered obligations are associated with arbitrary \emph{abstract}
state changes, removing the need for ad-hoc treatment of primitives,
and supporting abstraction and abstract atomicity.
Moreover, our layers do not represent causal dependencies between events,
but rather dependencies between liveness assumptions in a termination argument.
This reflects in our specifications, e.g.~a lock operation does not return an obligation in its post-condition.
Whether there is a need for an obligation linked to that lock is entirely
dependent on how the client will decide to use the lock.
Nevertheless, the specification precisely captures the termination guarantees
of lock operations.
Finally, obligations in the literature have a purely safety semantics,
from which one can only derive safety properties as non-blocking or deadlock-freedom.
Our obligations explain how to express proper liveness invariants,
how to blend them with the layers, and how to use them for proving termination.

\paragraph{Temporal Logics.}
There is substantial literature on using temporal logics to prove liveness
and termination of concurrent programs, e.g.~\cite{OwickiL82}.
By working directly at the level of traces with  liveness properties stated as
temporal logic formulas, this approach is very general.
It does however provide less guidance on how to prove programs,
and does not tackle the problem of abstract interfaces and proof reuse.
Our adoption of concurrent separation logic as the basis of our
reasoning achieves superior compositionality of the
reasoning including proof reuse.

\paragraph{History-based methods.}
The CertiKOS project~\cite{Shao17,Shao18} developed mechanised techniques for
the specification and verification of fine-grained low-level code
with explicit support for abstract atomicity and progress verification.
The approach is based on \emph{histories}:
the abstract state is a log of the abstract events of a trace; and
the specification of an atomic operation inserts exactly one event in the log.
Local reasoning is achieved by rely/guarantee through complex automata product constructions.
The framework is very expressive,

with the downside that specifications are more complex and difficult to read,
and verification requires manipulation of abstract traces/interleavings.
Our work is similar in aim and scope, but our strategy is different.
We try to specify/verify programs using the minimal machinery possible,
keeping the specifications as close to the developer's intuition as we can.
As a result, our specifications are more readable
(compare our fair-lock specification with the corresponding 30-line
specification from Fig.~7 in~\cite{Shao17}),
and our reasoning is simpler
(the layered obligation system leads to a more intuitive proof
compared to the proof of MCS locks in~\cite{Shao17}).\footnote{The proof is a variation of the one for CLH.}

\paragraph{Contextual refinement.}
Another approach to specify and prove progress
of concurrent systems is to prove refinement
between the implementation and simpler, abstract code acting as a specification~\cite{Tassarotti17,LiangF16,LiangF18}.
By making sure the refinement preserves progress properties,
one can represent the salient termination properties of the implementation
by the termination properties of the specification code.
The Iris implementation of this idea~\cite{Tassarotti17}
uses a non-contextual refinement,
which means that the refinement is proven
between the closed-world behaviour of implementation and specification code,
and does not necessarily carry over contexts.
This severely hinders proof reuse.
The only refinement-based work that is able to modularly verify
blocking code is the LiLi logic discussed below.

There has been work on extending linearizability,
characterised as a contextual refinement,
to support reasoning about progress properties,
e.g.~\cite{GotsmanY11}.
This work only supports non-blocking operations.
\citet{LiangHFS13} studies  the exact relationship between common
progress properties of fine-grained operations and contextual refinement.
The study of the contextual refinement induced by
our triple semantics is future work.

\paragraph{LiLi.}
The work closest to ours is LiLi~\cite{LiangF16,LiangF18}.
LiLi was the first concurrent separation logic to prove
progress specifications for linearizable concurrent objects
with internal blocking~\cite{LiangF16},
and it was then extended to handle external blocking~\cite{LiangF18}.
Although we share most of our goals with LiLi,
our approach differs in two important ways.

First, LiLi's goal is to prove a progress-preserving contextual refinement
between the implementation of a module and its specification.
Termination properties of implementation code are not represented directly,
but in terms of the termination properties of the specification code.
Although proof of clients of the module
have to be done outside of the LiLi logic
(there is no rule for parallel, nor for calling a module's operation)
such proofs would need to reprove the relevant termination properties of
the specification code so that the properties themselves become available in the proof.
Moreover, as we outlined in \cref{sec:overview} for CLH lock,
the specification code for blocking operations
may be non-atomic even in the case of linearizable operations.
Instead we aim at specifications that directly represent termination properties
as a logical statement that can be readily used in a client proof,
and in the proof of the implementation.
Our specification format obtains a crucial advantage:
it achieves abstraction and can represent atomicity for blocking operations,
enabling more scalable and reusable reasoning.

Second, LiLi's rely/guarantee incorporates a form of liveness invariants
through so-called \emph{definite actions}.
Definite actions require the identification of a logical global ``queue''
of threads where the thread at the front is always able to execute its action
and that action implies global progress.
This queue is maintained as shared auxiliary state manipulated through ghost code.
It is due to this global view that definite actions
can side-step the issue of circular reasoning.
Our layered subjective obligations push the idea much further,
obtaining sound liveness invariants that can be represented thread-locally
and without the need for ghost code, improving proof scalability.
The design choice of making both rely/guarantee and specification
represent blocking via liveness assumptions is the key to
making the blocking specifications directly usable in the proof system.
 \section{Conclusions and Future Work}
\label{sec:concl}
\label{sec:future}

We have introduced TaDA Live,
a concurrent separation logic for reasoning
compositionally about the termination
of fine-grained blocking concurrent programs, and proved a substantial
soundness result.
Our key contribution is our approach to abstract atomic blocking as
the reliance of termination on the liveness properties of the environment.
By wholly embracing this point of view,
we have designed a rely/guarantee principle that incorporates liveness invariants
using layered subjective obligations, a new form of local ghost
state, and have extended \tada's abstract atomic specifications to
provide total specifications for blocking programs using environment
liveness assumptions.  Through several case studies, we have illustrated
how our formalisation of abstract blocking
allows for the right level of abstraction in specification,
and strong thread-locality of the proofs.
The result is a verification system with scalable and reusable proofs.

The work presented in this paper opens a number of immediate
directions for future work
on concurrent separation logics.
A first direction is to extend \tadalive\ to prove general liveness properties
beyond termination.
A possible way to achieve this is to wrap a refinement calculus around
\tadalive's atomic specifications,
as was done in the safety case in TaDA Refine~\cite{tadarefine}.
Specifications would be able to sequentially compose atomic triples
and take fixpoints, thus being able to specify linear-time temporal properties
of infinite traces.
A second direction is to study general fork/join concurrency and provide
a generalisation of the liveness rely/guarantee necessary to accommodate
patterns typical of distributed/reactive systems,
where long-lived maintenance threads interact with an environment
to realise an operation's effect. 
A third direction is to 
transfer ideas from \tadalive\ to  the Iris framework~\cite{iris3},
to provide a Coq-mechanised environment for reasoning about the termination
of concurrent programs. More widely, we hope  that our emphasis on 
environment liveness invariants for proving termination will transfer 
to other forms of 
reasoning about blocking concurrent programs.

\begin{acks}
  We would like to thank
    Hongjin Liang,
    Xinyu Feng,
    Martin Bodin,
    Shale Xiong and
    Petar Maksimovic,
  for the helpful discussions and comments.
  We also thank the anonymous reviewers for their thorough critical reading of the paper and insightful feedback.
  This research was supported
by: 
     the EPSRC Programme Grant
    ``REMS: Rigorous Engineering for Mainstream Systems''
    (EP/K008528/1);
    by the European Union's Horizon 2020 research and innovation programme
    under the Marie Skłodowska-Curie project ``VeSPA'',
    grant agreement no.~795218;
    by a Department of Computing PhD Scholarship from Imperial;
by the UKRI Established Fellowship
    ``VeTSpec: Verified Trustworthy Software Specification''
    (EP/R034567/1);
and in the final stages by the ERC Consolidator Grant for the project ``RustBelt'',
    also funded under EU Horizon 2020, grant agreement no.~683289. 
\end{acks}

\label{paper-last-page}

\appendix
\makeatletter\def\@adddotafter#1{\llap{\normalfont\bfseries---~}#1\@addpunct{.}}\makeatother

We present here omitted definitions and details of proofs.
An extended version of this paper is also available at
\url{https://arxiv.org/abs/1901.05750}~\cite{arxiv}.

\compresspqsets

\section{Some Proofs Conventions}
\label{app:conventions}

\subsection{Specification abbreviations}

Here is a summary of all the abbreviations we use in writing specifications.
The full hybrid specification format is
\[
  \ATRIPLE m;\lvl;\actxt |=
    \A x \in X \eventually[k] X'.
    <\na{P}|\at{P}(x)>
    \cmd
    \E y.<\na{Q}(x,y)|\at{Q}(x,y)>
\]
The $\exists y$ quantification is a normal existential quantification
but its scope extends over both the Hoare and the atomic postconditions.
We omit it when $y$ does not occur in the triple.
\begin{align*}
  \PQ{x} &\is \PQ{x \in \Val}
  \\
  \PQ{x\in X} &\is \PQ{x \in X \eventually[\layBot] X}
  \\
  \PQ{x_1\in X_1 \eventually[k] X'_1,
      x_2\in X_2 \eventually[k] X'_2.}
      &\is
    \PQ{(x_1,x_2) \in (X_1 \times X_2)
          \eventually[k] (X'_1 \times X'_2).}
\end{align*}
An omitted pseudo-quantifier is to be understood
as the trivial pseudo-quantifier $\PQ{x \in \AVal \eventually[\layBot] \AVal}$,
for an unused $x$.

The triples
\begin{gather*}
  \TRIPLE m,\lvl,\actxt |- {P}\cmd{Q}
  \\
  \ATRIPLE m,\lvl,\actxt |- \A x\in X\eventually[k] X'.<P(x)>\cmd<Q(x)> 
  \label{specs:special-cases}
\end{gather*}
are abbreviated with
\begin{gather*}
  \ATRIPLE m;\lvl;\actxt |- <P|\emp> \cmd <Q|\emp>
  \\
  \forall \vec{v}_0\st
  \ATRIPLE m;\lvl;\actxt |-
    \A x \in X \eventually[k] X'.
    <\pvars{v}_0 \doteq \vec{v}_0 | P'(x)>
    \cmd
    \E \vec{v}_1.
    <\pvars{v}_0 \doteq \vec{v}_0 \land \pvars{v}_1 \doteq \vec{v}_1| Q'(x)>
\end{gather*}
respectively,
where
  $\pvars{v}_0 = \progvars(P(x))$,
  $\pvars{v}_1 = \progvars(Q(x)) \setminus \pvars{v}_0$,
$P'(x) = P(x)\subst{\pvars{v}_0->\vec{v}_0}$ and
$Q'(x) = Q(x)\subst{\pvars{v}_0->\vec{v}_0,\pvars{v}_1->\vec{v}_1}$
(for technical reasons the atomic pre/post-conditions in the general triples cannot contain program variables).
In other words, the program variables mentioned in the atomic pre/post-conditions refer to the value stored in them \emph{at the beginning}
of the execution of the command.
Most commonly, the program variables used this way are actually not
modified by the command.

\subsection{Guard and Obligation Algebras}

Defining a guard algebra can be tedious.
In program proofs, we will define guard algebras by
generating them from some \emph{guard constructors}
and some axioms defining the guard operation.

Consider two common guard patterns in \tadalive:
the use of an \emph{exclusive guard} and the
\emph{$\guard{u}$, $\guard{l}$, $\guard{k}$ pattern}
used to represent possession of a lock in ghost state.

An exclusive guard, $\gEx$, is very commonly used
to express some exclusive permission on some shared resource,
which cannot be composed with itself:
i.e.~$\guardUndef{\gEx\guardOp\gEx}$.
Local ownership of $\gEx$ is exclusive in that
no other thread can at the same time assert ownership of $\gEx$.
A ubiquitous use of this guard is in representing the resource offered by a module.

The $\guard{u}$, $\guard{l}$, $\guard{k}$ pattern is commonly used to
represent ownership of a lock guarding a resource. The thread records its
ownership of a lock by holding the ghost state $\guard{k}$,
which cannot be composed with the guard $\guard{u}$,
recording the lock is unlocked: $\guardUndef{\guard{u}\guardOp\guard{k}}$
The region holds the associated guard $\guard{l}$, which can
be recombined with the guard $\guard{k}$ once the thread releases the lock
to form the guard $\guard{u}$:
$\guard{u} = \guard{l} \guardOp \guard{k}$. 

We explain the construction of a guard or obligations algebra
from these axioms by introducing some unsurprising auxiliary definitions.

Given a set $X$,
the set $\Multiset(X) \is X \to \Nat$
is the set of \emph{multisets} over $X$;
$\msetEmpty$ is the empty multiset
(i.e.~the function mapping every element to 0) and
${\msetPlus} \from \Multiset(X) \times \Multiset(X) \to \Multiset(X)$
is multiset union (i.e.~the pointwise lifting of $+$).
The expression $\mset{x_1,\dots,x_n}$
denotes the multiset containing the elements
${x_1,\dots,x_n}$.
Given a set $X$, the \emph{free commutative monoid} over $X$
is the monoid $(\Multiset(X),\msetPlus,\msetEmpty)$.
Given a commutative monoid $(X,\pcmOp,\pcmZero)$ and
a congruence relation ${\cong} \subseteq X \times X $,
the \emph{quotient} $(X/_{\cong}, \pcmOp/_{\cong}, [\pcmZero]_{\cong})$
is a commutative monoid.
Given a commutative monoid $(X,\pcmOp,\pcmZero)$ and
a set $U \subseteq X$ with $\pcmZero \not\in U$,
the \emph{PCM over $X$ induced by $U$} is
$(X|_U, \pcmOp_{U}, \pcmZero)$
where
\[ {X|_U \is \set{x\in X| \forall u\in U\st\nexists y\in X\st x = u\pcmOp y}} \]
and for $x,y \in X|_U$,
$ x \pcmOp_U y = x \pcmOp y $ if $x \pcmOp y \in X|_U$, otherwise undefined.

For each guard algebra to be defined,
we will introduce a number of symbols
$\guard G_1,\dots,\guard G_n$,
called \emph{guard constructors}
each with some \emph{guard domain} $\dom(\guard G_i) \subseteq \AVal^{k_i}$
for some $k_i \in \Nat$.
They induce the set of \emph{guard terms}
$
  \Type{GT} \is
    \Union_{i=1}^n
      \set{\guard G_i(\vec{a}) |
        \vec{a} \in \dom(\guard G_i)
      }
$.
By specifying
  some guard constructors,
  a congruence
    ${\cong} \subseteq \Multiset(\Type{GT}) \times \Multiset(\Type{GT})$ and
  a set
    $U \subseteq \Multiset(\Type{GT})/_{\cong}$
one obtains the guard algebra
$
  ((\Multiset(\Type{GT})/_{\cong})|_U,
   (\msetPlus/_{\cong})_U,
   [\emptyset]_{\cong}
  )
$.

The guard constructors are specified by listing
their domains, writing
$ \guard{G}_{i} \from D_{i} $
to mean
$ \dom(\guard{G}_{i}) = D_{i} \subseteq \AVal^{k_{i}}$,
as, in certain cases, we may want to further restrict
the domain of the guard constructors to simplify the
reasoning.

The congruence $\cong$ is specified as the smallest congruence
satisfying given axioms of the form
\[
  \mset{
    \guard G_{i_1}(\vec{a}_{i_1})
      , \dots ,
    \guard G_{i_k}(\vec{a}_{i_k})
  }
  \cong
  \mset{
    \guard G_{j_1}(\vec{a}_{j_1})
      , \dots ,
    \guard G_{j_{k'}}(\vec{a}_{j_{k'}})
  }
\]
which we write using the syntax
\[
  \guard G_{i_1}(\vec{a}_{i_1})
    \pcmOp \dots \pcmOp
  \guard G_{i_k}(\vec{a}_{i_k})
  =
  \guard G_{j_1}(\vec{a}_{j_1})
    \pcmOp \dots \pcmOp
  \guard G_{j_{k'}}(\vec{a}_{j_{k'}})
\]

The set $U$ is specified as the smallest set satisfying given axioms of the form
\[
  {\bigl[\,\mset{
    \guard G_{i_1}(\vec{a}_{i_1})
    ,\dots,
    \guard G_{i_k}(\vec{a}_{i_k})
  }\,\bigl]_{\cong}}
  \in U  
\]
which we write using the syntax
\[ 
  \guardUndef{
    \guard G_{i_1}(\vec{a}_{i_1})
      \pcmOp \dots \pcmOp
    \guard G_{i_k}(\vec{a}_{i_k})
  }
\]

\begin{example}
    The guard algebra used in \cref{ex:distinguishing-client},
  is expressed by using two guard constructors with empty domain,
  $\guard{k}$ and $\guard{d}$, and axioms:
    $\guardUndef{\guard{k}\guardOp\guard{k}}$,
    $\guardUndef{\guard{d}\guardOp\guard{d}}$
  Note that with no congruence axioms,
  the induced congruence relation is equality.
  These induce the guard algebra 
  with elements
  $
    \Set{
      \msetEmpty,
      \mset{\guard{k}},
      \mset{\guard{d}},
      \mset{\guard{k},\guard{d}}
    }
  $.
\end{example}

\subsection{Levels}

Region levels are used to remove the possibility of
unsound duplication of resources by opening regions.
The presentation of the program proofs omits the level annotations to ease readability.
The levels can be unambiguously derived from the sequence of application
of rules \ref{rule:update-region} and \ref{rule:lift-atomic}.

To see the problem consider a generic region $\region[\lvl]{t}{\rid}(a)$;
we have
$
  \region[\lvl]{t}{\rid}(a) \equiv
    \region[\lvl]{t}{\rid}(a) * \region[\lvl]{t}{\rid}(a)
$: this is the essence of what it means for a region to be a shared resource.
When we open a region however, we obtain ownership of the contents of its interpretation $\rInt(\region[\lvl]{t}{\rid}(a))$;
the interpretation can contain resources that are not shared,
for example heap assertions, in which case we have
$
  \rInt(\region[\lvl]{t}{\rid}(a)) \not\equiv
    \rInt(\region[\lvl]{t}{\rid}(a)) * \rInt(\region[\lvl]{t}{\rid}(a))
    \equiv \False
$.
Without constraining levels, one could start with
$ \region[\lvl]{t}{\rid}(a) $,
produce the equivalent
$ \region[\lvl]{t}{\rid}(a) * \region[\lvl]{t}{\rid}(a) $,
open the first region assertion
with \ref{rule:update-region} or \ref{rule:lift-atomic},
then open the second region assertion and end up with $\False$.
Levels are a mean to avoid unsound derivations that use the above chain of implications.
A level $\lvl$ in the context of a judgement records that all the regions of
level $\lvl$ or higher might have been already opened and should not be opened again.
The rules that do open regions (\cref{rule:update-region,rule:lift-atomic})
can only open a region of level $\lvl$ 
if the level in the context is $\lvl+1$,
and they record the operation by setting the context level to $\lvl$,
so that the region cannot be opened again.

\subsection{Region type specifications}

\paragraph{Abstract state domain}
It can be tedious (and detrimental to readability) to always explicitly
write the domains of quantified variables in the assertions of program proofs,
especially when they can be easily inferred from context.
Consider the case of regions.
Some of the rules, for example~\ref{rule:make-atomic},
need the precise domain of the abstract state ($\exists x\in X$)
because it needs to match the pseudo-quantifier's domain ($\aall x\in X$).
To improve readability, we adopt the following strategy.
Suppose the region type $\rt$ has abstract state in the domain $A$.
We can define the interpretation function so that it constrains the domain
of the abstract state accordingly:
$ \rInt(\region[\lvl]{\rt}{\rid}(a)) = a \in A \land \cdots $.
Then we trivially have that
$
  \viewshift \lvl'; \actxt |=
    \exists a\st \region[\lvl]{\rt}{\rid}(a) =>
      \exists a\in A\st\region[\lvl]{\rt}{\rid}(a)
$.
We thus can omit the domains from existential quantification
and implicitly apply \cref{rule:consequence} whenever the domain
information is needed in the proof.

To further ease the specification of region types,
when defining a new region type we will introduce
the domain of the corresponding abstract state,
and omit the obvious constraint from the interpretation definition.

\paragraph{Fixed parameters}
It is very common to have a product domain as abstract state of regions,
as one needs to assemble in an abstract state many bits of information
that characterise region's state.
Typically, the abstract state domain $A$
can be seen as the product of two domains $F \times S$,
the domain of the \emph{fixed parameters} $F$ and
the domain of the \emph{non-fixed parameters} $S$.
(Both $F$ and $S$ can be themselves products of simpler domains.)
The fixed parameters are set at the point of creation of the region,
and can never be updated; they typically define the ``interface'' of the region.
For example, if the address of a lock module instance $x$ is the fixed parameter
of a hypothetical region $\region{lock}{\rid}(x,l)$
and $l \in \set{0,1}$ the non-fixed parameter
representing the state of the lock.
When introducing a new region type we will specify which parameters are fixed,
and they will be omitted from the region interference specification,
as they are left untouched by every transition.
For example, for the region $\region{lock}{\rid}(x,l)$ above,
we may write
$
  \guard{g}: (0, \oblZero) \interfTo (1, \obl{k} ) 
$ and
$
  \guard{g}: (1, \obl{k} ) \interfTo (0, \oblZero)
$
to denote
$
  \guard{g}: ((x,0), \oblZero) \interfTo ((x,1), \obl{k} )
$ and
$
  \guard{g}: ((x,1), \obl{k} ) \interfTo ((x,0), \oblZero)
$.

\paragraph{Interference protocols and atomicity contexts}
\cref{def:interference} requires $\regLTS[\rt]$ to be
monotone in the guards, reflexive and closed under obligation frames.
Since writing the whole function can be tedious and redundant,
we will only write a number of expressions of the form
\begin{equation}
  G: (a_1, O_1) \interfTo (a_2, O_2)
  \label{ex:interfto}
\end{equation}
which will set $\regLTS[\rt](G) \ni \set{\bigl((a_1, O_1), (a_2, O_2)\bigr)}$,
and implicitly complete the function
by closing $\regLTS[\rt]$ under the properties above.

Similarly, atomicity contexts associate to some region identifier
records $ \actxt(r) = (X,k,X',R) $
that have (unguarded) transition relations as their last component $R$.
We therefore borrow the syntax from \eqref{ex:interfto},
and write $ R = (a_1, O_1) \interfTo (a_2, O_2) $
to specify $ R $ as the minimal relation that include such relations
and is closed under obligation frames.

\subsection{Proof patterns}

There are some recurring patterns in \tadalive{} proofs,
which we summarise here to help the reader navigate the examples.

\paragraph{The exclusive guard}

Take for example a concurrent counter module.
Abstractly we have a (fixed) location $x$ for the module instance
and an abstract state $n \in \Nat$ representing the current value of the counter.
Since this is a concurrent counter it uses internally shared resources.
We therefore have a region
$ \region{cnt}{\rid}(x,n) $
encapsulating the shared internal resources of the counter.
From the perspective of the client, however,
at the moment of creation of the counter with, say,
an operation \code{makeCounter()},
the counter is exclusively owned by the client.
This, for example, is reflected in the fact that, until the client
shares the counter or invokes operations on it,
the value of the counter will be stably $0$.
To represent this fact, one typically defines an exclusive guard $\gEx$
guarding each transition of the region interference:
e.g.~$ \gEx: (n,O_1) \interfTo (m,O_2) $.
Then the \code{makeCounter()} operation can be given the specification
\[
  \TRIPLE |-
    {\emp}
      {\code{x := makeCounter()}}
    {\exists \rid\st \region{cnt}{\rid}(\pvar{x},0) * \guardA{\gEx}{\rid}}
\]
which gives to the client the stable assertion
$ \region{cnt}{\rid}(x,0) * \guardA{\gEx}{\rid} $.
(Note how $\region{cnt}{\rid}(x,0)$ is not stable.)
To re-share the counter, the client will create its own region
encoding the invariants governing the interaction over the counter
(and the other resources of the client)
the interpretation of which will contain
$\region{cnt}{\rid}(x,0) * \guardA{\gEx}{\rid}$.

Note that the assertion $\region{cnt}{\rid}(x,0) * \guardA{\gEx}{\rid}$
has a very different meaning if occurring in the \emph{atomic} precondition
of a triple, as opposed to the \emph{Hoare} precondition:
the resources in the atomic precondition are not owned by the local thread,
but only acquired instantaneously at the linearisation point.
For example, in the triple
\[
  \ATRIPLE |-
    \A n \in \Nat.
      <\region{cnt}{\rid}(\pvar{x},n) * \guardA{\gEx}{\rid}>
        {\code{incr(x)}}
      <\region{cnt}{\rid}(\pvar{x},n+1) * \guardA{\gEx}{\rid}>
\]
the exclusivity of $\gEx$ is only granted \emph{instantaneously} to the thread
acting on it atomically, i.e.~either the environment during the interference phase as allowed by the pseudo-quantifier, or the local thread at the linearisation point.

Since this pattern is ubiquitous,
we reserve the $\gEx$ guard constructor for this use, and
will omit the $\guardUndef{\gEx\guardOp\gEx}$
axiom when specifying guard algebras.

\subsection{Modules}

\tadalive\ is a logics that emphasizes modularity of the proofs.
One aspect of this is that when a program is naturally structured as a collection of modules, one would want the proof of correctness to be
decomposed into independent proofs of each module exporting some specifications for the externally accessible operations,
and a proof that the client of these modules is correct,
which depends only on these abstract module specifications.

In our model, a module is nothing but a conceptually related set of operations
$\pvar{f}_1 , \dots , \pvar{f}_n$
that are defined in a \code{let} statement:
\acode{let f$_1$(VARS$_1$) = CMD$_1$ in ...let f$_n$(VARS$_n$) = CMD$_n$ in CMD}.
Here $\cmd$ is what we call ``client'' of a module offering
operations $\pvar{f}_1 , \dots , \pvar{f}_n$.
The operation deals with let statements by populating a function $\functxt$
associating each function name $\pvar{f}_i$ to its formal parameters $\pvars{x}_i$ and its implementation $\cmd_i$.

Similarly, the proof of correctness of $\cmd$, will need to fetch the
abstract specifications of the functions (which appear as free names in $\cmd$)
from some mapping $\funspec$ from function names to their specifications.
The fact that the implementation of each operation satisfies its specification
is checked in the proof derivation for the let statement (\cref{rule:let})
but then the proof of the client and of the module are done separately.

For this reason, we present proofs of just a module against its abstract specifications, which can be used as if they were axioms in the proof of any client using them.
To talk about modules independently of their clients we introduce the notation
\acode{def f(x) \{CMD\}} which can be understood as populating
an entry of $\functxt$ for \pvar{f}.
We will then prove some specification for \pvar{f} which will populate
an entry of $\funspec$ for \pvar{f}.

In the proof of some client, we will recall the module specifications
that are assumed in $\funspec$, and use \cref{rule:call} to handle
the calls to the operations of the module.
We will omit from the proof outlines $\funspec$ and the applications of \cref{rule:call} for readability.

\subsection{Proof outlines}

In program proof outlines, we adopt a number of notational conventions.
First, unless it involves a viewshift or we want to highlight it,
we will apply \cref{rule:consequence} without mentioning it.
Similarly, we omit the obvious applications of
\cref{rule:var,,rule:call,,rule:subst}
and the axioms (i.e.~the rules associated with primitive commands).

Next, in outline such as

\begin{center}
\begin{proofoutline}
  \CTXT m; \lvl; \actxt |-\\
  \A x\in X \eventually X'.\\
  \PREC<P(x)>\\
  \begin{proofjump}["\textsc{outer}"]
    \PREC<P'(x)>\\
    \begin{proofjump}["\textsc{inner}"]
    $\;\vdots$
    \end{proofjump}\\
    \POST<Q'(x)>
  \end{proofjump}\\
  \POST<Q(x)>
\end{proofoutline}
\hspace{3em}
$
  \infer*[Right={outer},vcenter]{
    \infer*[Right={inner}]{
      \vdots
    }{
      \ATRIPLE m; \lvl; \actxt |-
        \A x\in X \eventually X'.<P'(x)> \cmd <Q'(x)>
    }
  }{
    \ATRIPLE m; \lvl; \actxt |-
      \A x\in X \eventually X'.<P(x)> \cmd <Q(x)>
  }
$
\end{center}

the specification of the inner step
inherits the context and the pseudo-quantifier of the specification
of the outer step, as in the derivation on the right.
 
\section{The \tadalive\ proof system}
\label{app:proof-system}

In this section, we present the full proof system of \tadalive.

For brevity we use the metavariable
$\pqsets{X}$ to range over expressions of the form
$ X_1 \eventually[k] X_2 $ and is used in rules when the pseudo-quantification
is simply copied verbatim from premise to conclusion.

In the rules we use the following abbreviation:
\begin{align*}
  \VALID \actxt |= \minLay{P}{k} & \quad \is \quad
    \forall r\in\RId \st
      \VALID \actxt |= P \implies \minLay{r}{k}
  \\
  k \laygeqq n &\quad\is\quad
    \forall k'\laygt k \st k' \laygeq n
\end{align*}
The \pre\lvl-safety condition is defined in \cref{app:lsafe}
and can be typically proven by using \cref{lemma:lsafe}.

\subsection{Liveness rules}
For reference we reproduce the liveness-related rules.
 
\begingroup
\centering

\begin{proofrule*}
  \infer*[right={LiveCG}]{
  \forall x \in X \st
    \VALID \lvl; \actxt |= \at{P}(x) * T \implies x \in X'
  \\\\
  \ENVLIVE n; \lvl; \actxt |- L -M->> T
  \\
  m \laygeq n \quad k \laygeqq n
  \\
  \progvars(L) \inters \modvars(\cmd) = \emptyset
  \\\\ \ATRIPLE[\funspec] m; \lvl; \actxt |-
    \A x \in X\eventually[k]X'. <\na{P}| \at{P}(x)> \cmd \E y. <\na{Q}(x, y)|\at{Q}(x)>
}{\phantom{{}\eventually[k]X'}
     \ATRIPLE[\funspec] m; \lvl; \actxt |-
         \A x \in X. <\na{P} * L| \at{P}(x)> \cmd \E y. <\na{Q}(x, y) * L|\at{Q}(x)>
}
 \end{proofrule*}

\begin{proofrule*}
  \infer*[right=While]{\forall \beta \le \beta_0 \st
    \ENVLIVE m(\beta); \lvl; \actxt |- L - M ->> T(\beta)
  \\\forall \beta \le \beta_0 \st
    \VALID \actxt |= \minLay{P(\beta)}{m(\beta)} \layleq m
  \\\\\forall \alpha\st
    \STABLE \actxt |=
      {\exists \alpha'\st L * M(\alpha') \land \alpha' \leq \alpha}
  \\\progvars(T,L,M) \cap \modvars(\cmd) = \emptyset
  \\\\\forall \beta \le \beta_0 \st
    \forall b \in \Bool \st
    \TRIPLE[\funspec] m; \lvl; \actxt |-
        {P(\beta) *  (b \dotimplies T(\beta)) \land \bexp}
        \cmd
        {\exists \gamma \st P(\gamma)
          \land \gamma \leq \beta
          * (b \dotimplies \gamma < \beta)}
}{\TRIPLE[\funspec] m; \lvl; \actxt |-
    {P(\beta_0) * L}
      {\acode{while(BEXP)\{CMD\}}}
    {\exists \gamma \st  P(\gamma) * L
      \land \neg \bexp
      \land \gamma \leq \beta_0}
}
 \end{proofrule*}

\begin{proofrule*}
  \infer*[right={Par}]{
  \TRIPLE[\funspec] m_1; \levl; \actxt |- {P_1}{\cmd_1}{Q_1}
  \\
  \VALID \actxt |= \minLay{Q_1}{m_2} \layleq m
  \\\\
  \TRIPLE[\funspec] m_2; \levl; \actxt |- {P_2}{\cmd_2}{Q_2}
  \\
  \VALID \actxt |= \minLay{Q_2}{m_1} \layleq m
}{\TRIPLE[\funspec] m; \levl; \acontext |-
    {P_1 \ssep P_2}
      {\cmd_1 \parallel \cmd_2}
    {Q_1 \ssep Q_2}
}
 \end{proofrule*}

\endgroup
\bigskip

\subsubsection{The Environment liveness rules}
The Environment liveness rules
use the $\decr$ condition (\cref{def:impr}) recalled here for convenience:

\begin{definition}[\/$\decr$]
  Given assertions $L(\alpha)$, $L'(\alpha)$ and $T$,
the condition
  $\decr[\actxt](L', L, T)$
holds if and only if, for arbitrary $\store \in \Store$, letting
\begin{align*}
        l(\alpha) &= \WorldSem{\actxt}{\store}{L(\alpha)} &
        l'(\alpha) &= \WorldSem{\actxt}{\store}{L'(\alpha)} &
        t &= \WorldSem{\actxt}{\store}{T * \True}
\end{align*}
the following holds:
\[
\forall \alpha_1, \alpha_2 \ge \alpha_1 \st
{\relyAt[\actxt]}(l'(\alpha_1)) \inters l(\alpha_2) \subseteq l'(\alpha_1) \union t
\]
 \end{definition}

We reproduce below for completeness
the rules to prove the environment liveness condition.
\begingroup
\centering

\begin{proofrule*}
  \infer*[right={EnvLive}]{
  \STABLE \lvl; \actxt |= L
  \\
  \VALID \lvl; \actxt |=
    L \implies L * \exists \alpha\st M(\alpha)
  \\\\
  \ENVLIVE m; \lvl; \actxt |- L * M(\alpha) : L * M(\alpha) -->> T
}{
  \ENVLIVE m; \lvl; \actxt |- L - M ->> T
}
 \end{proofrule*}

\begin{proofrules}
  \infer*[right=ECase,vcenter]{
  \ENVLIVE m; \levl; \actxt |-
    L(\alpha) : L_1(\alpha) -->> T
  \\\\
  \ENVLIVE m; \levl; \actxt |-
    L(\alpha) : L_2(\alpha) -->> T
}{
  \ENVLIVE m; \levl; \actxt |-
    L(\alpha) : L_1(\alpha) \lor L_2(\alpha) -->> T
}
   \and
  \infer*[right=EQuant,vcenter]{
  \forall x \in X \st
  \ENVLIVE m; \levl; \actxt |-
    L(\alpha) : L(x, \alpha) -->> T
}{
  \ENVLIVE m; \levl; \actxt |-
    L(\alpha) : \exists x \in X \st L(x, \alpha) -->> T
}
   \and
  \infer*[Right=LiveT,vcenter]{
  \forall \alpha. \VALID \actxt |= T'(\alpha) \implies T
}{
  \ENVLIVE m; \lvl; \actxt |-
    L(\alpha) : T'(\alpha) -->> T
}
 \end{proofrules}

\begin{proofrule*}
  \infer*[right=LiveO]{
  \decr[\actxt](L', L, T)
  \\
  \forall \alpha\st
    \VALID \actxt |= \minLayStrict{L'(\alpha)}{\lay(O(x))}
  \\\\
  \lvl < \lvlp
  \\
  \forall \alpha \st
    \SAT[\actxt] |- L'(\alpha) \implies
    \exists x\st \region[\lvl]{\rt}{\rid}(x) * \envObl{O(x)}{\rid} * \True
    \land m \laygt \lay(O(x))
}{ 
\ENVLIVE m; \lvlp; \actxt |-
  L(\alpha) :
  L'(\alpha)
  -->> T
}
 \end{proofrule*}

\begin{proofrule*}
  \infer*[right=LiveA]{
  \decr[\actxt](L', L, T)
  \\
  m \laygt k
  \\
  \forall \alpha\st
    \VALID \actxt |= \minLayStrict{L'(\alpha)}{k}
  \\\\
  (X \eventually[k] X') = \live(\actxt,r)
\\
  \lvl < \lvlp \\
  \SAT[\actxt] |- L'(\alpha) \implies
    \exists x \in X \setminus X' \st
    \region[\lvl]{\rt}{\rid}(x)
    * \done{r}{\lozenge}
    * \True
}{
  \ENVLIVE m; \lvlp; \actxt |-
    L(\alpha) :
    L'(\alpha)
    -->> T
}
 \end{proofrule*}

\endgroup

\subsection{Atomicity rules}
\label{app:atomicity-rules}

We first give the formal definition of \pre\lvl-safety and prove its properties,
and then give the general forms of
\cref{rule:lift-atomic,,rule:make-atomic,,rule:update-region},
which are the ones dealing with proving atomicity.

\subsubsection{The \pre\lvl-safety condition}
\label{app:lsafe}

The rules of \tadalive\ dealing with opening and closing regions 
(\cref{rule:update-region,rule:lift-atomic})
require the \pre\lvl-safety side condition for the postcondition.
While the definition of \pre\lvl-safety is technical, its intuition is simple:
those are the assertions that preserve their meaning
when interpreted at level~$\lvl$ or at level~$\lvl+1$.
The only possible contradictions arising by increasing the level come
from assertions about the state and environment obligations of
regions that are open at~$\lvl$ but not at~$\lvl+1$.

\begin{definition}[Havoc]
  Let $\lvl \in \Level$.
  The set $
    \clreg^{\lvl_2}_{\lvl_1}(\regMap) \is
      \set{r | \regMap(r) = (\wtv,\lvl,\wtv), \lvl_1 \leq \lvl < \lvl_2}
  $ is the set of region ids of~$\regMap$ that are
  closed at level~$\lvl_2$ but not at level~$\lvl_1$.
  We define the function on worlds:
  \[
    \havoc{\lvl}(h, \regMap, \guardMap, \atomMap, \oblMap, \envMap) \is
      \Set{
        (h, \regMap', \guardMap, \atomMap, \oblMap, \envMap') |\;
          \begin{matrix*}[l]
            \clreg^{\lvl+1}_{\lvl}(\regMap) = \set{\rid_{1}, \dots, \rid_{n}},
            \\
            \regMap(r_i) = (\rt_i,\wtv[2]),
            b_i \in \AVal,
            w_i \in \rIntSem[\rt_i]{r_i,\lvl,b_i},
            \\
            \regMap' = \regMap\map{r_1->(\rt_1,\lvl,b_1);;r_n->(\rt_n,\lvl,b_n)},
            \\
            O_i' \oblOp \oblMap[w_i](r_i) = \envMap(r_i),
            \envMap' \resgeq \envMap\map{r_1->O_1';;r_n->O_n'}
          \end{matrix*}
      }
  \]
  We extend it to a function on sets of worlds
  in the obvious way:
  $
    \havoc{\lvl}(p) \is
      \Union\limits_{w\in p} \havoc{\lvl}(w)
  $.
\end{definition}

\begin{definition}[\pre\lvl-safety]
  A set $p\in\UCWorld[\actxt]$ is \emph{\pre\lvl-safe} if 
  $ p = \havoc{\lvl}(p) $.
  An assertion~$P$ is \pre\lvl-safe, written $ \lsafe \lvl;\actxt |= P $ if,
  for all $\varsigma$, $\WorldSem{\actxt}{\varsigma}{P}$ is \pre\lvl-safe.
\end{definition}

Since proving \pre\lvl-safety in general
involves meddling with the semantics of assertions,
we provide the following lemma that can be used to immediately prove all the \pre\lvl-safety side conditions
involved in our program proofs.

\begin{lemma}
\label{lemma:lsafe-again}
  The properties below hold, for arbitrary~$\lvl\in\Level$:
  \begin{enumerate}
    \item $\emp$, $\vexp_1\mapsto \vexp_2$ and $\bexp$ are \pre\lvl-safe.
          \label{lm:lsafe-basic}
    \item $\guardA{G}{\rid}$ and $\locObl{O}{\rid}$ are both \pre\lvl-safe.
          \label{lm:lsafe-g-o}
    \item If\/ $\lvl' < \lvl$,
          then $
              \region[\lvl']{\rt}{\rid}(a) *
              \envObl{O}{\rid}
          $ is \pre\lvl-safe.
          \label{lm:lsafe-closed}
    \item If\/ $P,Q$ are both \pre\lvl-safe,
          then so are
            $P \land Q$, $P \lor Q$, and $P*Q$.
          \label{lm:lsafe-and-or-star}
    \item If $P(v)$ is \pre\lvl-safe for all $v\in\AVal$,
          then $\exists x\st P(x)$ is \pre\lvl-safe.
          \label{lm:lsafe-exists}
  \end{enumerate}
\end{lemma}

\subsubsection{Generalised atomicity rules}
The following rules are the general forms of 
\cref{rule:lift-atomic,,rule:make-atomic,,rule:update-region}
of \cref{fig:proof-rules}.

\begingroup
\centering
\begin{proofrule*}
  \begingroup
\specsdelim{big}
\infer*[right={MkAtomG}]{
  \lvl < \lvlp
  \\
  \rid \notin \dom(\actxt)
  \\
  \actxt' = \actxt\map{r -> (X,k,X',T)}
  \\
  T \subseteq \regLTS_{\rt}( \guard{G} )
  \\
  R = \iorel(T) \\
  \forall x\in X \st
  \STABLE \actxt |= {
    \region[\lvl]{\rt}{\rid}(x) * \guardA{\guard{G}}{\rid}
  }
  \\\\
  \TRIPLE[\funspec] m; \lvlp; \actxt' |-
    {\na{P} * \exists x \in X \st
        \region[\lvl]{\rt}{\rid}(x)
        * \done{\rid}{\blacklozenge}}
    \cmd
    {\exists x, y \st
        R(x,y) * \na{Q}(x,y) * \done{\rid}{(x,y)}}
}{
  \ATRIPLE[\funspec] m; \lvlp ; \actxt |-
    \A x \in X \eventually[k] X'.
      <\na{P} | \region[\lvl]{\rt}{\rid}(x) * \guardA{\guard{G}}{\rid} >
        \cmd
      \E y.
      < \na{Q}(x,y) |
          \region[\lvl]{\rt}{\rid}(y) * \guardA{\guard{G}}{\rid} * R(x,y) >
}
\endgroup
   \label{rule:make-atomic-gen}
\end{proofrule*}
\endgroup

\begingroup
\let\proofrulesize\footnotesize
\begin{proofrule*}
  \infer*[right={UpdRegG}]{
  r \in \dom(\actxt)
  \\
  \actxt' = \actxt\map{r->\bot}\\\\
  \VALID \actxt |= \na{P} \implies \empObl[\rid]
  \\
  \VALID \actxt |= \at{P}(x) \implies \empObl[\lvl{+}1]
  \\\\
  \VALID \actxt |= \na{Q}(x,y) \implies \empObl[\rid]
  \\
  \VALID \actxt |= Q_i(x,y,z) \implies \empObl[\lvl{+}1]
\\\\
  \lsafe \lvl;\actxt |= {\na{P}} \\
  \lsafe \lvl;\actxt |= {\at{P}(x)}
  \\\\
  \lsafe \lvl;\actxt |= {\na{Q}(x,y)}
  \\
  \lsafe \lvl;\actxt |= {\left(
      \begin{array}{@{}r@{}l@{}}
           & (R(x,z) \land Q_1(x,y,z)) \\
      \lor & (\ghost{R(x,z)}{x = z} \land Q_2(x,y))
      \end{array}
    \right)}
  \\\\
\Set{
    ((x, \obl{O}_0) , (z, \obl{O}_1(x,y))) |
      x\in X \land (R(x,z) \lor x=z) \land y \in Y(x)
  }
\subseteq \trrel(\actxt, r)
  \\\\
  {\ATRIPLE[\funspec] m; \levl; \actxt' |-
    \A x \in \pqsets{X}.
    < \na{P} |
      \begin{array}{@{}l@{}}
        \rInt(\region[\levl]{\rt}{\rid}(x)) \\
        {} * \at{P}(x) *
        \oblA{\obl{O}_0}{\rid}
      \end{array}
    >
    \cmd
    \E y. 
    < \na{Q}(x,y) \land y \in Y(x)
      \\ \quad\Big|\;
      \exists z\st
      \rInt(\region[\levl]{\rt}{\rid}(z)) * \oblA{\obl{O}_1(x,y)}{\rid} *
      \left(
        \begin{array}{@{}r@{}l@{}}
             & (R(x,z) \land Q_1(x, y, z)) \\
        \lor & (\ghost{R(x,z)}{x = z} \land Q_2(x, y))
        \end{array}
      \right)
    >}
}{
{\ATRIPLE[\funspec] m; \levl {+} 1;\actxt |-
    \A x \in \pqsets{X}.
      < \na{P} * \oblA{\obl{O}_0}{\rid}
       \\ \quad\big|\;
        \region[\levl]{\rt}{\rid}(x) * \at{P}(x) * \done{\rid}{\blacklozenge} >
      \cmd
    \E y.
    < \na{Q}(x,y) * \oblA{\obl{O}_1(x,y)}{\rid} \land y \in Y(x)
      \\ \quad\Big|\;
      \exists z\st
      \region[\levl]{\rt}{\rid}(z) *
      \left(
        \begin{array}{@{}r@{}l@{}}
                 & (R(x, z) \land Q_1(x, y, z) * \done{\rid}{(x,z)}) \\
        {}\lor{} & (\ghost{R(x, z)}{x = z}   \land \ghost{Q_1(x, y, z)}{Q_2(x, y)}
                * \done{\rid}{\blacklozenge})
        \end{array}
      \right)
    >}
}
   \label{rule:update-region-gen}
\end{proofrule*}
\endgroup

\begin{proofrule*}
  \infer*[right={LiftAG}]{
  \VALID \actxt |= \na{P} \implies \empObl[\rid]
  \\
  \VALID \actxt |= \na{Q}(x,y) \implies \empObl[\rid]
  \\\\
  \VALID \actxt |= \at{P}(x) \implies \empObl[\lvl{+}1]
  \\
  \VALID \actxt |= \at{Q}(x,y,z) \implies \empObl[\lvl{+}1]
  \\\\
  \lsafe \lvl;\actxt |= {\na{P}} \\
  \lsafe \lvl;\actxt |= {\at{P}(x)}
  \\\\
  \lsafe \lvl;\actxt |= {\na{Q}(x,y)} 
  \\
  \lsafe \lvl;\actxt |= {\at{Q}(x,y,z) \land R(x,z)}
  \\\\
  \rid \in \dom(\actxt) \implies R = \mathop{id}
  \\
\Set{ ((x, \obl{O}_1) , (z, \obl{O}_2(x,y))) | x\in X \land R(x,z) \land y\in Y(x) } 
    \subseteq \regLTS_{\rt}( \guard{G} )
  \\ 
  {\ATRIPLE[\funspec] m; \lvl; \actxt |-  
    \A x \in \pqsets{X}. 
      < \na{P} | 
        \rInt(\region[\lvl]{\rt}{\rid}(x)) * \at{P}(x) *
        \guardA{\guard{G}}{\rid} * \oblA{\obl{O}_1}{\rid} >
        \cmd
      \E y.
      < \na{Q}(x,y) \land y \in Y(x)
        \\ \quad\big|\; 
           \begin{array}{@{}r@{}l@{}}
           \exists z\st &
                  \rInt(\region[\lvl]{\rt}{\rid}(z)) * \at{Q}(x,y,z) \\ 
            {}*{} & \oblA{\obl{O}_2(x, y)}{\rid} \land R(x,z)
           \end{array}
      >}
}{
  {\ATRIPLE[\funspec] m; \lvl {+} 1; \actxt |-
    \A x \in \pqsets{X}.
      < \na{P} * \oblA{\obl{O}_1}{\rid} 
        | \region[\lvl]{\rt}{\rid}(x) * 
        \at{P}(x) * \guardA{\guard{G}}{\rid} >
      \cmd
      \E y. 
      < \na{Q}(x,y) * \oblA{\obl{O}_2(x,y)}{\rid} \land y \in Y(x)
        \\ \quad\big|\;
        \exists z\st
          \region[\lvl]{\rt}{\rid}(z) *
          \at{Q}(x,y,z) \land R(x,z) >}
}
   \label{rule:lift-atomic-gen}
\end{proofrule*}

\subsection{General forms}
The following rules are the general forms of 
some of the rules in \cref{fig:proof-rules}.

\begingroup
\centering

\begin{proofrule*}
  \infer*[right={Frame}]{
  \forall x \in X \st
    \VALID \actxt |= \minLay{\na{R} * \at{R}(x)}{m}
  \\
  \progvars(\na{R}) \inters \modvars(\cmd) = \emptyset
  ,
  \progvars(\at{R}(x)) = \emptyset
  \\
  \STABLE \actxt |= { \na{R} }
  \\
  \forall x \in X\st
    \STABLE \actxt |= { \at{R}(x) }
  \\
  \oblfree \lvl;\actxt |= {\at{R}(x)}
  \\\\
  \ATRIPLE[\funspec] m; \lvl; \actxt |-
    \A x \in \pqsets{X}.
    <\na{P}\phantom{{} * \na{R}} | \at{P}(x)\phantom{{} * \at{R}(x)}>
      \cmd
    \E y.<\na{Q}(x,y)\phantom{{} * \na{R}} | \at{Q}(x,y)\phantom{{} * \at{R}(x)}>
}{
  \ATRIPLE[\funspec] m; \lvl; \actxt |-
    \A x \in \pqsets{X}.
    <\na{P} * \na{R} | \at{P}(x) * \at{R}(x)>
    \cmd
    \E y.<\na{Q}(x,y) * \na{R} | \at{Q}(x,y) * \at{R}(x)>
}
   \label{rule:frame}
\end{proofrule*}

\begin{proofrule*}
  \infer*[right={AtomWG}]{
  \STABLE \actxt |= {\na{P} * P}
  \\
  \forall x\in X,y\st
    \STABLE \actxt |= {Q(x, y)}
  \\\\
  \ATRIPLE[\funspec] m; \levl; \actxt |-
    \A x \in \pqsets{X}.
    <\na{P} | P * \at{P}(x)>
      \cmd
    \E y.
    <\na{Q}(x,y) | Q(x,y) * \at{Q}(x,y)>
}{
  \ATRIPLE[\funspec] m; \levl; \actxt |-
    \A x \in \pqsets{X}.
    <\na{P} * P | \at{P}(x)>
      \cmd
    \E y.
    <\na{Q}(x,y) * Q(x,y) | \at{Q}(x,y)>
}
   \label{rule:atomicity-weak-gen}
\end{proofrule*}

\begin{proofrule*}
  \infer*[right={A$\exists$ElimG}]{
  \ATRIPLE[\funspec] m; \lvl; \actxt |-
    \A x \in \pqsets{X}, z \in Z.
    <\na{P} | \at{P}(x,z)>
      \cmd
    \E y.
    <\na{Q}(x,y) |  \at{Q}(x,y,z)>
  \phantom{, z \in Z}}{
  \ATRIPLE[\funspec] m; \lvl; \actxt |-
    \A x \in \pqsets{X}.
    <\na{P} | \exists z \in Z \ldotp \at{P}(x,z)>
      \cmd
    \E y.
    <\na{Q}(x,y) | \exists z \in Z \ldotp \at{Q}(x,y,z)>
}
   \label{rule:atomic-exists-elim-gen}
\end{proofrule*}

\begin{proofrule*}
  \infer*[right={LayWG}]{
  \ATRIPLE[\funspec] k_1; \lvl ; \actxt |-
    \A x \in \pqsets{X}.
    <\na{P} | \at{P}(x) >
      \cmd
    \E y.<\na{Q}(x, y) | \at{Q}(x, y)>
  \\
  k_1 \layleq k_2
}{
  \ATRIPLE[\funspec] k_2; \lvl ; \actxt |-
    \A x \in \pqsets{X}.
    <\na{P} | \at{P}(x) >
      \cmd
    \E y.<\na{Q}(x, y) | \at{Q}(x, y)>
}
   \label{rule:layer-weak-gen}
\end{proofrule*}

\endgroup

\subsection{Logical manipulation rules}

The rules below allow for basic logical manipulation.

\begingroup
\centering

\begin{proofrule*}
  \infer*[right={Cons}]{
  {\begin{array}{r@{\;}l}
  & \STABLE \phantom{\lvl;}\actxt |= {\na{P}}
  \\
  & \viewshift \levl; \actxt |= \na{P} => \na{P'}
  \\
  \forall x\in X\st&
{\def\vshift{\mathrel{{\Lleftarrow}\mkern-10mu{\Rrightarrow}}}
      \viewshift \levl; \actxt |= \at{P}(x) => \at{P'}(x)
    }
  \end{array}}
  \\
  {\begin{array}{r@{\;}l}
  \forall x\in X,y\st&
    \STABLE \phantom{\lvl;}\actxt |= {\na{Q}(x,y)}
  \\
  \forall x\in X,y\st&
    \viewshift \levl; \actxt |= \na{Q'}(x,y) => \na{Q}(x,y)
  \\
  \forall x\in X,y\st&
    \viewshift \levl; \actxt |= \at{Q'}(x,y) => \at{Q}(x,y)
  \end{array}}
  \\\\
  \ATRIPLE[\funspec] m; \levl; \actxt |-
    \A x \in \pqsets{X}.
      <\na{P'} | \at{P'}(x)>
        \cmd
      \E y.<\na{Q'}(x,y) | \at{Q'}(x,y)>
}{
   \ATRIPLE[\funspec] m; \levl; \actxt |-
      \A x \in \pqsets{X}.
        <\na{P} | \at{P}(x)>
          \cmd
        \E y.<\na{Q}(x,y) | \at{Q}(x,y)>
}
   \label{rule:consequence}
\end{proofrule*}

\begin{proofrule*}
  \infer*[right={$\exists$Elim}]{
  \forall v\in X\st\;
    \TRIPLE[\funspec] m; \lvl; \actxt |- {P(v)} \cmd {Q}
}{
  \TRIPLE[\funspec] m; \lvl; \actxt |-
    {\exists x\in X. P(x)} \cmd {Q}
}
   \label{rule:exists-elim}
\end{proofrule*}

\begin{proofrule*}
  \infer*[right={QL}]{
  \forall \ghost{m}{k} \layleq m \st
  \ATRIPLE[\funspec] k; \lvl ; \actxt |-
    \A x \in \pqsets{X}.
    <\na{P}(k) \land k \layleq m | \at{P}(k, x) >
      \cmd
    \E y.<\na{Q}(k, x, y) | \at{Q}(k, x, y)>
}{
  \forall k \layleq m \st
  \ATRIPLE[\funspec] m; \lvl ; \actxt |-
    \A x \in \pqsets{X}.
    <\na{P}(k) \phantom{{} \land k \layleq m} | \at{P}(k, x) >
      \cmd
    \E y.<\na{Q}(k, x, y) | \at{Q}(k, x, y)>
}
   \label{rule:quantify-layer}
\end{proofrule*}

\begin{proofrule*}
  \infer*[right={SubPq}]{
  f \from X \to Y
  \\
  Y' = f(X')
  \\
  \forall x\in X\st
    \VALID \actxt |= \at{P}'(x) \iff \at{P}(f(x))
  \\
  \forall x\in X,z\st
    \VALID \actxt |= \na{Q}(f(x),z) \implies \na{Q}'(x,z)
  \\
  \forall x\in X,z\st
    \VALID \actxt |= \at{Q}(f(x),z) \implies \at{Q}'(x,z)
  \\\\
  \ATRIPLE[\funspec] m; \levl ; \actxt |-
    \A y \in Y \eventually[k] Y'.
    < \na{P} | \at{P}(y) >
      \cmd
    \E z.
    < \na{Q}(y,z) | \at{Q}(y,z) >
}{
  \ATRIPLE[\funspec] m; \levl ; \actxt |-
    \A x \in X \eventually[k] X'.
    < \na{P} | \at{P}'(x) >
      \cmd
    \E z.
    < \na{Q}'(x,z) | \at{Q}'(x,z) >
}
   \label{rule:subst}
\end{proofrule*}

\begin{proofrule*}
  \infer*[right={LiveW}]{
  \ATRIPLE[\funspec] m; \lvl; \actxt |-
    \A x \in X \eventually[k] X''.
      <\na{P} | \at{P}(x)>
        \cmd
      \E y.<\na{Q}(x,y) | \at{Q}(x,y)>
    \\
    X' \subseteq X'' \subseteq X
}{
   \ATRIPLE[\funspec] m; \lvl; \actxt |-
      \A x \in X \eventually[k] X'.
        <\na{P} | \at{P}(x)>
          \cmd
        \E y.<\na{Q}(x,y) | \at{Q}(x,y)>
}
   \label{rule:liveness-weak}
\end{proofrule*}

\endgroup

\subsection{Axioms}

\begingroup
\centering

\begin{proofrule*}
  \infer*[right={Alloc}]{
}{
  \TRIPLE[\funspec] m; \lvl; \actxt |-
         {\vexp \; \dot{\ge} \; 0}
         {\code{x := alloc($\vexp$)}}
         {\Sep*_{i = 0}^{\vexp - 1} \pvar{x} + \pvar{i} \mapsto \wtv }
}
   \label{rule:alloc}
\end{proofrule*}

\begin{proofrule*}
  \infer*[right={Dealloc}]{
  }{
    \TRIPLE[\funspec] m; \lvl; \actxt |-
            {\vexp \mapsto \wtv}
            {\code{dealloc($\vexp$)}}
            {\emp}
  }
   \label{rule:dealloc}
\end{proofrule*}

\begin{proofrule*}
  \infer*[right={Read}]{
}{
  \ATRIPLE[\funspec] m; \lvl; \acontext |-
  \A v .
  <\vexp \mapsto v>
  {\code{x:=[$\vexp$]}}
  <\vexp \mapsto v \land \pvar{x} = v>
}
   \label{rule:read}
\end{proofrule*}

\begin{proofrule*}
  \infer*[right={Mutate}]{
}{
  \ATRIPLE[\funspec] m; \lvl; \actxt |-
  \A v.
  <\vexp_{1} \mapsto v>
  {\code{[$\vexp_{1}$]:=$\vexp_{2}$}}
  <\vexp_{1} \mapsto \vexp_{2}>
}
   \label{rule:mutate}
\end{proofrule*}

\begin{proofrule*}
  \infer*[right={CAS}]{
}{
  {\ATRIPLE[\funspec] m; \lvl; \actxt |-
    \A v.
    <\vexp_{1} \mapsto v>
    {\code{x := CAS($\vexp_{1}$,$\vexp_{2}$,$\vexp_{3}$)}}
    <
    \begin{array}{l}
      (\pvar{x} = 1 \land \vexp_{1} \mapsto \vexp_3 \land v = \vexp_2) \lor {} \\
      (\pvar{x} = 0 \land \vexp_{1} \mapsto \ghost{\vexp_3}{v} \land v \neq \vexp_2)
    \end{array}
    >}
}
   \label{rule:cas}
\end{proofrule*}

\begin{proofrule*}
  \infer*[right={FAS}]{ }{
   \ATRIPLE[\funspec] m; \levl; \actxt |-
      \A v.
        <\vexp_1 \mapsto v>
          \code{x := FAS($\vexp_1$,$\vexp_2$)}
        <\vexp_1\mapsto \vexp_2 \land \p{x} = v>
}
   \label{rule:fas}
\end{proofrule*}

\endgroup

\subsection{Standard Hoare rules}

\begingroup
\centering

\begin{proofrule*}
  \infer*[right=Seq]{\TRIPLE[\funspec] m; \lvl; \actxt |- {P} {\cmd_1} {R}
  \\
  \TRIPLE[\funspec] m; \lvl; \actxt |- {R} {\cmd_2} {Q}
}{\TRIPLE[\funspec] m; \lvl; \actxt |-
    {P} {\cmd_1\p{;}\cmd_2} {Q}
}
   \label{rule:sequential}
\end{proofrule*}

\begin{proofrule*}
  \infer*[right={If}]{
  \TRIPLE[\funspec] m; \lvl; \actxt |- {P \land \bexp} {\cmd_{1}} {Q} \\
  \TRIPLE[\funspec] m; \lvl; \actxt |- {P \land \neg \bexp} {\cmd_{2}} {Q}
}{
  \TRIPLE[\funspec] m; \lvl; \actxt |-
         {P}
         {\code{if($\bexp$)\{$\cmd_{1}$\}else\{$\cmd_{2}$\}}}
         {Q}
}
   \label{rule:if}
\end{proofrule*}

\begin{proofrule*}
  \infer*[right={Var}]{
  \pvar{x} \not\in \fv(\na{P}) \union \fv(\na{Q}) \union \fv(\vexp)
  \\
  \ATRIPLE[\funspec] m; \levl; \acontext |-
    \A x\in \pqsets{X}.
    <\na{P} \land \pvar{x} = \vexp | \at{P}(x)>
    {\cmd}
    <\na{Q}(x,y)|\at{Q}(x,y)>
}{
  \ATRIPLE[\funspec] m; \levl; \acontext |-
    \A x\in \pqsets{X}.
    <\na{P} | \at{P}(x)>
   {\acode{var x=EXP in CMD}}
    <\na{Q}(x,y)|\at{Q}(x,y)>
}
   \label{rule:var}
\end{proofrule*}

\begin{proofrule*}
  \infer*[right={Call}]{
  \bigl(
    \pvars{x},
    \SPEC m;\lvl;\actxt |=
      \A x \in \pqsets{X}. <\na{P}|\at{P}(x)>
      \E y. <\na{Q}(x,y,\pvar{ret})|\at{Q}(x,y)>
  \bigr)
  \in \funspec(\p{f})
}{
  \ATRIPLE[\funspec] m; \lvl;\actxt |-
    \A x \in \pqsets{X}.
      <\na{P}\smash{\subst{\pvars{x}->\vec{\vexp}}}
      |\at{P}(x)>
    {\code{z := f($\smash{\vec{\vexp}}$)}}
    \E y.
<\na{Q}(x,y,\pvar{z})
      |\at{Q}(x,y)>
}
   \label{rule:call}
\end{proofrule*}

\begin{proofrule*}
  \infer*[right={Let}]{
  \progvars(\spec)\subseteq\pvars{x}\union\set{\pvar{ret}}
  \\
  \p{f} \not\in \dom(\funspec)
  \\
  \funspec' = \funspec
    \map{ \p{f} -> (\pvars{x}, \spec ) }
  \\
  \JUDGE[\funspec] |- \cmd_1 : \spec_1
  \\
  \JUDGE[\funspec'] |- \cmd_2 : \spec_2
}{
  \JUDGE[\funspec] |-
    {\acode{let f($\pvars{x}$) = CMD$_1$ in CMD$_2$}} : \spec_2
}
   \label{rule:let}
\end{proofrule*}

\begin{proofrule*}
  \infer*[right={PrAt}]{
P, Q \in \Type{SL}
  \\
  \forall x\in X \st
    \TRIPLE[\funspec] \layBot; 0; \emptyset |-
      {P(x)} \cmd {Q(x)}
}{
   \ATRIPLE[\funspec] \layBot; 0; \emptyset |-
      \A x \in X.
        <P(x)> \atombra{\cmd} <Q(x)>
}
   \label{rule:primitive-atomic}
\end{proofrule*}

\endgroup

\subsection{On Stablity Checks}

A triple is well-defined, according to \cref{def:spec},
if the Hoare pre- and post-conditions are both stable assertions.
The rules all assume the triples in the premises are well-defined
and ensure that the triple in the conclusion is well-defined as well.
The only exceptions are rules
  \ref{rule:make-atomic-gen},
\ref{rule:subst}, and
  \ref{rule:exists-elim},
where the Hoare pre-/post-conditions should be checked for stability
to ensure the conclusion is a well-defined triple.
We omitted these stability checks from these rules to improve readability.

In practice, however, this way of handling stability has a drawback:
if one starts with a goal that has unstable pre-/post-conditions,
one would only see the mistake much further up in the proof,
typically at applications of \ref{rule:atomicity-weak} or \ref{rule:frame}
(which requires stability of the frames)
just before applications of the axioms.
Therefore, in practice, to make the proof fail early in case of mistakes,
one would want to additionally check stability
at the top-level goal, and applications of \ref{rule:parallel}.

\section{Case Study: Lock-Coupling Set}
\label{appendix:lock-coupling-set}

We develop the proof of a lock-coupling set module, which represents
a set of integer numbers using an ordered linked list.
The module interface presents three operations,
\p{add}, \p{remove} and \p{member} for adding and removing elements from the abstract
set representing the module's state, as well as checking membership of an integer
in this set.

Each cell of the linked list contains either a value from this set or~$\pm \infty$
(representing dummy beginning and end nodes respectively),
a pointer to a lock and a pointer to the next cell of the linked list
(null for the final cell, with value~$\infty$). The values
of the cells in the linked list are sorted in strictly increasing
order.

The value and lock associated with a cell in the linked list
are immutable, however, the module's protocol allows a thread
holding the lock associated with a cell to change the value of
the pointer to the next cell, allowing cells to be added and removed
from the linked list.

The internal operation \p{locate} performs a traversal of the
linked list using hand-over-hand locking so as to, given
some value $v$, find and lock the two adjacent cells  with values
$v'$ and $v''$ such that $v' < v \le v''$.
All the operations would use \p{locate} to obtain ownership
of the nodes that they need to modify.

To perform this hand-over-hand locking, the \p{locate} operation must hold
the lock associated with a cell while locking the lock associated with
the next, therefore the layers of the locks associated with each cell
of the linked list must strictly decrease as the list is traversed.

As we explained in \cref{sec:lock-coupling-set},
the example is challenging for the handling of layers.
Intuitively, we want to associate layers with each lock in the list,
in strictly decreasing order.
This represent the dependecies between the locks introduced by the
order of the traversal:
the release of lock at position~$i$ from the head depends on the liveness
of the lock at position~$i+1$.
This introduces two challenges: we need to associate different layers to each instance of a lock while the lock specifications mention fixed layers;
and we need to dynamically reassign layers to locks as the list grows.
As we already anticipated, we can solve both challenges by a suitable
generalisation of the lock specifications.
Let us first introduce this generalisation formally, and then use it for the proof of the lock-coupling set.

\subsection{Interlude: a Generalisation of Fair Lock Specifications}

We generalise the fair lock specifications we used for the CLH lock
in three ways:
\begin{enumerate}
  \item we parametrise the specifications with client-definable layer maps;
  \item we provide a viewshift to the client with which it is possible to reassign layers;
  \item we add the \p{deleteLock} operation since the lock-coupling set's \p{remove} operation disposes of the removed cells; we omit its implementation and proof as it is standard.
\end{enumerate}

First let us recall the definition of a layer map.
Given two partial orders
  $(\Layer_1, \layleq_1, \layTop_1, \layBot_1)$
and
  $(\Layer_2, \layleq_2, \layTop_2, \layBot_2)$,
a function $ \laymap \from \Layer_1 \to \Layer_2  $
is \emph{strictly monotone} if
$
  \forall m,n\in \Layer_1 \st
    m \laylt_1 n \implies \laymap(m) \laylt_2 \laymap(n)
$.
A \emph{layer map} ${\laymap \from \Layer_1 \layto \Layer_2}$ is a strictly monotone function between the two partial orders.

Let $\lvlOf{clh}-2$ be the level of the $\rt[lclh]$ region used in the proof of the CLH lock.
We generalise the client-facing
CLH~lock specifications as follows:
\begin{align*}
  \exists &(\LayerOf{clh}, \layleqOf{clh}, \layTopOf{clh}, \layBotOf{clh}) \st
    \forall \laymap \from \LayerOf{clh} \layto \Layer \st
    \\&\,
\ATRIPLE \laymap(\layTopOf{clh}); \lvlOf{clh} |-
        \A l \in \set{0, 1} \eventually[{\laymap(\layBotOf{clh})}] \set{0}.
          < \ap{P}(s,\pi) | \ap{L}_{\laymap}(s, \pvar{x}, l) >
            \code{lock(x)}
          < \ap{P}(s,\pi) | \ap{L}_{\laymap}(s, \pvar{x}, 1) \land l = 0 >
    \\&
      \ATRIPLE \laymap(\layBotOf{clh}); \lvlOf{clh} |-
        <\ap{L}_{\laymap}(s, \pvar{x}, 1)>
          \code{unlock(x)}
        <\ap{L}_{\laymap}(s, \pvar{x}, 0)>
    \\&
      \TRIPLE \laymap(\layBotOf{clh}); \lvlOf{clh} |-
        {\emp}
          {\p{makeLock()}}
        {\exists s \st \ap{L}_{\laymap}(s, \ret, 0) * \ap{P}(s,1)}
    \\&
      \TRIPLE \laymap(\layBotOf{clh}); \lvlOf{clh} |-
        {\ap{L}_{\laymap}(s, \pvar{x}, \wtv) * \ap{P}(s,1)}
          {\p{deleteLock(x)}}
        {\emp}
\end{align*}
In particular, the abstract predicate
$ \ap{L}_{\laymap}(s, \pvar{x}, l) $
represents a lock resource with abstract identifier~$ s \in \AId[clh] $
(i.e.~a pair of region identifiers; the client will treat this type opaquely),
concrete address~$x \in \Addr$, and abstract state $l\in\set{0,1}$.

Moreover, the specifications would export to the client
the following viewshifts, for every ${\lvl \geq \lvlOf{clh}}$,
and every $\laymap,\laymap' \from \LayerOf{clh} \layto \Layer$:
\begin{align}
  &
  \viewshift \lvl |=
    \ap{L}_{\laymap}(s, x, l)
    * \ap{L}_{\laymap'}(s, x', l')
      => \False
    \label{vshift:L-excl}
  \\&
  \viewshift \lvl |=
    \ap{L}_{\laymap}(s, x, l) * \ap{P}(s,1) =>
    \ap{L}_{\laymap'}(s, x, l) * \ap{P}(s,1)
    \label{vshift:L-laymap}
\end{align}
Note that the naming choice here suggests CLH as the implementation to keep
the discussion grounded,
but the specification would be the same for any other fair lock implementation.

We now sketch the modifications needed to adapt the proof of CLH presented in
\cref{examples:clh_lock} to prove the generalised specification.

First, we pick, just as in \cref{examples:clh_lock},
$\LayerOf{clh} = \Nat \dunion \set{\layBotOf{clh}, \layTopOf{clh}}$.
We then need to parametrise the two regions with
a layer map~$\laymap \from \LayerOf{clh} \layto \Layer$,
for an arbitrary~$\Layer$.
We include it in the regions abstract state:
$\region{clh}{\rid}(\rid', x, \laymap, h, l, o)$
and $\region{lclh}{\rid'}(x, \laymap, h, l, o, t)$.
The abstract predicate for the lock can the be defined as:
\[
\ap{L}_{\laymap}(s, x, l) \is
  \exists \rid,\rid'\st s=(\rid,\rid') \land
  \exists o, h \st
  \region{clh}{\rid}(\rid', x, \laymap, h, l, o) * \guardA{\gEx}{\rid'}
\]

We similarly parametrise every obligation with a layer map as well,
obtaining obligations
$\obl{o}_{\laymap}(o,t)$ and $\obl{p}_{\laymap}(t)$ with layers
$\lay(\obl{o}_{\laymap}(o,t)) = \laymap(\layBotOf{clh})$
$\lay(\obl{p}_{\laymap}(t)) = \laymap(t)$.

The protocol of the regions is extended
by having each transition preserve the layer map.
Before extending the protocol with a transition that can update the layer map,
we motivate the need for fractional permissions by showing what goes wrong without them.
Suppose we just provide a transition, guarded by $\gEx$, to update the current layer map to an arbitrary new one, and define $\ap{P}(s,\pi) = \emp$.
With this protocol it would be impossible to prove the layer-map-altering viewshift~\eqref{vshift:L-laymap}.
The reason lies in the definition of the interpretation of $\rt[clh]$:
\[
  \rInt(\region{clh}{\rid}(\rid', x, \laymap, h, l, o)) \is
  \exists t \in \Nat \st
  \region{lclh}{\rid'}(x, \laymap, h, l, o, t) *
  \guardA{\guard{e}}{\rid'} *
  \smash{\Sep*^{t - 1}_{i=o+1} \envObl{\obl{p}_{\laymap}(i)}{\rid'}}
\]
In the case where $t\neq o$,
which represents the case where there are threads enqueued
waiting to acquire the lock,
the interpretation ensures that the environment will contain obligations
$ \obl{p}_{\laymap}(i) $ for each issued ticket~$i$.
When we try to prove the viewshift, we need to obtain 
$\envObl{\obl{p}_{\laymap'}(i)}{\rid'}$ with the new layer map,
which can be obtained only by creating out-of-thin-air
the corresponding $\locObl{\obl{p}_{\laymap'}(i)}{\rid'}$ resources.
These would be created in the local state, leaving us with
$
  \region{clh}{\rid}(\rid', x, \laymap, h, l, o) * \guardA{\gEx}{\rid'} *
  \smash{\Sep*^{t - 1}_{i=o+1} \locObl{\obl{p}_{\laymap'}(i)}{\rid'}}
$
which cannot be viewshifted to the desired
$ \ap{L}_{\laymap'}(s, x, l) * \ap{P}(s,1) $
since there is no way to get rid of the local obligations.
Conceptually this encodes the following fact:
if we were to remap the layers of the lock when other threads are queued,
the obligations held by those thread would become unfulfillable,
and we would inherit copies of them with the new mapping, which we also would not be able to fulfil on behalf of the other threads.

To resolve this impasse, we need to allow the layer map to be update only
when there is no thread queued to acquire the lock.
This way we would have $t=o$ and so no environment obligation laying around.
We cannot achieve this by exposing the queue in the abstract state of the lock,
however, without loosing the atomicity of the lock specifications.
With the introduction of fractional permissions, giving the right to enqueue to the lock, we can encode the emptiness of the queue by asserting we are the only one with that right.

To achieve this technically, we start by encoding fractional permissions as a guard algebra.
We introduce guards $\guard{f}_\pi$ with the axioms
$\guard{f}_0 = \guardZero$ and
$\guard{f}_{\pi_1 + \pi_2} = \guard{f}_{\pi_1} \guardOp \guard{f}_{\pi_2}$.
We then define the abstract predicate
$
  \ap{P}(s,\pi) =
    \bigl(
      \exists \rid,\rid' \st
        s = (\rid,\rid') \land
        \guardA{\guard{f}_\pi}{\rid'}
    \bigr).
$

For technical reasons explained later we introduce guards $\guard{g}_\pi$
with exactly the same axioms as the $\guard{f}$ guards.
To encode the fact that full permissions imply empty queue,
we adapt the interpretation of $\rt[lclh]$ as follows:
\begin{multline*}
  \rInt(\region{lclh}{\rid'}(x, \laymap, h, l, o, t)) \is
    \exists ns \st \;
      x \mapsto h, \lfun{last}(ns) *
      h \mapsto l * \lfun{ones}(ns) *
      \guardA{\guard{q}(ns, o)}{\rid'} *
      \locObl{\obl{o}_{\laymap}(o, t)}{\rid'} * {} \\
      \exists \pi \st
        \guardA{\guard{f}_{\pi}}{\rid'} *
        \guardA{\guard{r}_{1 - \pi}}{\rid'} *
        (\pi = 0 \dotimplies t = o)
        \land  t - o = \lstLen{ns} \land ns(0) = h
\end{multline*}
From
$
  \region{lclh}{\rid'}(x, \laymap, h, l, o, t) * \guardA{\guard{f}_{1}}{\rid'}
$
we can deduce that $\pi=0$ inside the region interpretation,
and hence $t=o$.

Finally, we add to the protocol of $\rt[lclh]$ the possibility of updating
the layer map when owning full permissions:
\[
  \guard{f}_{1} :
    ((\laymap, h, l, o, t), \oblZero)
      \interfTo
    ((\laymap', h, l, o, t), \oblZero)
\]

The reason for including the $\guard{r}_\pi$ is as follows.
When a thread enqueues on the lock,
it gives up a non-trivial fraction of
the $\guard{f}_\pi$ permission it owns to be able to make $t \neq o$.
When it dequeues, it should get back that fraction;
the $\guard{r}_\pi$ guards are obtained as ``leftovers'' when putting $\guard{f}_\pi$ in the region's interpretation.
Those are proof that the region interpretation has at least $\guard{f}_\pi$ in it when we want to get it back.

Adapting the proof of \cref{examples:clh_lock} to use these generalised definitions is a routine application of standard \tada\ patterns.
The satisfiability of the layer constraints
is preserved by strict monotonicity of layer maps.

\subsection{Correctness of the Lock-Coupling Set}

\paragraph{Code} The implementation of the module's operations is
in \cref{app:lc-set-code} with the implementation of the constructor
\code{makeSet} in \cref{app:lc-set-code-locate}.
We write \code{dealloc(c,3)} for the deallocation of the 3 contiguous cells from
address \code{c}.
The auxiliary operation \code{locate} (also in \cref{app:lc-set-code-locate})
is meant to only be used internally.
The code uses a ``record'' syntax for readability,
desugared as follows:
\begin{align*}
  \pvar{x.lock} & \is \code{[x]} &
  \pvar{x.val} & \is \code{[x + 1]} &
  \pvar{x.next} & \is \code{[x + 2]}
\end{align*}

\begin{figure}[tb]
\small
\begin{tabular}{c@{\hspace{5em}}c@{\hspace{5em}}c}
  {\codefromfile[numbers=left]{examples/lc-set/code-add}} &
  {\codefromfile[numbers=left]{examples/lc-set/code-remove}} &
  {\codefromfile[numbers=left]{examples/lc-set/code-member}}
\end{tabular}
\caption{Implementation of the lock-coupling set operations.}
\label{app:lc-set-code}
\end{figure}

\begin{figure}[tb]
  \small \begin{tabular}{c@{\hspace{5em}}c@{\hspace{5em}}c}
  {\codefromfile[numbers=left]{examples/lc-set/code-make}}&
  {\codefromfile[numbers=left,linerange={1-11}]{examples/lc-set/code-locate}}&
  {\codefromfile[numbers=left,linerange={12-24},firstnumber=last]{examples/lc-set/code-locate}}
  \end{tabular}
  \caption{Implementation of \code{makeSet} and the internal \code{locate} operation.}
\label{app:lc-set-code-locate}
\end{figure}

\paragraph{Specifications}
The abstract predicate $\ap{LCSet}(s, x, S)$ represents a lock-coupling set at address~$x$ abstractly representing the set~$S$.
\begin{align*}
  &\TRIPLE \layBotOf{lc} |-
      {\emp} {\code{makeSet()}}
      {\exists s\st \ap{LCSet}(s, \ret, \emptyset)}
  \\
  &\ATRIPLE \layTopOf{lc} |-
    \A S \in \powerset(\Int).
      <\ap{LCSet}(s, \pvar{x}, S) \land \pvar{e} \in \Int>
        \code{add(x,e)}
      <\ap{LCSet}(s, \pvar{x}, S \cup \set{\pvar{e}})>
  \\
  &\ATRIPLE \layTopOf{lc} |-
    \A S \in \powerset(\Int).
      <\ap{LCSet}(s, \pvar{x}, S) \land \pvar{e} \in \Int>
        \code{remove(x,e)}
      <\ap{LCSet}(s, \pvar{x}, S \setminus \set{\pvar{e}})> 
  \\
  &\ATRIPLE \layTopOf{lc} |-
    \A S \in \powerset(\Int).
      <\ap{LCSet}(s, \pvar{x}, S) \land \pvar{e} \in \Int>
        \code{member(x,e)}
      <\ap{LCSet}(s, \pvar{x}, S) \land \ret = (\pvar{e} \in S)>    
\end{align*}

\paragraph{Region Types}
This proof will utilise two region types: $\region{lcset}{\rid}(\rid', x, hl, shl, S)$ and $\region{lclist}{\rid}(x, hl, shl, ls)$ where
$\rid' \in \RId$,
$x, hl \in \Addr$,
$shl \in \AId[clh]$,
$S \in \powerset(\Int)$,
$ls \in ((\Int \union \set{\infty, -\infty}) \times \set{0,1} \times (\Nat \union \set{\topOf}))^{*}$.
Here $\rid'$, $x$, $hl$ and $shl$ are fixed parameters of both regions.
The lock-coupling set resource is abstractly represented by the predicate
\[
  \ap{LCSet}(s, x, S) \is
  \exists \rid, \rid', hl, shl\st s = (\rid, \rid', hl, shl) \land
    \region{lcset}{\rid}(\rid', x, hl, shl, S) * \guardA{\guard{e}}{\rid}
\]

\paragraph{Guards}
We introduce a number of guards that are used to represent ownership of information
regarding nodes of the linked list.
To ease readability, we will adopt a record notation for tuples
(i.e.~tuples with named positions).
In particular, we will make a record~$n$ with the following information for each node:
  an address ($\fldOf{n}{addr} \in \Addr$),
  a lock address ($\fldOf{n}{lck} \in \Addr$),
  a lock abstract identifier ($\fldOf{n}{lid} \in \AId[clh]$),
  a value ($\fldOf{n}{val} \in \Int \union \set{\infty, -\infty} $), and
  a layer ($\fldOf{n}{lay} \in \Nat \union \set{\topOf}$).
A guard $\guard{c}(vs)$ records the list~$vs$ of values represented by the linked list.
An unlocked node is represented by the guard $\guard{u}(n)$
where~$n$ is a record of the value, lock address/id, layer associated
with the node of the list at address~$\fldOf{n}{addr}$.
So in particular, the cell at~$\fldOf{n}{addr}$
would store the tuple $(\fldOf{n}{lck},\fldOf{n}{val},\fldOf{n}{nxt})$
and we will have the resource $\ap{L}_{\laymap[\fldOf{n}{lay}]}(\fldOf{n}{lck},\fldOf{n}{lid},l)$
associated with its lock (we will explain the layer map $\laymap[\fldOf{n}{lay}]$ when introducing the region interpretations).
A locked node is represented by two guards
$\guard{l}(n, a)$ and $\guard{k}(n, a)$,
following the usual pattern for locks.
These guards additionally store the address~$a$ of the next node
which is stable if we hold the lock at~$\fldOf{n}{lck}$.
Moreover, assuming~$m$ is the node following~$n$,
if we hold the lock at~$\fldOf{n}{lck}$,
we know that all the information in~$m$ is stable
(i.e.~everything but the address of the node following~$m$).
To represent this we make use of a guard~$ \guard{w}(m) $.

The following axioms reflect the operations we desire to perform on the nodes.
For locking/unlocking a non-terminal node, when~$\lvar{vs}' \neq []$:
\begin{multline*}
\guard{c}(\lvar{vs} \lstPlus [\fldOf{n}{val}, \fldOf{m}{val}] \lstPlus \lvar{vs}') \guardOp \guard{u}(n) \guardOp \guard{l}(m, a') = {} \\
\guard{c}(\lvar{vs} \lstPlus [\fldOf{n}{val}, \fldOf{m}{val}] \lstPlus \lvar{vs}') \guardOp \guard{l}(n, \fldOf{m}{addr}) \guardOp \guard{k}(n, \fldOf{m}{addr}) \guardOp \guard{l}(m, a') \guardOp \guard{w}(m)
\end{multline*}

For locking/unlocking the last node:
\[
  \guard{c}(\lvar{vs} \lstPlus v) \guardOp \guard{u}(n) = \guard{c}(\lvar{vs} \lstPlus v) \guardOp \guard{l}(n, \p{null}) \guardOp \guard{k}(n, \p{null})
\]

For inserting a node~$m$ between~$n_1$ and~$n_2$:
\begin{multline*}
  \guard{c}(\lvar{vs} \lstPlus [\fldOf{n_1}{val}, \fldOf{n_2}{val}] \lstPlus \lvar{vs}') \guardOp \guard{l}(n_1, \fldOf{n_2}{addr}) \guardOp \guard{k}(n_1, \fldOf{n_2}{addr}) \guardOp \guard{w}(n_2) = {} \\
  \guard{c}(\lvar{vs} \lstPlus [\fldOf{n_1}{val}, \fldOf{m}{val}, \fldOf{n_2}{val}] \lstPlus \lvar{vs}') \guardOp \guard{l}(n_1, \fldOf{m}{addr}) \guardOp \guard{k}(n_1, \fldOf{m}{addr}) \guardOp
  \guard{u}(m) \guardOp \guard{w}(m)
\end{multline*}

For deleting a node~$m$:
\begin{multline*}
\guard{c}(\lvar{vs} \lstPlus [\fldOf{n}{val}, \fldOf{m}{val}] \lstPlus \lvar{vs}') \guardOp \guard{l}(n, \fldOf{m}{addr}) \guardOp \guard{k}(n, \fldOf{m}{addr}) \guardOp \guard{l}(m, a) \guardOp \guard{k}(m, a) = {} \\
\guard{c}(\lvar{vs} \lstPlus [\fldOf{n}{val}] \lstPlus \lvar{vs}') \guardOp \guard{l}(n, a) \guardOp \guard{k}(n, a)
\end{multline*}

Then the following axioms keep the guard's information for the nodes consistent:
\begin{align*}
  \fldOf{n}{val} \not\in \lvar{vs} &\implies 
    \guardUndef{\guard{c}(\lvar{vs}) \guardOp \guard{k}(n)}
  \\
  \fldOf{n}{val} = \fldOf{n'}{val} &\implies 
    \guardUndef{\guard{k}(n, \wtv) \guardOp \guard{u}(n')}
  \\
  \fldOf{n}{val} = \fldOf{n'}{val} &\implies 
    \guardUndef{\guard{k}(n, \wtv) \guardOp \guard{k}(n', \wtv)}
  \\
  (\fldOf{n}{val} = \fldOf{n'}{val} \land (a\neq a' \lor n\neq n')) &\implies 
  \guardUndef{\guard{k}(n,a) \guardOp \guard{l}(n',a')}
  \\
  (\fldOf{n}{addr}=\fldOf{n'}{addr} \land n\neq n') &\implies 
    \guardUndef{\guard{k}(\wtv, \fldOf{n}{addr}) \guardOp \guard{w}(n) \guardOp \guard{l}(n', \wtv)}
\end{align*}

\paragraph{Layers and Obligations}
We use the layer structure
$\LayerOf{lc} \is (\Nat \union \set{\topOf, \botOf}) \times \LayerOf{clh}$
(where $\forall n \in \Nat \st \topOf > n > \botOf$),
ordered by the lexicographic ordering~$\le$ and
with
$\layTopOf{lc} \is (\topOf, \layTopOf{clh})$ and
$\layBotOf{lc} \is (\botOf, \layBotOf{clh})$.
Roughly, take a non-initial node~$n$ which is at position $\layer$ from the
end of the list; we will associate with it the layer $(\layer,\layTopOf{clh})$,
which is guaranteed to be strictly greater than any layer associated
with the nodes following~$n$ in the list.
Intuitively, no matter what $ \LayerOf{clh} $ has been chosen for the proof of the implementation of locks,
there are enough layers between $(\layer,\layTopOf{clh})$
and $(\layer+1,\layTopOf{clh})$ to allow the proof of the lock of~$n$
not to conflict with the lock of the node ahead.

We construct obligations out of the atoms
$\obl{k}(\layer)$ (representing the ``key'' of the lock associated with the layer $\layer$)
and $\obl{f}(\layer)$ (representing a ``free'' spot at layer $\layer$)
for $\layer \in \Nat \union \set{\topOf}$.
We set
$\lay(\obl{k}(\layer)) \is (\layer, \layBotOf{clh})$
and
$\lay(\obl{f}(\layer)) = \layTopOf{lc}$.
We also define an obligation acting as a ``reservoir'' of atoms:
\begin{align*}
  \obl{r}(\bar\layer) & \is
    \set{\obl{k}(\layer) | \bar\layer \leq \layer \in \Nat}
    \union
    \set{\obl{f}(\layer) | \bar\layer \leq \layer \in \Nat}
  & \lay(\obl{r}(\bar\layer)) &= (\bar\layer, \layBotOf{clh})
\end{align*}
We can always split a pair of~$\obl{f}$ and~$\obl{k}$ atoms from the
reservoir:
$
  \obl{r}(\bar\layer) =
    \obl{r}(\bar\layer+1) \oblOp \obl{k}(\bar\layer)\oblOp \obl{f}(\bar\layer)
$.

\paragraph{Interference protocol}
The guard-labelled transition system of the region $\region{lcset}{\rid}(\rid', x, hl, shl, S)$ is:
\begin{align*}
\guard{e} \; &: \; \forall v \ldotp (S, \oblZero) \interfTo (S \cup \set{v}, \oblZero) \\
\guard{e} \; &: \; \forall v \ldotp (S, \oblZero) \interfTo (S \setminus \set{v}, \oblZero)
\end{align*}and the guard-labelled transition system of the region $\region{lclist}{\rid}(x, hl, shl, ls)$ is:
\begin{align*}
\guard{e} :{}&
    ((-\infty, 0, \topOf) \lstPlus ls, \oblZero) \\
    \interfTo{}&
    ((-\infty, 1, \topOf) \lstPlus ls, \obl{k}(\topOf) \oblOp \obl{f}(\layer))
  \\\guard{k}(n, \wtv) :{}&
    (ls \lstPlus (\fldOf{n}{val}, 1, \layer) \lstPlus (v', 0, \layer') \lstPlus ls', \oblZero)\\
    \interfTo{}&
    (ls \lstPlus (\fldOf{n}{val}, 1, \layer) \lstPlus (v', 1, \layer') \lstPlus ls', \obl{k}(\layer'))
  \\\guard{k}(n, \wtv) :{}&
    (ls \oplus (\fldOf{n}{val}, 1, \layer) \lstPlus (v', 1, \layer') \lstPlus ls', \obl{k}(\layer) \oblOp \obl{f}(\layer^{\dagger}) \oblOp \obl{k}(\layer'))\\
    \interfTo{}&
    (ls \oplus (\fldOf{n}{val}, 0, \layer) \lstPlus (v', 1, \layer^{\dagger}) \lstPlus ls', \obl{k}(\layer^{\dagger}) \oblOp \obl{f}(\layer'))
    && \layer > \layer^{\dagger} > \layer'
  \\\guard{k}(n, \wtv) :{}&
    (ls \oplus (\fldOf{n}{val}, 1, \layer) \lstPlus (v', 1, \layer') \lstPlus ls', \obl{k}(\layer) \oblOp \obl{f}(\layer^{\dagger}) \oblOp \obl{k}(\layer'))\\
    \interfTo{}&
    (ls \oplus (\fldOf{n}{val}, 1, \layer) \lstPlus (v^{\dagger}, 0, \layer^{\dagger}) \lstPlus (v', 1, \layer') \lstPlus ls', \obl{k}(\layer) \oblOp \obl{k}(\layer'))
    && \layer > \layer^{\dagger} > \layer', v < v^{\dagger} < v'
  \\\guard{k}(n, \wtv) :{}&
    (ls \oplus (\fldOf{n}{val}, 1, \layer) \lstPlus (v', 1, \layer') \lstPlus ls', \obl{k}(\layer) \oblOp \obl{f}(\layer^{\dagger}) \oblOp \obl{k}(\layer'))\\
    \interfTo{}&
    (ls \oplus (\fldOf{n}{val}, 1, \layer) \lstPlus ls', \obl{k}(\layer))
    && \\\guard{k}(n, \wtv) :{}&
    (ls \oplus (\fldOf{n}{val}, 1, \layer) \lstPlus (v', 1, \layer') \lstPlus ls', \obl{k}(\layer) \oblOp \obl{f}(\layer^{\dagger}) \oblOp \obl{k}(\layer'))\\
    \interfTo{}&
    (ls \oplus (\fldOf{n}{val}, 1, \layer) \lstPlus (v', 1, \layer') \lstPlus ls', \obl{k}(\layer) \oblOp \obl{k}(\layer'))
    && \layer > \layer^{\dagger} > \layer'
  \\\guard{k}(n, \wtv) :{}&
    (ls \lstPlus (\fldOf{n}{val}, 1, \layer) \lstPlus ls', \obl{k}(\layer)) \\
    \interfTo{}&
    (ls \lstPlus (\fldOf{n}{val}, 0, \layer) \lstPlus ls', \oblZero)
\end{align*}
They represent, in order:
the acquisition of the first lock, obtaining both the key for that lock
and a ``free'' layer spot;
the acquisition of a next lock, obtaining its key;
the release of the previous lock, swapping layer of the next with the free one;
the insertion of a node which gets assigned the free layer between the two adjacent locks held (used by \p{add});
the deletion of a node that also drops the non-needed free layer spot (used by \p{remove});
the drop of a non-needed free layer spot (used by \p{member});
the release of a lock.

\paragraph{Region interpretation}
The lock-coupling set internally
respresents the elements of the set with a lock-coupling
linked list. To represent these in ghost state,
we will use a list of quadruples of each cells value,
the state of the associated lock, as well as its layer
and region identifier. We introduce the predicate~$\lvar{ord}$
which verifies that the value in the list are in strictly increasing
order, while the layers of the associated locks are in strictly
decreasing order:
\[
\lvar{ord}(ls) \is 
\begin{cases}
  \True \CASE ls = [\wtv]\\
  v < v' \land \layer > \layer' \land \lvar{ord}((v', l', \layer'):ls')
  \CASE ls = (v, \_, \layer):(v', l', \layer'):ls'
\end{cases}
\]
We also introduce a function that allows us to extract
the associated set of values from such a list, $\lvar{vals}$,
and a function that similarly allows us to extract
a list of just the values, retaining their order,
$\lvar{lvals}$:
\begin{align*}
  \lvar{vals}(ls) & \is
  \begin{cases}
    \emptyset \CASE ls = [] \\
    \set{v}\dunion \lvar{vals}(ls')
    \CASE ls = (v, \wtv[2]) \lstPlus ls'
  \end{cases} \\
  \lvar{lvals}(ls) & \is
  \begin{cases}
    [] \CASE ls = [] \\
    v \lstPlus \lvar{vals}(ls')
    \CASE ls = (v, \wtv[2]) \lstPlus ls'
  \end{cases}
\end{align*}
The interpretation of the outer region is a straightforward wrapper around
the inner one.
\begin{align*}
  \rInt(\region{lcset}{\rid}(\rid', x, hl, shl, S)) & \is
    \exists ls \st
      \region{lclist}{\rid'}(x, hl, shl, ls) *
      \guardA{\guard{e}}{\rid'} *
      \lfun{envK}_{\rid'}(ls) \land {} \\ &\qquad
      S \dunion \set{-\infty, \infty} = \lvar{vals}(ls) \land
      \lvar{ord}(ls) \land ls = ((-\infty, \_, \topOf) \lstPlus \_)
  \\
  \lfun{envK}_{\rid}(ls) & \is
  \begin{cases}
    \emp \CASE ls = [] \\
    (l = 1 \dotimplies \envObl{\obl{k}(\layer)}{\rid}) * \lfun{envK}_{\rid}(ls') \CASE ls = (v, l, \layer) \lstPlus ls'
  \end{cases}
\end{align*}
As usual, the outer region has two purposes:
hiding internal state so that the operations can be seen as abstractly atomic,
and keeping track of the obligations held by threads.

The interpretation of the inner region encapsulates the concrete
heap cells and the lock-related guards and obligations:
\begin{align*}
  \rInt(\region{lclist}{\rid}(x, hl, \lvar{shl}, \lvar{ls})) & \is
    \exists n_0,l_0,\dots,n_{k+1},l_{k+1}\st \exists \bar\layer\in\Nat\st\\
      &\qquad \guardA{\guard{c}(\lvar{lvals}(\lvar{ls}))}{\rid} *
              \lvar{ls} = [(\fldOf{n_0}{val},l_0,\fldOf{n_0}{lay}), \dots, (\fldOf{n_{k+1}}{val},l_{k+1},\fldOf{n_{k+1}}{lay})] \land {}\\
      &\qquad \fldOf{n_0}{val} = -\infty \land \fldOf{n_{k+1}}{val} = \infty \land \fldOf{n_0}{lay} = \topOf \land \fldOf{n_{k+1}}{lay} = 0 \land{}\\
      &\qquad \fldOf{n_0}{addr} = x \land \fldOf{n_0}{lck} = hl \land \fldOf{n_0}{lid} = \lvar{shl} \land{}\\
      &\qquad \oblA{\obl{r}(\bar{\layer})}{\rid} \land \bar\layer > \fldOf{n_1}{lay} * \lvar{Node}^0_r(n_0,l_0,\fldOf{n_1}{addr}) * {}\\
      &\qquad \lvar{Nodes}_r(\bar\layer, l_0, [(n_1,l_1),\dots,(n_{k+1},l_{k+1})],\bar\ell)
\end{align*}
where the resources associated with each node are described by the following
auxiliary predicates:
\begin{align*}
  \lvar{Node}^0_r(n, l, a) & \is
    (\fldOf{n}{addr} \mapsto \fldOf{n}{lck}, -\infty, a) * \\
    &\qquad
    \ap{L}_{\laymap[1]}(\fldOf{n}{lid}, \fldOf{n}{lck}, l) *
    \exists \pi>0. \ap{P}(\fldOf{n}{lid}, \pi) * \locObl{\obl{f}(\topOf)}{\rid} * {}\\
    &\qquad \bigl(
              (l=0 \land \guardA{\guard{u}(n)}{\rid} * \locObl{\obl{k}(\topOf)}{\rid})
              \lor
              (l=1 \land \guardA{\guard{l}(n,a)}{\rid})
            \bigr)
  \\
  \lvar{Nodes}_r(\layer_p, l_p, [(n, l)]) & \is
    (\fldOf{n}{addr} \mapsto \fldOf{n}{lck}, \fldOf{n}{val}, \p{null}) *
    \lvar{Gaps}_{\rid}(\layer_p,\fldOf{n}{lay}) * {} \\
    &\qquad
    \ap{L}_{\laymap[\fldOf{n}{lay}]}(\fldOf{n}{lid}, \fldOf{n}{lck}, l) *
    \ap{P}(\fldOf{n}{lid}, \ifte{l_p{=}1}{\onehalf}{1}) * \locObl{\obl{f}(\fldOf{n}{lay})}{\rid} * {}\\
    &\qquad \bigl(
              (l=0 \land \guardA{\guard{u}(n)}{\rid} * \locObl{\obl{k}(\fldOf{n}{lay})}{\rid})
              \lor
              (l=1 \land \guardA{\guard{l}(n,\p{null})}{\rid})
            \bigr)
  \\
  \lvar{Nodes}_r(\layer_p, l_p, [(n, l), (\pr{n}, l')] \lstPlus \lvar{ns}) & \is
    (\fldOf{n}{addr} \mapsto \fldOf{n}{lck}, \fldOf{n}{val}, \fldOf{\pr{n}}{addr}) *
    \lvar{Gaps}_{\rid}(\layer_p,\fldOf{n}{lay}) * {} \\
    &\qquad
    \ap{L}_{\laymap[\fldOf{n}{lay}]}(\fldOf{n}{lid}, \fldOf{n}{lck}, l) *
    \ap{P}(\fldOf{n}{lid}, \ifte{l_p{=}1}{\onehalf}{1}) * \locObl{\obl{f}(\fldOf{n}{lay})}{\rid} * {}\\
    &\qquad \bigl(
              (l=0 \land \guardA{\guard{u}(n)}{\rid} * \locObl{\obl{k}(\fldOf{n}{lay})}{\rid})
              \lor
              (l=1 \land \guardA{\guard{l}(n,\fldOf{\pr{n}}{addr})}{\rid})
            \bigr) * {}\\
    & \qquad \lvar{Nodes}_r(\fldOf{n}{lay}, l, (\pr{n}, l') \lstPlus \lvar{ns})
  \\
  \lvar{Gaps}_{\rid}(\layer_1,\layer_2) & \is
    \Sep*_{\layer = \layer_1+1}^{\layer_2-1}
      \pars*{\locObl{\obl{k}(\layer)}{\rid} \lor \locObl{\obl{k}(\layer)\oblOp\obl{f}(\layer)}{\rid}}
\end{align*}
The layer map $\laymap[\layer]$ maps the layers of $\LayerOf{clh}$
to the ones of $\LayerOf{lc}$ as follows:
\[
\laymap[\layer](k) = (\layer, k)
\]

\paragraph{Proof of\/ \code{locate}}
We use the following specification for the internal operation \code{locate}:
\[
\TRIPLE \layTopOf{lc} |-
        {\exists S \ldotp \region{lcset}{\rid}(\rid', \pvar{x}, hl, shl, S)}
        {\code{locate(x,e)}}
        {
            \exists S \st 
            \region{lcset}{\rid}(\rid', \pvar{x}, hl, shl, S) * 
            \mathit{Loc}(r', \p{x}, \p{e},\ret)
        }
\]
where $\mathit{Loc}(r', \p{x}, \p{e}, p)$ represents the ownership
of two adjacent list nodes representing value~$v$ and~$v'$ with
$v < \pvar{e} \le v'$ (where \p{e} is the value we wanted to locate in the list):
\begin{align*}
  \mathit{Loc}(r', \p{x}, \p{e}, p) &\is
    \exists n_1,n_2,n_3,\layer \st \fldOf{n_1}{addr} = p \land {}\\
    &
    \fldOf{n_1}{val} < \pvar{e} \le \fldOf{n_2}{val} \land
    \fldOf{n_1}{lay} > \layer > \fldOf{n_2}{lay} \land {} \\
    &
    \guardA{\guard{k}(n_1, \fldOf{n_2}{addr})}{\rid'} *
    \oblA{\obl{k}(\fldOf{n_1}{lay})}{\rid'} *
    \oblA{\obl{f}(\layer)}{\rid'} * {} \\
    &
    \guardA{\guard{k}(n_2, \fldOf{n_3}{addr})}{\rid'} *
    \oblA{\obl{k}(\fldOf{n_2}{lay})}{\rid'} *
    \guardA{\guard{w}(n_2)}{\rid'} *
    \ap{P}(\fldOf{n_2}{lid}, \onehalf) * {} \\
    &
    (\fldOf{n_3}{addr} \neq \p{null} \dotimplies
      (\guardA{\guard{w}(n_3)}{\rid'} * \ap{P}(\fldOf{n_3}{lid}, \onehalf))
    )
\end{align*}

\begin{mathfig}[p]
  \begin{proofoutline}
  \TITLE{Proof of \code{locate(x,e)}:}
   \CTXT \layTopOf{lc}; \emptyset |- \\
   \PREC{{\exists S \st \region{lcset}{\rid}(\rid', \pvar{x}, hl, shl, S)}} \\
   \CODE{p \:= x;} \\
   \ASSR{{\exists S \st \region{lcset}{\rid}(\rid', \pvar{x}, hl, shl, S)} \land \pvar{p} = \pvar{x}} \\
   \CODE{pl \:= p.lock;} \\
   \ASSR{{\exists S \st \region{lcset}{\rid}(\rid', \pvar{x}, hl, shl, S)} *
     \exists \pi > 0 \st \ap{P}(shl, \pi) * {} \\
     \exists l \st
     \region{lclist}{\rid'}(\pvar{x}, hl, shl, (-\infty, l, \topOf) \lstPlus \_) *
     l = 1 \dotimplies \envObl{\obl{k}(\topOf)}{\rid'} \land
     \pvar{p} = \pvar{x} \land \pvar{pl} = hl
   } \\
   \CODE{lock(pl);} \\
   \ASSR{
     \exists \layer^{\dagger}, \layer', S, c, scl \st
     \region{lcset}{\rid}(\rid', \pvar{x}, hl, shl, S) *
     \guardA{\guard{k}(\pvar{p}, hl, shl, -\infty,\topOf, c)}{\rid'} *
     \oblA{\obl{k}(\topOf)}{\rid'} * {} \\
     \guardA{\guard{w}(c, \_, scl, \wtv, \layer')}{\rid'} *
     \ap{P}(scl, \onehalf) *
     \oblA{\obl{f}(\layer^{\dagger})}{\rid'} \land \topOf > \layer^{\dagger}  > \layer'
   } \\
   \CODE{c \:= p.next;} \\
   \ASSR{
     \exists \layer^{\dagger}, \layer', S, scl \st
     \region{lcset}{\rid}(\rid', \pvar{x}, hl, shl, S) *
     \guardA{\guard{k}(\pvar{p}, hl, shl, -\infty, \topOf, \pvar{c})}{\rid'} *
     \oblA{\obl{k}(\topOf)}{\rid'} * {} \\
     \guardA{\guard{w}(\pvar{c}, \_, scl, \wtv, \layer')}{\rid'} *
     \ap{P}(scl, \onehalf) *
     \oblA{\obl{f}(\layer^{\dagger})}{\rid'} \land \topOf > \layer^{\dagger} > \layer'
   } \\
   \CODE{cl \:= c.lock;} \\
   \ASSR{
     \exists \layer^{\dagger}, \layer', S, scl, v' \st
     \region{lcset}{\rid}(\rid', \pvar{x}, hl, shl, S) *
     \guardA{\guard{k}(\pvar{p}, hl, shl, -\infty,\topOf, \pvar{c})}{\rid'} *
     \oblA{\obl{k}(\topOf)}{\rid'} * {} \\
     \guardA{\guard{w}(\pvar{c}, \pvar{cl}, scl, v',\layer')}{\rid'} *
     \ap{P}(scl, \onehalf) *
     \oblA{\obl{f}(\layer^{\dagger})}{\rid'} \land \topOf > \layer^{\dagger} > \layer' * {} \\
     \exists l \st
     \region{lclist}{\rid'}(\pvar{x}, hl, shl, \_ \lstPlus (v', l, \layer') \lstPlus \_) *
     l = 1 \dotimplies \envObl{\obl{k}(\layer')}{\rid'} 
   } \\
   \CODE{lock(cl);} \\
   \ASSR{
     \exists \layer^{\dagger}, \layer', S, scl, snl, n, v' \st
     \region{lcset}{\rid}(\rid', \pvar{x}, hl, shl, S) *
     \guardA{\guard{k}(\pvar{p}, hl, shl, -\infty, \topOf, \pvar{c})}{\rid'} *
     \oblA{\obl{k}(\topOf)}{\rid'} * {} \\
      \guardA{\guard{w}(\pvar{c}, \wtv[3])}{\rid'} *
     \guardA{\guard{k}(\pvar{c}, \pvar{cl}, scl, v',\layer', n)}{\rid'} *
     \oblA{\obl{k}(\layer')}{\rid'} *
     \ap{P}(scl, \onehalf) * {} \\
     (n \neq \p{null} \dotimplies (\guardA{\guard{w}(n, \_, snl, \_)}{\rid'} * \ap{P}(snl, \onehalf))) * 
     \oblA{\obl{f}(\layer^{\dagger})}{\rid'} \land 
     \topOf > \layer^{\dagger} > \layer'
   } \\
   \CODE{v \:= c.value;} \\
   \ASSR{
     \exists \layer^{\dagger}, \layer', S, scl, snl, n \st
     \region{lcset}{\rid}(\rid', \pvar{x}, hl, shl, S) *
     \guardA{\guard{k}(\pvar{p}, hl, shl, -\infty, \topOf, \pvar{c})}{\rid'} *
     \oblA{\obl{k}(\topOf)}{\rid'} * {} \\
     \guardA{\guard{w}(\pvar{c}, \wtv[4])}{\rid'} *
     \guardA{\guard{k}(\pvar{c}, \pvar{cl}, scl, \pvar{v},\layer', n)}{\rid'} *
     \oblA{\obl{k}(\layer')}{\rid'} *
     \ap{P}(scl, \onehalf) * {} \\
     (n \neq \p{null} \dotimplies (\guardA{\guard{w}(n, \_, snl, \_)}{\rid'} * \ap{P}(snl, \onehalf))) *
     \oblA{\obl{f}(\layer^{\dagger})}{\rid'} \land 
     \topOf > \layer^{\dagger} > \layer'
   } \\
   \ASSR{
     \exists \layer, \layer^{\dagger}, \layer', S, scl, snl, n, v \st
     \region{lcset}{\rid}(\rid', \pvar{x}, hl, shl, S) *
     \guardA{\guard{k}(\pvar{p}, \wtv[2], v,\layer, \pvar{c})}{\rid'} *
     \oblA{\obl{k}(\layer)}{\rid'} * {} \\
     \guardA{\guard{w}(\pvar{c}, \wtv[4])}{\rid'} *
     \guardA{\guard{k}(\pvar{c}, \_, scl, \pvar{v},\layer', n)}{\rid'} *
     \oblA{\obl{k}(\layer')}{\rid'} *
     \ap{P}(scl, \onehalf) * {} \\
     (n \neq \p{null} \dotimplies (\guardA{\guard{w}(n, \_, snl, \wtv[2])}{\rid'} * \ap{P}(snl, \onehalf))) * 
     \oblA{\obl{f}(\layer^{\dagger})}{\rid'} \land 
     \layer > \layer^{\dagger} > \layer'  \land v < \pvar{e} \land v < \pvar{v}
   } \\
   \CODE{while(v < e) \{} \\
   \quad \CODE{pl := p.lock;} \\
   \quad \CODE{c' := c.next;} \\
   \quad \CODE{cl' := c'.lock;} \\
   \quad \CODE{lock}(\pvar{cl'}); \\
   \quad \CODE{v := c'.val;} \\
   \quad \CODE{unlock(pl)}; \\
   \quad \CODE{p := c;} \\
   \quad \CODE{c := c';} \\
   \CODE{\}} \\
   \ASSR{
     \exists \layer, \layer^{\dagger}, \layer', S, scl, snl, c, n, v, v' \st
     \region{lcset}{\rid}(\rid', \pvar{x}, hl, shl, S) *
     \guardA{\guard{k}(\pvar{p}, \wtv[2], v, \layer, c)}{\rid'} *
     \oblA{\obl{k}(\layer)}{\rid'} * {} \\
     \guardA{\guard{w}(c, \wtv[4])}{\rid'} *
     \guardA{\guard{k}(c, \_, scl, v',\layer', n)}{\rid'} *
     \oblA{\obl{k}(\layer')}{\rid'} *
     \ap{P}(scl, \onehalf) * {} \\
     (n \neq \p{null} \dotimplies (\guardA{\guard{w}(n, \_, snl, \wtv[2])}{\rid'} * \ap{P}(snl, \onehalf))) * 
     \oblA{\obl{f}(\layer^{\dagger})}{\rid'} \land 
     \layer > \layer^{\dagger} > \layer'  \land v < \pvar{e} \le v'
   } \\
   \CODE{ret := p;} \\
   \ASSR{
     \exists \layer, \layer^{\dagger}, \layer', S, scl, snl, c, n, v, v' \st
     \region{lcset}{\rid}(\rid', \pvar{x}, hl, shl, S) *
     \guardA{\guard{k}(\ret, \wtv[2], v,\layer, c)}{\rid'} *
     \oblA{\obl{k}(\layer)}{\rid'} * {} \\
     \guardA{\guard{w}(c, \wtv[3])}{\rid'} *
     \guardA{\guard{k}(c, \_, scl, v',\layer', n)}{\rid'} *
     \oblA{\obl{k}(\layer')}{\rid'} *
     \ap{P}(scl, \onehalf) * {} \\
     (n \neq \p{null} \dotimplies (\guardA{\guard{w}(n, \_, snl, \wtv[2])}{\rid'} * \ap{P}(snl, \onehalf))) * 
     \oblA{\obl{f}(\layer^{\dagger})}{\rid'} \land 
     \layer > \layer^{\dagger} > \layer'  \land v < \pvar{e} \le v'
   } \\
\end{proofoutline}   \caption{Proof outline of \code{locate}.}
  \label{fig:lc-set-locate-outline}
\end{mathfig}

\begin{mathfig}[p]
\begin{proofoutline}
  \CTXT \layTopOf{lc}; \emptyset |- \\
  \PREC{
    \exists \layer, \layer^{\dagger}, \layer', S, scl, snl, n, v \st
    \region{lcset}{\rid}(\rid', \pvar{x}, hl, shl, S) *
    \guardA{\guard{k}(\pvar{p}, \wtv[2], v,\layer, \pvar{c})}{\rid'} *
    \oblA{\obl{k}(\layer)}{\rid'} * {} \\
    \guardA{\guard{w}(\pvar{c}, \wtv[4])}{\rid'} *
    \guardA{\guard{k}(\pvar{c}, \_, scl, \pvar{v},\layer', n)}{\rid'} *
    \oblA{\obl{k}(\layer')}{\rid'} *
    \ap{P}(scl, \onehalf) * {} \\
    (n \neq \p{null} \dotimplies (\guardA{\guard{w}(n, \_, snl, \wtv[2])}{\rid'} * \ap{P}(snl, \onehalf))) * 
    \oblA{\obl{f}(\layer^{\dagger})}{\rid'} \land 
    \layer > \layer^{\dagger} > \layer'  \land v < \pvar{e} \land v < \pvar{v}
  } \\
  \CODE{while(v < e) \{} \\
  $\forall \beta \st$ \CTXT \layTopOf{lc}; \emptyset |- \\
  \quad \PREC{
    \exists \layer, \layer^{\dagger}, \layer', S, scl, snl, n, v \st
    \region{lcset}{\rid}(\rid', \pvar{x}, hl, shl, S) *
    \guardA{\guard{k}(\pvar{p}, \wtv[2], v,\layer, \pvar{c})}{\rid'} *
    \oblA{\obl{k}(\layer)}{\rid'} * {} \\
    \guardA{\guard{w}(\pvar{c}, \wtv[4])}{\rid'} *
    \guardA{\guard{k}(\pvar{c}, \_, scl, \pvar{v},\layer', n)}{\rid'} *
    \oblA{\obl{k}(\layer')}{\rid'} *
    \ap{P}(scl, \onehalf) *
    \guardA{\guard{w}(n, \_, snl, \wtv[2])}{\rid'} * \ap{P}(snl, \onehalf) * {} \\
    \oblA{\obl{f}(\layer^{\dagger})}{\rid'} \land 
    \beta \ge \layer > \layer^{\dagger} > \layer' \land v < \pvar{v} < \pvar{e}
  } \\
  \quad \CODE{pl \:= p.lock;} \\
  \quad \CODE{c' \:= c.next;} \\
  \quad \CODE{cl' := c'.lock;} \\
  \quad \ASSR{
    \exists \layer, \layer^{\dagger}, \layer', \layer'', S, scl, snl, v, v'' \st
    \region{lcset}{\rid}(\rid', \pvar{x}, hl, shl, S) *
    \guardA{\guard{k}(\pvar{p}, \pvar{pl}, \_, v,\layer, \pvar{c})}{\rid'} *
    \oblA{\obl{k}(\layer)}{\rid'} * {} \\
    \guardA{\guard{w}(\pvar{c}, \wtv[4])}{\rid'} *
    \guardA{\guard{k}(\pvar{c}, \_, scl, \pvar{v},\layer', \pvar{c'})}{\rid'} *
    \oblA{\obl{k}(\layer')}{\rid'} *
    \ap{P}(scl, \onehalf) *
    \guardA{\guard{w}(\pvar{c'}, \pvar{cl'}, snl, v'',\layer'')}{\rid'} * \ap{P}(snl, \onehalf) * {} \\
    \oblA{\obl{f}(\layer^{\dagger})}{\rid'} *
    \exists l\st
    \region{lclist}{\rid'}(\pvar{x}, hl, shl, \_ \lstPlus (v'', l, \layer'') \lstPlus \_) *
    l = 1 \dotimplies \envObl{\obl{k}(\layer'')}{\rid'} \land {} \\
    \beta \ge \layer > \layer^{\dagger} > \layer' > \layer'' \land \pvar{v} < \pvar{e} \land \pvar{v} < v''
  } \\
  \quad \CODE{lock(cl');} \\
  \quad \ASSR{
    \exists \layer, \layer^{\dagger}, \layer', \layer'', S, scl, snl, n, v, v'' \st
    \region{lcset}{\rid}(\rid', \pvar{x}, hl, shl, S) *
    \guardA{\guard{k}(\pvar{p}, \pvar{pl}, \_, v,\layer, \pvar{c})}{\rid'} *
    \oblA{\obl{k}(\layer)}{\rid'} * {} \\
    \guardA{\guard{w}(\pvar{c}, \wtv[4])}{\rid'} *
    \guardA{\guard{k}(\pvar{c}, \_, scl, \pvar{v},\layer', \pvar{c'})}{\rid'} *
    \oblA{\obl{k}(\layer')}{\rid'} *
    \ap{P}(scl, \onehalf) *
    \guardA{\guard{w}(\pvar{c'}, \pvar{cl'}, snl, v'',\layer'')}{\rid'} * {} \\
    \guardA{\guard{k}(\pvar{c}', \_, snl, v'',\layer'', n)}{\rid'} *
    \oblA{\obl{k}(\layer'')}{\rid'} * \ap{P}(snl, \onehalf) * {} \\
    n \neq \p{null} \dotimplies (\exists snl' \st \guardA{\guard{w}(n, \_, snl', \wtv[2])}{\rid'} * \ap{P}(snl', \onehalf)) *
    \oblA{\obl{f}(\layer^{\dagger})}{\rid'} \land {} \\
    \beta \ge \layer > \layer^{\dagger} > \layer' > \layer'' \land \pvar{v} < \pvar{e} \land \pvar{v} < v''
  } \\
  \quad \CODE{v \:= c'.val;} \\
  \quad \ASSR{
    \exists \layer, \layer^{\dagger}, \layer', \layer'', S, scl, snl, n, v, v', \pvar{v} \st
    \region{lcset}{\rid}(\rid', \pvar{x}, hl, shl, S) *
    \guardA{\guard{k}(\pvar{p}, \pvar{pl}, \_, v,\layer, \pvar{c})}{\rid'} *
    \oblA{\obl{k}(\layer)}{\rid'} * {} \\
    \guardA{\guard{w}(\pvar{c}, \wtv[4])}{\rid'} *
    \guardA{\guard{k}(\pvar{c}, \_, scl, v',\layer', \pvar{c'})}{\rid'} *
    \oblA{\obl{k}(\layer')}{\rid'} *
    \ap{P}(scl, \onehalf) *
    \guardA{\guard{w}(\pvar{c'}, \pvar{cl'}, snl, \pvar{v},\layer'')}{\rid'} * {} \\
    \guardA{\guard{k}(\pvar{c}', \_, snl, \pvar{v},\layer'', n)}{\rid'} *
    \oblA{\obl{k}(\layer'')}{\rid'} * \ap{P}(snl, \onehalf) * {} \\
    n \neq \p{null} \dotimplies (\exists snl' \st \guardA{\guard{w}(n, \_, snl', \wtv[2])}{\rid'} * \ap{P}(snl', \onehalf)) *
    \oblA{\obl{f}(\layer^{\dagger})}{\rid'} \land {} \\
    \beta \ge \layer > \layer^{\dagger} > \layer' > \layer'' \land v' < \pvar{e} \land v' < \pvar{v}
  } \\
  \quad \CODE{unlock(pl);} \\
  \quad \ASSR{
    \exists \layer, \layer^{\dagger}, \layer', \layer'', S, scl, snl, n, v, v', \pvar{v} \st
    \region{lcset}{\rid}(\rid', \pvar{x}, hl, shl, S) * {} \\
    \guardA{\guard{k}(\pvar{c}, \_, scl, v',\layer^{\dagger}, \pvar{c'})}{\rid'} *
    \oblA{\obl{k}(\layer^{\dagger})}{\rid'} *
    \guardA{\guard{w}(\pvar{c'}, \pvar{cl'}, snl, v'',\layer'')}{\rid'} * {} \\
    \guardA{\guard{k}(\pvar{c}', \_, snl, \pvar{v},\layer'', n)}{\rid'} *
    \oblA{\obl{k}(\layer'')}{\rid'} * \ap{P}(snl, \onehalf) * {} \\
    n \neq \p{null} \dotimplies (\exists snl' \st \guardA{\guard{w}(n, \_, snl', \wtv[2])}{\rid'} * \ap{P}(snl', \onehalf)) *
    \oblA{\obl{f}(\layer')}{\rid'} \land {} \\
    \beta \ge \layer > \layer^{\dagger} > \layer' > \layer'' \land v' < \pvar{e} \land v' < \pvar{v}
  } \\
  \quad \CODE{p \:= c;} \\
  \quad \CODE{c \:= c';} \\
  \quad \POST{
    \exists \layer, \layer^{\dagger}, \layer', S, scl, snl, n, v \st
    \region{lcset}{\rid}(\rid', \pvar{x}, hl, shl, S) *
    \guardA{\guard{k}(\pvar{p}, \wtv[2], v,\layer, \pvar{c})}{\rid'} *
    \oblA{\obl{k}(\layer)}{\rid'} * {} \\
    \guardA{\guard{w}(\pvar{c}, \wtv[4])}{\rid'} *
    \guardA{\guard{k}(\pvar{c}, \_, scl, \pvar{v},\layer', n)}{\rid'} *
    \oblA{\obl{k}(\layer')}{\rid'} *
    \ap{P}(scl, \onehalf) * {} \\
    (n \neq \p{null} \dotimplies (\guardA{\guard{w}(n, \_, snl, \wtv[2])}{\rid'} * \ap{P}(snl, \onehalf))) * 
    \oblA{\obl{f}(\layer^{\dagger})}{\rid'} \land 
    \layer > \layer^{\dagger} > \layer'  \land v < \pvar{e} \land v < \pvar{v}
  } \\
  \CODE{\}} \\
  \POST{
     \exists \layer, \layer^{\dagger}, \layer', S, scl, snl, c, n, v, v' \st
     \region{lcset}{\rid}(\rid', \pvar{x}, hl, shl, S) *
     \guardA{\guard{k}(\pvar{p}, \wtv[2], v,\layer, c)}{\rid'} *
     \oblA{\obl{k}(\layer)}{\rid'} * {} \\
     \guardA{\guard{w}(c, \wtv[4])}{\rid'} *
     \guardA{\guard{k}(c, \_, scl, v',\layer', n)}{\rid'} *
     \oblA{\obl{k}(\layer')}{\rid'} *
     \ap{P}(scl, \onehalf) * {} \\
     (n \neq \p{null} \dotimplies (\guardA{\guard{w}(n, \_, snl, \wtv[2])}{\rid'} * \ap{P}(snl, \onehalf))) * 
     \oblA{\obl{f}(\layer^{\dagger})}{\rid'} \land 
     \layer > \layer^{\dagger} > \layer'  \land v < \pvar{e} \le v'
   }
\end{proofoutline} \caption{Details of while loop in \code{locate}.}
\label{fig:lc-set-locate-while}
\end{mathfig}

\begin{mathfig}[tbp]
  \begin{proofoutline}
\CTXT \layTopOf{lc}; \emptyset |- \\
    \PREC{
      \exists \layer, \layer^{\dagger}, \layer', \layer'', S, scl, snl, v, v'' \st
      \region{lcset}{\rid}(\rid', \pvar{x}, hl, shl, S) *
      \guardA{\guard{k}(\pvar{p}, \pvar{pl}, \_, v,\layer, \pvar{c})}{\rid'} *
      \oblA{\obl{k}(\layer)}{\rid'} * {} \\
      \guardA{\guard{w}(\pvar{c}, \wtv[4])}{\rid'} * 
      \guardA{\guard{k}(\pvar{c}, \_, scl, \pvar{v},\layer', \pvar{c'})}{\rid'} *
      \oblA{\obl{k}(\layer')}{\rid'} *
      \ap{P}(scl, \onehalf) *
      \guardA{\guard{w}(\pvar{c'}, \pvar{cl'}, snl, v'',\layer'')}{\rid'} * \ap{P}(snl, \onehalf) * {} \\
      \oblA{\obl{f}(\layer^{\dagger})}{\rid'} *
      \exists l \st
      \region{lclist}{\rid'}(\pvar{x}, hl, shl, \_ \lstPlus (v'', l, \layer'') \lstPlus \_) *
      l = 1 \dotimplies \envObl{\obl{k}(\layer'')}{\rid'} \land {} \\
      \beta \ge \layer > \layer^{\dagger} > \layer' > \layer'' \land \pvar{v} < \pvar{e} \land \pvar{v} < v''
    } \\
    \begin{proofjump*}[rule:exists-elim,rule:atomicity-weak,rule:atomic-exists-elim,rule:lift-atomic,rule:quantify-layer,rule:frame]
      \label{step:locate-lock-hoare}
      \CTXT (\layer', \layTopOf{clh}); \emptyset |- \\
      \A ls, ls' \in \Int^{*}, \layer'', l \in \set{0, 1}. \\
      \PREC<
        \exists l, \layer'' \st
        \region{lclist}{\rid'}(\pvar{x}, hl, shl, \_ \lstPlus (v'', l, \layer'') \lstPlus \_) *
        \ap{P}(snl, \onehalf) * {} \\
        \guardA{\guard{w}(\pvar{c'}, \pvar{cl'}, snl, v'', \layer'')}{\rid'} *
        l = 1 \dotimplies \envObl{\obl{k}(\layer'')}{\rid'} \land {\layer' > \layer''}
      |
        \region{lclist}{\rid'}(\pvar{x}, hl, shl, ls \oplus ((v'', l, \layer'') \lstPlus ls')) *
        \guardA{\guard{e}}{\rid'} * {} \\
        \guardA{\guard{k}(\pvar{c}, \_, scl, \pvar{v},\layer', \pvar{c'})}{\rid'}
        \land {\layer' > \layer''}
      > \\
      \begin{proofjump}[rule:liveness-check]
        \CTXT (\layer', \layTopOf{clh}); \emptyset |- \\
        \A ls, ls' \in \Int^{*}, l \in \set{0, 1} \eventually[(\layer'', \layBotOf{clh})] \set{0}. \\
        \PREC<
          \ap{P}(snl, \onehalf) |
          \region{lclist}{\rid'}(\pvar{x}, hl, shl, ls \oplus ((v'', l, \layer'') \lstPlus ls')) *
          \guardA{\guard{e}}{\rid'} * {} \\
          \guardA{\guard{k}(\pvar{c}, \_, scl, \pvar{v},\layer', \pvar{c'})}{\rid'} *
          \guardA{\guard{w}(\pvar{c'}, \pvar{cl'}, snl, v'',\layer'')}{\rid'} 
          \land {\layer' > \layer''}
        > \\
        \begin{proofjump*}[rule:lift-atomic,rule:quantify-layer,rule:consequence,rule:frame]
          \label{step:locate-lock}
          \CTXT (\layer'', \layTopOf{clh}); \emptyset |- \\
          \A l \in \set{0, 1} \eventually[(\layer'', \layBotOf{clh})] \set{0}. \\
          \PREC<
          \ap{P}(snl, \onehalf) |
          \ap{L}_{\laymap_{\layer''}}(snl, \pvar{cl'}, l)
          > \\
          \CODE{lock(cl');} \\
          \POST<\ap{P}(snl, \onehalf) |
          \ap{L}_{\laymap_{\layer''}}(snl, \pvar{cl'}, 1) \land l = 0
          > 
        \end{proofjump*} \\
        \POST<
          \ap{P}(snl, \onehalf) |
          \exists n \st
          \region{lclist}{\rid'}(\pvar{x}, hl, shl, ls \oplus ((v'', 1, \layer'') \lstPlus ls')) *
          \guardA{\guard{e}}{\rid'} * {} \\
          \guardA{\guard{k}(\pvar{c}, \_, scl, \pvar{v}, \layer', \pvar{c'})}{\rid'} *
          \guardA{\guard{k}(\pvar{c}', \_, snl, v'', \layer'', n)}{\rid'} *
          \oblA{\obl{k}(\layer'')}{\rid'} * {} \\
          n \neq \p{null} \dotimplies (\exists snl' \st \guardA{\guard{w}(n, \_, snl', \wtv[2])}{\rid'} * \ap{P}(snl', \onehalf))
        >
      \end{proofjump} \\
      \POST<
        \exists l, \layer'' \st
        \region{lclist}{\rid'}(\pvar{x}, hl, shl, \_ \lstPlus (v'', l, \layer'') \lstPlus \_) * {} \\
        \ap{P}(snl, \onehalf) *
        \guardA{\guard{w}(\pvar{c'}, \pvar{cl'}, snl, v'',\layer'')}{\rid'} * {} \\
        l = 1 \dotimplies \envObl{\obl{k}(\layer'')}{\rid'} \land {\layer' > \layer''}
      |
        \exists n \st
        \region{lclist}{\rid'}(\pvar{x}, hl, shl, ls \oplus ((v'', 1, \layer'') \lstPlus ls')) *
        \guardA{\guard{e}}{\rid'} * {} \\
        \guardA{\guard{k}(\pvar{c}, \_, scl, \pvar{v},\layer', \pvar{c'})}{\rid'} *
        \guardA{\guard{k}(\pvar{c}', \_, snl, v'',\layer'', n)}{\rid'} *
        \oblA{\obl{k}(\layer'')}{\rid'} * {} \\
        n \neq \p{null} \dotimplies (\exists snl' \st \guardA{\guard{w}(n, \_, snl', \wtv[2])}{\rid'} * \ap{P}(snl', \onehalf))
      >
    \end{proofjump*} \\
    \POST{
      \exists \layer, \layer^{\dagger}, \layer', \layer'', S, scl, snl, n, v, v'' \st
      \region{lcset}{\rid}(\rid', \pvar{x}, hl, shl, S) * {} \\
      \guardA{\guard{k}(\pvar{p}, \pvar{pl}, \_, v,\layer, \pvar{c})}{\rid'} *
      \oblA{\obl{k}(\layer)}{\rid'} *
      \guardA{\guard{w}(\pvar{c}, \wtv[4])}{\rid'} * {} \\
      \guardA{\guard{k}(\pvar{c}, \_, scl, \pvar{v},\layer', \pvar{c'})}{\rid'} *
      \oblA{\obl{k}(\layer')}{\rid'} *
      \ap{P}(scl, \onehalf) * \guardA{\guard{w}(\pvar{c'}, \wtv[4])}{\rid'} * {} \\
      \guardA{\guard{k}(\pvar{c}', \_, snl, v'',\layer'', n)}{\rid'} *
      \oblA{\obl{k}(\layer'')}{\rid'} * \ap{P}(snl, \onehalf) * {} \\
      n \neq \p{null} \dotimplies (\exists snl' \st \guardA{\guard{w}(n, \_, snl', \wtv[2])}{\rid'} * \ap{P}(snl', \onehalf)) *
      \oblA{\obl{f}(\layer^{\dagger})}{\rid'} \land {} \\
      \beta \ge \layer > \layer^{\dagger} > \layer' > \layer'' \land \pvar{v} < \pvar{e} \land \pvar{v} < v''
    }
\end{proofoutline} \caption{Details of \code{lock} in while loop of \code{locate}.\\
\cref{step:locate-lock-hoare} is \explainproofjump{step:locate-lock-hoare}.\\
\cref{step:locate-lock} is \explainproofjump{step:locate-lock}.}
\label{fig:lc-set-locate-livec}
\end{mathfig}

The outline of the proof of \code{locate} is shown in 
\cref{fig:lc-set-locate-outline}.
\Cref{fig:lc-set-locate-while,fig:lc-set-locate-livec} show the details of the derivation for the crucial steps, i.e.~the while loop and the acquisition of a lock.
In the outlines,
we expand the record notation to tuples,
e.g.~$\guard{k}(\fldOf{n}{addr}, \fldOf{n}{lck}, \fldOf{n}{lid}, \fldOf{n}{val}, \fldOf{n}{lay}, a)$.
We detail here the application of \ref{rule:liveness-check}.
The associated environment liveness condition is proved by:
\[
\infer*[right={\ref{rule:envlive}}]{
  \infer*[right={\ref{rule:envlive-case}}]{
    \infer*[right={\ref{rule:envlive-target}}]{
      \forall \alpha \st \VALID \actxt |= T'(\alpha) \implies T
    }{
      \ENVLIVE (\layer', \layTopOf{clh}); \lvl; \actxt |- L(\alpha) : T'(\alpha) -->> T
    } \quad
    \infer*[right={\ref{rule:envlive-quant}}]{
      (\ref{lc-set:live-o})
    }{
      \ENVLIVE (\layer', \layTopOf{clh}); \lvl; \actxt |- L(\alpha) : L_1(\alpha) -->> T }
  }{
    \ENVLIVE (\layer', \layTopOf{clh}); \lvl; \actxt |- L(\alpha) : L(\alpha) -->> T }
}{
  \ENVLIVE (\layer', \layTopOf{clh}); \lvl; \actxt |- L - M ->> T
}
\]
where $L(\alpha) \is L * M(\alpha)$ and
\begin{align*}
M(\alpha) &\is
  \exists l\st
    \region{lclist}{\rid'}(\pvar{x}, hl, shl, \_ \oplus (v'', l, \_) \lstPlus \_) \land \alpha = l
\\
L &\is
\exists l, \layer'' \st
  \region{lclist}{\rid'}(\pvar{x}, hl, shl, \_ \oplus (v'', l, \layer'') \lstPlus \_) * 
  \guardA{\guard{w}(\pvar{c'}, \pvar{cl'}, snl, v'',\layer'')}{\rid'}
  \\&\qquad * 
  l = 1 \dotimplies \envObl{\obl{k}(\layer'')}{\rid} \land \layer' > \layer''
\\
L_{1}(\alpha) &\is
\exists \layer'' \st L'_{\layer''}(\alpha)
\\
L'_{\layer''}(\alpha) &\is
 \region{lclist}{\rid'}(\pvar{x}, hl, shl, \_ \oplus (v'', 1, \layer'') \lstPlus \_) * 
  \guardA{\guard{w}(\pvar{c'}, \pvar{cl'}, snl, v'',\layer'')}{\rid'}
  \\&\qquad * \envObl{\obl{k}(\layer'')}{\rid} \land \layer' > \layer'' \land \alpha = 1
\\
T'(\alpha) &\is
\exists \layer'' \st
 \region{lclist}{\rid'}(\pvar{x}, hl, shl, \_ \oplus (v'', 0, \layer'') \lstPlus \_) * 
 \guardA{\guard{w}(\pvar{c'}, \pvar{cl'}, snl, v'',\layer'')}{\rid'} \land {} \\
 &\qquad \layer' > \layer'' \land \alpha = 0
\end{align*}

\begin{equation}\small
\infer*[right={\ref{rule:envlive-quant}}]{
\infer*[right={\ref{rule:envlive-obl}}]{
  \decr[\actxt](L'_{\layer''}, L, T) \\
  \forall \alpha\st
  \VALID \actxt |= \minLayStrict{L'_{\layer''}(\alpha)}{\lay(\obl{k}(\layer''))} \\
  \forall \alpha \st
  \SAT[\actxt] |- L'_{\layer''}(\alpha) \implies
  {
      \region[\lvlp]{lclist}{\rid'}(\pvar{x}, \pvar{hl}, shl, \_ \lstPlus (v'', 1, \layer'') \lstPlus \_)
      * \envObl{\obl{k}(\layer'')}{\rid} * \True \land
      (\layer', \layTopOf{clh}) \laygt \lay(\obl{k}(\layer''))
  }
}{
  \ENVLIVE (\layer', \layTopOf{clh}); \lvl; \actxt |- L(\alpha) : L'_{\layer''}(\alpha) -->> T
}}{
  \ENVLIVE (\layer', \layTopOf{clh}); \lvl; \actxt |- L(\alpha) : L_{1}(\alpha) -->> T
}
\label{lc-set:live-o}
\end{equation}

\paragraph{Proof of \code{add}}
The proof of the \code{add} operation builds on the specification
of \code{locate}. We show its outline in~\cref{fig:lc-set-add-outline}
with a more detailed derivation showing how the first unlock operation is handled
in \cref{fig:lc-set-add-unlock}.

\begin{mathfig}[ptb]
\begin{proofoutline}
\TITLE{Proof of \code{add(x,e)}:}
  \CTXT \layTopOf{lc}; \emptyset |-
  \A S \in \mathcal{P}(\Int). \\
  \PREC<\ap{LCSet}(s, \pvar{x}, S)> \\
  \begin{proofjump}[rule:consequence,"Sub $s=(r\text{,}r'\text{,}hl)$"]
    \PREC<\region{lcset}{\rid}(\rid', \pvar{x}, hl, S) * \guardA{\guard{e}}{\rid}> \\
    \begin{proofjump}[rule:make-atomic]
      \CTXT \layTopOf{lc};
            \actxt = {\map{\rid->(\powerset(\Int),\layBot,\powerset(\Int), \set{((S, \oblZero),(S\union\set{\pvar{e}}, \oblZero)) | S\subseteq \Int} )}}
            |- \\
      \PREC{\exists S \st
        \region{lcset}{\rid}(\rid', \pvar{x}, hl, S) * \done{\rid}{\blacklozenge}
      } \\
      \CODE{p \:= locate(x, e);} \\
      \ASSR{
        \exists \layer, \layer^{\dagger}, \layer', S, scl, snl, c, n, v, v' \st
        \region{lcset}{\rid}(\rid', \pvar{x}, hl, shl, S) * \done{\rid}{\blacklozenge} *
        \guardA{\guard{k}(\pvar{p}, \wtv[2], v,\layer, c)}{\rid'} *
        \oblA{\obl{k}(\layer)}{\rid'} * {} \\
        \guardA{\guard{w}(c, \wtv[4])}{\rid'} *
        \guardA{\guard{k}(c, \_, scl, v',\layer', n)}{\rid'} *
        \oblA{\obl{k}(\layer')}{\rid'} *
        \ap{P}(scl, \onehalf) * {} \\
        (n \neq \p{null} \dotimplies (\guardA{\guard{w}(n, \_, snl, \wtv[2])}{\rid'} * \ap{P}(snl, \onehalf))) * 
        \oblA{\obl{f}(\layer^{\dagger})}{\rid'} \land 
        \layer > \layer^{\dagger} > \layer'  \land v < \pvar{e} \le v'
      } \\
      \begin{proofjump}[rule:exists-elim]
\ASSR{
          \exists S \st
          \region{lcset}{\rid}(\rid', \pvar{x}, hl, shl, S) * \done{\rid}{\blacklozenge} *
          \guardA{\guard{k}(\pvar{p}, \wtv[2], v, \layer, c)}{\rid'} *
          \oblA{\obl{k}(\layer)}{\rid'} * {} \\
          \guardA{\guard{w}(c, \wtv[4])}{\rid'} *
          \guardA{\guard{k}(c, \_, scl, v',\layer', n)}{\rid'} *
          \oblA{\obl{k}(\layer')}{\rid'} *
          \ap{P}(scl, \onehalf) * {} \\
          (n \neq \p{null} \dotimplies (\guardA{\guard{w}(n, \_, snl, \wtv[2])}{\rid'} * \ap{P}(snl, \onehalf))) *
          \oblA{\obl{f}(\layer^{\dagger})}{\rid'} \land 
          \layer > \layer^{\dagger} > \layer'  \land v < \pvar{e} \le v'
        } \\
        \begin{proofjump}[rule:frame]
\ASSR{
            \exists S \st
            \region{lcset}{\rid}(\rid', \pvar{x}, hl, shl, S) * \done{\rid}{\blacklozenge} *
            \guardA{\guard{k}(\pvar{p}, \wtv[2], v,\layer, c)}{\rid'} *
            \oblA{\obl{k}(\layer)}{\rid'} *
            \guardA{\guard{w}(c, \wtv[4])}{\rid'} * {} \\
            \guardA{\guard{k}(c, \_, scl, v', n)}{\rid'} *
            \oblA{\obl{k}(\layer')}{\rid'} *
            \ap{P}(scl, \onehalf) *
            \oblA{\obl{f}(\layer^{\dagger})}{\rid'} \land 
            \layer > \layer^{\dagger} > \layer'  \land v < \pvar{e} \le v'
          } \\
          \CODE{c \:= p.next;} \\
\CODE{v \:= c.val;} \\
          \ASSR{
            \exists S \st
            \region{lcset}{\rid}(\rid', \pvar{x}, hl, shl, S) * \done{\rid}{\blacklozenge} *
            \guardA{\guard{k}(\pvar{p}, \wtv[2], v,\layer, \pvar{c})}{\rid'} *
            \oblA{\obl{k}(\layer)}{\rid'} *
            \guardA{\guard{w}(\pvar{c}, \wtv[4])}{\rid'} * {} \\
            \guardA{\guard{k}(\pvar{c}, \_, scl, v',\layer', n)}{\rid'} *
            \oblA{\obl{k}(\layer')}{\rid'} *
            \ap{P}(scl, \onehalf) *
            \oblA{\obl{f}(\layer^{\dagger})}{\rid'} \land \layer > \layer^{\dagger} > \layer'  \land v < \pvar{e} \le v' \land \pvar{v} = v'
          } \\
          \CODE{if(v != e) \{} \\
          \begin{proofindent}
            \ASSR{
              \exists S \st
              \region{lcset}{\rid}(\rid', \pvar{x}, hl, shl, S) * \done{\rid}{\blacklozenge} *
              \guardA{\guard{k}(\pvar{p}, \wtv[2], v,\layer, \pvar{c})}{\rid'} *
              \oblA{\obl{k}(\layer)}{\rid'} *
              \guardA{\guard{w}(\pvar{c}, \wtv[4])}{\rid'} * {} \\
              \guardA{\guard{k}(\pvar{c}, \_, scl, v',\layer', n)}{\rid'} *
              \oblA{\obl{k}(\layer')}{\rid'} *
              \ap{P}(scl, \onehalf) *
              \oblA{\obl{f}(\layer^{\dagger})}{\rid'} \land 
              \layer > \layer^{\dagger} > \layer'  \land v < \pvar{e} < v'
            } \\
            \CODE{n \:= alloc(3);} \\
            \CODE{nl \:= makeLock();} \\
            \CODE{n.lock \:= nl;} \\
            \CODE{n.val \:= e;} \\
            \CODE{n.next \:= c;} \\
            \ASSR{
              \exists S \st
              \region{lcset}{\rid}(\rid', \pvar{x}, hl, shl, S) * \done{\rid}{\blacklozenge} *
              \guardA{\guard{k}(\pvar{p}, \wtv[2], v,\layer, \pvar{c})}{\rid'} *
              \oblA{\obl{k}(\layer)}{\rid'} *
              \guardA{\guard{w}(\pvar{c}, \wtv[4])}{\rid'} * {} \\
              \guardA{\guard{k}(\pvar{c}, \_, scl, v',\layer', n)}{\rid'} * 
              \oblA{\obl{k}(\layer')}{\rid'} *
              \ap{P}(scl, \onehalf) *
              \oblA{\obl{f}(\layer^{\dagger})}{\rid'} * {} \\
              \exists s \st \pvar{n} \mapsto \pvar{nl}, \pvar{e}, \pvar{c} *
              \ap{L}_{\laymap_{\layer^{\dagger}}}(s, \pvar{nl}, 0) * \ap{P}(s, 1) \land 
              \layer > \layer^{\dagger} > \layer'  \land v < \pvar{e} < v'
            } \\
            \CODE{p.next \:= n;} \\
            \ASSR{
              \exists S, S', snl \st
              \region{lcset}{\rid}(\rid', \pvar{x}, hl, shl, S') * \done{\rid}{(S, S \union \set{\pvar{e}})} *
              \guardA{\guard{k}(\pvar{p}, \wtv[2], v,\layer, \pvar{n})}{\rid'} *
              \oblA{\obl{k}(\layer)}{\rid'} * {} \\
              \guardA{\guard{w}(\pvar{n}, \_, snl, \wtv[2])}{\rid'} *
              \ap{P}(snl, \onehalf) *
              \guardA{\guard{k}(\pvar{c}, \_, scl, v',\layer', n)}{\rid'} *
              \oblA{\obl{k}(\layer')}{\rid'} \land
              \layer > \layer' 
            }
          \end{proofindent} \\
          \CODE{\}} \\
          \ASSR{
            \exists S, S', snl \st
            \region{lcset}{\rid}(\rid', \pvar{x}, hl, shl, S') * \done{\rid}{(S, S \union \set{\pvar{e}})} *
            \guardA{\guard{k}(\pvar{p}, \wtv[2], v,\layer, \pvar{n})}{\rid'} *
            \oblA{\obl{k}(\layer)}{\rid'} * {} \\
            \guardA{\guard{w}(\pvar{n}, \_, snl, \wtv[2])}{\rid'} *
            \ap{P}(snl, \onehalf) 
            \guardA{\guard{k}(\pvar{c}, \_, scl, v',\layer', n)}{\rid'} *
            \oblA{\obl{k}(\layer')}{\rid'} \land
            \layer > \layer' 
          } \\
          \CODE{pl \:= p.lock;} \\
          \CODE{cl \:= c.lock;} \\
          \ASSR{
            \exists S, S', snl \st
            \region{lcset}{\rid}(\rid', \pvar{x}, hl, shl, S') * \done{\rid}{(S, S \union \set{\pvar{e}})} *
            \guardA{\guard{k}(\pvar{p}, \pvar{pl}, \_, v,\layer, \pvar{n})}{\rid'} *
            \oblA{\obl{k}(\layer)}{\rid'} * {} \\
            \guardA{\guard{w}(\pvar{n}, \_, snl, \wtv[2])}{\rid'} *
            \ap{P}(snl, \onehalf) 
            \guardA{\guard{k}(\pvar{c}, \pvar{cl}, scl, v',\layer', n)}{\rid'} *
            \oblA{\obl{k}(\layer')}{\rid'} \land
            \layer > \layer' 
          }
        \end{proofjump} \\
        \ASSR{
          \exists S, S, snl' \st
          \region{lcset}{\rid}(\rid', \pvar{x}, hl, shl, S') * \done{\rid}{(S, S \union \set{\pvar{e}})} * {} \\
          \guardA{\guard{k}(\pvar{p}, \pvar{pl}, \_, v,\layer, \pvar{n})}{\rid'} *
          \oblA{\obl{k}(\layer)}{\rid'} *
          \guardA{\guard{w}(\pvar{n}, \_, snl', \wtv[2])}{\rid'} *
          \ap{P}(snl', \onehalf) * {} \\
          \guardA{\guard{k}(\pvar{c}, \pvar{cl}, \_, v',\layer', n)}{\rid'} *
          \oblA{\obl{k}(\layer')}{\rid'} * {} \\
          (n \neq \p{null} \dotimplies (\guardA{\guard{w}(n, \_, snl, \wtv[2])}{\rid'} * \ap{P}(snl, \onehalf))) \land
          \layer > \layer' 
        } \\
        \CODE{unlock(cl);} \\
        \ASSR{
          \exists S, S', snl' \st
          \region{lcset}{\rid}(\rid', \pvar{x}, hl, shl, S') * \done{\rid}{(S, S \union \set{\pvar{e}})} * {} \\
          \guardA{\guard{k}(\pvar{p}, \pvar{pl}, \_, v,\layer, \pvar{n})}{\rid'} *
          \oblA{\obl{k}(\layer)}{\rid'} * 
          \guardA{\guard{w}(\pvar{n}, \_, snl', \wtv[2])}{\rid'} *
          \ap{P}(snl', \onehalf)             
        } \\
        \CODE{unlock(pl);} \\
        \POST{
          \exists S \st
          \done{\rid}{(S, S \cup \set{\pvar{e}})}
        }
      \end{proofjump} \\
      \POST{
          \exists S \st
          \done{\rid}{(S, S \cup \set{\pvar{e}})}
      }
    \end{proofjump} \\
    \POST<\region{lcset}{\rid}(\rid', hl, x, S \cup \set{\pvar{e}}) * \guardA{\guard{e}}{\rid}>
  \end{proofjump} \\
  \POST<\ap{LCSet}(s, \pvar{x}, S \cup \set{\pvar{e}})>
\end{proofoutline} \caption{Proof outline of \code{add} operation.}
\label{fig:lc-set-add-outline}
\end{mathfig}

\begin{mathfig}[ptb]
\begin{proofoutline}
  \CTXT \layTopOf{lc}; \actxt |- \\
  \PREC{
    \exists S, S', snl' \st
    \region{lcset}{\rid}(\rid', \pvar{x}, hl, shl, S') * \done{\rid}{(S, S \union \set{\pvar{e}})} * {} \\
    \guardA{\guard{k}(\pvar{p}, \pvar{pl}, \_, v,\layer, \pvar{n})}{\rid'} *
    \oblA{\obl{k}(\layer)}{\rid'} *
    \guardA{\guard{w}(\pvar{n}, \_, snl', \wtv[2])}{\rid'} *
    \ap{P}(snl', \onehalf){} \\
    \guardA{\guard{k}(\pvar{c}, \pvar{cl}, \_, v',\layer', n)}{\rid'} *
    \oblA{\obl{k}(\layer')}{\rid'} * {} \\
    (n \neq \p{null} \dotimplies (\guardA{\guard{w}(n, \_, snl, \wtv[2])}{\rid'} * \ap{P}(snl, \onehalf))) \land
    \layer > \layer' 
  } \\
  \begin{proofjump}[rule:quantify-layer,rule:consequence,rule:frame,rule:exists-elim]
    \CTXT (\layer', \layBotOf{clh}); \actxt |- \\
    \PREC{
      \exists S \st
      \region{lcset}{\rid}(\rid', \pvar{x}, hl, shl, S) *
      \guardA{\guard{k}(\pvar{c}, \pvar{cl}, \_, v',\layer', n)}{\rid'} *
      \oblA{\obl{k}(\layer')}{\rid'} * {}\\
      (n \neq \p{null} \dotimplies (\guardA{\guard{w}(n, \_, snl, \wtv[2])}{\rid'} * \ap{P}(snl, \onehalf)))
    } \\
    \begin{proofjump}[rule:atomic-exists-elim]
      \CTXT (\layer', \layBotOf{clh}); \actxt |- \\
      \A S \in \mathcal{P}(\Int) . \\
      \PREC<
      \region{lcset}{\rid}(\rid', \pvar{x}, hl, shl, S) *
      \guardA{\guard{k}(\pvar{c}, \pvar{cl}, \_, v',\layer', n)}{\rid'} *
      \oblA{\obl{k}(\layer')}{\rid'} * {}\\
      (n \neq \p{null} \dotimplies (\guardA{\guard{w}(n, \_, snl, \wtv[2])}{\rid'} * \ap{P}(snl, \onehalf)))
      > \\
      \begin{proofjump}[rule:lift-atomic,rule:atomic-exists-elim,rule:frame]
        \CTXT (\layer', \layBotOf{clh}); \actxt |- \\
        \A ls, ls' \in ((\Int \dunion \set{-\infty,\infty}) \times \set{0,1} \times \Nat)^{*} . \\
        \PREC<
        \region{lclist}{\rid'}(\pvar{x}, hl, shl, ls \lstPlus (v', 1, \layer') \lstPlus ls') *
        \guardA{\guard{k}(\pvar{c}, \pvar{cl}, \_, v',\layer', n)}{\rid'} *
        \oblA{\obl{k}(\layer')}{\rid'} *  {} \\
        (n \neq \p{null} \dotimplies (\guardA{\guard{w}(n, \_, snl, \wtv[2])}{\rid'} * \ap{P}(snl, \onehalf)))
        > \\
        \begin{proofjump*}[rule:lift-atomic,rule:consequence,rule:frame]\label{step:add-unlock-lift}\CTXT (\layer', \layBotOf{clh}); \actxt |- \\
          \PREC<\ap{L}_{\laymap_{\layer'}}(s, la, 1)> \\
          \CODE{unlock(cl);} \\
          \PREC<\ap{L}_{\laymap_{\layer'}}(s, la, 0)>
        \end{proofjump*} \\
        \POST<
        \region{lclist}{\rid'}(\pvar{x}, hl, shl, ls \lstPlus (v', 0, \layer') \lstPlus ls')
        > \\
      \end{proofjump} \\
      \POST<\region{lcset}{\rid}(\rid', \pvar{x}, hl, shl, S)>
    \end{proofjump} \\
    \POST{\exists S \st \region{lcset}{\rid}(\rid', \pvar{x}, hl, shl, S)}
  \end{proofjump} \\
  \POST{
    \exists S, S', snl' \st
    \region{lcset}{\rid}(\rid', \pvar{x}, hl, shl, S') * \done{\rid}{(S, S \union \set{\pvar{e}})} * {} \\
    \guardA{\guard{k}(\pvar{p}, \pvar{pl}, \_, v,\layer, \pvar{n})}{\rid'} *
    \oblA{\obl{k}(\layer)}{\rid'} *
    \guardA{\guard{w}(\pvar{n}, \_, snl', \wtv[2])}{\rid'} *
    \ap{P}(snl', \onehalf)
  }
\end{proofoutline} \caption{Details of \code{unlock(cl)} in \code{add}.
  \cref{step:add-unlock-lift} is \explainproofjump{step:add-unlock-lift}.}
\label{fig:lc-set-add-unlock}
\end{mathfig}

\paragraph{Proof of \code{makeSet}, \code{member} and \code{remove}}
We omit the proofs of the \code{makeSet}, \code{member} and \code{remove} operations
as they do not add much to the presentation.
\code{makeSet} can be proved as standard by keeping track of the nodes
created locally and with a final viewshift to create the two nested regions
representing an empty set.
The hard part of the proofs of \code{member} and \code{remove} is the call
to \code{locate} which has been already presented in detail.
The rest is handled analogously to \code{add}.

\clearpage

\section{Programming Language Definition}
\label{app:command-semantics}

We will make regular use of partial functions.
We write $X \pto Y$
for the set of
partial function from $X$ to $Y$, and  $X \finpto Y$ for the set of finite
partial function.
Given $f \from X \pto Y$, we write $f(x) = \bot$  if $f$ is undefined
on $x$,  and
$\dom(f) \is \set{x | f(x) \neq \bot}$.
We will use the notation $\map{x_1 -> y_1;;x_n -> y_n}$
for the finite function that maps each of the $x_i$ to $y_i$
and is undefined on any other input.
Given elements $x\in X$ and  $y \in Y$, and functions
  $f\from X \pto Y$ and
  $g\from X' \pto Y'$,
we define the functions $f\map{x -> y}$ and  $ f\dunion g$ by:
\begin{align*}
  (f\map{x -> y}) (z) & \is
      \begin{cases}
        y \CASE z = x\\
        f(z) \OTHERWISE
      \end{cases}
  \\
 (f\dunion g)(x)  & \is
      \begin{cases}
        f(x) \CASE x \in \dom(f)\\
        g(x) \CASE x \in \dom(g)
      \end{cases}
  & \text{if} \dom(f)\inters\dom(g) = \set{}
\end{align*}
We write $f\map{x->\bot}$ for the partial function that is undefined on $x$
but otherwise behaves  like $f$.
The union of two partial function $f \union g$ is a well-defined partial function as long as
$f(x) = g(x)$ where their domains overlap.

We use
the \emph{set of Booleans},
  $\Bool \is \set{\True,\False} \ni b, b_1, b_2$,
a \emph{set of values},
  $\Val \is  \Int \cup \Bool \ni v, v_1, v_2, \cdots$,
a  \emph{set of program variables},
  $\PVar \ni \pvar{x}, \pvar{y}, \cdots$,
and a \emph{set of function names},
  $\FName \ni \code{f}, \code{g}, \cdots$.
The set $\PVar$ contains a special element,
$\pvar{ret}$, that holds a function's return value.
Heap addresses are represented by natural numbers, $\Addr \is \Nat$.
The natural numbers in $\Val$ represent both numeric values and heap addresses.

\begin{definition}[Numeric and Boolean Expressions]
\label{def:expressions}
  Let $\Type{Vars}$ be an arbitrary  set of variables,
  and $\Type{Values}$ and arbitrary set of values.
  The \emph{set of numerical expressions},
    $\Exp{\Type{Vars},\Type{Values}} \ni \vexp, \vexp_1, \vexp_2, \cdots$,
  and the \emph{set of boolean expressions},
    $\BExp{\Type{Vars},\Type{Values}} \ni \bexp, \bexp_1, \bexp_2, \cdots$,
  are defined by the grammars:

  \begin{grammar}
    \vexp \is v
            | x
            | \vexp + \vexp
            | \vexp - \vexp
            | \vexp * \vexp
            | \cdots
    & \text{where } v \in \Type{Values}, x \in \Type{Vars}
    \\
    \bexp \is b
            | x
            | \neg \bexp
            | \bexp \land \bexp
            | \vexp = \vexp
            | \vexp < \vexp
            | \cdots
    & \text{where } b \in \Bool, x \in \Type{Vars}
  \end{grammar}

  The numeric and Boolean program expressions  are defined
  by the sets $\Exp{\PVar,\Val}$ and $\BExp{\PVar,\Val}$ respectively.
  In \cref{sec:assertions}, we
  also work with logical expressions built from both program and
  logical variables and values,
  hence the reason for the expression
  definition defined over an arbitrary variable and value sets.

  The functions $\fve$ and $\fvb$ provide  the sets of free variables for
  the  numeric and Boolean  expressions respectively. They are defined
  inductively on the structure of expressions by:

  \vspace{\abovedisplayskip}\par \noindent
  \begin{minipage}{.5\linewidth}
  \[
  \begin{array}{l@{}l}
    \fve(v) = \emptyset & \mathllap{v \in \Type{Values} }\\
    \fve(\pvar{x}) = \set{\pvar{x}} & \mathllap{\pvar{x} \in \Type{Vars}}\\
    \fve(\vexp_{1} + \vexp_{2}) = \fve(\vexp_{1}) \cup \fve(\vexp_{2}) & \\
    \fve(\vexp_{1} - \vexp_{2}) = \fve(\vexp_{1}) \cup \fve(\vexp_{2}) & \\
    \fve(\vexp_{1} * \vexp_{2}) = \fve(\vexp_{1}) \cup \fve(\vexp_{2}) & \\
    \cdots
  \end{array}
  \]
  \end{minipage}\begin{minipage}{.5\linewidth}
  \[
  \begin{array}{l@{}l}
    \fvb(b) = \emptyset & \mathllap{b \in \set{\True, \False} }\\
    \fvb(\pvar{x}) = \set{\pvar{x}} & \mathllap{\pvar{x} \in \Type{Vars}}\\
    \fvb(\neg \bexp) = \fvb(\bexp) & \\
    \fvb(\bexp_{1} \land \bexp_{2}) = \fvb(\bexp_{1}) \cup \fvb(\bexp_{2}) & \\
    \fvb(\vexp_{1} = \vexp_{2}) = \fve(\vexp_{1}) \cup \fve(\vexp_{2}) & \\
    \fvb(\vexp_{1} < \vexp_{2}) = \fve(\vexp_{1}) \cup \fve(\vexp_{2}) & \\
    \cdots
  \end{array}
  \]
  \end{minipage}
\end{definition}

\begin{mathfig}
  \[
    \begin{array}{r@{\;}c@{\;}l@{\hspace{5em}}l}
      \cmd  &::=  & \acode{skip}                          & (\text{skip})\\
      & | & \acode{x:=EXP}                                & (\text{assignment})\\
      & | & \acode{x:=[EXP]}                              & (\text{read})\\
      & | & \acode{[EXP]:=EXP}                            & (\text{write})\\
      & | & \acode{x:=CAS(EXP,EXP,EXP)}                   & (\text{compare-and-swap})\\
      & | & \acode{x:=FAS(EXP,EXP,EXP)}                   & (\text{fetch-and-set})\\
      & | & \acode{x:=alloc(EXP)}                         & (\text{allocate})\\
      & | & \acode{dealloc(EXP)}                          & (\text{deallocate})\\
      & | & \acode{CMD;CMD}                               & (\text{sequential composition})\\
      & | & \cmd \parallel \cmd                           & (\text{parallel composition}) \\
      & | & \acode{let f(VARS)=CMD in CMD}                & (\text{function definition})\\
      & | & \acode{var x = EXP in CMD}                    & (\text{local variable binding}) \\
      & | & \acode{if(BEXP)\{CMD\}else\{CMD\}}            & (\text{if}) \\
      & | & \acode{while(BEXP)\{CMD\}}                    & (\text{while loop})\\
      & | & \acode{x:=f($\vec{\vexp}$)}                   & (\text{function call})\\
      & | & \acode{<<CMD>>}                               & (\text{primitive atomic block})
    \end{array}
  \]
  \caption{Syntax of commands}
  \label{app:fig:commands}
\end{mathfig}

\begin{definition}[Commands]
  The \emph{set of commands}, $\Cmd \ni \cmd$, is defined by the grammar
  in \cref{app:fig:commands} where
    $\vexp \in \Exp{\PVar,\Val}$,
    $\bexp \in \BExp{\PVar,\Val}$,
    $\pvar{x} \in \PVar$,
    $\pvars{x} \in \PVar^*$ is a list of pairwise distinct variables, and
    $\p{f} \in \FName$.
\end{definition}

We use [$\vexp$] to denote the value of the heap cell with address
given by $\vexp$.  In \cref{fig:freevars-mods}, we define operators
$\fv$ and $\mods$, which identify the variables that a command can
access and the variables that are potentially modified by a command,
respectively.
In a command $\cmd_1 \parallel \cmd_2$, we apply a
strong syntactic restriction that
${\mods(\cmd_{1}) = \mods(\cmd_{2}) = \emptyset}$. Each individual
thread is still able to modify variables that are created locally and
to modify shared heap cells, but are not allowed to modify the free
variables.\footnote{To lift this restriction, one could use standard techniques,
  such as ``variables as resources''~\cite{BornatCY06}.
  Our restriction minimises the noise
  generated by handling local state in the formalisation
  of the model and the assertions.
  Note that expressivity is not really limited by our restriction:
  any local variable in the scope common to both threads, that needs to be modified,
  can be instead implemented by using a shared memory cell.
}
In a function definition
$\acode{let f(x$_1$,$\ldots$,x$_n$)=$\cmd_1$ in\ $\cmd_2$}$,
we use the natural restriction
$\fv(\cmd_1) \subseteq \set{ \pvar{x}_1,\dots,\pvar{x}_n, \pvar{ret}}$.
Also for simplicity, we assume each function name is given a definition at most once.
The function $\funnames\from\Cmd \to \powerset(\FName)$
returns the function names occurring in $\Cmd$
that are not bound by a \code{let}.

\begin{figure}[tb]
  \centering\footnotesize\noindent \begin{minipage}{.5\linewidth}
\[
\let\fv\progvars
\begin{array}{l}
  \fv(\p{skip}) = \emptyset \\
  \fv(\p{x\:= $\vexp$}) = \set{\pvar{x}} \cup \fve(\vexp) \\
  \fv(\p{x\:= [$\vexp$]}) = \set{\pvar{x}} \cup \fve(\vexp) \\
  \fv(\p{[$\vexp_{1}$]\:= $\vexp_{2}$}) = \fve(\vexp_{1}) \cup \fve(\vexp_{2}) \\
  \fv(\p{x\:= CAS($\vexp_{1}$,$\vexp_{2}$,$\vexp_{3}$)}) = \\
    \qquad\set{\pvar{x}} \cup \fve(\vexp_{1}) \cup \fve(\vexp_{2}) \cup \fve(\vexp_{3}) \\
  \fv(\p{x\:= alloc($\vexp$)}) = \set{\pvar{x}} \cup \fve(\vexp) \\
  \fv(\p{dealloc($\vexp$)}) = \fve(\vexp) \\
  \fv(\p{let $\pvar{f}$($\vec{\pvar{x}}$)\,=\,$\cmd_{\pvar{f}}$\;in\;$\cmd$}) = \fv(\cmd) \\
  \fv(\p{var $\pvar{x} = \vexp$ in $\cmd$}) = (\fv(\cmd) \setminus \set{\pvar{x}}) \cup \fve(\vexp)\\
  \fv(\p{if($\bexp$)\{$\cmd_{1}$\}else\{$\cmd_{2}$\}}) = \fvb(\bexp) \cup \fv(\cmd_{1}) \cup \fv(\cmd_{2}) \\
  \fv(\p{while($\bexp$)\{$\cmd$\}}) = \fvb(\bexp) \cup \fv(\cmd) \\
  \fv(\p{x\:= f($\vec{\vexp}$)}) = \set{\pvar{x}} \cup \fve(\vexp) \\
  \fv(\cmd_1\p;\cmd_2) = \fv(\cmd_1) \cup \fv(\cmd_2) \\
  \fv(\cmd_{1} \parallel \cmd_{2}) = \fv(\cmd_{1}) \cup \fv(\cmd_{2})
\end{array}
\]
\end{minipage}\begin{minipage}{.5\linewidth}
\[
\begin{array}{l}
  \mods(\p{skip}) = \emptyset \\
  \mods(\p{x\:= $\vexp$}) = \set{\pvar{x}} \\
  \mods(\p{x\:= [$\vexp$]}) = \set{\pvar{x}} \\
  \mods(\p{[$\vexp_{1}$]\:= $\vexp_{2}$}) = \emptyset \\
  \mods(\p{x\:= CAS($\vexp_{1}$,$\vexp_{2}$,$\vexp_{3}$)}) = \set{\pvar{x}} \\
  \mods(\p{x\:= alloc($\vexp$)}) = \set{\pvar{x}} \\
  \mods(\p{dealloc($\vexp$)}) = \emptyset \\
  \mods(\p{let $\pvar{f}$($\vec{\pvar{x}}$) = $\cmd_{\pvar{f}}$ in $\cmd$}) = \mods(\cmd) \\
  \mods(\p{var $\pvar{x} = \vexp$ in $\cmd$}) = \mods(\cmd) \setminus \set{\pvar{x}} \\
  \mods(\p{if($\bexp$)\{$\cmd_{1}$\}else\{$\cmd_{2}$\}}) = \mods(\cmd_{1}) \cup \mods(\cmd_{2}) \\
  \mods(\p{while($\bexp$)\{$\cmd$\}}) = \mods(\cmd) \\
  \mods(\p{x\:= f($\vec{\vexp}$)}) = \set{\pvar{x}} \\
  \mods(\cmd_1\p;\cmd_2) = \mods(\cmd_1) \cup \mods(\cmd_2) \\
  \mods(\cmd_{1} \parallel \cmd_{2}) = \mods(\cmd_{1}) \cup \mods(\cmd_{2})
\end{array}
\]
\end{minipage}
   \caption{The sets of free and modified program variables}
  \label{fig:freevars-mods}
\end{figure}

\begin{definition}[Variable Store]
  A \emph{program variable store}, $\store \in \Store \is \PVar
  \pto \Val$, is a finite partial function from program variables to values.
The \emph{right-biased union} of variable stores,
$\store_1 \sunion \store_2$,
is defined~by:
  \[
  (\store_1 \sunion \store_2)(\pvar{x}) =
  \begin{cases}
    \store_2(\pvar{x}) \CASE \pvar{x} \in \dom(\store_2) \\
    \store_1(\pvar{x}) \OTHERWISE
  \end{cases}
  \]
\end{definition}

\begin{definition}[Expression evaluation]
\label{def:expr-eval}
  Let $\varsigma \from \Type{Vars} \finpto \Type{Values}$ be an arbitrary function from an
  arbitray set of variables to values.
  The \emph{numeric expression evaluation function},
  $\esem{\hole}(\varsigma) \from \Exp{\Type{Vars},\Type{Values}} \to \Type{Values}$,
  and
  the \emph{Boolean expression evaluation function},
  ${\bsem{\hole}(\varsigma) \from \BExp{\Type{Vars,\Type{Values}}} \to \Bool}$,
  are defined~by:
  \begin{align*}
    \esem{v}(\varsigma) & = v &
      \bsem{b}(\varsigma) & = b \\
    \esem{\pvar{x}}(\varsigma) & = \varsigma(\pvar{x}) &
      \bsem{\neg \bexp}(\varsigma) & = \neg \bsem{\bexp}(\varsigma) \\
   \esem{\vexp_1 + \vexp_2}(\varsigma) & = \esem{\vexp_1}(\varsigma) + \esem{\vexp_2}(\varsigma) &
      \bsem{\bexp_1 \land \bexp_2}(\varsigma) & = \bsem{\bexp_1}(\varsigma) \land \bsem{\bexp_2}(\varsigma) \\
    \esem{\vexp_1 - \vexp_2}(\varsigma) & = \esem{\vexp_1}(\varsigma) - \esem{\vexp_2}(\varsigma) &
      \bsem{\vexp_1 = \vexp_2}(\varsigma) & = (\esem{\vexp_1}(\varsigma) = \esem{\vexp_2}(\varsigma)) \\
    \esem{\vexp_1 \cdot \vexp_2}(\varsigma) & = \esem{\vexp_1}(\varsigma) \cdot \esem{\vexp_2}(\varsigma) &
      \bsem{\vexp_1 < \vexp_2}(\varsigma) & = (\esem{\vexp_1}(\varsigma) < \esem{\vexp_2}(\varsigma)) \\
    & \cdots &
    & \cdots
  \end{align*}
\end{definition}
The program expressions are evaluated using program store $\store \in \Store$.
In \cref{sec:assertions}, we also work with
logical expressions which are evaluated over both program and logical
variables and values.
The right-biased union of stores is used to describe
how, when nesting scopes, a variable occurrence is bound by the innermost binder surrounding it.
The notation  \acode{var x1,x2...,xn in CMD} denotes
\acode{var x1=0 in var x2=0 in ... var xn=0 in CMD}.

\begin{definition}[Heap]
  A \emph{heap}, $h \in \Heap \is \Addr \finpto \Val$, is a finite partial function from addresses to values.
  The \emph{set of  heaps}, $\Heap$,  forms a PCM
    $(\Heap, \dunion, \set{\emptyset})$
  with $h_1 \dunion h_2$ defined only if
  ${\dom(h_1) \inters \dom(h_2) = \emptyset}$.
\end{definition}

\begin{definition}[Function Implementation Context]
  A \emph{function implementation context},
  $
    \functxt \in \FunCtxt \is
      \FName \pto (\PVar^{*}, \Cmd)
  $,
  is a finite partial function from function names to
  pairs comprising a finite list of distinct variables and a command.
\end{definition}
We write $
  \functxt(\pvar{f}) = (\pvars{x}, \cmd ),
$ where variable list $\pvars{x}$ represents the
function arguments  and $\cmd$ represents the function body.
We use the notation $\functxt_{\mathsf{var}}$ and $\functxt_{\mathsf{cmd}}$
to refer to the arguments and function body of \pvar{f} respectively.

In order to describe the behaviour of local variable binding and
function calls,
we define program states which extends commands with variable stores. For
example, the program state  $(\store, \cmd)$ indicates that the
command $\cmd$ is evaluated in the current  store updated with the
variables in $\store$.

\begin{definition}[Program States]
  The \emph{set of program states}, $\PState \ni \pstate, \pstate_1, \pstate_2,
  \cdots$  is defined by the grammar:
  \begin{grammar}
    \pstate
      \is \checkmark
      |   (\store, \pstate)
      |   \pstate; \cmd
      |   \acode{let f(VARS) = CMD in\ $\pstate$}
      |   \pstate \parallel \pstate
      |   \cmd
  \end{grammar}
\end{definition}

The   $\checkmark$ indicates a terminated program. It is a technical
device so that every $\cmd \in \Cmd$, including $ \code{skip} $,
takes at least one step.

In the operational semantics,
we need to keep track of which thread is originating each step
to be able to define later concepts of fairness of the scheduling.
We do this tracking using \emph{thread identifiers}
$ t \in \TId \is \set{\threadL,\threadR}^* $
which are strings of letters 
  $\threadL$ (for the left thread) and
$\threadR$ (for the right thread).
$\epsilon$ will be used to denote the thread identifier which is an empty sequence.
Intuitively, such a string identifies a single thread as the path in the syntax tree of parallel compositions at which the thread is found.

\begin{definition}[Command Semantics]
  A \emph{scheduler annotation} $ \sched $ is an element of the set
  \[
    \Type{Sched} \is
      \set{\LocOf{t} | t \in \TId}
        \dunion
      \set{\env}
      .
  \]
A~\emph{program configuration} $\pconf$ is an element of the set
  $ \Type{PConf} \is (\Store \times \Heap \times \PState) \dunion \set{\fault} $.
  Let $\functxt \in \FunCtxt$.
  The \emph{operational semantics of the commands} is given by the labelled relation,
  $
    {\redto} \subseteq \Type{PConf} \times \Type{Sched} \times \Type{PConf}
  $,
  defined in \cref{fig:oper-semantics} and \cref{fig:oper-semantics-fail}.
  We write $ a \step{\sched} b$ for $ (a,\sched,b) \in {\redto} $.
  We also define $ {\locsteps} \is (\union_{t \in \TId} \step{\loc_{t}} )^* $.
\end{definition}

To simplify the development, in our programming language
the initial state's store assigns arbitrary values
to the free variables of a program.
With such assumption, every reference to a local variable
will be in the domain of the current store.
This ensures that in every application of the rules
in \cref{fig:oper-semantics} and \cref{fig:oper-semantics-fail} to construct a trace,
the evaluations of (boolean) expressions are well-defined.

\begin{figure}[tbp]
  \thisfloatpagestyle{empty}\centering\footnotesize \begin{mathpar}
  \infer{
  }{
    \store, h, \p{skip} \step{\LocOf{\epsilon}} \store, h, \checkmark
  }
\and \infer{
  }{
    \store, h, \p{x\:= \vexp}
      \step{\LocOf{\epsilon}}
    \store[\pvar{x} \mapsto \esem{\vexp}(\store)], h, \checkmark
  }
\\ \infer{
    \esem{\vexp}(\store) \in \dom(h)
  }{
    \store, h, \p{x\:= [$\vexp$]}
      \step{\LocOf{\epsilon}}
    \store[\pvar{x} \mapsto h(\esem{\vexp}(\store))], h, \checkmark
  }
\and \infer{
    \esem{\vexp_1}(\store) \in \dom(h)
  }{
    \store, h, \p{[$\vexp_1$]\:= $\vexp_2$}
      \step{\LocOf{\epsilon}}
    \store, h\Map{{\esem{\vexp_1}(\store)} -> {\esem{\vexp_2}(\store)}}, \checkmark
  }
\and \infer{
    \esem{\vexp_1}(\store) \in \dom(h)
      \and
    h(\esem{\vexp_1}(\store)) = \esem{\vexp_2}(\store)
  }{
    \store, h, \p{x\:= CAS($\vexp_1$,$\vexp_2$,$\vexp_3$)}
      \step{\LocOf{\epsilon}}
    \store\map{{\pvar{x}} -> 1}, h\map{{\esem{\vexp_1}(\store)} -> {\esem{\vexp_3}(\store)}}, \checkmark
  }
\and \infer{
    \esem{\vexp_1}(\store) \in \dom(h) \and
    h(\esem{\vexp_1}(\store)) \neq \esem{\vexp_2}(\store)
  }{
    \store, h, \p{x\:= CAS($\vexp_1$,$\vexp_2$,$\vexp_3$)}
      \step{\LocOf{\epsilon}}
    \store\map{{\pvar{x}} -> 0}, h, \checkmark
  }
\and \infer{
    a=\esem{\vexp_1}(\store) \in \dom(h) \and
    v=\esem{\vexp_2}(\store)
  }{
    \store, h, \p{x\:= FAS($\vexp_1$,$\vexp_2$)}
      \step{\LocOf{\epsilon}}
    \store\map{{\pvar{x}} -> h(a)}, h\map{a->v}, \checkmark
  }
\and \infer{
    l = \esem{\vexp}(\store) \and
    l > 0 \and
    \set{r, r + 1, \cdots, r + l - 1} \cap \dom(h) = \emptyset \and
    v_0, v_1, \cdots, v_{l - 1} \in \Val
  }{
    \store, h, \p{x\:= alloc($\vexp$)}
      \step{\LocOf{\epsilon}}
    \store[\pvar{x} \mapsto r], h\bigl[r \mapsto v_0, r + 1 \mapsto v_1, \cdots, r + l - 1 \mapsto v_{l - 1}\bigr], \checkmark
  }
\\ 

\infer{
    \esem{\vexp}(\store) \in \dom(h)
  }{
    \store, h, \p{dealloc($\vexp$)}
      \step{\LocOf{\epsilon}}
    \store, h\map{{\esem{\vexp}(\store)} -> \bot}, \checkmark
  }
\end{mathpar}
\begin{mathpar}
\infer{
    \store, h, \pstate \step[\functxt']{\LocOf{t}} \store', h', \pstate' \and
    \functxt' = \functxt[\pvar{f} \mapsto (\vec{\pvar{x}}, \cmd_{\pvar{f}})]
  }{
    \store, h, \p{let\ f($\vec{\pvar{x}}$)\,=\,$\cmd_{\pvar{f}}$ in\ $\pstate$}
      \step{\LocOf{t}}
    \store', h', \p{let\ f($\vec{\pvar{x}}$)\,=\,$\cmd_{\pvar{f}}$ in\ $\pstate'$}
  }
\and \infer{
  }{
    \store, h, \p{let\ f($\vec{\pvar{x}}$)\,=\,$\cmd_{\pvar{f}}$ in\ \checkmark} \step{\LocOf{\epsilon}} \store, h, \checkmark
  }
\\ \infer{
  }{
    \store, h, \p{var $\pvar{x} = \vexp$ in\ $\cmd$}
      \step{\LocOf{\epsilon}}
    \store, h, ([\pvar{x} \mapsto \esem{\vexp}(\store)], \cmd)
  }
\and \infer{
  }{
    \store, h, (\store', \checkmark)
      \step{\LocOf{\epsilon}}
    \store, h, \checkmark
  }
\and \infer{
    \store \sunion \store_1, h, \pstate
      \step{\LocOf{t}}
    \store' \sunion \store_1', h', \pstate' \and
    \dom(\store) = \dom(\store') \and
    \dom(\store_1) = \dom(\store_1 ')
  }{
    \store, h, (\store_1, \pstate)
      \step{\LocOf{t}}
    \store', h', (\store_1', \pstate')
  }
\and \infer{
    \bsem{\bexp}(\store)
  }{
    \store, h, \p{if(\bexp)\{$\cmd_1$\}else\{$\cmd_2$\}}
      \step{\LocOf{\epsilon}}
    \store, h, \p{$\cmd_1$}
  }
\and \infer{
    \neg \bsem{\bexp}(\store)
  }{
    \store, h, \p{if(\bexp)\{$\cmd_1$\}else\{$\cmd_2$\}}
      \step{\LocOf{\epsilon}}
    \store, h, \p{$\cmd_2$}
  }
\and \infer{
    \bsem{\bexp}(\store)
  }{
    \store, h, \p{while(\bexp)\{\cmd\}}
      \step{\LocOf{\epsilon}}
    \store, h, \p{\cmd;while(\bexp)\{\cmd\}}
  }
\and \infer{
    \neg \bsem{\bexp}(\store)
  }{
    \store, h, \p{while(\bexp)\{\cmd\}}
      \step{\LocOf{\epsilon}}
    \store, h, \p{\checkmark}
  }
\and \infer{
    \functxt(\p{f}) = (\vec{\pvar{x}}, \cmd)
  }{
    \store, h, \p{y\:= f($\vec{\vexp}$)}
      \step{\LocOf{\epsilon}}
    \store, h, \p{var ret=0 in (var $\vec{\pvar{x}} = \vec{\vexp}$ in \cmd); y\:=ret}
  }
  \and \infer{
    \store, h, \p{$\pstate_1$} \step{\LocOf{t}} \store', h', \p{$\pstate_1'$}
  }{
    \store, h, \p{$\pstate_1$;$\cmd_2$} \step{\LocOf{t}} \store', h', \p{$\pstate_1'$;$\cmd_2$}
  }
\and \infer{
  }{
    \store, h, \p{$\checkmark$;$\cmd$}
      \step{\LocOf{\epsilon}}
    \store, h, \p{$\cmd$}
  }
\and \infer{
    \store, h, \pstate_1
      \step{\LocOf{t}}
      \store', h', \pstate_1'
  }{
    \store, h, \pstate_1 \parallel \pstate_2
      \step{\LocOf{\threadL t}}
    \store', h', \pstate_1' \parallel \pstate_2
  }
\and \infer{
    \store, h, \pstate_2
      \step{\LocOf{t}}
      \store', h', \pstate_2'
  }{
    \store, h, \pstate_1 \parallel \pstate_2
      \step{\LocOf{\threadR t}}
    \store', h', \pstate_1 \parallel \pstate_2'
  }
\and \infer{
  }{
    \store, h, \checkmark \parallel \checkmark
      \step{\LocOf{\epsilon}}
    \store, h, \checkmark
  }
\and \infer{
    \store, h, \cmd \locsteps \store', h', \checkmark
  }{
    \store, h, \atombra{\cmd} \step{\LocOf{\epsilon}} \store', h', \checkmark
  }
\and \infer{
    h'\in \Heap
  }{
    \store, h, \pstate \envstep \store, h', \pstate
  }
\end{mathpar}
   \caption{The small-step operational semantics}
  \label{fig:oper-semantics}
\end{figure}

\begin{figure}[tbp]
  \thisfloatpagestyle{empty}\centering\footnotesize \begin{mathpar}
  \infer{
    \esem{\vexp}(\store) \not\in \dom(h)
  }{
    \store, h, \p{x\:= [$\vexp$]} \step{\LocOf{\epsilon}} \fault
  }
\and \infer{
    \esem{\vexp_1}(\store) \not\in \dom(h)
  }{
    \store, h, \p{[$\vexp_1$]\:= $\vexp_2$} \step{\LocOf{\epsilon}} \fault
  }
\and \infer{
    \esem{\vexp_1}(\store) \not\in \dom(h)
  }{
    \store, h, \p{x\:= CAS($\vexp_1$,$\vexp_2$,$\vexp_3$)}
      \step{\LocOf{\epsilon}}
    \fault
  }
\\ \infer{
    \esem{\vexp_1}(\store) \not\in \dom(h)
  }{
    \store, h, \p{x\:= FAS($\vexp_1$,$\vexp_2$)}
      \step{\LocOf{\epsilon}}
    \fault
  }
\and \infer{
    \esem{\vexp}(\store) \not\in \dom(h)
  }{
    \store, h, \p{dealloc($\vexp$)} \step{\LocOf{\epsilon}} \fault
  }
\\ \infer{
    \store, h, \pstate \step[\functxt']{\LocOf{t}} \fault \and
    \functxt' = \functxt[\pvar{f} \mapsto (\vec{\pvar{x}}, \cmd_{\pvar{f}})]
  }{
    \store, h, \p{let\ f($\vec{\pvar{x}}$)\,=\,$\cmd_{\pvar{f}}$ in\ $\pstate$} \step{\LocOf{t}} \fault
  }
\and \infer{
    \store \sunion \store', h, \pstate \step{\LocOf{t}} \fault
  }{
    \store, h, (\store', \pstate) \step{\LocOf{t}} \fault
  }
  \\ \infer{
    \p{f} \not\in \dom(\functxt)
  }{
    \store, h, \p{y\:= f($\vec{\vexp}$)} \step{\LocOf{\epsilon}} \fault
  }
\and \infer{
    \store, h, \p{$\pstate_1$} \step{\LocOf{t}} \fault
  }{
    \store, h, \p{$\pstate_1$;$\cmd_2$} \step{\LocOf{t}} \fault
  }
\and \infer{
    \store, h, \pstate_1 \step{\LocOf{t}} \fault
  }{
    \store, h, \pstate_1 \parallel \pstate_2 \step{\LocOf{\threadL t}} \fault
  }
\\ \infer{
    \store, h, \cmd \locsteps \fault
  }{
    \store, h, \atombra{\cmd} \step{\LocOf{\epsilon}} \fault
  }
\and \infer{
    \store, h, \pstate_2 \step{\LocOf{t}} \fault
  }{
    \store, h, \pstate_1 \parallel \pstate_2 \step{\LocOf{\threadR t}} \fault
  }
\and \infer{
    \pconf \in \PConf
  }{
    \pconf \envstep \fault
  }
\end{mathpar}
   \caption{The small-step operational semantics, failure cases}
  \label{fig:oper-semantics-fail}
\end{figure}

\begin{definition}[Threads]
  Given a program state $\pconf \in \PConf$,
  the set $\threads(\pconf)$ is the set of threads of $\pconf$
  that can take a step.
  The function $\threads\from\PConf \to \powerset(\TId)$
  is defined as follows:
  \begin{align*}
    \threads(\fault) &\is \emptyset \\
    \threads(\pconf) &\is
    \set{ t \in \TID |
      \pconf \step[\functxt]{\LocOf{t}} \_
    }
  \end{align*}
\end{definition}

\begin{definition}[Program Traces and Fairness]
\label{def:prog-traces}
\label{def:fairness}
  We call \emph{program traces},
  the infinite sequences of the form
  $
    \pconf[0] \schpl[0] \pconf[1] \schpl[1] \cdots
  $
  where,
for all $i\in\Nat$,
      $\pconf[i] \in \Type{PConf}$,
      ${\schpl[i]} \in \Type{Sched}$.
  We use $\ptrace$ for ranging over infinite suffixes of program traces
  and $\PTrace$ for the set of all program traces.
  For a program trace
    $ \ptrace = \pconf[0] \schpl[0] \pconf[1] \schpl[1] \cdots $,
  we define $\trAt[\ptrace]{i} \is (\pconf[i],\schpl[i])$,
and $\trFrom[\ptrace]{i} \is \pconf[i] \schpl[i] \pconf[i+1] \schpl[i+1] \cdots$.
We define the \emph{set of \pre\functxt-program traces}
  \[
  \PTrace_{\functxt} \is
  \set{ \pconf[0]\schpl[0]\pconf[1]\schpl[1] \cdots
    | \forall i\in\Nat\st
    \pconf[i] \step{\schpl[i]} \pconf[i+1]
  }.
  \]
  A program trace
    $ (\pconf[0]\schpl[0]\pconf[1]\schpl[1]\cdots) \in \PTrace[\functxt] $
  is \emph{(weakly) fair} if and only if:
  \begin{gather}
    \forall i \in \Nat \st
      \forall t \in \threads(\pconf[i]) \st
        \exists j \ge i\st
          ({\schpl[j]} = \LocOf{t} \lor
          \pconf[j] = \fault)
    \\
    \forall i \in \Nat \st
      \exists j \ge i \st \schpl[j] = \env
  \end{gather}
  That is: a trace is fair if, at any point in time,
  every thread that can take a step (and the environment)
  will eventually be scheduled.
\end{definition}

The open-world program semantics defines the behaviour of a command
when run concurrently with an arbitrary environment.
This semantics interleaves steps from two ``players'':
  the local thread given by the   $\loc$  relation; and its
  environment given by the $\env$ relation,  respectively.

\begin{definition}[Open World Semantics]
\label{def:prog-semantics}
  We call \emph{traces}
  the infinite sequences
  $
    {\conf[0] \pl[0] \conf[1] \pl[1] \cdots}
  $
  where,
for all $i\in\Nat$,
      $\conf[i] \in \Conf \is (\Store \times \Heap) \union \set{\fault}$,
      ${\pl[i]} \in \set{\loc,\env}$.
  We use $\trace$ for ranging over infinite suffixes of traces
  and $\Trace$ for the set of all traces.
  For a trace $ \trace = \conf[0] \pl[0] \conf[1] \pl[1] \cdots $,
  we define $\trAt{i} \is (\conf[i],\pl[i])$,
and $\trFrom{i} \is \conf[i] \pl[i] \conf[i+1] \pl[i+1] \cdots$.
The function $\stripTr{\hole} \from \PTrace \to \Trace$ is defined by
  $
    \stripTr{\pconf[0]\schpl[0]\pconf[1]\schpl[1]\cdots}
      \is
        \conf[0] \pl[0] \conf[1] \pl[1] \cdots
  $
  where
  \begin{align*}
    \conf[i] &\is
      \begin{cases}
        (\store, h) \CASE \pconf[i] = (\store,h,\wtv[2])\\
        \fault      \CASE \pconf[i] = \fault
      \end{cases}
    &
    \pl[i] &\is
      \begin{cases}
        \loc \CASE {\schpl[i]} \in \Type{Sched}\setminus \set{\env}\\
\env \CASE {\schpl[i]} = {\env} \\
      \end{cases}
  \end{align*}

  The \emph{open-world program semantics function},
  $
    \sem{\hole}_{\functxt} \from \Cmd \to \powerset(\Trace)
  $
  is the function such that
  \[
    \sem{\cmd}_{\functxt}
      \is
      \Set{ \stripTr{\pconf[0] \ptrace}
        |
          (\pconf[0]\ptrace) \in \PTrace[\functxt],
          \fv(\cmd) \subseteq \dom(\store_0),
          \pconf[0] = (\store_0,\wtv,\cmd),
          \pconf[0]\ptrace \text{ is fair}
      }
  \]
  The notation $\sem{\cmd}$ is syntactic sugar for $\sem{\cmd}_{\emptyset}$.
\end{definition}

\begin{definition}
  A trace $\trace \in \Trace$ is \emph{locally terminating},
  written $\locterm(\trace)$,
  if it contains finitely many occurrences of $\loc$.
\end{definition}

\begin{remark}[Design of semantics]
  We made some design choices in crafting this semantics,
  with the motivation of making manipulation easier in the proofs.
The first choice is to model environmental steps explicitly.
  These steps drive the argument about progress in the presence of blocking,
  where the local thread is not able to make progress in isolation
  but is relying on the environment actively performing some state changes
  that would lead to local progress.

  The second choice we highlight is that the semantics of a program
  only contains infinite traces.
  This might seem odd when the goal is proving termination.
  Traces that locally terminate simply have an infinite tail of environment steps.
  To simulate a closed system one can select for the traces where the environment steps preserve the heaps.
  More importantly, we strip the information about threads and program state,
  which means that information about when the local thread terminated
  (in the form of $\checkmark$ or $\EndOf{t}$)
  has been erased.
  However, by construction,
  traces obtained from fair program traces
  can only contain finitely many local steps
  if the program terminated,
  justifying our definition of local termination.
\end{remark}

\begin{example}
  The traces in
  $
    \sem[\big]{\p{[x]\:=y}}
  $
  can be characterised as follows.
  They all start from some configuration $(\store,h_0)$
  with $ x,y\in\dom(\store) $.
  A (possibly zero) finite number of environment steps follow;
  these steps preserve the store, but arbitrarily alter the heap,
  or they lead to a fault, terminating the trace with an infinite tail of
  $\fault \env \fault \env \cdots$ steps.
  If no fault happened,
  a local step is taken from some configuration $(\store,h)$
  for an arbitrary $h\in \Heap$.
  If $\store(x)\not\in\dom(h)$
    then the local step leads to a fault,
    leading again to a $\fault \env \fault \env \cdots$ tail.
  Otherwise,
    it leads to the configuration $(\store, h\map{\store(x)->\store(y)})$.
  After that there is an infinite number of environment steps,
  which again preserve the store but arbitrarily mutate the heap,
  or lead to an infinite fault tail.
\end{example}

 \section{Soundness}
\label{app:soundness}

In this section, we provide the details of the soundness of three rules: \ref{rule:liveness-check}, \ref{rule:parallel}, \ref{rule:while}, \ref{rule:frame}, \ref{rule:envlive-obl} and \ref{rule:envlive-pq}. These are the only proof rules in TaDA-Live that bring in non-trivial liveness information. All other proof rules follow in the same way as for TaDA, with the liveness constraints on the traces being identical between the antecedent and consequent of such rules or being trivial in the case of command axioms. We will focus particularly on the liveness argument for these rules.

We start by giving some technical definitions omitted from the main text,
and then move to the soundness argument.

\subsection{Atomic World Rely}
\label{app:atomic-rely}

Recall that the \emph{atomic world rely relation}, $\relyAt[\actxt]$,
coincides with the smallest reflexive relation closed under
the rules of the world rely (\cref{fig:world-rely}),
with the restriction that \cref{rule:rely-interf,rule:rely-linpt} can be
applied at most once per region identifier.

\begin{definition}[Atomic World Rely Relation]
  \label{def:atomic-rely}
  The atomic world rely relation, $\relyAt[\actxt]$,
  is defined as $\relyAt[\actxt] = \rely[\actxt]^\emptyset$,
  where $\rely[\actxt]^{R}$, taking $R \subseteq \RId$,
  is defined to be the smallest
  reflexive relation closed under:
  \begin{mathpar}
    \infer*[Right=wr$_1$]{
      \guardMap(\rid) \compat G
      \\
      ((a_1,O_1), (a_2,O_2)) \in \glts_{\rt}(G)
      \\
      \atomMap(\rid) \in \set{\blacklozenge, \lozenge}
        \implies a_2 \in \safe(\actxt,\rid)
      \\
      O_2 \compat \oblMap(r)
      \\
      \rid \not\in R
      \\
      (h,
      \regMap \map{ \rid -> (\rt, \lvl, a_2)},
      \guardMap,
      \atomMap,
      \oblMap,
      \envMap\map{ \rid -> O_2 })
      \rely[\actxt]^{R \dunion \set{\rid}} w'
    }{
      (h,
        \regMap \map{ \rid -> (\rt, \lvl, a_1)},
        \guardMap,
        \atomMap,
        \oblMap,
        \envMap\map{ \rid -> O_1 })
      \rely[\actxt]^{R}
      w'
    }
  \\\infer*[Right=wr$_2$]{
      ((a_1 ,O_1), (a_2,O_2)) \in \trrel(\actxt,\rid)
      \\
      O_2 \compat \oblMap(r)
      \\
      \rid \not\in R
      \\
      (h,
      \regMap \map{ \rid -> (\rt, \lvl, a_2)},
      \guardMap,
      \atomMap \map{ \rid -> (a_1, a_2)},
      \oblMap,
      \envMap \map{ \rid -> O_2 })
      \rely[\actxt]^{R \dunion \set{\rid}}
      w'
    }{
      (h,
        \regMap \map{ \rid -> (\rt, \lvl, a_1)},
        \guardMap,
        \atomMap \map{ \rid -> \lozenge},
        \oblMap,
        \envMap \map{ \rid -> O_1 })
      \rely[\actxt]^{R}
      w'
    }
  \end{mathpar}
\end{definition}

\subsection{Environment Liveness Judgement Semantics}
\label{app:envlive-semantics}

We give semantics to the judgements defined in \cref{fig:envlive}.

\newcommand{\sat}{\lfun{sat}}
\newcommand{\act}{\lfun{active}}

\begin{definition}[Auxiliary Environment Liveness Judgement Semantics]
  Let $m \in \Layer, \lvl \in \Level, \actxt \in ACtxt, L, L' \in \Ord \to \Assrt, T \in \Assrt$ such that
  \begin{itemize}
  \item $\STABLE \levl; \actxt |= {\exists \alpha \st L(\alpha)}$.
  \item $\forall \alpha \st \viewshift \lvl; \actxt |= {L'(\alpha)} => {L(\alpha)} $.
  \end{itemize}
  and let
  \begin{align*}
    t_{\store} &= \WorldSem{\actxt}{\store}{T * \True} &
    l_{\store}(\alpha) &= \WorldSem{\actxt}{\store}{L(\alpha)} &
l'_{\store}(\alpha) &= \WorldSem{\actxt}{\store}{L'(\alpha)}
  \end{align*}
  Then, the auxiliary semantic environmental liveness judgement
  $
    \SENVLIVE m; \lvl; \actxt |= L : L' - ->> T
  $
  holds when, for arbitrary $\store \in \Store$, there exist
  $P \subseteq \powerset(\Ord \to \UCWorld[\actxt])$ such that
  $l'_{\store}(\alpha) = \bigcup\limits_{\lvar{lf} \in P} \lvar{lf}(\alpha)$ and
  for all $\lvar{lf} \in P$, either $\forall\alpha\st \lvar{lf}(\alpha) \subseteq t$ or 
  there exists some $\rid \in \RId$ and 
  \[
  \pobl\in \AObl_{<m} \dunion \Set{\live(\actxt,\rid) | \lay(\live(\actxt,\rid)) < m}
  \]
  such that
  \begin{itemize}
    \item $\forall \alpha \in \Ord, w \in \lvar{lf}(\alpha) \st \act_{\rid;\lvl}(w, \pobl)$
    \item
      $
      \forall \alpha_1, \alpha_2 \ge \alpha_1 \st
      \relyAt[\actxt](\lvar{lf}(\alpha_1)) \inters l_{\store}(\alpha_2) \subseteq \lvar{lf}(\alpha_1) \union t
      $
  \end{itemize}
  hold, where:
  \[
  \act_{\rid;\lvl}(w, \pobl) \is
  \begin{cases}
    \lfun{dep}_{\rid;\lvl}(w, \pobl) \land
    \envMap[w](\rid) \resgeq \pobl & \pobl \in \AObl \\
   
    \lfun{dep}_{\rid;\lvl}(w, \pobl) \land
    \atomMap[w](\rid) \compat \lozenge \land
    \astate[w](\rid) \in X \setminus X' & \pobl = X \eventually[k] X'
  \end{cases}
  \]
  and
  \[
  \lfun{dep}_{\rid;\lvl}(w, \pobl) \is
  \forall \rid' \in \dom(\oblMap[w]) \st \lay(\oblMap[w](\rid')) > \lay(\pobl) \land \level[w](\rid) < \lvl
  \]

\end{definition}

\begin{definition}[Environment Liveness Judgement Semantics]
  The semantic environmental liveness judgement:
  \[
    \SENVLIVE m; \lvl; \actxt |= L - \smash{M} ->> T
  \]
  where $m \in \Layer, \lvl \in \Level, \actxt \in ACtxt, L \in \Assrt, M \in \Ord \to \Assrt, T \in \Assrt$, 
  holds when
  \begin{align*}
    & \STABLE \levl; \actxt |= L \\
    & \VALID \levl; \actxt |= L \implies \exists \alpha \st L * M(\alpha) \\
    & \SENVLIVE m; \lvl; \actxt |= L * M(\alpha) : L * M(\alpha) - ->> T
  \end{align*}
\end{definition}

\begin{theorem}
  \label{thrm:aux-env-live-soundness}
  For arbitrary $m \in \Layer, \lvl \in \Level$,
  $\actxt$ an atomicity context,
  $L, L' \in \Ord \to \Assrt, T \in \Assrt$
  such that
  \begin{align}
    \STABLE \levl; \actxt |= {\exists \alpha \st L(\alpha)}
    \label{ass:envlive-l-stable}\\
  \forall \alpha \st \viewshift \lvl; \actxt |= {L'(\alpha)} => {L(\alpha)}
  \label{ass:envlive-implication}
  \end{align}
  if 
    $\, \ENVLIVE m; \lvl; \actxt |- L : L' - ->> T$,
  then
    $\SENVLIVE m; \lvl; \actxt |= L : L' - ->> T$.
\end{theorem}

\begin{proof}
  Assuming $\, \ENVLIVE m; \lvl; \actxt |- L : L' - ->> T$ and
  taking $\store \in \Store$ arbitrary,
  the proof proceeds by induction on the structure of
  derivation trees of the auxiliary environmental liveness
  condition.
  We start of with the bases cases: \ref{rule:envlive-obl},
  \ref{rule:envlive-pq} and \ref{rule:envlive-target}.

  \paragraph{case \ref{rule:envlive-obl}}
  In this case, for some $\rid \in \RId, \rt \in \RType,
  \lvlp \in \Level, \obl{O} \in \AObl$,
  \begin{align}
    & \decr[\actxt](L', L, T)
    \label{ass:liveO-impr}
    \\
    & m \laygt \lay(O) \label{ass:liveO-ctxt-lay} \\
    & \forall \alpha\st \VALID \actxt |= \minLayStrict{L'(\alpha)}{\lay(O)}
    \label{ass:liveO-ctxt-minlay}
    \\
    & \lvlp < \lvl
    \label{ass:liveO-level} \\
    & \forall \alpha \st
    \SAT[\actxt] |- L'(\alpha) \implies
    \exists x \st \region[\lvlp]{\rt}{\rid}(x) * \envObl{O}{\rid} * \True
    \label{ass:liveO-impl-oblig}
  \end{align}
  hold. From this, we need to show
  $\SENVLIVE m; \lvl; \actxt |= L(\alpha) : L'(\alpha) -->> T$.
  
  Let $P = \set{l'_{\store}(\alpha)}$,
  clearly the union of the elements of this set equals $l'_{\store}(\alpha)$
  as required. Assuming $l'_\store(\alpha) \not\subseteq t_\store$ and
  setting $\pobl = \obl{O}$, which is in $\AObl_{<m}$ given
  (\ref{ass:liveO-ctxt-lay}), it suffices to show
  \begin{align}
    & \forall \alpha \in \Ord, w \in l'_{\store}(\alpha) \st \act_{\rid;\lvl}(w, \obl{O})
    \label{thrm:aux-env-live-soundness-oblo-act} \\
    & \forall \alpha_1, \alpha_2 \ge \alpha_1 \st
    \relyAt[\actxt](l'_{\store}(\alpha_1)) \inters l_{\store}(\alpha_2) \subseteq l'_{\store}(\alpha_1) \union t
    \label{thrm:aux-env-live-soundness-oblo-prog}
  \end{align}
  hold to complete the proof.

  We start off by showing that (\ref{thrm:aux-env-live-soundness-oblo-act}) holds.
  Taking $\alpha \in \Ord$ and $w \in l'_{\store}(\alpha)$ arbitrary,
  given (\ref{ass:liveO-ctxt-minlay}), it is clear that $\lay(\oblMap[w](\rid)) \ge \lay(\obl{O})$
  holds and, given (\ref{ass:liveO-level}) and (\ref{ass:liveO-impl-oblig}), $\level[w](\rid) < \lvl$ holds.
  From these two conclusions, we can infer that $\lfun{dep}_{\rid;\lvl}(w, \obl{O})$ holds.
  Then, from (\ref{ass:liveO-impl-oblig}), it is clear that
  $\envMap[w](\rid) \resgeq \obl{O}$ holds, and therefore, $\act_{\rid;\lvl}(w, \obl{O})$.

  Finally, (\ref{thrm:aux-env-live-soundness-oblo-prog}) follows immediately from
  (\ref{ass:liveO-impr}) and the definition of $\decr$.
  
  \paragraph{case \ref{rule:envlive-pq}}
  In this case, for some $\rid \in \RId, \rt \in \RType,
  \lvlp \in \Level$,
  \begin{align}
    & \decr[\actxt](L', L, T)
    \label{ass:liveA-impr}
    \\
    & m \laygt k \label{ass:liveA-ctxt-lay} \\
    & \forall \alpha\st
    \VALID \actxt |= \minLayStrict{L'(\alpha)}{k}
    \label{ass:liveA-ctxt-minlay}
    \\
    & \live(\actxt,r) = X \eventually[k] X' \label{ass:liveO-live-actxt} \\
    & \lvlp < \lvl
    \label{ass:liveA-level} \\
    & \forall \alpha \st
    \SAT[\actxt] |- L'(\alpha) \implies
    \exists x \st \region[\lvlp]{\rt}{\rid}(x) * \done{\rid}{\lozenge} \land
    x \in X \setminus X' * \True
    \label{ass:liveA-impl-oblig}
  \end{align}
  hold. From this, we need to show
  $\SENVLIVE m; \lvl; \actxt |= L(\alpha) : L'(\alpha) -->> T$.
  
  Let $P = \set{l'_{\store}(\alpha)}$,
  clearly the union of the elements of this set equals $l'_{\store}(\alpha)$
  as required. Assuming $l'_\store(\alpha) \not\subseteq t_\store$ and
  setting $\pobl = X \eventually[k] X'$, which is in
  $\Set{\live(\actxt,r) | \lay(\live(\actxt,r)) < m}$ given
  (\ref{ass:liveA-ctxt-lay}) and (\ref{ass:liveO-live-actxt}),
  it suffices to show
  \begin{align}
    & \forall \alpha \in \Ord, w \in l'_{\store}(\alpha) \st \act_{\rid;\lvl}(w, X \eventually[k] X')
    \label{thrm:aux-env-live-soundness-obla-act} \\
    & \forall \alpha_1, \alpha_2 \ge \alpha_1 \st
    \relyAt[\actxt](l'_{\store}(\alpha_1)) \inters l_{\store}(\alpha_2) \subseteq l'_{\store}(\alpha_1) \union t
    \label{thrm:aux-env-live-soundness-obla-prog}
  \end{align}
  to complete the proof.

  We start off by showing that (\ref{thrm:aux-env-live-soundness-obla-act}) holds.
  Taking $\alpha \in \Ord$ and $w \in l'_{\store}(\alpha)$ arbitrary,
  given (\ref{ass:liveA-ctxt-minlay}), it is clear that $\lay(\oblMap[w](\rid)) \ge \lay(X \eventually[k] X')$
  holds and, given (\ref{ass:liveA-level}) and (\ref{ass:liveA-impl-oblig}), $\level[w](\rid) < \lvl$ holds.
  From these two conclusions, we can infer that $\lfun{dep}_{\rid;\lvl}(w, X \eventually[k] X')$ holds.
  Then, from (\ref{ass:liveA-impl-oblig}), it is clear that
  $\atomMap[w](\rid) \compat \lozenge \land \astate[w](\rid) \in X \setminus X'$ holds, and therefore, $\act_{\rid;\lvl}(w, X \eventually[k] X')$.

  Finally, (\ref{thrm:aux-env-live-soundness-obla-prog}) follows immediately from
  (\ref{ass:liveA-impr}) and the definition of $\decr$.
    
  \paragraph{case \ref{rule:envlive-target}}
  In this case
  $
  \forall \alpha. \VALID \actxt |= L'(\alpha) \implies T
  $
  holds. From this, we need to show
  $
  \SENVLIVE m; \lvl; \actxt |= L(\alpha) : L'(\alpha) -->> T
  $.
  Let $P = \set{l'_\store(\alpha)}$,
  clearly the union of the elements of this set equals $l'_{\store}(\alpha)$ as required.
  From $\forall \alpha. \VALID \actxt |= L'(\alpha) \implies T$, clearly
  $l'_\store(\alpha) \subseteq t_\store$, therefore
  $
  \SENVLIVE m; \lvl; \actxt |= L(\alpha) : L'(\alpha) -->> T
  $ holds, as required.

  Finally, we complete this theorem's proof with a
  proof of the soundness of the one inductive case, \ref{rule:envlive-quant}.
  Note that \ref{rule:envlive-case} can be derived directly from
  \ref{rule:envlive-quant}.

  \paragraph{case \ref{rule:envlive-quant}}
  In this case, $L'(\alpha) = \exists x \in X \st L''(x, \alpha)$
  for some $L'' \in X \times \Ord \to \Assrt$ and
  \begin{align}
    \forall x \in X \st
    \ENVLIVE m; \levl; \actxt |-
    L(\alpha) : L''(x, \alpha) -->> T
    \label{ass:live-quant-forall}
  \end{align}
  hold. Letting
  \[
    l'_{x,\store}(\alpha) = \WorldSem{\actxt}{\store}{L''(x, \alpha)}
  \]
  From (\ref{ass:live-quant-forall}), for any $x \in X$
  there exists $P_x \subseteq \mathcal{P}(\World[\actxt])$
  such that $l'_{x,\store}(\alpha) = \bigcup P_x$ with the appropriate
  conditions holding for each $\lvar{lf} \in \Union\limits_{x \in X} P_x$.

  Setting $P = \Union\limits_{x \in X} P_x$, given the definition
  of $L'(\alpha)$, clearly $l'_\store(\alpha) = \bigcup P$
  and for each $l \in P$, there exists some $x \in X$ such that $l \in P_x$
  and therefore, the appropriate properties hold due to
  (\ref{ass:live-quant-forall}) as required.

  By induction on the structure of
  derivation trees of the auxiliary environmental liveness
  condition, $\SENVLIVE m; \lvl; \actxt |= L : L' - ->> T$
  holds, as required.
  
 \end{proof}

\begin{theorem}
  If 
    $\ENVLIVE m; \lvl; \actxt |- L - M ->> T$
  then
    $\SENVLIVE m; \lvl; \actxt |= L - M ->> T$.
\end{theorem}

\begin{proof}
  This theorem follows trivially from theorem
  \ref{thrm:aux-env-live-soundness}.
\end{proof}
 
\subsection{Soundness of \ref{rule:frame}}

For the rest of the section, we let
\begin{align*}
  \spec &= \SPEC m; \lvl; \actxt |=
  \A x \in \pqsets{X}.
  <\na{P} | \at{P}(x)>
  \E y.<\na{Q}(x,y) | \at{Q}(x,y)> \\
  \spec' &=  \SPEC m; \lvl; \actxt |=
  \A x \in \pqsets{X}.
  <\na{P} * \na{R} | \at{P}(x) * \at{R}(x)>
  \E y.<\na{Q}(x,y) * \na{R} | \at{Q}(x,y) * \at{R}(x)>
\end{align*}
such that
\begin{align*}
  & \STABLE \actxt |= { \na{R} } \\
  & \forall x \in X\st \STABLE \actxt |= { \at{R}(x) } \\
  & \forall x \in X \st  \VALID \actxt |= \at{R}(x) \implies \empObl[\lvl]
\end{align*}

\begin{lemma}
  \label{lemma:AFPU_extend}
  For arbitrary $\levl \in \Level$, $\actxt$ and atomicity context, $h_0, h_1 \in \Heap$, $p, q \in \UCWorld[\actxt]$ and $r \in \View[\actxt]$:
  \[
  \update h_0 -> h_1 |= p -> q \implies \update h_0 -> h_1 |= p * r -> q * r
  \]
\end{lemma}
\begin{proof}
  Assume $\update h_0 -> h_1 |= p -> q$, which is equivalent to:
  \[
  \forall f \in \UCWorld[\actxt] \st
  h_1 \in \sem{p_1 * f}_{\lvl}
  \implies
  h_2 \in \sem{p_2 * {\relyAt[\actxt]}(f)}_{\lvl}
  \]
  Substituting $f = r * f'$, this is equivalent to:
  \[
  \forall f' \in \UCWorld[\actxt] \st
  h_1 \in \sem{p_1 * r * f'}_{\lvl}
  \implies
  h_2 \in \sem{p_2 * {\relyAt[\actxt]}(r * f')}_{\lvl}
  \]
  As $r \in \View[\actxt]$, ${\relyAt[\actxt]}(r * f') = r * {\relyAt[\actxt]}(f')$
  holds, and therefore, as required:
  \[
  \forall f' \in \UCWorld[\actxt] \st
  h_1 \in \sem{p_1 * r * f'}_{\lvl}
  \implies
  h_2 \in \sem{p_2 * r * {\relyAt[\actxt]}(f')}_{\lvl}
  \qedhere
  \]
\end{proof}

\begin{lemma}
  \label{lemma:AFPU_unfold_heap}
  For arbitrary $\levl \in \Level$, $\actxt$ and atomicity context, $h_0, h_1 \in \Heap$, $p, q \in \UCWorld[\actxt]$ and $f \in \View[\actxt]$:
  \[
  h_0 \in \sem{p * f} \land \update h_0 -> h_1 |= p -> q \implies h_1 \in \sem{q * f}
  \]
\end{lemma}
\begin{proof}
  To start off, assume
    $h_0 \in \sem{p * f}$ and
    $\update h_0 -> h_1 |= p -> q$.
  Clearly, this second assumption entails $\update h_0 -> h_1 |= p ->* q$, which is equivalent to:
  \[
  \forall f \in \View[\actxt] \st h_0 \in \sem{p * f} \implies h_1 \in \sem{q * f}
  \]
  Chosing the initial $f$ and applying the first assumption yields $h_1 \in \sem{q * f}$ as required.
\end{proof}

\begin{definition}
  For arbitrary $V \subset \PVar$ and $\trace \in \Trace$,
  we define the predicate $\lfun{noMods}_{V}(\trace)$,
  indentifying traces that only modify the program variables
  in~$V$:
  \[
  \lfun{noMods}_{V}((\store_0, h_0) \pl (\store_1, h_1)\trace') \is
  \lfun{noMods}_{V}((\store_1, h_1)\trace') \land
  \forall \pvar{v} \in V \st
  \store_0(\pvar{v}) = \store_1(\pvar{v})
  \]
\end{definition}

\begin{definition}
  For $V \subseteq \PVar$:
  $
  \Trace_{V} \is \set{\trace \in \Trace | \lfun{noMods}_{V}(\trace)}.
  $
\end{definition}

\begin{lemma}
  \label{lemma:mod_vars}
  Given $\cmd \in \Cmd, \functxt \in \FunCtxt$ and
  $V \subseteq \PVar$ arbitrary such that
  $V \inters \mods(\cmd) = \emptyset$, then:
  $
  \sem{\cmd}_{\functxt} \subseteq \Trace_{V}.
  $
\end{lemma}

\begin{proof}
  Easy coinduction on the small-step operational
  semantics of commands.
\end{proof}

\begin{definition}We define an auxiliary operation that takes a Hoare frame $\na{r} \in \View[\actxt]$
  and an atomic frame $\at{r} \in \UCWorld[\actxt]$
  and applies the frames at each position of a specification trace,
  if the heaps at each position are compatible with said frames (and returns the empty set otherwise).
  \begin{multline*}
  \unframe{\na{r}, \at{r}}{((\store, h, \na{p}, \at{p}, v) \pl \strace)} \is {} \\
    \Set{
      (\store, h, \na{p} * \na{r}, \at{p} * \at{r}, v) \pl \strace'
      |
      \begin{array}{c}
        \strace' \in (\unframe{\na{r}, \at{r}}{\strace}) \land {} \\
        h \in \sem{\na{p} * \na{r} * \at{p}(v) * \at{r}(v) * \True}_\lvl
      \end{array}
    }
  \end{multline*}
  This can be lifted to sets of specification traces, $\straces \subseteq \STrace$:
  \[
  \unframe{\na{r}, \at{r}}{\straces} \is
  \bigcup\limits_{\strace \in \straces} \unframe{\na{r}, \at{r}}{\strace}
  \]
\end{definition}

\begin{lemma}
  \label{lemma:frame_safety}
  For arbitrary $(\store_0, h_0)\trace \in \Trace_{\fv(\na{R})}$, $p_h \in \View[\actxt]$, $v_0 \in \AVal'$ and $\straces \in \mathcal{P}(\STrace)$, then
  \begin{multline*}
    h_0 \in \sem{\na{p} * \na{r} * \at{p}(v_0) * \at{r}(v_0) * \True} \land
    \accept{(\store_0, h_0)\;\trace}{\na{p}}{\at{p}}{v_0} : \straces \implies {} \\
    \accept[\spec']{(\store_0, h_0)\;\trace}{\na{p} * \na{r}}{\at{p} * \at{r}}{v_0} :
    \unframe{\na{r}, \at{r}}{\straces}
  \end{multline*}
  holds, where
  \begin{align*}
    \na{r} &= \WorldSem{\actxt}{\store_0}{\na{R}} \\
    \at{p}(v) &=
    \begin{cases}
      \WSem[\actxt]{\at{P}(v) \land v \in X} \CASE x \in \AVal \\
      \wEmp[\actxt] \OTHERWISE
    \end{cases} \\
    \at{r}(v) &=
    \begin{cases}
      \WSem[\actxt]{\at{R}(v) \land v \in X} \CASE x \in \AVal \\
      \wEmp[\actxt] \OTHERWISE
    \end{cases}
  \end{align*}
\end{lemma}
\begin{proof}
  Taking $(\store_0, h_0)\trace \in \Trace_{\fv(\na{R}) \inters \PVar}$, $p_h \in \View[\actxt]$, $v_0\in \AVal'$ and
  $\straces \in \mathcal{P}(\STrace)$ arbitrary such that:
  \begin{align}
    & h_0 \in \sem{\na{p} * \na{r} * \at{p}(v_0) * \at{r}(v_0) * \True} \label{ass:frame_safety_heap_ass} \\
    & \accept{(\store_0, h_0)\;\trace}{\na{p}}{\at{p}}{v_0} : \straces \label{ass:frame_safety_frameless_ass}
  \end{align}

  The proof proceeds by coinduction on the structure of $\trace$. We consider the rules can
  apply from the trace safety judgement: \ref{rule:stutter}, \ref{rule:linpt}, \ref{rule:env},
  \ref{rule:env2} and \ref{rule:env-fault}.

  \paragraph{Case \ref{rule:stutter}} In this case,
  $(\store_0, h_0)\trace = (\store_0, h_0)\loc(\store_1, h_1)\trace'$ and
  $\straces = (\store_0, h_0, \na{p}, \at{p}, v)\loc\straces'$.
  From \eqref{ass:frame_safety_frameless_ass}, for some $\na{p}' \in \View[\actxt]$, the following hold:
  \begin{align}
    & \update
    h_0 -> h_1 |= \na{p} * \at{p}(v_0) -> \na{p'} * \at{p}(v_0) \label{ass:stutter_frame_fpu}
    \\
    & \accept{(\store_1, h_1)\;\trace'}{\na{p'}}{\at{p}}{v_0} : \straces' \label{ass:stutter_frame_cont}
    \\
    & \term(\trace')
    \implies \exists v_1, v_2 \st
    v = \DONE(v_1,v_2) \land
    \na{p}' = \WorldSem{\actxt}{\store_1}{\na{Q}(v_1,v_2)} \label{ass:frame_done}
  \end{align}

  Given that $\na{r}, \at{r}(v_0) \in \View[\actxt]$, using \cref{lemma:AFPU_extend},
  \eqref{ass:stutter_frame_fpu} implies
  \begin{align}
    \update h_0 -> h_1 |= \na{p} * \na{r} * \at{p}(v_0) * \at{r}(v_0) ->
    \na{p'} * \na{r} * \at{p}(v_0) * \at{r}(v_0) \label{goal:stutter_frame_fpu}
  \end{align}
  By \cref{lemma:AFPU_unfold_heap}, \eqref{ass:frame_safety_heap_ass} and \eqref{ass:stutter_frame_fpu} imply:
  \begin{align}
    h_1 \in \sem{\na{p}' * \na{r} * \at{p}(v_0) * \at{r}(v_0) * \True} \label{goal:stutter_frame_heap}
  \end{align}

  Given that $(\store_0, h_0)\trace \in \Trace_{\fv(\na{R}) \inters \PVar}$, $\forall \pvar{v} \in \fv(\na{R}) \inters \PVar \st \store_o(\pvar{v}) = \store_1(\pvar{v})$ holds,
  and therefore:
  \begin{align}
    \na{r} = \WorldSem{\actxt}{\store_1}{\na{R}} \label{goal:stutter_frame_rh_pres}
  \end{align}
  
  From this, given \eqref{goal:stutter_frame_rh_pres}, \eqref{goal:stutter_frame_heap} and
  \eqref{ass:stutter_frame_cont} and using the inductive assumption, we derive:
  \begin{align}
    \accept[\spec']{(\store_1, h_1)\;\trace'}{\na{p'} * \na{r}}{\at{p} * \at{r}}{v_0} : \unframe{\na{r}, \at{r}}{\straces'}
    \label{goal:stutter_frame_cont}
  \end{align}

  Finally, assuming $\term(\trace')$, given \eqref{ass:frame_done},
  we know $\exists v_1, v_2 \st v = \DONE(v_1,v_2) \land \na{p}' = \WorldSem{\actxt}{\store_1}{\na{Q}(v_1,v_2)}$.
  From this and \eqref{goal:stutter_frame_rh_pres}, we infer that
  $\na{p}' * \na{r} = \WorldSem{\actxt}{\store_1}{\na{Q}(v,v') * \na{R}}$ and therefore,
  \begin{align}
    \term(\trace') \implies \exists v_1, v_2 \st v = \DONE(v_1,v_2) \land \na{p}' * \na{r} = \WorldSem{\actxt}{\store_1}{\na{Q}(v,v') * \na{R}}
    \label{goal:stutter_frame_done}
  \end{align}

  From \eqref{goal:stutter_frame_fpu}, \eqref{goal:stutter_frame_cont} and \eqref{goal:stutter_frame_done} by \ref{rule:stutter},
  $\accept[\spec']{(\store_0, h_0)\;\trace}{\na{p} * \na{r}}{\at{p} * \at{r}}{v_0} : \unframe{\na{r}, \at{r}}{\straces}$ holds as required. 

  \paragraph{Case \ref{rule:linpt}}
  In this case,
  $(\store_0, h_0)\trace = (\store_0, h_0)\loc(\store_1, h_1)\trace'$ and
  $\straces = (\store_0, h_0, \na{p}, \at{p}, v)\loc\straces'$.
  From \eqref{ass:frame_safety_frameless_ass}, the following hold for some $v' \in \AVal$:
  \begin{align}
    & \update
    h_0 -> h_1 |=     \na{p} \ssep \at{p}(v_0)
    ->
    \na{q'} \ssep \WorldSem{\actxt}{}{\at{Q}(v_0,v')} \label{ass:linpt_frame_fpu}
    \\
    & \accept{(\store_1, h_1)\;\trace'}{\na{q'}}{\emp}{\DONE(v_0,v')} : \straces' \label{ass:linpt_frame_cont}
    \\
    &\term(\trace')
    \implies
    \na{q}' = \WorldSem{\actxt}{\store_1}{\na{Q}(v_0,v')} \label{ass:linpt_frame_done}
  \end{align}

  Given that $\na{r}, \at{r}(v_0) \in \View[\actxt]$, using \cref{lemma:AFPU_extend},
  \eqref{ass:linpt_frame_fpu} implies
  \begin{align}
    \update h_0 -> h_1 |= \na{p} * \na{r} * \at{p}(v_0) * \at{r}(v_0) ->
    \na{q'} * \na{r} * \WorldSem{\actxt}{}{\at{Q}(v,v')} * \at{r}(v_0) \label{goal:linpt_frame_fpu}
  \end{align}
  By \cref{lemma:AFPU_unfold_heap}, \eqref{ass:frame_safety_heap_ass} and \eqref{ass:linpt_frame_fpu} imply:
  \begin{align}
    h_1 \in \sem{\na{q}' * \na{r} * \WorldSem{\actxt}{}{\at{Q}(v_0,v')} * \at{r}(v_0) * \True}
    \label{goal:linpt_frame_heap}
  \end{align}

  Given that $(\store_0, h_0)\trace \in \Trace_{\fv(\na{R}) \inters \PVar}$, $\forall \pvar{v} \in \fv(\na{R}) \inters \PVar \st \store_o(\pvar{v}) = \store_1(\pvar{v})$ holds,
  and therefore:
  \begin{align}
    \na{r} = \WorldSem{\actxt}{\store_1}{\na{R}} \label{goal:linpt_frame_rh_pres}
  \end{align}

  From this, given \eqref{goal:linpt_frame_rh_pres}, \eqref{goal:linpt_frame_heap} and
  \eqref{ass:linpt_frame_cont} and using the inductive assumption, we derive:
  \begin{align}
    \accept[\spec']{(\store_1, h_1)\;\trace'}{\na{q'} * \na{r}}{\at{p} * \at{r}}{\DONE(v_0, v')} : \unframe{\na{r}, \at{r}}{\straces'}
    \label{goal:linpt_frame_cont}
  \end{align}

  Finally, assuming $\term(\trace')$, given \eqref{ass:linpt_frame_done},
  we know $\na{q}' = \WorldSem{\actxt}{\store_1}{\na{Q}(v_0,v')}$.
  From this and \eqref{goal:linpt_frame_rh_pres}, we infer that
  $\na{q}' * \na{r} = \WorldSem{\actxt}{\store_1}{\na{Q}(v_0,v') * \na{R}}$ and therefore,
  \begin{align}
    \term(\trace') \implies \na{q}' * \na{r} = \WorldSem{\actxt}{\store_1}{\na{Q}(v_0,v') * \na{R}}
    \label{goal:linpt_frame_done}
  \end{align}

  From \eqref{goal:linpt_frame_fpu}, \eqref{goal:linpt_frame_cont} and \eqref{goal:linpt_frame_done} by \ref{rule:linpt},
  $\accept[\spec']{(\store_0, h_0)\;\trace}{\na{p} * \na{r}}{\at{p} * \at{r}}{v_0} : \unframe{\na{r}, \at{r}}{\straces}$ holds as required. 
  
  \paragraph{case \ref{rule:env}} In this case,
  $(\store_0, h_0)\trace = (\store_0, h_0)\env(\store_1, h_1)\trace'$ and
  \[\straces =  \textstyle\Union \Set{
      (\store,h_1, \na{p},\at{p},v_0) \env \straces[v']' | v' \in X, E(v') }.\]
  From \eqref{ass:frame_safety_frameless_ass} we have that
  $
  \forall v' \in X\st
  E(v')
  \implies
  \accept{(\store, h_2)\;\trace}{\na{p}}{\at{p}}{v'} : \straces[v']
  $.
Taking $v' \in X$ arbitrary and, assuming $E(v')$ given some $\overline{\pe},\overline{\pe}'$, for the goal specification:
  \begin{align*}
    &h_0 \in \sem{ \na{p} \ssep \na{r} \ssep \at{p}(v_0) \ssep \at{r}(v_0) \ssep \overline{\pe}}_\lvl
    \\
    &\update h_0 -> h_1 |= \at{p}(v_0) \ssep \at{r}(v_0) \ssep \overline{\pe} ->
    \at{p}(v') \ssep \at{r}(v')  \ssep \overline{\pe}'
  \end{align*}
  It then suffices to show that $\accept[\spec']{(\store_1, h_1)\;\trace'}{\na{p} * \na{r}}{\at{p} * \at{r}}{v'} : \straces[v']$ holds. This follows from \cref{lemma:AFPU_extend} by using
  $
    \pe  = \overline{\pe} \ssep \na{r} \ssep \at{r}(v_0)
  $ and $
    \pe' = \overline{\pe}' \ssep \na{r} \ssep \at{r}(v')
  $, yielding:
  \begin{align*}
    &h_1 \in \sem{ \na{p} \ssep \at{p}(v) \ssep \pe}_\lvl &
    &\update h_1 -> h_2 |= \at{p}(v) \ssep \pe ->
    \at{p}(v')  \ssep \pe'
  \end{align*}
  as required.

  \paragraph{Case \ref{rule:env2}} This case follows similarly to the \ref{rule:env} case.
  
  \paragraph{Case \ref{rule:env-fault}} This case is trivially true.
\end{proof}

\begin{lemma}
  \label{lemma:frame-liveness}
  Letting
  \begin{align*}
    \na{r} &= \WorldSem{\actxt}{\store_0}{\na{R}} \\
    \at{r}(v) &=
    \begin{cases}
      \WSem[\actxt]{\at{R}(v) \land v \in X} \CASE x \in \AVal \\
      \wEmp[\actxt] \OTHERWISE
    \end{cases}
  \end{align*}
  and assuming $\VALID \actxt |= \minLay{\na{R} * \at{R}(x)}{m}$, then,
  given $\straces \subseteq \STrace$:
  \[
  \forall \wtrace' \in \wtr{\unframe{\na{r}, \at{r}}{\straces}} \st
  \goodenv_{\spec'}(\wtrace') \implies \exists \wtrace \in \wtr{\straces} \st
  \goodenv_{\spec}(\wtrace)
  \]
\end{lemma}

\begin{proof}
  Taking $\wtrace' \in  \wtr{\unframe{\na{r}, \at{r}}{\straces}}{}$
  arbitrary such that $\goodenv_{\spec'}(\wtrace')$. This implies that:
  \begin{align*}
    \forall \pobl \in \POLay{m}^{\spec'}\st
      \;&\textbf{if }
        \forall \rid, O \in \AObl_{\le\lay(\pobl)}\st
        \forall i\in \Nat\st
          \exists j \geq i\st
            \neg \locheld_\lvl(\rid, O,\trAt[\wtrace']{j})
      \\
      &\textbf{then }
        \forall i\in \Nat\st
          \exists j \geq i\st
            \neg \envheld_\lvl(\smash{\pobl},\trAt[\wtrace']{j})
  \end{align*}

  As $\wtrace' \in  \wtr{\unframe{\na{r}, \at{r}}{\straces}}{}$,
  there must be some $\strace \in \straces$ such that
  $\wtrace' \in  \wtr{\unframe{\na{r}, \at{r}}{\strace}}{}$.
  Taking $\wtrace \in \wtr{\strace}{}$ arbitrary, to show $\goodenv_{\spec}(\wtrace)$,
  take $\pobl \in \POLay{m}^{\spec}$ arbitrary such that
  \[
  \forall \rid, O \in \AObl_{\le\lay(\pobl)}\st
  \forall i\in \Nat\st
  \exists j \geq i\st
  \neg \locheld_\lvl(\rid, O,\trAt[\wtrace]{j})
  \]

  Given $\VALID \actxt |= \minLay{\na{R} * \at{R}(x)}{m}$ and
  the definition of $\unframe*$,
  the following holds:
  \[
  \forall \rid, O \in \AObl_{\le\lay(\pobl)}\st
  \forall i\in \Nat\st
  \exists j \geq i\st
  \neg \locheld_\lvl(\rid, O,\trAt[\wtrace']{j})
  \]

  Now, from $\goodenv_{\spec'}(\wtrace')$:
  \[
  \forall i\in \Nat\st
  \exists j \geq i\st
  \neg \envheld_\lvl(\smash{\pobl},\trAt[\wtrace']{j})
  \]
  From this, as required:
  \[
  \forall i\in \Nat\st
  \exists j \geq i\st
  \neg \envheld_\lvl(\smash{\pobl},\trAt[\wtrace]{j})
  \]

\end{proof}

\begin{theorem}[Soundness of \ref{rule:frame}]
  \label{soundness:frame}
  Assuming
  \begin{align}
    & \forall x \in X \st \VALID \actxt |= \minLay{\na{R} * \at{R}(x)}{m} \label{ass:layer}
  \end{align}
  and given arbitrary $\cmd \in \Cmd$ such that
  \begin{align}
    & \progvars(\na{R}) \inters \modvars(\cmd) = \emptyset
    \label{ass:frame_progvars}
  \end{align}
  and arbitrary $\funspec\in\FunSpec$ such that
  \[
  \JUDGE[\funspec] |= \cmd : \spec
  \]
  then
  \[
  \JUDGE[\funspec] |= \cmd : \spec'
  \]
\end{theorem}

\begin{proof}

  To start off, as $\STABLE \actxt |= { \na{R} }$, clearly $\na{P} \ssep \na{R} \in  \Stable[\actxt]$,
  $\forall x \in X, y \st \na{Q}(x,y) \ssep \na{R} \in  \Stable[\actxt]$,
  $\forall x \in X \st  \VALID \actxt |= \at{P}(x) * \at{R}(x) \implies \empObl[\lvl]$ and
  $\forall x \in X, y \st  \VALID \actxt |= \at{Q}(x,y) * \at{R}(x) \implies \empObl[\lvl]$
  therefore, $\spec' \in \Spec$.

  Taking $\cmd \in \Cmd$ arbitrary such that \eqref{ass:frame_progvars} holds,
  $\funspec \in \FunSpec$ arbitrary such that $\JUDGE[\funspec] |= \cmd : \spec$ holds
  and arbitrary $\functxt \in \FunCtxt$ such that $\JUDGE |= \functxt : \funspec$ holds,
  then $\sem{\cmd}_{\functxt} \subseteq \Sem{\spec}$.
  From \cref{lemma:mod_vars} and \eqref{ass:frame_progvars},
  we can also infer that $\sem{\cmd}_{\functxt} \subseteq \Trace_{\progvars(\na{R})}$
  and therefore, it is clear that $\sem{\cmd}_{\functxt} \subseteq \Trace_{\progvars(\na{R})} \inters \Sem{\spec}$.
  From this, we know that it is sufficient to show that
  $\Trace_{\progvars(\na{R})} \inters \Sem{\spec} \subseteq \Sem{\spec'}$,
  to show that $\sem{\cmd}_{\functxt} \subseteq \Sem{\spec'}$ holds,
  and therefore, $\JUDGE[\funspec] |= \cmd : \spec'$, as required.
  
  Therefore, taking $(\store_0,h_0)\trace \in \Sem{\spec} \inters \Trace_{\progvars(\na{R})}$ arbitratry,
  it is sufficient to show $(\store_0,h_0)\trace \in \Sem{\spec'}$.

  Let
  \begin{align*}
    \na{p}    &= \WorldSem{\actxt}{\store_0}{\na{P}}\\
    \na{r}    &= \WorldSem{\actxt}{\store_0}{\na{R}}\\
    \at{p}(v) &=
    \begin{cases}
      \WSem[\actxt]{\at{P}(v) \land v \in X} \CASE x \in \AVal \\
      \wEmp[\actxt] \OTHERWISE
    \end{cases} \\
    \at{r}(v) &=
    \begin{cases}
      \WSem[\actxt]{\at{R}(v) \land v \in X} \CASE x \in \AVal \\
      \wEmp[\actxt] \OTHERWISE
    \end{cases}
  \end{align*}
  
  To show $(\store_0,h_0)\trace \in \Sem{\spec'}$, for some arbitrary $v_0 \in X$, assume $h_0 \in \sem{\na{p} \ssep \na{r} \ssep \at{p}(v_0) \ssep \at{r}(v_0) \ssep \True}_\lvl$, from which it follows that $h_0 \in \sem{\na{p} \ssep \at{p}(v_0) \ssep \True}_\lvl$. Then, as  $(\store_0,h_0)\trace \in \Sem{\spec}$, for some $\straces \subseteq \STrace$:
  \begin{align}
  &\accept{(\store_0, h_0)\;\trace}{\na{p}}{\at{p}}{v_0} : \straces \label{given:safety} \\
  &\forall \wtrace \in \wtr{\straces} \st \goodenv_{\spec}(\wtrace) \implies \locterm(\wtrace) \label{given:liveness}
  \end{align}

  From \cref{lemma:frame_safety} and \eqref{given:safety},
  $\accept{(\store_0, h_0)\;\trace}{\na{p} * \na{r}}{\at{p} * \at{r}}{v_0} : \unframe{\na{r}, \at{r}}{\straces}$.
  To reach the goal now, it suffices to show that for some arbitrary
  $\wtrace' \in \wtr{\unframe{\na{r}, \at{r}}{\straces}}{}$:
  \[
  \goodenv_{\spec'}(\wtrace') \implies \locterm(\wtrace')
  \]
  This holds trivially from \cref{lemma:frame-liveness}, \eqref{ass:layer} and \eqref{given:liveness}.
  
\end{proof}

Note that \cref{soundness:frame} has the side condition $\progvars(\na{R}) \inters \modvars(\cmd) = \emptyset$
rather than $\forall x \in X \st \progvars(\na{R}, \at{R}(x)) \inters \modvars(\cmd) = \emptyset$ as in \ref{rule:frame}.
This is because this theorem applies to \tadalive\ specifications without the syntactic sugar that permits
program variables to be directly referenced in the atomic precondition and postcondition of a \tadalive\ hybrid triple.
The side condition present in the \ref{rule:frame} rule permits it to be applied directly to sugared hybrid specifications,
as it guarantees that the necessary side condition for the corresponding desugared specification holds.
 
\subsection{Soundness of \ref{rule:liveness-check}}

Let
\begin{align*}
    \spec_{\twoheadrightarrow} &= \SPEC m;\lvl;\actxt |= \A x \in X \eventually[k] X'. <\na{P}|\at{P}(x)> \E y. <\na{Q}(x,y)|\at{Q}(x,y)> \\
    \spec &= \SPEC m;\lvl;\actxt |= \A x \in X. <\na{P} * L|\at{P}(x)> \E y. <\na{Q}(x,y) * L|\at{Q}(x,y)>
\end{align*}
where $L \in \Stable[\actxt]$.

\newcommand{\goodLoc}{\lfun{goodLoc}}

\begin{definition}
  For atomicity context $\actxt$ and layer $m$ from the context of $\spec$
  and sets $X$ and $X$' as well as the layer $k$ from the pseudo-quantifier of $\spec$,
  let
  \[
  \PObl^{\spec} \is
    \Set{(\rid, O) | \rid \in \RId, O \in \AObl} \dunion
    \Set{(\rid,\live(\actxt,\rid)) | \rid \in \dom(\actxt)} \dunion
    \set{X \eventually[k] X'}
  \]
  Then $\goodenv_\spec(\strace)$ predicate checks whether the environment
  is satisfying the liveness assumptions of the specification:
  \begin{align*}
  \goodenv_\spec(\wtrace) \is
    \forall \pobl \in \POLay{m}^{\spec}\st
      \;&\textbf{if }
        \forall \rid, O \in \AObl_{\le\lay(\pobl)}\st
        \forall i\in \Nat\st
          \exists j \geq i\st
            \neg \locheld_\lvl(\rid, O,\trAt[\wtrace]{j})
      \\
      &\textbf{then }
        \forall i\in \Nat\st
          \exists j \geq i\st
            \neg \envheld_\lvl(\smash{\pobl},\trAt[\wtrace]{j})
  \end{align*}
\end{definition}

\begin{lemma}
  \label{lemma:liveC-key-liveness-pq}
  Given $M \in \Ord \to \UCWorld[\actxt], T \in \Assrt, n \le m, k$
  such that
  \begin{align}
    & \SENVLIVE n; \lvl; \actxt |= L -\smash{M}->> T
    \label{ass:liveC-key-liveness-pq-envlive} \\
    & \forall x \in X \st \VALID \lvl; \actxt |= P(x) * T \implies x \in X'
    \label{ass:liveC-key-liveness-pq-T-imp-X'}
  \end{align}
  hold. Take $(\store_0, h_0) \trace \in \Trace$ and let
  \begin{align*} 
    \at{p}(v) &= \begin{cases}
      \WSem[\actxt]{\at{P}(v) \land v \in X} \CASE x \in \AVal \\
      \wEmp[\actxt] \OTHERWISE
    \end{cases} \\
    l &= \WorldSem{\actxt}{\store_0}{L} \\
    l(\alpha) &= \WorldSem{\actxt}{\store_0}{L * M(\alpha)} \\
    t &= \WorldSem{\actxt}{\store_0}{T * \True}
  \end{align*}
  Taking arbitrary $\na{p}' \in \View[\actxt], \straces \subseteq \STrace$ and $v_0 \in X$ such that
  \begin{align}
    & h_0 \in \sem{\na{p}' * l * \at{p}(v_0) * \True}_\lvl
    \label{ass:liveC-key-liveness-pq-init} \\
    & \accept[\spec_{\twoheadrightarrow}]{(\store_0, h_0) \trace}{\na{p}'}{\at{p}}{v_0} : \straces
    \label{ass:liveC-key-liveness-pq-traces-ant}\\
    & \accept[\spec]{(\store_0, h_0) \trace}{\na{p}' * l}{\at{p}}{v_0} : \unframe{l, \emp}{\straces}
    \label{ass:liveC-key-liveness-pq-traces-cons}
  \end{align}
  for arbitrary $\wtrace' \in \wtr{\unframe{l, \emp}{\straces}}{\lvl;\actxt}$ such that
  \begin{align}
    & \goodenv_{\spec}(\wtrace')
    \label{ass:liveC-key-liveness-pq-liveEnv}
  \end{align}
  there exists $\wtrace \in \wtr{\straces}{\lvl;\actxt}$, such that, if
  \begin{align}
    \forall \pobl' \in \RId \times \AObl_{\le k}\st
    \forall i\in \Nat\st \exists j \ge i\st \neg \locheld_{\lvl}(\pobl',\trAt[\wtrace]{j})
    \label{ass:liveC-key-liveness-pq-liveEnv-pre}
  \end{align}
  then
  $
  \forall i \in \Nat \st
  \exists j \ge i \st \neg \envheld_{\lvl}(X \eventually[k] X',\trAt[\wtrace]{j}).
  $
\end{lemma}

\begin{proof}
  Taking $(\store_0, h_0) \trace \in \Trace, \na{p}' \in \View[\actxt], \straces \subseteq \STrace$ and $v_0 \in X$
  arbitrary such that (\ref{ass:liveC-key-liveness-pq-init}), (\ref{ass:liveC-key-liveness-pq-traces-ant}) and
  (\ref{ass:liveC-key-liveness-pq-traces-cons}) hold and
  $\wtrace' \in \wtr{\unframe{l, \emp}{\straces}}{\lvl;\actxt}$
  satisfying (\ref{ass:liveC-key-liveness-pq-liveEnv}).
  From this,
  we know that, there exists $\wtrace \in \wtr{\straces}{\lvl;\actxt}$ and $\lvar{ls} \in l^{\omega}$,
  such that, for all $i \in \Nat$:
  \begin{align}
    & \lvar{ls}(i) \relyAt[\actxt] \lvar{ls}(i + 1)
    \label{ass:liveC-key-liveness-pq-step} \\
    & \exists h \in \Heap, \na{w}, \at{w}, \en{w} \in \View[\actxt], v \in \Val' \st
    \begin{array}{c}
      \wtrace(i) = ((h,\na{w}, \at{w}, \en{w} \worldSsep \lvar{ls}(i), v), \_) \land {} \\
      \wtrace'(i) = ((h,\na{w} \worldSsep \lvar{ls}(i), \at{w}, \en{w}, v), \_)
    \end{array}
    \label{ass:liveC-key-liveness-trace-rel}
  \end{align}
  hold. If (\ref{ass:liveC-key-liveness-pq-liveEnv-pre}) does not hold, our proof
  is complete, otherwise, from (\ref{ass:liveC-key-liveness-pq-envlive}), there exists some
  $P \subseteq \powerset(\Ord \to \UCWorld[\actxt])$ such
  that $\forall \alpha \st l(\alpha) = \Union_{\lvar{lf} \in P} \lvar{lf}(\alpha)$.

  We now show by transfinite induction on $\alpha \in \Ord$, that
  \[
  \forall \alpha \in \Ord, i \in \Nat \st
  \lvar{ls}(i) \in l(\alpha) \implies
  \exists j \ge i \st \neg \envheld_{\lvl}(X \eventually[k] X',\trAt[\wtrace]{j})
  \]

  \begin{proofscript}
  \item[Base case ($\alpha = 0$):]
  Take $i \in \Nat$ and assume
  $
  \lvar{ls}(i) \in l(0)
  $.
  Since $l(0) = \Union_{\lvar{lf} \in P} \lvar{lf}(0)$,
  for some $\lvar{lf} \in P$ we have $\lvar{ls}(i) \in \lvar{lf}(0)$.
  We now assume, towards a contradiction,
  that $
    \forall j \ge i \st
      \envheld_{\lvl}(X \eventually[k] X',\trAt[\wtrace]{j})
  $
  and therefore $
    \forall j \ge i \st
      \exists v \in X \setminus X' \st \trAt[\wtrace]{j} = ((\wtv[4], v),\_)
  $.
  We now demonstrate, that under this assumption, by induction on $j \ge i$ that
  \begin{align}
    \forall j \ge i \st \lvar{ls}(j) \in \lvar{lf}(0)
    \label{int:liveC-key-liveness-base-worst-case}
  \end{align}
  Assume that for $j \ge i$, $\lvar{ls}(j) \in \lvar{lf}(0)$
  holds. From (\ref{ass:liveC-key-liveness-pq-step}), $\lvar{ls}(j + 1) \in {\relyAt[\actxt]}(\lvar{lf}(0))$
  holds and from (\ref{ass:liveC-key-liveness-pq-envlive}), setting $\alpha_1 = 0$,
  $\relyAt[\actxt](\lvar{lf}(0)) \subseteq \lvar{lf}(0) \union t$,
  therefore, either $\lvar{ls}(j + 1) \in \lvar{lf}(0)$ or $\lvar{ls}(j + 1) \in t$ hold.
  In the latter case, from (\ref{ass:liveC-key-liveness-pq-T-imp-X'}),
  $\exists v \in X' \st \wtrace(j + 1) = ((\wtv[4], v), \_)$, a contradiction.
  Therefore, $\lvar{ls}(j + 1) \in \lvar{lf}(0)$ holds,
  proving \eqref{int:liveC-key-liveness-base-worst-case}.

  From (\ref{ass:liveC-key-liveness-pq-envlive}), there exists some $\rid \in \RId$ and
  $\pobl \in \AObl_{<n} \dunion \Set{\live(\actxt,\rid) | \lay(\live(\actxt,\rid)) < n}$
  such that:
  \begin{align}
    \forall w \in \lvar{lf}(0) \st \act_{\rid;\lvl}(w, \pobl)
    \label{int:liveC-key-liveness-act}
  \end{align}

  Taking $\rid' \in \RId$ and $\obl{O} \in \AObl_{\le\lay(\pobl)}$ arbitrary,
  since $\lay(\pobl) < n$, $\lay(\obl{O}) < n$ and therefore $\lay(\obl{O}) < k$ and $\lay(\obl{O}) < m$ hold.
  As $\lay(\obl{O}) < k$ holds, given (\ref{ass:liveC-key-liveness-pq-liveEnv-pre}),
  for some $j' \ge i$, $\neg \locheld_{\lvl}((\rid', \obl{O}),\trAt[\wtrace]{j'})$ holds.
  Given (\ref{ass:liveC-key-liveness-trace-rel}), we know
  $\trAt[\wtrace]{j'} = (h, \na{w}, \at{w}, \en{w} \worldSsep \lvar{ls}(j'), v)$ and
  $\trAt[\wtrace']{j'} = (h, \na{w} \worldSsep \lvar{ls}(j'), \at{w}, \en{w}, v)$ for some
  $h \in \Heap, \na{w}, \at{w}, \en{w} \in \World[\actxt], v \in X \setminus X'$, and
  therefore, $\oblMap[\na{w}](\rid') \not\resgeq \obl{O}$.
  Given (\ref{int:liveC-key-liveness-base-worst-case}), $\lvar{ls}(j') \in \lvar{lf}(0)$
  holds, and, therefore, given (\ref{int:liveC-key-liveness-act}), we know that:
  \begin{align}
    & \lay(\oblMap[\lvar{ls}(j')](\rid')) > \lay(\pobl)
    \label{int:liveC-key-liveness-loc-obl-lay}\\
    & \level[\lvar{ls}(j')](\rid) < \lvl
    \label{int:liveC-key-liveness-region-lvl}
  \end{align}
  Given (\ref{int:liveC-key-liveness-loc-obl-lay}), $\oblMap[\lvar{ls}(j')](\rid') \not\resgeq \obl{O}$,
  as otherwise, $\lay(\oblMap[\lvar{ls}(j')](\rid')) \le \lay(\obl{O})$ and therefore, as
  $\lay(\obl{O}) < \lay(\pobl)$, we reach a contradiction.
  Then, from $\oblMap[\na{w}](\rid) \not\resgeq \obl{O}$
  and $\oblMap[\lvar{ls}(j')](\rid) \not\resgeq \obl{O}$, we know
  $\oblMap[\na{w} \worldSsep \lvar{ls}(j')](\rid) \not\resgeq \obl{O}$,
  as $\obl{O} \in \AObl$, and therefore, $\neg \locheld_{\lvl}((\rid',\obl{O}),\trAt[\strace']{j'})$.
  From this, it follows that
  \begin{align}
    \forall \pobl' \in \RId \times \AObl_{\le \lay(\pobl)}, i \in \Nat \st
    \exists j \ge i \st \neg \locheld_{\lvl}(\pobl',\trAt[\wtrace']{j})
    \label{int:liveC-key-liveness-base-loc-obl-eventually}
  \end{align}
  holds.

  Letting $\trAt[\wtrace']{i} = ((h^i, \na{w}^i \worldSsep \lvar{ls}(i), \at{w}^i, \en{w}^i, v^i), \_)$,
  first, consider the case $\pobl \in \AObl_{<n}$. As an obligation's
  composition with itself within the obligation algebra of the region type
  of the shared region $\rid$ must be undefined, one of
  \begin{align*}
    \locheld_{\lvl}(\pobl,\trAt[\wtrace']{i}) \land \neg \envheld_{\lvl}(\pobl,\trAt[\wtrace']{i}) \\
    \neg \locheld_{\lvl}(\pobl,\trAt[\wtrace']{i}) \land \envheld_{\lvl}(\pobl,\trAt[\wtrace']{i}) \\
    \neg \locheld_{\lvl}(\pobl,\trAt[\wtrace']{i}) \land \neg \envheld_{\lvl}(\pobl,\trAt[\wtrace']{i})
  \end{align*}
  holds, as if both the environment and local worlds hold $\pobl$, their composition
  would be undefined.

  If $\neg \locheld_{\lvl}(\pobl,\trAt[\strace']{i}) \land \neg \envheld_{\lvl}(\pobl,\trAt[\strace']{i})$
  holds, then $\oblMap[\na{w}^i \worldSsep \lvar{ls}(i)](\rid) \not\resgeq \pobl$ and
  $\oblMap[\en{w}^i](\rid) \not\resgeq \pobl$ hold. 
  From (\ref{int:liveC-key-liveness-base-worst-case}), we know that $\lvar{ls}(i) \in \lvar{lf}(0)$,
  which, given (\ref{int:liveC-key-liveness-act}) and $\pobl \in \AObl_{<n}$
  implies that $\envMap[\lvar{ls}(i)](\rid) \resgeq \pobl$.
  Given (\ref{int:liveC-key-liveness-region-lvl}) and
  the invariant on atomic resources, we also know that $\oblMap[\at{w}^i](\rid) \not\resgeq \pobl$,
  and therefore, since we know that
  $\envMap[\na{w}^i \worldSsep ls(i) \worldSsep \at{w}^i \worldSsep \en{w}^i](\rid) = \oblZero$,
  given the definition of $\worldSsep$, we reach a contradiction, as required.

  Otherwise, if $\neg \locheld_{\lvl}(\pobl,\trAt[\wtrace']{i}) \land \envheld_{\lvl}(\pobl,\trAt[\wtrace']{i})$
  holds, from (\ref{ass:liveC-key-liveness-pq-liveEnv})
  \begin{align*}
      \;&\textbf{if }
        \forall \pobl' \in \RId \times \AObl_{\le\lay(\pobl)}, i \in \Nat\st
          \exists j \geq i\st
            \neg \locheld(\pobl',\trAt[\wtrace']{j})
      \\
      &\textbf{then }
        \forall i\in \Nat\st
          \exists j \geq i\st
            \neg \envheld(\smash{\pobl},\trAt[\wtrace']{j})
  \end{align*}
  holds, and therefore, given (\ref{int:liveC-key-liveness-base-loc-obl-eventually}),
  there exists some minimal $j \geq i$, such that
  $\neg \envheld(\smash{\pobl},\trAt[\wtrace']{j + 1})$ holds.
  Since $j$ is minimal,
  $\forall k \in \Nat \st i \le k \le j \implies \envheld(\smash{\pobl},\trAt[\wtrace']{k})$
  also holds. From this, letting
  \begin{align*}
    \trAt[\wtrace']{j} &= ((h^j, \na{w}^j \worldSsep \lvar{ls}(j), \at{w}^{j}, \en{w}^{j}, v^{j}), \pl)\\
    \trAt[\wtrace']{j+1} &= ((h^{j+1}, \na{w}^{j+1} \worldSsep \lvar{ls}(j + 1), \at{w}^{j+1}, \en{w}^{j+1}, v^{j+1}), \_)
  \end{align*}
  we know that $\oblMap[\en{w}^{j}](\rid) \resgeq \pobl$ and
  $\oblMap[\en{w}^{j+1}](\rid) \not\resgeq \pobl$.
  Then, given (\ref{int:liveC-key-liveness-region-lvl}) and the invariant
  on atomic resources, we know $\oblMap[\at{w}^{j}](\rid) = \oblZero$ and
  $\oblMap[\at{w}^{j+1}](\rid) = \oblZero$, and therefore, by the definition of $\worldSsep$,
  $\envMap[\na{w}^{j} \worldSsep \lvar{ls}(j) \worldSsep \at{w}^{j}](\rid) \resgeq \pobl$ and
  $\envMap[\na{w}^{j+1} \worldSsep \lvar{ls}(j + 1) \worldSsep \at{w}^{j+1}](\rid) \not\resgeq \pobl$.
  
  If $\pl = \loc$, by construction of $\wtrace'$, $\en{w}^{j} \relyAt[\actxt] \en{w}^{j+1}$
  holds, and therefore, from the definition of $\relyAt[\actxt]$,
  $\oblMap[\en{w}^{j}](\rid) = \oblMap[\en{w}^{j + 1}](\rid)$,
  a contradiction.
  
  Otherwise, if $\pl = \env$, by construction of $\wtrace'$, $\na{w}^{j} \worldSsep \lvar{ls}(j) \relyAt[\actxt] \na{w}^{j+1} \worldSsep \lvar{ls}(j+1)$
  holds, and therefore, from the definition of $\relyAt[\actxt]$,
  $\oblMap[\na{w}^{j} \worldSsep \lvar{ls}(j)](\rid) = \oblMap[\na{w}^{j + 1} \worldSsep \lvar{ls}(j+1)](\rid)$.
  Given (\ref{ass:liveC-key-liveness-pq-step}), $\oblMap[\lvar{ls}(j)](\rid) = \oblMap[\lvar{ls}(j+1)](\rid)$,
  and therefore, from the definition of $\worldSsep$, $\oblMap[\na{w}^{j}](\rid) = \oblMap[\na{w}^{j+1}](\rid)$.
  Given that we know that
  $\oblMap[\na{w}^j \worldSsep ls(j) \worldSsep \at{w}^j](\rid) \compat \oblMap[\en{w}^j](\rid)$,
  then clearly $\oblMap[\na{w}^j \worldSsep ls(j) \worldSsep \at{w}^j](\rid) \not\resgeq \pobl$,
  from which it follows that
  $\oblMap[\na{w}^{j+1} \worldSsep ls(j+1) \worldSsep \at{w}^{j+1}](\rid) \not\resgeq \pobl$. From
  this and (\ref{int:liveC-key-liveness-base-worst-case}) it is clear that
  $\envMap[\na{w}^{j+1} \worldSsep ls(j+1) \worldSsep \at{w}^{j+1}](\rid) \resgeq \pobl$.
  However, given that $\oblMap[\en{w}^{j+1}](\rid) \not\resgeq \pobl$,
  from the definition of $\worldSsep$, it must be the case that
  $\envMap[\na{w}^{j+1} \worldSsep ls(j+1) \worldSsep \at{w}^{j+1} \worldSsep \en{w}^{j+1}](\rid) \resgeq \pobl$,
  a contradiction.

  Finally, the case $\locheld_{\lvl}(\pobl,\trAt[\wtrace']{i}) \land \neg \envheld_{\lvl}(\pobl,\trAt[\wtrace']{i})$
  reaches a contradiction similarly. To finish the base case, it now suffices to consider the case
  $\pobl = X \eventually[k] X'$. Given (\ref{ass:liveC-key-liveness-pq-liveEnv}) and
  (\ref{int:liveC-key-liveness-base-loc-obl-eventually}), we know that for some $j \ge i$,
  $\neg \envheld_{\lvl}(\smash{\pobl},\trAt[\wtrace']{j})$ holds.
  Letting $\trAt[\wtrace']{j} = ((\_,\na{w}^{j} \worldSsep \lvar{ls}(j),\wtv[3]),\_)$,
  given (\ref{int:liveC-key-liveness-act}), $\level[\lvar{ls}(j)](\rid) < \lvl$ and
  $\astate[\lvar{ls}(j)](\rid) \in X \setminus X'$ hold.
  Given $\level[\lvar{ls}(j)](\rid) < \lvl$ and $\neg \envheld_{\lvl}(\smash{\pobl},\trAt[\wtrace']{j})$,
  $\astate[\lvar{ls}(j)](\rid) \in X'$ holds, a contradiction.

  By contradiction, the base case holds as required.
  
  \item[Inductive case:]
  Take $\alpha \in \Ord, i \in \Nat$ and assume
  $
  \lvar{ls}(i) \in l(\alpha)
  $.
  Since $l(\alpha) = \Union_{\lvar{lf} \in P} \lvar{lf}(\alpha)$
  holds for some $\lvar{lf} \in P$,
  we have $\lvar{ls}(i) \in \lvar{lf}(\alpha)$.
  Now assume, towards a contradiction,
  that $
    \forall j \ge i \st
      \envheld_{\lvl}(X \eventually[k] X',\trAt[\wtrace]{j})
  $
  and therefore $
    \forall j \ge i \st
      \exists v \in X \setminus X' \st \trAt[\wtrace]{j} = ((\wtv[4], v),\_)
  $.
  We now demonstrate, that under this assumption the following holds:
  \begin{align}
    (\forall j \ge i \st \lvar{ls}(j) \in \lvar{lf}(\alpha)) \lor
    (\exists j > i, \beta < \alpha \st \lvar{ls}(j) \in l(\beta))
    \label{int:liveC-key-liveness-ind-worst-case}
  \end{align}
  To show this, it is sufficient to show
  that $\neg (\exists j > i, \beta < \alpha \st \lvar{ls}(j) \in l(\beta))$
  implies $\forall j \ge i \st \lvar{ls}(j) \in \lvar{lf}(\alpha)$.
We proceed to prove $\forall j \ge i \st \lvar{ls}(j) \in \lvar{lf}(\alpha)$
  by induction on $j \ge i$. The base case holds by our assumptions.
  Now for the inductive case, assume that for $j \ge i$, $\lvar{ls}(j) \in \lvar{lf}(\alpha)$
  holds. From (\ref{ass:liveC-key-liveness-pq-step}), $\lvar{ls}(j + 1) \in {\relyAt[\actxt]}(\lvar{lf}(\alpha))$
  holds and from (\ref{ass:liveC-key-liveness-pq-envlive}), setting $\alpha_1 = \alpha$,
  $\relyAt[\actxt](\lvar{lf}(\alpha)) \subseteq \lvar{lf}(\alpha) \union \Union_{\beta < \alpha} l(\beta) \union t$,
  therefore, either $\lvar{ls}(j + 1) \in \lvar{lf}(\alpha)$, $\lvar{ls}(j+1) \in \Union_{\beta < \alpha} l(\beta)$ or
  $\lvar{ls}(j + 1) \in t$ hold. In the case where $\lvar{ls}(j + 1) \in t$ holds, from (\ref{ass:liveC-key-liveness-pq-T-imp-X'}),
  $\exists v \in X' \st \wtrace(j + 1) = ((\wtv[4], v), \_)$, a contradiction
  and in the case where $\lvar{ls}(j+1) \in \Union_{\beta < \alpha} l(\beta)$ holds,
  we reach a contradiction with $\neg (\exists j > i, \beta < \alpha \st \lvar{ls}(j) \in l(\beta))$
  which implies $\forall j > i, \beta < \alpha \st \lvar{ls}(j) \not\in l(\beta)$.
  Therefore, $\lvar{ls}(j + 1) \in \lvar{lf}(0)$ holds, as required, completing the proof by induction.

  The inductive case then follows from (\ref{int:liveC-key-liveness-ind-worst-case}).
  The goal follows similarly to the base case for the first disjunct
  and by inductive assumption in the second.
  \qedhere
  \end{proofscript}
\end{proof}

\begin{lemma}
  \label{lemma:liveC-key-liveness}
  Given $M \in \Ord \to \UCWorld[\actxt], T \in \Assrt, n \le m, k$
  such that
  \begin{align}
    & \SENVLIVE n; \lvl; \actxt |= L -\smash{M}->> T
    \label{ass:liveC-key-liveness-envlive} \\
    & \forall x \in X \st \VALID \lvl; \actxt |= P(x) * T \implies x \in X'
    \label{ass:liveC-key-liveness-T-imp-X'}
  \end{align}
  hold. Take $(\store_0, h_0) \trace \in \Trace$ and let
  \begin{align*}
\forall v\in \AVal\st
    \at{p}(v) &= \WorldSem{\actxt}{}{v\in X \land \at{P}(v)} \\
    l &= \WorldSem{\actxt}{\store_0}{L} \\
    t &= \WorldSem{\actxt}{\store_0}{T * \True}
  \end{align*}
  Taking arbitrary $\na{p}', \en{p} \in \View[\actxt], \straces \subseteq \STrace$ and $v_0 \in X$ such that
  \begin{align}
    & h_0 \in \sem{\na{p}' * l * \at{p}(v_0) * \True * \en{p}}_\lvl
    \label{ass:liveC-key-liveness-init} \\
    & \accept[\spec_{\twoheadrightarrow}]{(\store_0, h_0) \trace}{\na{p}'}{\at{p}}{v_0} : \straces
    \label{ass:liveC-key-liveness-traces-ant}\\
    & \accept[\spec]{(\store_0, h_0) \trace}{\na{p}' * l}{\at{p}}{v_0} : \unframe{l, \emp}{\straces}
    \label{ass:liveC-key-liveness-traces-cons}
  \end{align}
  and arbitrary $\wtrace' \in \wtr{\unframe{l, \emp}{\straces}}{}$ such that
  \begin{align}
    & \goodenv_{\spec}(\wtrace')
    \label{ass:liveC-key-liveness-liveEnv}
  \end{align}
  holds, then, there exists $\wtrace \in \wtr{\straces}{}$, such that:
  \[
  \goodenv_{\spec_{\twoheadrightarrow}}(\wtrace)
  \]
  
\end{lemma}

\begin{proof}
  This lemma follows straightforwardly from lemma \ref{lemma:liveC-key-liveness-pq}.
\end{proof}

\begin{theorem}
  Taking $n \in \Layer, T \in \Assrt$ and $M \in \Ord \to \Assrt$ such that $m \laygeq n, k \laygeqq n$ and
  \begin{align}
    & \SENVLIVE n; \lvl; \actxt |= L -\smash{M}->> T
    \label{ass:liveC-envlive} \\
    & \forall x \in X \st \VALID \lvl; \actxt |= \at{P}(x) * T \implies x \in X'
    \label{ass:liveC-t-pq}
  \end{align}
  Then, for any $\funspec\in\FunSpec$ and $\cmd \in \Cmd$, if
  \begin{align}
    & \progvars(L) \inters \modvars(\cmd) = \emptyset
    \label{ass:liveC-progvars} \\
    & \JUDGE[\funspec] |= \cmd : \spec_{\twoheadrightarrow}
  \end{align}
  then
  \[
  \JUDGE[\funspec] |= \cmd : \spec
  \]
\end{theorem}

\begin{proof}
  Taking $n \in \Layer, T \in \Assrt$ and $M \in \Ord \to \Assrt$ arbitrary such that $m \laygeq n, k \laygeqq n$,
  (\ref{ass:liveC-envlive}) and (\ref{ass:liveC-t-pq}) hold. Then, to start off, given (\ref{ass:liveC-envlive}),
  $\STABLE \actxt |= { L }$ holds, and therefore, $\na{P} \ssep L \in  \Stable[\actxt]$.
  From this we can infer, $\spec \in \Spec$.

  Then, taking $\cmd \in \Cmd$ arbitrary such that (\ref{ass:frame_progvars}) holds,
  $\funspec \in \FunSpec$ arbitrary such that $\JUDGE[\funspec] |= \cmd : \spec_{\twoheadrightarrow}$ holds
  and arbitrary $\functxt \in \FunCtxt$ such that $\JUDGE |= \functxt : \funspec$ holds,
  then $\sem{\cmd}_{\functxt} \subseteq \Sem{\spec_{\twoheadrightarrow}}$.
  From \cref{lemma:mod_vars} and (\ref{ass:liveC-progvars}),
  we can also infer that $\sem{\cmd}_{\functxt} \subseteq \Trace_{\progvars(L)}$
  and therefore, it is clear that
  $\sem{\cmd}_{\functxt} \subseteq \Trace_{\progvars(L)} \inters \Sem{\spec_{\twoheadrightarrow}}$.
  From this, we know that it is sufficient to show that
  $\Trace_{\progvars(\na{R})} \inters \Sem{\spec_{\twoheadrightarrow}} \subseteq \Sem{\spec}$,
  to show that $\sem{\cmd}_{\functxt} \subseteq \Sem{\spec}$ holds,
  and therefore, $\JUDGE[\funspec] |= \cmd : \spec$, as required.

  Therefore, taking $(\store_0,h_0)\trace \in \Sem{\spec_{\twoheadrightarrow}} \inters \Trace_{\progvars(L)}$ arbitratry,
  it is sufficient to show $(\store_0,h_0)\trace \in \Sem{\spec}$. Let
  \begin{align*}
    \hspace{-1cm}
    \na{p} &= \WorldSem{\actxt}{\store_0}{\na{P}} &
    \at{p}(v) &=
    \begin{cases}
      \WSem[\actxt]{\at{P}(v) \land v \in X} \CASE x \in \AVal \\
      \wEmp[\actxt] \OTHERWISE
    \end{cases} \\
    m(\alpha) &= \WorldSem{\actxt}{\store_0}{M(\alpha)} &
    l(\alpha) &= \WorldSem{\actxt}{\store_0}{L * M(\alpha)} \\
    l &= \WorldSem{\actxt}{\store_0}{L} &
    t &= \WorldSem{\actxt}{\store_0}{T} &
  \end{align*}
  To show $(\store_0,h_0)\trace \in \Sem{\spec}$, assume for some $v_0 \in X$, $h_0 \in \sem{\na{p} \ssep \at{p}(v_0) \ssep l \ssep \True}_\lvl$. Then, given $(\store_0,h_0)\trace \in \Sem{\spec_{\twoheadrightarrow}}$, for some $\straces \in \powerset(\STrace)$:
  \[
  \accept[\spec_{\twoheadrightarrow}]{(\store_0, h_0)\;\trace}{\na{p}}{\at{p}}{v_0} : \straces \land
  \forall \wtrace \in \wtr{\straces}{} \st \goodenv_{\spec_{\twoheadrightarrow}}(\wtrace) \implies \locterm(\wtrace)
  \]
  Given \cref{lemma:frame_safety} and that the definition of the trace safety judgement does not depend
  on the good states of a specification, $X'$, clearly,
  $\accept[\spec]{(\store_0, h_0)\;\trace}{\na{p} * l}{\at{p}}{v_0} : \unframe{L, \emp}{\straces}$ holds.
  To complete the proof, it suffices show that
  \[
  \forall \wtrace \in \wtr{\unframe{L, \emp}{\straces}}{} \st
  \goodenv_{\spec_{\twoheadrightarrow}}(\wtrace) \implies \goodenv_{\spec}(\wtrace)
  \]
This follows straightfowardly from \cref{lemma:liveC-key-liveness}.
\end{proof}
 
\subsection{Soundness of \ref{rule:while}}

\begin{definition}[Concrete trace sequence operator]
  \[
  \trace = \trace_1 \trseq \trace_2 \iff
  \begin{array}{l}
    (\neg \locterm(\trace_1) \land \trace = \trace_1) \lor {} \\
    \left(
    \begin{array}{l}
      \exists \store \in \Store, h \in \Heap, \trace_1'\loc(\store,h)\trace_1'', (\store,h)\trace_2' \in \mathcal{\Trace} \st \\
      \trace_1 = \trace_1'\loc(\store,h)\trace_1'' \land \trace_2 = (\store,h)\trace_2' \land \term((\store,h)\trace_1'') \land \trace = \trace_1'\loc(\store, h)\loc(\store, h)\trace_2'
    \end{array}
    \right)
  \end{array}
  \]
  A similarly defined overloading of this operator exists for specification traces, $\strace_1 \trseq \strace_2$ and the obvious lifting to sets $\straces_1 \trseq \straces_2$.
\end{definition}

\begin{lemma}
  \label{lemma:while_trace_destruct}
  For arbitrary $\functxt \in \FunCtxt$, $(\store_0, h_0)\trace \in \sem{\p{while($\bexp$)\{\cmd\}}}_{\functxt}$, either $\neg \Sem[B]{\bexp}_{\store_0}$ and $(\store_0, h_0)\trace \in \sem{\p{skip}}_{\functxt}$,
  or $\Sem[B]{\bexp}_{\store_0} $ and there exists $(\store_0, h_0)\trace' \in \sem{\cmd}_\functxt$ and $\trace'' \in \sem{\p{while($\bexp$)\{\cmd\}}}_{\functxt}$,
  such that $(\store_0, h_0)\trace = (\store_0, h_0)loc(\store_0, h_0)\trace' \trseq \trace''$.
\end{lemma}
\begin{proof}
  Straightforward by induction on $\step[\functxt]{\_}$.
\end{proof}

\begin{lemma}
  \label{lemma:safety_post}
  Given an arbitrary specification
  \[
  \spec = \SPEC \levl; \actxt |= \A x \in X \eventually X'. <\na{P}|\at{P}(x)> \E y. <\na{Q}(x,y)|\at{Q}(x,y)>
  \]
  for an arbitrary trace $(\store_0,h_0)\trace \in \Trace$, let
  \begin{align*}
    \na{p}    &= \WorldSem{\actxt}{\store_0}{\na{P}}
    &
    \at{p}(v) &= \WorldSem{\actxt}{\store_0}{\at{P}(v)}
  \end{align*}
  If for some $v \in X$ and $\straces \in \mathcal{P}(\STrace)$,
  $
    h_0 \in \sem{\na{p} \ssep \at{p}(v) \ssep \True}
  $ and $
    \accept[\spec]{(\store_0,h_0)\trace}{\na{p}}{\at{p}}{v} : \straces
  $, then:
  \begin{multline}
    \forall \strace \in \straces \st \forall i \in \Nat \st
      \term(\trFrom[\strace]{i}) \implies
        \exists h \in \Heap, \store \in \Store,
        \na{p} \in \View[\actxt], v \in X, v' \in \AVal \st\\
          \strace(i) = (\store, h, \na{p}, \emp, \DONE(v,v')) \land
          \na{p} = \WorldSem{\actxt}{\store}{\na{Q}(v, v')} \land
          h \in \sem{\na{p} * \True}_\lvl
  \end{multline}
\end{lemma}

\begin{proof}
    Straightforward by induction on the specification semantics rules.
\end{proof}

For the rest of the section, let
\begin{gather*}
  \spec'(\beta, b)  = \HSPEC m; \lvl; \actxt |= {P(\beta) *  (b \dotimplies T(\beta)) \land \bexp} {\exists \gamma \st P(\gamma) \land \gamma \leq \beta * (b \dotimplies \gamma < \beta)}\\
  \spec(\beta_0)   = \HSPEC m; \lvl; \actxt |= {P(\beta_0) * L} {\exists \beta \st  P(\beta) * L \land \neg \bexp \land \beta_0 \ge \beta}
\end{gather*}

\begin{lemma}
  \label{lemma:while_safety}
  Take $\functxt \in \FunCtxt$ and $\beta_0 \in \Ord$ arbitrary and take $(\store_0,h_0)\trace \in \sem{\p{while($\bexp$)\{\cmd\}}}_{\functxt}$ such that $\Sem[B]{\bexp}_{\store_0}$. Let
  \begin{align*}
    p'(\beta, b) &= \WorldSem{\actxt}{\store_0}{P(\beta) *  (b \dotimplies T)} &
    l            &= \WorldSem{\actxt}{\store_0}{L}
  \end{align*}

  As $\Sem[B]{\bexp}_{\store_0}$, by lemma \ref{lemma:while_trace_destruct}, there exists $(\store_0, h_0)\trace' \in \sem{\cmd}_\functxt$ and $(\store_1,h_1)\trace'' \in \sem{\p{while($\bexp$)\{\cmd\}}}_{\functxt}$, such that $(\store_0, h_0)\trace = (\store_0, h_0)\trace' \trseq (\store_1,h_1)\trace''$. If, for arbitrary $\beta' \le \beta \le \beta_0$ and $\straces', \straces'' \in \mathcal{P}(\STrace)$, there exists $b \in \Bool$, such that:
  \begin{align*}
    & h_0 \in \sem{p'(\beta, b) * l * \True}_\lvl \\
    & \accept[\spec'(\beta, b)]{(\store_0, h_0)\trace'}{p'(\beta,b)}{\emp}{1} : \straces' \\
    & \accept[\spec(\beta')]{(\store_1,h_1)\trace''}{p'(\beta', \False) * l}{\emp}{1} : \straces'' \end{align*}
  then:
  \[
  \accept[\spec(\beta)]{(\store_0,h_0)\trace}{p'(\beta, \False) * l}{\emp}{1} : \straces' \trseq \straces''
  \]
  and one of the following hold:
  \begin{gather*}
    \locterm((\store_0, h_0)\trace) \\
    \forall \wtrace \in \wtr{\straces' \trseq \straces''}{} \st \neg \goodenv_{\spec(\beta)}(\wtrace) \\
    \forall \strace \in \straces' \trseq \straces'' \st \forall i \in \Nat \st \exists j \ge i, \beta \st \strace(j) = (\store, h, p'(\beta,\False), \emp, 1)
  \end{gather*}
  
\end{lemma}

\begin{proof}
  This lemma is proven by coinduction on the structure of $(\store_0,h_0)\trace$. First, assume:
  \begin{align}
    & h_0 \in \sem{p'(\beta, b) * l * \True}_\lvl \label{eq:while_pre}\\
    & \accept[\spec'(\beta, b)]{(\store_0, h_0)\trace'}{p'(\beta,b)}{\emp}{1} : \straces' \label{eq:while_first_it}\\
    & \accept[\spec(\beta')]{(\store_1,h_1)\trace''}{p'(\beta', \False) * l}{\emp}{1} : \straces'' \label{eq:while_cont} \end{align}

  As, clearly, $\forall \beta \st p'(\beta, \True) \subseteq p'(\beta, \False)$,
  using (\ref{eq:while_pre}), (\ref{eq:while_first_it}), (\ref{eq:while_cont}),
  \cref{lemma:frame_safety} and \cref{lemma:safety_post}, by coinduction,
  we can derive:
  \[
  \accept[\spec(\beta)]{(\store_0,h_0)\trace}{p'(\beta, \False) * l}{\emp}{1} : \straces' \trseq \straces''
  \]

  Now, split on $\locterm((\store_0,h_0)\trace)$. If $\locterm((\store_0,h_0)\trace)$, then the goal holds, otherwise, split again on $\locterm((\store_0,h_0)\trace')$. If $\neg \locterm((\store_0,h_0)\trace')$, then $\straces' \trseq \straces'' = \straces'$, so from (\ref{eq:while_first_it}), $\forall \wtrace \in \wtr{\straces' \trseq \straces''}{} \st \goodenv_{\spec'(\beta, b)}(\wtrace) \implies \locterm(\wtrace)$, from this we infer that
  $\forall \wtrace \in \wtr{\straces' \trseq \straces''}{} \st \neg \goodenv_{\spec'(\beta, b)}(\wtrace)$. Given that the definition of $\goodenv$ only reference the pseudo-quantifier,
  context layer and atomicity context of the parametrising specification, this clearly implies our goal, $\forall \wtrace \in \wtr{\straces' \trseq \straces''}{} \st \neg \goodenv_{\spec(\beta)}(\wtrace)$, as required.
  Otherwise, $\neg \locterm((\store_1,h_1)\trace'')$. To not terminate, the while loop must iterate at least one more time, as $(\store_1,h_1)\trace''$ is a fair trace, therefore $\Sem[B]{\bexp}_{\store_1}$ holds. We can then use lemma \ref{lemma:while_trace_destruct} and our coinductive assumption to obtain
  $
  h_1 \in \sem{p'(\beta, b) * l * \True}_\lvl
  $
  and that one of the following holds:
  \begin{gather*}
    \forall \wtrace \in \wtr{\straces''}{} \st \neg \goodenv_{\spec(\beta')}(\wtrace) \\
    \forall \strace \in \straces'' \st \forall i \in \Nat \st \exists j \ge i, \beta \st \strace(j) = (\store, h, p'(\beta,\False), \emp, 1)
  \end{gather*}
  If the first holds, then $\forall \strace \in \straces' \trseq \straces'' \st \goodenv_{\spec(\beta)}(\strace) \implies \locterm(\strace)$, so the goal is proven; if the second holds, then
  from $h_1 \in \sem{p'(\beta, b) * l * \True}_\lvl$:
  \[
  \forall \strace \in \straces' \trseq \straces'' \st \forall i \in \Nat \st \exists j \ge i, \beta \st \strace(j) = (\store, h, p'(\beta,\False), \emp, 1)
  \qedhere
  \]
\end{proof}

\begin{theorem}
  Given
  \begin{align}
    & \forall \beta \le \beta_0 \st \forall b \in \set{0,1}\st \JUDGE[\funspec] |= \cmd : \spec'(\beta, b) \label{ass:while-body}\\
    & \forall \beta \le \beta_0 \st \SENVLIVE m(\beta); \lvl; \actxt |= L - M ->> T(\beta) \label{ass:while-envlive}\\
    & \forall \alpha\st \STABLE \actxt |= {\exists \alpha'\st L * M(\alpha') \land \alpha' \leq \alpha} \label{ass:while-stable-dec}\\
    & \STABLE \actxt |= L \\
    & \forall \beta \le \beta_0 \st \VALID \actxt |= \minLay{P(\beta)}{m(\beta)} \layleq m \\
    & \progvars(T,L,M) \cap \modvars(\cmd) = \emptyset
  \end{align}
  then:
  \[
  \JUDGE[\funspec] |= {\acode{while(BEXP)\{CMD\}}} : \spec(\beta_0)
  \]
\end{theorem}

\begin{proof}
  Taking $\functxt \in \FunCtxt$ arbitrary such that $\JUDGE |= \functxt : \funspec$ and
  $(\store_0,h_0)\trace \in \sem{\p{while($\bexp$)\{\cmd\}}}_{\functxt}$ arbitary.
  We need to show $(\store_0,h_0)\trace \in \sem{\spec(\beta_0)}$. Let
  \begin{align*}
    p'(\beta, b) &= \WorldSem{\actxt}{\store_0}{P(\beta) *  (b \dotimplies T)} &
    l            &= \WorldSem{\actxt}{\store_0}{L}
  \end{align*}
  To reach the goal, assume $h_0 \in \sem{p'(\beta_0, \False) * l * \True}_\lvl$. By lemma \ref{lemma:while_safety}, in the case that $\Sem[B]{\bexp}_{\store_0}$, and our assumptions, there exists $\straces \subseteq \STrace$:
  \[
  \accept[\spec]{(\store_0,h_0)\trace}{\WorldSem{\actxt}{\store_0}{P(\beta_0) * L}}{\emp}{1} : \straces
  \]
    and one of the following holds:
  \begin{gather*}
    \locterm((\store_0, h_0)\trace) \\
    \forall \strace \in \straces \st \neg \goodenv_\spec(\strace) \\
    \forall \strace \in \straces \st \forall i \in \Nat \st \exists j \ge i, \beta \st \strace(j) = (\store, h, p'(\beta,\False), \emp, 1)
  \end{gather*}
  In the first case, $\forall \strace \in \straces \st \locterm(\strace)$, therefore, $\forall \strace \in \straces \st \goodenv_{\spec(\beta_0)}(\strace) \implies \locterm(\strace)$, as required. In the second, $\forall \strace \in \straces \st \goodenv_{\spec(\beta_0)}(\strace) \implies \locterm(\strace)$ clearly also holds. Finally, we consider the third case. Take $\strace \in \straces$ arbitrary and asume $\goodenv_{\spec(\beta_0)}(\strace)$. Now, for a contradiction, assume $\neg \locterm(\strace)$. In this case, due to (\ref{ass:while-envlive}), with an argument similar to that in the soundness of (\ref{rule:liveness-check}), at every point, every $\strace \in \straces$ eventually reaches a state satisfying $T(\beta_0)$. This must eventually be stable due to the metric stabily decreasing due to assumption (\ref{ass:while-stable-dec}), holding till the next iteration, at which point, the loop variant decreases due to (\ref{ass:while-body}) with $b = \True$. By repeating this argument with the continuation, by well-foundness of ordinals, the while loop must eventually terminate if $\goodenv(\strace)$ holds, leading to a contradiction. Therefore, in all cases, $\forall \strace \in \straces \st \goodenv_{\spec(\beta_0)}(\strace) \implies \locterm(\strace)$ holds, as required, concluding the proof.
\end{proof}

\subsection{Soundness of \ref{rule:parallel}}

\def\pe{p_{\mathsf{e}}}

\newcommand{\sbowtie}{\overset{\raisebox{-3pt}[0pt]{\scriptsize\textsf{s}}}{\bowtie}}
\newcommand{\wbowtie}{\overset{\raisebox{-3pt}[0pt]{\scriptsize\textsf{w}}}{\bowtie}}

\begin{definition}[Bowtie operator]
The bowtie operator, $\bowtie$, which interleaves the subjective traces of two commands executed in parallel into a command from their combined perspective:
\begin{align*}
    (\store, h)\env\trace_1' \bowtie (\store,h)\env\trace_2' &= (\store,h)\env(\trace_1' \bowtie \trace_2') \\
    (\store, h)\env\trace_1' \bowtie (\store,h)\loc\trace_2' &= (\store,h)\loc(\trace_1' \bowtie \trace_2') \\
    (\store, h)\loc\trace_1' \bowtie (\store,h)\env\trace_2' &= (\store,h)\loc(\trace_1' \bowtie \trace_2')
  \end{align*}
  All other cases are undefined.
\end{definition}

\begin{definition}[Specification Bowtie operator]
The specification bowtie operator, $\sbowtie$, which interleaves the subjective specification traces of two commands executed in parallel into a command from their combined perspective:
\begin{align*}
    (\store, h, p_1, \emp, 1)\env\trace_1' \sbowtie (\store, h, p_2, \emp, 1)\env\trace_2' &= (\store, h, p_1 \ssep p_2, \emp, 1)\env(\trace_1' \sbowtie \trace_2') \\
    (\store, h, p_1, \emp, 1)\env\trace_1' \sbowtie (\store, h, p_2, \emp, 1)\loc\trace_2' &= (\store, h, p_1 \ssep p_2, \emp, 1)\loc(\trace_1' \sbowtie \trace_2') \\
    (\store, h, p_1, \emp, 1)\loc\trace_1' \sbowtie (\store, h, p_2, \emp, 1)\env\trace_2' &= (\store, h, p_1 \ssep p_2, \emp, 1)\loc(\trace_1' \sbowtie \trace_2')
  \end{align*}
  All other cases are undefined.
\end{definition}

\begin{lemma}
  \label{lemma:bowtie-comb}
  For any $\functxt \in \FunCtxt$:
  \[
    \forall \trace \in \sem{\cmd_1 || \cmd_2}_\functxt \st
      \exists \trace_1 \in \sem{\cmd_1}_\functxt,
              \trace_2 \in \sem{\cmd_2}_\functxt \st
        \trace = \trace_1 \bowtie \trace_2
  \]\unskip
\end{lemma}
\begin{proof}
  Straightforward by induction on $\step[\functxt]{\_}$.
\end{proof}

\begin{lemma}
\label{lemma:mods-const-tr}
For any trace $(\store_0,h_0) \;\trace\; (\store_1,h_1) \;\trace' \in \sem{\cmd}_{\functxt}$, we have
$\forall \pvar{x} \in \PVar \setminus \mods(\cmd)\st \store_0(\pvar{x}) = \store_1(\pvar{x})$.
\end{lemma}

\begin{proof}
  Straightforward by induction on the length of the trace.
\end{proof}

For the rest of the section, we name the specifications involved in the
\ref{rule:parallel} rule as follows:
\begin{align*}
  \spec_1  &= \HSPEC m_1;\lvl;\actxt |= {P_1}  {Q_1} &
  \spec_2  &= \HSPEC m_2;\lvl;\actxt |= {P_2}  {Q_2} &
  \spec    &= \HSPEC m;\lvl;\actxt |= {P_1 \ssep P_2}  {Q_1 \ssep Q_2}
\end{align*}

\begin{lemma}
  \label{lemma:parallel-safety}
  
  For arbitrary $(\store_0,h_0)\trace, (\store_0,h_0)\trace_1, (\store_0,h_0)\trace_2 \in \Trace$, $\straces_1, \straces_2 \in \mathcal{P}(\STrace)$, $v_1, v_2 \in \set{1, \DONE(1,1)}$, and, $p_1', p_2' \in \View_\actxt$ , then:
  \[
  \left.
  \begin{aligned}
    (\store_0, h_0)\trace = (\store_0, h_0)\trace_1 \bowtie (\store_0, h_0)\trace_2 \\
    \accept[\spec_1]{(\store_0,h_0)\trace_1}{p_1'}{\emp}{v_1} : \straces_1 \\
    \accept[\spec_2]{(\store_0,h_0)\trace_2}{p_2'}{\emp}{v_2} : \straces_2 \\
    h_0 \in \sem{p_1' \ssep p_2' \ssep \True}_\lvl \\
    \term((\store_0,h_0)\trace_1) \implies p_1' = \WorldSem{\actxt}{\store}{Q_1}\\
    \term((\store_0,h_0)\trace_2) \implies p_2' = \WorldSem{\actxt}{\store}{Q_2}
  \end{aligned} \right\}
  \implies
  \begin{aligned}
    \exists
      \straces &\in \mathcal{P}(\STrace),
      v \in \set{1, \DONE(1,1)} \st
    \\
    &\quad\accept{(\store, h)\trace}{p_1' \ssep p_2'}{\emp}{v} : \straces \land {} \\
    &\quad
    \forall \strace \in \straces \st \exists \strace_1 \in \straces_1, \strace_2 \in \straces_2 \st \strace = \strace_1 \sbowtie \strace_2 \land {} \\
    &\quad
      (v_1 = \DONE(1,1) \land v_2 = \DONE(1,1)) \iff v = \DONE(1,1)
  \end{aligned}
  \]
\end{lemma}

\begin{proof}
  This lemma is proven by coinduction on the structure of $(\store_0, h_0)\trace$.

  The trace either starts with a local, or an environment step. We split on the two cases:
    \begin{proofscript}
    \item[Case~$(\store,h)\trace = (\store,h)\env(\store,h')\trace'$:] Take $(\store_0,h_0)\trace_1, (\store_0,h_0)\trace_2 \in \Trace$, $\straces_1, \straces_2 \in \mathcal{P}(\STrace)$, $v_1, v_2 \in \set{1, \DONE(1,1)}$, and, $p_1', p_2' \in \View_\actxt$ arbitrary, and assume:
      \begin{align}
        & (\store_0, h_0)\trace = (\store_0, h_0)\trace_1 \bowtie (\store_0, h_0)\trace_2  \label{eq:par_env_bowtie_eq} \\
        & \accept[\spec_1]{(\store_0,h_0)\trace_1}{p_1'}{\emp}{v_1} : \straces_1 \label{eq:par_env_t_1} \\
        & \accept[\spec_2]{(\store_0,h_0)\trace_2}{p_2'}{\emp}{v_2} : \straces_2 \label{eq:par_env_t_2} \\
        & h_0 \in \sem{p_1' \ssep p_2' \ssep \True}_\lvl \label{eq:par_env_sat} \\
        & \term((\store_0,h_0)\trace_1) \implies p_1' = \WorldSem{\actxt}{\store}{Q_1} \label{eq:par_env_done_t_1} \\
        & \term((\store_0,h_0)\trace_2) \implies p_2' = \WorldSem{\actxt}{\store}{Q_2} \label{eq:par_env_done_t_2}
      \end{align}
      Given \eqref{eq:par_env_bowtie_eq} and the definition of $\bowtie$:
      \begin{align*}
        & (\store_0, h_0)\trace_1 = (\store_0, h_0)\env(\store_0,h')\trace_1' \\
        & (\store_0, h_0)\trace_2 = (\store_0, h_0)\env(\store_0,h')\trace_2' \\
        & (\store_0, h')\trace' = (\store_0,h')\trace_1' \bowtie (\store_0,h')\trace_2'
      \end{align*}
      Now to prove the goal, consider the case $v_1 = \DONE(1,1)$ and $v_2 = \DONE(1,1)$. In this case, take $v = \DONE(1,1)$, so \ref{rule:env2} must hold for the goal as well as \eqref{eq:par_env_t_1} and \eqref{eq:par_env_t_2}.
      Note that this choice of $v$ immediately satisfies the third conjunct of the goal. To show \ref{rule:env2} holds for the goal, given some $\pe, \pe' \in \View_\actxt$, assume:
      \[
      h_0 \in \sem{ p_1' \ssep p_2' \ssep \pe } \land \update h_0 -> h' |= p_1' \ssep p_2' \ssep \pe -> p_1' \ssep p_2' \ssep \pe'
      \]
      By substitution, this implies both:
      \begin{align*}
        \exists \pe,\pe' \st h_0 \in \sem{ p_1' \ssep \pe } \land \update h -> h' |= p_1' \ssep \pe -> p_1' \ssep \pe' \\
        \exists \pe,\pe' \st h_0 \in \sem{ p_2' \ssep \pe } \land \update h -> h' |= p_2' \ssep \pe -> p_2' \ssep \pe'
      \end{align*}
      Given \ref{eq:par_env_t_1} and \ref{eq:par_env_t_2}, these imply:
      \begin{align*}
        & \accept[\spec_1]{(\store_0,h')\trace_1'}{p_1'}{\emp}{v_1} : \straces_1' & \straces_1 = (\store_0,h_0, p_1,\emp,\DONE(1,1)) \env \straces'_1 \\
        & \accept[\spec_2]{(\store_0,h')\trace_2}{p_2'}{\emp}{v_2} : \straces_2' &
        \straces_2 = (\store_0,h_0, p_2,\emp,\DONE(1,1)) \env \straces'_2
      \end{align*}
      Assumption~\eqref{eq:par_env_sat} and $\update h_0 -> h' |= p_1' \ssep p_2' \ssep \pe -> p_1' \ssep p_2' \ssep \pe'$ yield:
      \begin{align}
      h' \in \sem{p_1' \ssep p_2' \ssep \True}_\lvl \label{eq:par_sat_ind}
      \end{align}
      Now, by using the inductive assumption, as \eqref{eq:par_env_done_t_1} and \eqref{eq:par_env_done_t_2} clearly imply the same assertions for $(\store_0,h')\trace_1'$ and $(\store_0,h')\trace_2'$ respectively, for some $\straces' \in \mathcal{P}(\STrace)$:
      \begin{align}
        & \accept{(\store_0, h')\trace'}{p_1' \ssep p_2'}{\emp}{v} : \straces' \\
        & \forall \strace \in \straces' \st \exists \strace_1 \in \straces_1', \strace_2 \in \straces_2' \st \strace = \strace_1 \sbowtie \strace_2
      \end{align}
      From this first consequence:
      \[
      \accept[\spec]{(\store_0,h)\trace}{p_1' \ssep p_2'}{\emp}{v} : \straces
      \]
      holds, where $\straces = (\store_0, h, p_1' \ssep p_2', \emp, v)\env\straces'$. This is the first conjunct of the goal.

      Finally, taking $\strace \in \straces$ arbitrary, there exists $\strace' \in \straces'$ such that $\strace = (\store_0, h, p_1' \ssep p_2', \emp, v)\env\strace'$. From the second consequence of our inductive assumption, it follows that there exist $\strace_1' \in \straces_1'$ and $\strace_2' \in \straces_2'$ such that $\strace' = \strace_1' \sbowtie \strace_2'$. Then, from the definitions of $\straces_1$ and $\straces_2$, $(\store_0,h_0, p_1,\emp,\DONE(1,1))\env\strace_1' \in \straces_1$ and $(\store_0,h_0, p_2,\emp,\DONE(1,1))\env\strace_2' \in \straces_2$ hold, and $\strace = (\store_0,h_0, p_1,\emp,\DONE(1,1))\env\strace_1' \sbowtie (\store_0,h_0, p_2,\emp,\DONE(1,1))\env\strace_2' \in \straces_2$ holds as required.
      
      Other cases for $v_1$, $v_2$ follow similarly.
    \item[Case~$(\store,h)\trace = (\store,h)\loc(\store,h')\trace'$:]
      Here the variable store does not change as $\mods(\cmd_1 || \cmd_2) = \emptyset$, due to \cref{lemma:mods-const-tr} and the syntactic restriction on parallel commands, requiring both threads to not modify the value of any variable.
      To prove the goal, take $(\store_0,h_0)\trace_1, (\store_0,h_0)\trace_2 \in \Trace$, $\straces_1, \straces_2 \in \mathcal{P}(\STrace)$, $v_1, v_2 \in \set{1, \DONE(1,1)}$, and, $p_1', p_2' \in \View_\actxt$ arbitrary, and assume:
      \begin{align}
        & (\store_0, h_0)\trace = (\store_0, h_0)\trace_1 \bowtie (\store_0, h_0)\trace_2  \label{eq:par_loc_bowtie_eq} \\
        & \accept[\spec_1]{(\store_0,h_0)\trace_1}{p_1'}{\emp}{v_1} : \straces_1 \label{eq:par_loc_t_1} \\
        & \accept[\spec_2]{(\store_0,h_0)\trace_2}{p_2'}{\emp}{v_2} : \straces_2 \label{eq:par_loc_t_2} \\
        & h_0 \in \sem{p_1' \ssep p_2' \ssep \True}_\lvl \label{eq:par_loc_sat}\\
        & \term((\store_0,h_0)\trace_1) \implies p_1' = \WorldSem{\actxt}{\store}{Q_1} \label{eq:par_loc_done_t_1} \\
        & \term((\store_0,h_0)\trace_2) \implies p_2' = \WorldSem{\actxt}{\store}{Q_2} \label{eq:par_loc_done_t_2}
      \end{align}
      Given \eqref{eq:par_loc_bowtie_eq} and the definition of $\bowtie$, either:
      \begin{align*}
        & (\store_0, h_0)\trace_1 = (\store_0, h_0)\loc(\store_0,h')\trace_1' \\
        & (\store_0, h_0)\trace_2 = (\store_0, h_0)\env(\store_0,h')\trace_2'
      \end{align*}
      or:
      \begin{align*}
        & (\store_0, h_0)\trace_1 = (\store_0, h_0)\env(\store_0,h')\trace_1' \\
        & (\store_0, h_0)\trace_2 = (\store_0, h_0)\loc(\store_0,h')\trace_2'
      \end{align*}
      and in both cases:
      \[
      (\store_0, h')\trace' = (\store_0,h')\trace_1' \bowtie (\store_0,h')\trace_2'
      \]
      Consider the first case, the second will follow symmetrically. Assume that the \ref{rule:stutter} rule holds for $\accept[\spec_1]{(\store,h)\loc(\store, h')\trace_1}{p_1'}{\emp}{v_1} : \straces_1$, then, for some $p_1'' \in \View[\actxt]$:
      \begin{align*}
      & \update h_0 -> h' |= p_1' -> p_1'' \\
      & \accept[\spec_1]{(\store_0, h')\;\trace_1'}{p_1''}{\emp}{v_1} : \straces_1' \\
      & \term((\store_0,h')\trace_1') \implies v_1 = \DONE(1,1) \land p_1'' = \WorldSem{\actxt}{\store_0}{Q_1}
      \end{align*}
      where $\straces_1 = (\store_0,h_0, p_1,\emp,v_1)\loc\straces_1'$.
      Given \ref{eq:par_loc_sat} and $\update h_0 -> h' |= p_1' -> p_1''$, $h' \in \sem{p_1'' \ssep p_2' \ssep \True}_\lvl$ holds. Given $\update h_0 -> h' |= p_1' -> p_1''$,  $\update h_0 -> h' |= p_1' \ssep p_2 -> p_1'' \ssep p_2$, also holds. Using this and \ref{rule:env} or \ref{rule:env2}:
      \[
      \accept[\spec_2]{(\store, h')\;\trace_2'}{p_2'}{\emp}{v_2} : \straces_2'
      \]
      where $\straces_2 = (\store_0,h_0, p_2,\emp,v_2)\loc\straces_2'$.
      Now using the inductive assumption, as, once again, \eqref{eq:par_loc_done_t_1} and \eqref{eq:par_loc_done_t_2} clearly imply the same assertions for
      $(\store_0,h')\trace_1'$ and $(\store_0,h')\trace_2'$ respectively, for some $\straces' \in \mathcal{P}(\STrace)$:
      \begin{align}
        & \accept[\spec]{(\store_0, h')\trace'}{p_1'' \ssep p_2'}{\emp}{v} : \straces' \land {} \\
        & \forall \strace \in \straces' \st \exists \strace_1 \in \straces_1', \strace_2 \in \straces_2' \st \strace = \strace_1 \sbowtie \strace_2 \land {} \\
        & (v_1 = \DONE(1,1) \land v_2 = \DONE(1,1)) \iff v = \DONE(1,1) \label{eq:par_loc_lin}
      \end{align}
      The second and third consequents imply the equivalent conjuncts of the goal with the same method as in the $\env$ case and directly respectively.
      As we have shown $\update h -> h' |= p_1' * p_2 -> p_1'' * p_2$ holds, using the \ref{rule:stutter} rule, to show that $\accept[\spec]{(\store_0, h_0)\trace}{p_1' \ssep p_2'}{\emp}{v} : \straces$ holds, where $\straces = (\store_0,h_0, p_1' \ssep p_2',\emp,v)\loc\straces'$, it suffices to show:
      \[
      \term((\store_0,h')\trace') \implies v = \DONE(1,1) \land p_1'' \ssep p_2' = \WorldSem{\actxt}{\store_0}{Q_1 \ssep Q_2}
      \]
      Assuming $\term((\store_0,h')\trace')$ holds, then $\term((\store_0,h')\trace_1')$ and $\term((\store_0,h')\trace_2')$ hold.
      From this it follows that $v_1, v_2 = \DONE(1,1)$, so, due to \eqref{eq:par_loc_lin}, $v = \DONE(1,1)$.

      Finally, due to $\term((\store_0,h')\trace_1')$ and $\term((\store_0,h')\trace_2')$, $p_1'' = \WorldSem{\actxt}{\store_0}{Q_1}$ and $p_2' = \WorldSem{\actxt}{\store}{Q_2}$ hold respectively, yielding $p_1'' \ssep p_2' = \WorldSem{\actxt}{\store_0}{Q_1 \ssep Q_2}$, as required.
      
      The \ref{rule:linpt} rule follows similarly.
      \qedhere
    \end{proofscript}
\end{proof}

\begin{theorem}

  Given
  \begin{align}
    \TRIPLE[\funspec] m_1; \levl; \actxt |= {P_1}{\cmd_1}{Q_1} \label{ass:parallel_1} \\
    \TRIPLE[\funspec] m_2; \levl; \actxt |= {P_2}{\cmd_2}{Q_2} \label{ass:parallel_2} \\
    \minLayer[\levl;\acontext]{Q_1}{m_2} \layleq m \\
    \minLayer[\levl;\acontext]{Q_2}{m_1} \layleq m
  \end{align}
  then:
  \[
  \TRIPLE[\funspec] m; \levl; \actxt |= {P_1 \ssep P_2}{\cmd_1 || \cmd_2}{Q_1 \ssep Q_2}
  \]
  
\end{theorem}
\begin{proof}
  Taking $\functxt \in \FunCtxt$ arbitrary such that $\JUDGE |= \functxt : \funspec$, from \eqref{ass:parallel_1} and \eqref{ass:parallel_2}, $\sem{\cmd_1}_{\functxt} \subseteq \sem{\spec_1}$ and $\sem{\cmd_2}_{\functxt} \subseteq \sem{\spec_2}$ hold. Given an arbitrary $(\store_0, h_0)\trace \in \sem{\cmd_1 || \cmd_2}_\functxt$, need to show $(\store_0, h_0)\trace \in \sem{\spec}$. Let
  \begin{align*}
    p_1 &= \WorldSem{\actxt}{\store_0}{P_1} &
    p_2 &= \WorldSem{\actxt}{\store_0}{P_2}
  \end{align*}
  To reach the goal, assume $h_0 \in \sem{p_1 \ssep p_2 \ssep \True}_\lvl$. Then $h_0 \in \sem{p_1 \ssep \True}_\lvl$ and $h_0 \in \sem{p_2 \ssep \True}_\lvl$ hold. From \ref{lemma:bowtie-comb} and the definition of $\bowtie$, there exists $(\store_0, h_0)\trace_1 \in \sem{\cmd_1}_\functxt$ and $(\store_0, h_0)\trace_2 \in \sem{\cmd_2}_\functxt$ such that $(\store_0, h_0)\trace = (\store_0, h_0)\trace_1 \bowtie (\store_0, h_0)\trace_2$. As $\sem{\cmd_1}_{\functxt} \subseteq \sem{\spec_1}$ and $\sem{\cmd_2}_{\functxt} \subseteq \sem{\spec_2}$, $(\store_0, h_0)\trace_1 \in \sem{\spec_1}$ and $(\store_0, h_0)\trace_2 \in \sem{\spec_2}$ hold. Now, as $h_0 \in \sem{p_1 \ssep \True}_\lvl$ and $h_0 \in \sem{p_2 \ssep \True}_\lvl$, then for some $\straces_1, \straces_2 \in \mathcal{P}(\STrace)$:
  \begin{align*}
    & \accept{(\store_0, h_0)\trace_1}{p_1}{\emp}{1} : \straces_1 &
    \forall \strace_1 \in \wtr{\straces_1}{} \st & \goodenv_\spec(\strace_1) \implies \locterm(\strace_1)
    \\
    & \accept{(\store_0, h_0)\trace_2}{p_2}{\emp}{1} : \straces_2 &
    \forall \strace_2 \in \wtr{\straces_2}{} \st & \goodenv_\spec(\strace_2) \implies \locterm(\strace_2)
  \end{align*}
  As all commands must take at least one step, $\neg \term((\store_0, h_0)\trace_1)$ and $\neg \term((\store_0, h_0)\trace_2)$ hold, therefore:
  \begin{align*}
    & \term((\store_0,h_0)\trace_1) \implies p_1 = \WorldSem{\actxt}{\store}{Q_1}\\
    & \term((\store_0,h_0)\trace_2) \implies p_2 = \WorldSem{\actxt}{\store}{Q_2}
  \end{align*}
  hold. 
  Now, using lemma \ref{lemma:parallel-safety}, there exists $\straces \in \mathcal{P}(\STrace)$ such that:
  \begin{mathpar}
    \accept{(\store_0, h_0)\trace}{p_1 \ssep p_2}{\emp}{1} : \straces
  \end{mathpar}
  and for any $\strace \in \straces$, there exist $\strace_1 \in \straces_1$ and $\strace_2 \in \straces_2$, such that $\strace = \strace_1 \sbowtie \strace_2$. It now suffices to show that
  $\forall \wtrace \in \wtr{\straces}{} \st \goodenv_\spec(\wtrace) \implies \locterm(\wtrace)$. Take $\strace\in\straces$ arbitrary and $\strace_1 \in \straces_1$ and $\strace_2 \in \straces_2$ such that $\strace = \strace_1 \sbowtie \strace_2$ and
  $\wtrace \in \wtr{\strace}, \wtrace_1 \in \wtr{\strace_1}{}, \wtrace_2 \in \wtr{\strace_2}{}$.
  From above:
  \begin{align}
    & \goodenv_{\spec_1}(\wtrace_1) \implies \locterm(\wtrace_1)
    \label{ass:liveness_condition_1} \\
    & \goodenv_{\spec_2}(\wtrace_2) \implies \locterm(\wtrace_2)
    \label{ass:liveness_condition_2}
  \end{align}
  holds. To reach the goal, split on $\locterm(\wtrace_1)$ and $\locterm(\wtrace_2)$.
  \begin{proofscript}
  \item[Case~$\locterm(\wtrace_1) \land \locterm(\wtrace_2)$:] In this case, clearly $\locterm(\wtrace)$ holds, therefore $\goodenv_{\spec}(\wtrace) \implies \locterm(\wtrace)$ holds trivially,
    as required.
  \item[Case~$\locterm(\wtrace_1) \land \neg \locterm(\wtrace_2)$:]
    From $\neg \locterm(\wtrace_2)$, by \eqref{ass:liveness_condition_2}, $\neg \goodenv_{\spec_2}(\wtrace_2)$ holds:
    \begin{multline*}
      \exists \pobl \in \POLay{m_2}^{\spec_2}\st
      (\forall O \in \ObLay{\lay(\pobl)}\st \forall i\in \Nat\st \exists j \geq i\st \neg \locheld(O,\trAt[\wtrace_2]{j})) \land {} \\
      (\exists i\in \Nat\st \forall j \geq i\st \envheld(\smash{\pobl},\trAt[\wtrace_2]{j}))
    \end{multline*}
    As $\locterm(\wtrace_1)$, by \cref{lemma:safety_post}, there exists some $i_1 \in \Nat$,
    an index after which the trace $\wtrace_1$ only performs $\env$ steps, in particular,
    for any $j \ge i_1$, $\wtrace_1(j) = (\store, h, \na{w}^1, \at{w}^1, \DONE(1,1))$,
    where $\na{w}^1 \in \WorldSem{\store}{\actxt}{Q_1}$ and $\at{w}^1 \in \wEmp[\actxt]$.
    Therefore $\wtrace(j) = (\store, h, \na{w}^1 \worldSsep \na{w}^2, \at{w}^1 \worldSsep \at{w}^2, \DONE(1,1))$,
    where $\wtrace_2(j) = (\store, h, \na{w}^2, \at{w}^2, \DONE(1,1))$, such that $\at{w}^2 \in \wEmp[\actxt]$.
    Given $\minLayer[\levl;\acontext]{Q_1}{m_2}$, it is clear that:
    \begin{multline*}
      (\forall O \in \ObLay{\lay(\pobl)}\st \forall i\in \Nat\st \exists j \geq i\st \neg \locheld(O,\trAt[\wtrace_2]{j})) \implies {} \\
      (\forall O \in \ObLay{\lay(\pobl)}\st \forall i\in \Nat\st \exists j \geq i\st \neg \locheld(O,\trAt[\wtrace]{j}))
    \end{multline*}
    and similarly as $\pobl \in \POLay{m_2}^{\spec_2}$:
    \[
      (\exists i\in \Nat\st \forall j \geq i\st \envheld(\smash{\pobl},\trAt[\wtrace_2]{j})) \implies (\exists i\in \Nat\st \forall j \geq i\st \envheld(\smash{\pobl},\trAt[\wtrace]{j}))
    \]
    As $m_2 \le m$, $\pobl \in \POLay{m}^{\spec}$. Finally, from this $\neg \goodenv_{\spec}(\wtrace)$ holds, implying $\goodenv_{\spec}(\wtrace) \implies \locterm(\wtrace)$, as required.
  \item[Case~$\neg \locterm(\strace_1) \land \locterm(\strace_2)$:]
    Similarly to the previous case.
  \item[Case~$\neg \locterm(\wtrace_1) \land \neg \locterm(\wtrace_2)$:] Given \eqref{ass:liveness_condition_1} and \eqref{ass:liveness_condition_2},
    we can infer $\neg \goodenv_{\spec_1}(\wtrace_1)$ and $\neg \goodenv_{\spec_2}(\wtrace_2)$.
    Assume $\goodenv_{\spec}(\wtrace)$ for a contradiction.
    From $\neg \goodenv_{\spec_1}(\wtrace_1)$, for some $\pobl \in \POLay{m_1}^{\spec}$:
    \[
      (\forall O \in \ObLay{\lay(\pobl)}\st \forall i\in \Nat\st \exists j \geq i\st \neg \locheld(O,\trAt[\wtrace_1]{j})) \land (\exists i\in \Nat\st \forall j \geq i\st \envheld(\smash{\pobl},\trAt[\wtrace_1]{j}))
    \]
    From this and $\goodenv_{\spec}(\wtrace)$, there is some $i \in \Nat$ such that:
    \[
    \forall j \ge i \st \locheld(\pobl, \wtrace_2(j))
    \]
    From $\neg \goodenv_{\spec_2}(\wtrace_2)$, for some $\pobl' \in \POLay{m_2}^{\spec_2}$:
    \[
      (\forall O \in \ObLay{\lay(\pobl')}\st \forall i\in \Nat\st \exists j \geq i\st \neg \locheld(O,\trAt[\wtrace_2]{j})) \land (\exists i\in \Nat\st \forall j \geq i\st \envheld(\smash{\pobl'},\trAt[\wtrace_2]{j}))
    \]
    Given that $\forall j \ge i \st \locheld(\pobl, \wtrace_2(j))$, for $\forall O \in \ObLay{\lay(\pobl')}\st \forall i\in \Nat\st \exists j \geq i\st \neg \locheld(O,\trAt[\wtrace_2]{j})$ to hold,
    it must be the case that $\lay(\pobl) > \lay(\pobl')$. This argument can be repeated ad-infinitum, which, by the well-foundedness of layers, leads to a contradiction, and
    therefore $\neg \goodenv_{\spec}(\wtrace)$ holds.
    This implies $\goodenv_{\spec}(\wtrace) \implies \locterm(\wtrace)$.
  \end{proofscript}

  From these cases, we deduce that $\forall \wtrace \in \wtr{\straces} \st \goodenv_{\spec}(\wtrace) \implies \locterm(\wtrace)$.
  
  From this, we can infer $(\store_0, h_0)\trace \in \sem{\spec}$ and consequently, $\sem{\cmd_1 || \cmd_2}_\functxt \subseteq \sem{\spec}$, as required.
\end{proof}

\end{document}